\documentclass[runningheads,envcountsect,orivec]{llncs}
\usepackage[dvipsnames]{xcolor}
\usepackage{amsmath}
\usepackage{amssymb}
\usepackage{xspace}
\usepackage{graphicx}
\usepackage{latexsym}
\usepackage{listings}
\usepackage{multirow}
\usepackage{suffix}
\usepackage{url}
\usepackage{mathptmx}
\usepackage{mathrsfs}
\usepackage{comment}
\usepackage{enumerate}
\usepackage{txfonts}
\usepackage{hyperref}
\usepackage{color}      % use if color is used in text

\usepackage{float}
\floatstyle{ruled}
\restylefloat{figure}

\spnewtheorem{THM}[theorem]{Theorem}{\bfseries}{\itshape}
\spnewtheorem{PRO}[theorem]{Proposition}{\bfseries}{\itshape}
\spnewtheorem{COR}[theorem]{Corollary}{\bfseries}{\itshape}
\spnewtheorem{LEM}[theorem]{Lemma}{\bfseries}{\itshape}
\spnewtheorem{CON}[theorem]{Conjecture}{\bfseries}{\rmfamily}
\spnewtheorem{DEF}[theorem]{Definition}{\bfseries}{\rmfamily}
\spnewtheorem{REM}[theorem]{Remark}{\bfseries}{\rmfamily}
\spnewtheorem{EX}[theorem]{Example}{\bfseries}{\rmfamily}

\newcommand{\myrm}{}

\newcommand{\lrangle}[1]{\langle #1 \rangle}
\newcommand{\blrangle}[1]{\big\langle #1 \big\rangle}

\newcommand{\brtext}[1]{[\textrm{\small #1}]}

\newcommand{\ltsrule}[1]{{\footnotesize \lrangle{\textrm{#1}}}}
\newcommand{\eltsrule}[1]{{\footnotesize [\textrm{#1}]}}
\newcommand{\trule}[1]{{\footnotesize\brtext{#1}}}
\newcommand{\orule}[1]{{\scriptsize{\brtext{#1}}}}

\newcommand{\noi}{\noindent}

\newcommand{\myparagraph}[1]{\paragraph{\textbf{#1}}}

\newcommand{\secref}[1]{Section~\ref{#1}}
\newcommand{\defref}[1]{Definition~\ref{#1}}
\newcommand{\figref}[1]{Figure~\ref{#1}}
\newcommand{\thmref}[1]{Theorem~\ref{#1}}

\newcommand{\propref}[1]{Proposition~\ref{#1}}
\newcommand{\appref}[1]{Appendix~\ref{#1}}
\newcommand{\lemref}[1]{Lemma~\ref{#1}}

\newcommand{\bnfis}{\;\;::=\;\;}
\newcommand{\bnfbar}{\;\;\;|\;\;\;}

\newcommand{\tree}[2]{
\ensuremath{\displaystyle
		\frac
		{
			#1
		}{
			#2
		}
	}
}

\newcommand{\mytilde}[1]{\widetilde{#1}}

\newcommand{\set}[1]{\{#1\}}
\newcommand{\es}{\emptyset}

\newcommand{\setbar}{\ \ |\ \ }

\newcommand{\dom}[1]{\mathtt{dom}(#1)}

\newcommand{\freev}[1]{\lrangle{#1}}
\newcommand{\boundv}[1]{(#1)}

\newcommand{\send}[1]{\overline{#1}}

\newcommand{\inact}{\mathbf{0}}

\newcommand{\Par}{\;|\;}
\newcommand{\news}[1]{(\nu\ #1)}
\newcommand{\newsp}[2]{(\nu\ #1)(#2)}
\newcommand{\varp}[1]{#1}
\newcommand{\rvar}[1]{\mathcal{#1}}
\newcommand{\recp}[2]{\mu \varp{#1}. #2}

\newcommand{\Def}{\sessionfont{def}\ }

\newcommand{\defeq}{\stackrel{\Def}{=}}

\newcommand{\repl}{\ast\,}

\newcommand{\bn}[1]{\mathtt{bn}(#1)}
\newcommand{\fn}[1]{\mathtt{fn}(#1)}
\newcommand{\ofn}[1]{\mathsf{ofn}(#1)}
\newcommand{\fv}[1]{\mathtt{fv}(#1)}

\newcommand{\fpv}[1]{\mathtt{fpv}(#1)}

\newcommand{\subj}[1]{\mathtt{subj}(#1)}

\newcommand{\scong}{\equiv}

\newcommand{\wb}{\approx^{H}}

\newcommand{\wbc}{\approx}
\newcommand{\wbf}{\approx^{C}}
\newcommand{\fwb}{\wbf}
\newcommand{\hwb}{\wb}
\newcommand{\WB}{\approx}

\newcommand{\red}{\longrightarrow}
\newcommand{\Red}{\rightarrow\!\!\!\!\!\rightarrow}

\def\subst#1#2{\{\raisebox{.5ex}{\small$#1$}\! / \mbox{\small$#2$}\}}

\newcommand{\hole}{-}
\newcommand{\context}[2]{#1[#2]}

\newcommand{\barb}[1]{\downarrow_{#1}}
\newcommand{\Barb}[1]{\Downarrow_{#1}}
\newcommand{\nbarb}[1]{\not\downarrow_{#1}}
\newcommand{\nBarb}[1]{\not\Downarrow_{#1}}

\newcommand{\sessionfont}[1]{\mathtt{#1}}
\newcommand{\vart}[1]{\mathsf{#1}}

\newcommand{\ssep}{.}
\newcommand{\shsep}{.}
\newcommand{\outses}{!}
\newcommand{\inpses}{?}
\newcommand{\selses}{\triangleleft}
\newcommand{\brases}{\triangleright}
\newcommand{\dual}[1]{\overline{#1}}
\newcommand{\cat}{\cdot}

\newcommand{\bacc}[2]{#1 \boundv{#2} \shsep}
\newcommand{\breq}[2]{\send{#1} \freev{#2} \shsep}

\newcommand{\bout}[2]{#1 \outses \freev{#2} \ssep}
\newcommand{\bbout}[2]{#1 \outses \blrangle{#2} \ssep}
\newcommand{\binp}[2]{#1 \inpses \boundv{#2} \ssep}
\newcommand{\bsel}[2]{#1 \selses #2 \ssep}
\newcommand{\bbra}[2]{#1 \brases \set{#2}}

\newcommand{\tfont}[1]{\mathtt{#1}}
\newcommand{\tsep}{;}

\newcommand{\chtype}[1]{\lrangle{#1}}

\newcommand{\outtype}{\outses}
\newcommand{\inptype}{\inpses}

\newcommand{\trec}[2]{\mu\vart{#1}.#2}
\newcommand{\tvar}[1]{\vart{#1}}

\newcommand{\tinact}{\tfont{end}}

\newcommand{\btout}[1]{\outtype \lrangle{#1} \tsep}
\newcommand{\bbtout}[1]{\outtype \big\langle{#1}\big\rangle \tsep}
\newcommand{\btinp}[1]{\inptype (#1) \tsep}
\newcommand{\bbtinp}[1]{\inptype \big({#1}\big) \tsep}
\newcommand{\btsel}[1]{\oplus \set{#1}}

\newcommand{\btbra}[1]{\& \set{#1}}

\newcommand{\proves}{\vdash}
\newcommand{\hastype}{\triangleright}

\newcommand{\outlts}{\outses}
\newcommand{\inplts}{\inpses}
\newcommand{\sellts}{\oplus}
\newcommand{\bralts}{\&}

\newcommand{\bactout}[2]{#1 \outlts \freev{#2}}

\newcommand{\bactinp}[2]{#1 \inplts \freev{#2}}
\newcommand{\bactsel}[2]{#1 \sellts #2}
\newcommand{\bactbra}[2]{#1 \bralts #2}

\newcommand{\by}[1]{\stackrel{#1}{\longrightarrow}}
\newcommand{\By}[1]{\stackrel{#1}{\Longrightarrow}}

\newcommand{\hby}[1]{\stackrel{#1}{\longmapsto}}
\newcommand{\Hby}[1]{\stackrel{#1}{\Longmapsto}}

\newcommand{\stau}{\tau_{\mathsf{s}}}
\newcommand{\btau}{\tau_{\beta}}
\newcommand{\dtau}{\tau_{\mathsf{d}}}

\newcommand{\betatran}{$\beta-$transition\xspace}

\newcommand{\bistyp}{\rightleftharpoons}

\newcommand{\map}[1]{[\!\![#1]\!\!]}
\newcommand{\pmap}[2]{\ensuremath{[\!\![#1]\!\!]^{#2}}}
\newcommand{\pmapp}[3]{\ensuremath{[\!\![#1]\!\!]^{#2}_{#3}}}
\newcommand{\tmap}[2]{\ensuremath{(\!\!\langle#1\rangle\!\!)^{#2}}}
\newcommand{\mapt}[1]{\ensuremath{(\!\!\langle#1\rangle\!\!)}}
\newcommand{\mapa}[1]{\ensuremath{\{\!\!\{#1\}\!\!\}}}

\newcommand{\auxmap}[2]{\ensuremath{\big\lfloor\!\!\big\lfloor#1\big\rfloor\!\!\big\rfloor_{#2}}}

\newcommand{\mapchar}[2]{\ensuremath{[\!\!(#1)\!\!]^{#2}}}
\newcommand{\omapchar}[1]{\ensuremath{[\!\!(#1)\!\!]_{\mathsf{c}}}}

\newcommand{\enco}[1]{\big\langle #1\big\rangle}

\newcommand{\calc}[5]{\lrangle{#1, #2, #3, #4, #5}}
\newcommand{\tyl}[1]{\ensuremath{\mathcal{#1}}}

\newcommand{\dualof}{\ \mathsf{dual}\ }
\newcommand{\iso}{\ \mathsf{iso}\ }

\newcommand{\HOp}{\ensuremath{\mathsf{HO}\pi}\xspace}
\newcommand{\sessp}{\ensuremath{\pi}\xspace}
\newcommand{\pHOp}{\ensuremath{\mathsf{HO}\vec{\pi}}\xspace}
\newcommand{\HO}{\ensuremath{\mathsf{HO}}\xspace}

\newcommand{\HOpp}{\ensuremath{\mathsf{HO}\pi^{\mathsf +}}\xspace}
\newcommand{\pHOpp}{\ensuremath{\mathsf{HO}\vec{\pi}^{\mathsf +}}\xspace}

\newcommand{\CAL}{\ensuremath{\mathsf{C}}\xspace}

\newcommand{\minussh}{\ensuremath{\mathsf{-sh}}\xspace}

\newcommand{\ST}{\ensuremath{\mathsf{ST}}\xspace}

\newcommand{\Proc}{\ensuremath{\diamond}}

\newcommand{\appl}[2]{#1\,#2}
\newcommand{\abs}[2]{\lambda #1.\,#2}

\newcommand{\lollipop}{\multimap}
\newcommand{\sharedop}{\rightarrow}

\newcommand{\lhot}[1]{#1\!\! \lollipop\!\! \diamond}
\newcommand{\shot}[1]{#1\!\! \sharedop\!\! \diamond}

\newcommand{\vmap}[1]{(\!|\!|#1|\!|\!)}

\newcommand{\triggerarrow}{\Leftarrow}%{\leftarrow\!\!\!\!\!\!\!\leftarrow}
\newcommand{\htrigger}[2]{#1 \triggerarrow #2}
\newcommand{\ftrigger}[3]{#1 \triggerarrow #2:#3}

\newcommand{\hotrigger}[4]{\binp{#1}{#2} \newsp{#3}{\appl{#2}{#3} \Par \bout{\dual{#3}}{#4} \inact}}
\newcommand{\fotrigger}[5]{\binp{#1}{#2} \newsp{#3}{\mapchar{#4}{#3} \Par \bout{\dual{#3}}{#5} \inact}}
\newcommand{\horel}[6]{#1; #2 \proves #3 #4 #5 \proves #6}

\newcommand{\mhorel}[7]{
	\begin{array}{rcll}
		#1; \es; &#2& \proves& #3\\
			#4 &#5& \proves & #6 #7
	\end{array}
}

\newcommand{\nhorel}[7]{
	\begin{eqnarray}
		#1; \es; &#2 \proves& #3 \nonumber \\
			#4 &#5 \proves & #6
		\label{#7}
	\end{eqnarray}
}

\newcommand{\nonhosyntax}[1]{\colorbox{lightgray}{\ensuremath{#1}}}

\newcommand{\X}{\varp{X}}

\newcommand{\C}{\ensuremath{{\Bbb C}}}

\newcommand{\suc}{\textrm{succ}}

\newif\iftodo\todotrue
\newif\ifny\nyfalse
\newcommand{\NY}[1]
{\ifny{\color{purple}{#1}}\else{#1}\fi}

\newif\ifdm\dmtrue
\newif\ifrhu\rhutrue
\newif\ifjp\jptrue
\newcommand{\jp}[1]
{\ifjp{\color{red}{JP: #1}}\else{#1}\fi}

\newcommand{\stytraargi}[8]{\ensuremath{#1; #3 \proves_{#7} #4 \hby{#2}_{#8} #5 \proves_{#7} #6 }}

\newcommand{\wtytraargi}[8]{\ensuremath{#1; #3 \proves_{#7} #4 \Hby{#2}_{#8}  #5 \proves_{#7} #6 }}
\newcommand{\wbbarg}[7]{\ensuremath{#1; #3 \proves_{#7} #4 \WB_{#7} #5 \proves_{#7} #6 }}

\newcommand{\stytra}[6]{\ensuremath{#1; #3 \proves #4 \hby{#2} #1; #5 \proves #6 }}

\newcommand{\wtytra}[6]{\ensuremath{#1; #3 \proves #4 \Hby{#2} #1; #5 \proves #6 }}

\newcommand{\wbb}[6]{\ensuremath{#1; #3 \proves #4 \WB #1; #5 \proves #6 }}

\newcommand{\PHOpp}{\ensuremath{\mathsf{HO}{\vec{\pi}}^{+}}\xspace}
\newcommand{\PHOp}{\ensuremath{\mathsf{HO}{\vec{\pi}}}\xspace}

\begin{document}
%\pagestyle{plain} %% Jorge: Uncommented this

%%%%%%%%%%%%%%%%%%%%%%%%%%%%%%%%%%%%%%%%%%%%%%%%%%%%%%%%%%%%%%%%%%%%%%%%%%%%%
% Title.
%%%%%%%%%%%%%%%%%%%%%%%%%%%%%%%%%%%%%%%%%%%%%%%%%%%%%%%%%%%%%%%%%%%%%%%%%%%%%

\title {Core Higher-Order Session Processes: \\
	Tractable Equivalences and Relative Expressiveness\thanks{Last Revision: \today}
}\titlerunning{\today}\authorrunning{\today}
\author{
	Dimitrios Kouzapas\inst{1,2}
        \and 
        Jorge A. P\'{e}rez\inst{3}
	\and
	Nobuko Yoshida\inst{1}
}
\institute {
Imperial College London	
\and
University of Glasgow
\and 
University of Groningen
}
\maketitle
%%%%%%%%%%%%%%%%%%%%%%%%%%%%%%%%%%%%%%%%%%%%%%%%%%%%%%%%%%%%%%%%%%%%%%%%%%%%%
% Abstract.
%%%%%%%%%%%%%%%%%%%%%%%%%%%%%%%%%%%%%%%%%%%%%%%%%%%%%%%%%%%%%%%%%%%%%%%%%%%%%
% !TEX root = main.tex
\begin{abstract}
This work proposes %efficient 
tractable
bisimulations 
for the higher-order $\pi$-calculus with session primitives (\HOp) and 
offers a complete 
study of the expressivity of its most significant subcalculi.
First we develop three typed bisimulations, which are shown to 
coincide with contextual equivalence.
These characterisations  
demonstrate that observing as inputs
only a specific finite set of higher-order values (which inhabit session types) suffices 
to reason about \HOp processes. 
Next, we identify \HO, 
a minimal, second-order  subcalculus of \HOp in which 
%does {\em not} equip with 
higher-order applications/abstractions, name-passing,
and recursion are absent.
We show that 
%two fully abstract encodings:
\HO can encode 
%the $n$-order \HOp
 \HOp extended with higher-order applications and abstractions
and 
that
a first-order session $\pi$-calculus can encode 
\HOp. %, fully abstractly.  
Both encodings are fully abstract.
We also 
prove that 
%then prove a  non-encodability result from 
the session $\pi$-calculus
with passing of shared names 
cannot be encoded 
into \HOp without shared names. 
We show that $\HOp$, $\HO$, and $\sessp$ are  equally expressive; 
the expressivity of \HO enables effective reasoning about 
typed equivalences for
higher-order processes.
\end{abstract}

%%%%%%%%%%%%%%%%%%%%%%%%%%%%%%%%%%%%%%%%%%%%%%%%%%%%%%%%%%%%%%%%%%%%%%%%%%%%%
% Main sections.
%%%%%%%%%%%%%%%%%%%%%%%%%%%%%%%%%%%%%%%%%%%%%%%%%%%%%%%%%%%%%%%%%%%%%%%%%%%%%
%\todofalse
%\iftodo
%\input{contribution}
%\else\fi

\setcounter{tocdepth}{2}
\tableofcontents
%\newpage

% !TEX root = main.tex

%\myparagraph{Key points}
%\begin{enumerate}[1.]
%%	\item	Session $\pi$ calculus with process passing. DONE
%%	\item	Identify session $\pi$ and process passing subcalculi and their polyadic variants. DONE
%%	\item	Bisimulation theory for higher-order session semantics. DONE
%%	\item	New triggered bisimulation, related to J\&R's. DONE
%%	\item   Elementary values key to characterizations of behavioural equivalence. DONE
%	\item	Types provide techniques to prove completeness without matching. \jp{TBD}
%	\item	We are interested in encodings with properties a la Gorla. 
%                We extended them to typed setting. \jp{TBD}
%%	\item	Encode name-passing to pure process abstraction calculus, with name abstractions. DONE
%%	\item	Type of the recursion encoding uses non tail recursive type $\trec{t}{\btinp{t} \tinact}$. DONE
%%	\item	Encode higher-order semantics to first order semantics. DONE
%%	\item	Negative result. Cannot encode shared names using only shared names.
%%	\item   Extensions with higher-order abstractions and polyadicity also explored. DONE
%\end{enumerate}

%\smallskip 
%
%\myparagraph{Important things to explain}
%Explain our \HO is very small without containg name passing 
%\[ 
%\abs{x}.P \quad \appl{x}{u}
%\]

%Explain we input only characteristic processes.  
%
%\[
%\lambda x.\mapchar{S}{x}
%\]

%\subsection{Higher-Order Session Calculi}

\section{Introduction}
By combining features from the $\lambda$-calculus and the $\pi$-calculus, 
in \emph{higher-order process calculi} exchanged values may contain  processes. 
In this paper, we consider higher-order calculi with \emph{session primitives},
thus enabling the specification of reciprocal exchanges (protocols) 
for higher-order mobile processes, 
which can be verified via type-checking using \emph{session types}~\cite{honda.vasconcelos.kubo:language-primitives}.
The study of higher-order concurrency has received significant attention, 
from untyped and typed perspectives (see, e.g.,~\cite{ThomsenB:plachoasgcfhop,SangiorgiD:expmpa,San96int,JeffreyR05,MostrousY15,DBLP:journals/iandc/LanesePSS11,DBLP:conf/icalp/LanesePSS10,DBLP:conf/esop/KoutavasH11,XuActa2012}).
Although models of session-typed 
communication with features of higher-order concurrency exist~\cite{tlca07,DBLP:journals/jfp/GayV10},
their  \emph{tractable behavioural equivalences} and \emph{relative expressiveness}
remain little understood. 
Clarifying their status is not only useful for, 
e.g.,~justifying non-trivial mobile protocol
optimisations, but also for transferring key reasoning techniques
between (higher-order) session calculi. Our discovery 
is that \emph{linearity} of session types plays a vital role to 
offer new equalities and fully abstract encodability, 
which to our best knowledge have not been proposed before.   

The main higher-order language in our work, denoted \HOp,
extends the higher-order $\pi$-calculus~\cite{SangiorgiD:expmpa} with session primitives:
it contains constructs for synchronisation on shared names, 
recursion, 
name abstractions (i.e., functions from name identifiers  to processes, 
denoted $\lambda x.P$) and applications 
{(denoted $(\lambda x.P)a$)};
and session communication (value passing and
labelled choice using linear names). 
We study two significant subcalculi of \HOp, 
{which}
distil higher- and first-order mobility:
the \HO-calculus, which is \HOp without recursion and name passing, and 
the session \sessp-calculus {(here denoted~\sessp)}, which is \HOp without abstractions and applications.  
While \sessp is, 
in essence, the calculus in~\cite{honda.vasconcelos.kubo:language-primitives}, 
this paper shows that \HO  is a new core calculus 
for higher-order session concurrency.

In the first part of the paper, we address tractable behavioural equivalences
for \HOp.
A well-studied behavioural equivalence in the higher-order setting 
is \emph{context bisimilarity}~\cite{San96H},
a labelled characterisation of reduction-closed, barbed congruence, 
which offers an appropriate discriminative power at the price of heavy universal quantifications in output clauses.
Obtaining alternative characterisations 
is thus a recurring issue 
in the study of higher-order calculi. 
Our approach 
shows that protocol specifications given by session types are 
essential to  limit 
the behaviour of higher-order session processes. 
Exploiting elementary processes inhabiting session types, 
this limitation is formally enforced by 
a refined (typed) labelled transition system (LTS)
that narrows down the spectrum of allowed process behaviours, 
thus enabling tractable reasoning techniques. 
Two tractable characterisations of bisimilarity 
are then introduced. 
Remarkably, using session types we prove that these %tractable
bisimilarities coincide with context bisimilarity, without using
operators for 
name-matching.

We then move on to assess the expressivity 
 of \HOp, \HO, and \sessp as delineated by typing. 
We establish strong correspondences between 
these calculi  via type-preserving, fully abstract encodings up to 
behavioural equalities. While encoding \HOp 
into the $\pi$-calculus preserving session types 
(extending  known  results for untyped processes) is 
%{already}
significant, 
our main contribution is 
an encoding of \HOp into \HO, where name-passing is absent.  

We illustrate the essence of encoding name passing into \HO: 
to encode name output, we ``pack''
the name to be passed around into a suitable abstraction; 
upon reception, the receiver must ``unpack'' this object following a precise protocol.
More precisely, our encoding 
{of name passing}
in \HO is given as:
\begin{center}
\begin{tabular}{rcll}
  $\map{\bout{a}{b} P}$	&$=$&	$\bout{a}{ \abs{z}{\,\binp{z}{x} (\appl{x}{b})} } \map{P}$ \\
  $\map{\binp{a}{x} Q}$	&$=$&	$\binp{a}{y} \newsp{s}{\appl{y}{s} \Par \bout{\dual{s}}{\abs{x}{\map{Q}}} \inact}$
\end{tabular}
\end{center}
\noi where $a,b$ are names; $s$ and $\dual{s}$ are 
linear names (called \emph{session endpoints});
%$\lambda x.P$ is a name abstraction of $P$; $\appl{x}{a}$ is a name application; 
$\bout{a}{V} P$ and 
$\binp{a}{x} P$ denote an output and input at~$a$;   
and $\newsp{s}P$ is hiding. 
A (deterministic) reduction between   endpoints 
$s$ and $\dual{s}$ guarantees name $b$ is properly unpacked.
Encoding a recursive process $\recp{X}{P}$ is \NY{also} challenging, for 
the linearity of endpoints in $P$ must be preserved.
We encode recursion with non-tail recursive session types; for this 
we apply recent advances on the theory of session duality~\cite{TGC14,DBLP:journals/corr/abs-1202-2086}.

We further extend our encodability results to 
i)~\HOp with \emph{higher-order} abstractions (denoted \HOpp) 
and to ii)~\HOp with polyadic name passing and abstraction (\pHOp); and to
their super-calculus  (\PHOpp) (equivalent to the calculus in~\cite{tlca07}). 
A further result shows that 
shared names
%as required in the session establishment phase,
strictly add expressive power 
to session calculi. 
\figref{fig:express} summarises %our expressivity 
these results.

\begin{figure}[t]
\centering
\includegraphics[scale=1]{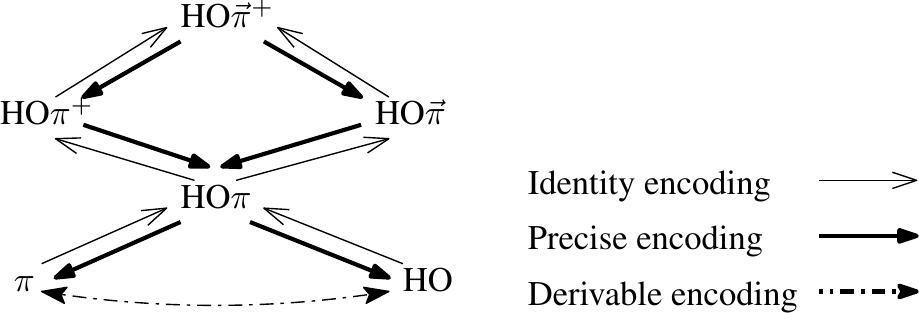}
\vspace{1mm}

	\caption{Encodability in Higher-Order Session Calculi. 
	Precise encodings are defined in \defref{def:goodenc}.
	\label{fig:express}}
\end{figure}

\smallskip

\myparagraph{Outline / Contributions.} This paper 
is structured as follows:
\begin{enumerate}[$\bullet$]
%	\item	\secref{sec:overview} overviews key ideas of our tractable bisimulations.

	\item	\secref{sec:calculus} presents the higher-order session calculus \HOp and its 
		subcalculi \HO and~\sessp. 

	\item	\secref{sec:types} gives the type system
		and states type soundness for \HOp and its variants.

	\item	\secref{sec:behavioural} 
		develops \emph{higher-order} and \emph{characteristic} bisimilarities, our two
		tractable characterisations of contextual equivalence which 
		alleviate the issues of context bisimilarity~\cite{San96H}. These 
		relations are shown to coincide in \HOp (\thmref{the:coincidence}).

	\item	\secref{sec:enc} defines \emph{precise (typed) encodings} by extending encodability criteria 
		studied for
		untyped processes~(e.g.~\cite{DBLP:journals/iandc/Gorla10,DBLP:conf/icalp/LanesePSS10}).

	\item	\secref{sec:positive} and \secref{sec:negative}
		gives encodings of \HOp into \HO and of \HOp into \sessp.
		These encodings 
		are shown to be \emph{precise} (\propref{prop:prec:HOp_to_HO} and \propref{prop:prec:HOp_to_p}).
		Mutual encodings between \sessp and \HO are derivable; 
		all these calculi are thus equally expressive.
		Exploiting determinacy and typed equivalences,
		we also prove the non-encodability of shared names
		into linear names (\thmref{thm:negative}).

	\item	\secref{sec:extension} studies extensions of \HOp. We show that 
		both \HOpp (the extension with higher-order applications) 
		and \pHOp (the extension with polyadicity) are encodable in \HOp
		(\propref{prop:prec:HOpp_to_HOp} and \propref{prop:prec:pHOp_to_HOp}).
		This connects our work 
		to the existing
		higher-order session calculus in~\cite{tlca07} (here denoted  $\pHOpp$).

	\item	\secref{sec:related} reviews related works.
The appendix collects proofs of the main results.
\end{enumerate}
%\noi
%The paper is self-contained. 
%{\bf\em Additional related work, more examples, omitted definitions, and  proofs 
%can be found
%are 
%in~\cite{KouzapasPY15}.} 

% !TEX root = main.tex

\section{The Higher-Order Session $\pi$-Calculus (\HOp)}
\label{sec:calculus}

We introduce the 
\emph{Higher-Order Session $\pi$-Calculus} (\HOp).
\HOp includes both name- and abstraction-passing operators
as well as recursion; it corresponds to a subcalculus 
of the language
studied by Mostrous and Yoshida in~\cite{tlca07,MostrousY15}. 
Following the literature~\cite{JeffreyR05},
for simplicity of the presentation
we concentrate on the second-order call-by-value \HOp.  
(In \secref{sec:extension} we consider the extension of 
\HOp with general higher-order abstractions 
and polyadicity in name-passing/abstractions.)
We also introduce two subcalculi of \HOp.
In particular, we define the 
core higher-order session
calculus (\HO), which 
%. The \HO calculus is  minimal: it 
includes constructs for shared name synchronisation and 
constructs for session establish\-ment/communication and 
(monadic) name-abstraction, but lacks name-passing and recursion.

Although minimal, in \secref{s:expr}
the abstraction-passing capabilities of \HOp will prove 
expressive enough to capture key features of session communication, 
such as delegation and recursion.

\subsection{Syntax} 

The syntax for $\HOp$ processes is given in \figref{fig:syntax}.

\myparagraph{Identifiers.}
We use $a,b,c, \dots$ to range over shared names, and
$s, \dual{s}, \dots$ to range over session names
whereas $m, n, t, \dots$ range over shared or session names.
We define dual session endpoints $\dual{s}$,
with the dual operator defined as
$\dual{\dual{s}} = s$ and $\dual{a} = a$.
Intuitively, names $s$ and $\dual{s}$ are dual  \emph{endpoints}.
Name and abstraction variables are uniformly denoted with $x, y, z, \dots$;
we reserve $k$ for name variables and we sometimes write $\underline{x}$ for abstraction variables.
Recursive variables are denoted with $\varp{X}, \varp{Y} \dots$.
An abstraction $\abs{x}{P}$ is a process $P$ with bound variable $x$.
Symbols $u, v, \dots$ range over names or variables. Furthermore
we use $V, W, \dots$ to denote transmittable values; either channels $u, v$ or
abstractions.

\myparagraph{Terms.} 
The name-passing constructs of \HOp include the
$\pi$-calculus prefixes for sending and receiving values $V$.
Process $\bout{u}{V} P$ denotes the output of value $V$
over channel $u$, with continuation $P$;
process $\binp{u}{x} P$ denotes the input prefix on channel $u$ of a value
that it is going to be substituted on variable $x$ in continuation $P$. 
Recursion is expressed by the primitive recursor $\recp{X}{P}$,
which binds the recursive variable $\varp{X}$ in process $P$.
Process $\appl{V}{u}$ is the application
process; it binds channel $u$ on the abstraction $V$.
Prefix $\bsel{u}{l} P$ selects label $l$ on channel $u$
and then behaves as $P$.
Given $i \in I$ process $\bbra{u}{l_i: P_i}_{i \in I}$ offers a choice
on labels $l_i$ with continuation $P_i$.
The calculus also includes standard constructs for 
the inactive process $\inact$, 
parallel composition $P_1 \Par P_2$, and 
name restriction $\news{n} P$.
Session name restriction $\news{s} P$ simultaneously 
binds endpoints $s$ and $\dual{s}$ in $P$.
We use $\fv{P}$ and $\fn{P}$ to denote a set of free 
variables and names, respectively; 
and assume $V$ in $\bout{u}{V}{P}$ does not include free recursive 
variables $\varp{X}$. 
Furthremore, a well-formed process relies on assumptions for
guarded recursive processes.
If $\fv{P} = \emptyset$, we call $P$ {\em closed}.
We write $\mathcal{P}$ for the set of all well-formed
processes.

\begin{figure}[t!]
\[
	\begin{array}{lcl}
		\begin{array}{lrcl}
			\textrm{(Processes)}&
			P, Q	&\bnfis&	\bout{u}{V} P \bnfbar \binp{u}{x} P \bnfbar \appl{V}{u} \\
			&	&\bnfbar&	\bsel{u}{l} P \bnfbar \bbra{u}{l_i: P_i}_{i \in I} \bnfbar \inact\\
			&	&\bnfbar&	P \Par Q \bnfbar \news{n} P \bnfbar \nonhosyntax{\varp{X} \bnfbar  \recp{X}{P}}
		\end{array}
		&\quad &
		\begin{array}{lrcl}
			
			\textrm{(Names)}	& n, m, t	&\bnfis&	a, b \bnfbar s, \dual{s} \\
			\textrm{(Identifiers)}	& u, v		&\bnfis& 	n \bnfbar x, y, z, k \\
			\textrm{(Values)}	& V, Q		&\bnfis&	\nonhosyntax{u} \bnfbar \abs{x}{P}
		\end{array}
	\end{array}
\]
	\caption{Syntax for $\HOp$ (The definition of \HO lacks the constructs in \colorbox{lightgray}{grey}) \label{fig:syntax}}
\end{figure}

\subsection{Sub-calculi}

\label{subsec:subcalculi}

We identify two main sub-calculi of $\HOp$
that will form the basis of our study:
\begin{definition}[Sub-calculi of \HOp]\myrm
	We let $\CAL \in \set{\HOp, \HO, \sessp}$ with:
	\begin{enumerate}[-]
		\item	{\em Core higher-order session calculus (\HO)}:
			The sub-calculus \HO uses only abstraction passing, i.e.,~values
			in \figref{fig:syntax}
			are defined as in the non-gray syntax;
			$V \bnfis \abs{x} P$ and does not use the primitive
			recursion constructs, $\varp{X}$ and $\recp{X}{P}$.

		\item	{\em Session $\pi$-calculus (\sessp)}:
			The sub-calculus \sessp uses only name-passing constructs, i.e.,~values
			in \figref{fig:syntax}
			are defined as $V \bnfis u$, and does not use applications
			$\appl{x}{u}$.
	\end{enumerate}
	We write $\CAL^\minussh$ to denote a sub-calculus without shared names,
	i.e.,~identifiers in \figref{fig:syntax} are defined as $u,v \bnfis s, \dual{s}$.
\end{definition}
Thus, while \sessp is essentially the standard session $\pi$-calculus
as defined in the literature~\cite{honda.vasconcelos.kubo:language-primitives,GH05},
\HO can be related to a subcalculus of higher-order process calculi as studied
in the untyped~\cite{SangiorgiD:expmpa,SaWabook,JeffreyR05}
and typed settings~\cite{tlca07,mostrous09sessionbased,MostrousY15}.
In \secref{sec:positive} we show that 
$\HOp$, $\HO$, and $\sessp$ have the same expressivity.

\subsection{Operational Semantics}

\label{subsec:reduction_semantics}

The operational semantics for \HOp is standardly given as a \emph{reduction relation},
supported by a {structural congruence} relation, denoted $\scong$. 
Structural congruence is 
the least congruence that satisfies the commutative monoid $(\mathcal{P}, \Par, \inact)$:
\[
\begin{array}{c}
	P \Par \inact \scong P
	\qquad
	P_1 \Par P_2 \scong P_2 \Par P_1
	\qquad
	P_1 \Par (P_2 \Par P_3) \scong (P_1 \Par P_2) \Par P_3
\end{array}
\]
\noi satisfies $\alpha$-conversion:
\[
	P_1 \scong_\alpha P_2\ \textrm{  implies  }\ P_1 \scong P_2
\]
\noi and furthermore, satisfies the rules:
\[
\begin{array}{c}
	n \notin \fn{P_1}\ \textrm{  implies  }\ P_1 \Par \news{n} P_2 \scong \news{n}(P_1 \Par P_2)
	\\[2mm]
	\news{n} \inact \scong \inact
	\qquad \quad
	\news{n} \news{m} P \scong \news{m} \news{n} P
	\qquad \quad
	\recp{X}{P} \scong P\subst{\recp{X}{P}}{\varp{X}}
%	\qquad
%	\appl{(\abs{x} P)}{u} \scong P \subst{x}{u}
\end{array}
\]
\noi The first rule is describes scope opening for names.
Restricting of a name in an inactive process has no effect.
Furthermore, we can permute name restrictions.
Recursion is defined in structural congruence terms;
a recursive term $\recp{X}{P}$ is structurally
equivalent to its unfolding.

Structural congruence is extended to support values,
i.e.,~is the least congruence over processes and values
that satisfies $\cong$ for processes and, furthermore:
\[
	\abs{x}{P_1} \scong_\alpha \abs{y}{P_2}\ \textrm{ implies }\ \abs{x}{P_1} \scong \abs{y}{P_2}
	\qquad \quad
	P_1 \scong P_2 \textrm{ implies } \abs{x}{P_1} \scong \abs{x}{P_2}
\]
\noi This way, abstraction values are congruent up-to $\alpha$-conversion.
Furthermore, two congruent processes can construct congruent
abstractions.

%Name application is also
%defined in structural congruence terms where structural
%congruence defines that application $\appl{(\abs{x} P)}{u}$ 
%substitutes variable $x$ with name $u$ over abstraction $\abs{x} P$.

% !TEX root = ../main.tex
\begin{figure}[t!]
\[
\begin{array}{c}
	\begin{array}{rcll}
		(\abs{x}{P}) \, u  & \red & P \subst{u}{x} & \orule{App}
		\\[2mm]
		\bout{n}{V} P \Par \binp{\dual{n}}{x} Q & \red & P \Par Q \subst{V}{x} & \orule{Pass}
		\\[2mm]
		\bsel{n}{l_j} Q \Par \bbra{\dual{n}}{l_i : P_i}_{i \in I} & \red & Q \Par P_j ~~(j \in I)~~  & \orule{Sel}
	\end{array}
	\\[8mm]
	\begin{array}{c}
		\tree{
			P \red P'
		}{
			\news{n} P \red \news{n} P'
		}\ \orule{Sess}
		\quad \qquad
		\tree{
			P \red P'
		}{
			P \Par Q \red  P' \Par Q
		}\ \orule{Par}
		\quad \qquad
		\tree{
			P \scong \red \scong P'
		}{
			P \red P'
		}\ \orule{Cong}
	\end{array}
\end{array}
\]
	\caption{Reduction semantics for \HOp. \label{fig:reduction}}
\end{figure}

\figref{fig:reduction} defines
the operational semantics for the \HOp.
$\orule{App}$ is a name application.
Rule~$\orule{Pass}$ defines value passing where
value $V$ is being send on channel $n$ to its dual endpoint $\dual{n}$
(for shared interactions $\dual{n} = n$).
As a result of the value passing reduction the continuation of the 
receiving process substitutes the receiving variable $x$ with $V$.
Rule~$\orule{Sel}$ is the standard rule for labelled choice/selection;
given an index set $I$,
a process selects label $l_j, j \in I$ on channel $n$ over a set of
labels $\set{l_i}_{i \in I}$ that are offered by a parallel process
on the dual session endpoint $\dual{n}$.
Remaining rules define congruence 
with respect to parallel composition (rule $\orule{Par}$)
and name restriction (rule $\orule{Ses}$).
Rule $\orule{Cong}$ defines closure under structural congruence.
We write $\red^\ast$ for a multi-step reduction.

% !TEX root = main.tex
%\newpage
\section{Session Types for $\HOp$}
\label{sec:types}

In this section we define a session typing system for
$\HOp$ and establish its main properties. We use as
a reference the type system for higher-order session processes 
developed by Mostrous and Yoshida~\cite{tlca07,mostrous09sessionbased,MostrousY15}.
Our system is simpler than that in~\cite{tlca07}, in order to distil the key
features of higher-order communication in a session-typed setting.

%%%%%%%%%%%%%%%%%%%%%%%%%%%%%%%%%%%%%%%%%%%%%%%%%%
%  SYNTAX FOR TYPES
%%%%%%%%%%%%%%%%%%%%%%%%%%%%%%%%%%%%%%%%%%%%%%%%%%

\subsection{Syntax}

We define the syntax of session types for \HOp.

\begin{definition}[Syntax of Types]\myrm
	\label{def:types}
	The syntax of types is defined on the types for sessions $S$,
	and the types for values $U$:
	\[
	\begin{array}{lrcl}
		\textrm{(value)} & U & \bnfis &		\nonhosyntax{C}  \bnfbar L 
		\\

		\textrm{(name)} & C & \bnfis &	S \bnfbar \chtype{S} \bnfbar \chtype{L}
		\\

		\textrm{(abstr)} & L & \bnfis &	\shot{C} \bnfbar \lhot{C}
		\\

		\textrm{(session)} & S,T & \bnfis & 	\btout{U} S \bnfbar \btinp{U} S
							\bnfbar \btsel{l_i:S_i}_{i \in I} \bnfbar \btbra{l_i:S_i}_{i \in I}\\
					&&\bnfbar&	\trec{t}{S} \bnfbar \vart{t}  \bnfbar \tinact
	\end{array}
	\]
\end{definition}
\noi \myparagraph{Types for Values.}
Types for values range over symbol $U$ which includes
first-order types $C$ and higher-order types $L$.
First-order types $C$ are used to type names;
session types $S$ type session names and shared types
$\chtype{S}$ or $\chtype{L}$ type shared names that
carry session values and higher-order values, respectively.
Higher-order types $L$ are used to type abstraction values;
$\shot{C}$ and $\lhot{C}$ denote
shared and linear abstraction types, respectively.

\myparagraph{Session Types.}
The syntax of session types $S$ follows the usual
(binary) session types with
recursion~\cite{honda.vasconcelos.kubo:language-primitives,GH05}.
An {\em output type} $\btout{U} S$ is assigned to a name that
first sends a value of type $U$ and then follows
the type described by $S$.
Dually, the {\em input type} $\btinp{U} S$ is assigned to a name
that first receives a value of type $U$ and then continues as $S$. 
Session types for labelled choice and selection, 
%are standard: they are 
written $\btbra{l_i:S_i}_{i \in I}$ and $\btsel{l_i:S_i}_{i \in I}$, respectively,
require a set of types $\set{S_i}_{i \in I}$ that correspond to a set of
labels $\set{i \in I}_{i \in I}$. 
{\em Recursive session types} are defined using the primitive recursor.
We require type variables to be \emph{guarded}; this means, e.g., that type~$\trec{t}{\vart{t}}$ is not allowed.
Type $\tinact$ is the termination type.
We let $\mathsf{T}$ to be the set of all well-formed types and
\ST to be the set of all well-formed session types.

Types of \HO exclude $\nonhosyntax{C}$ from 
value types of \HOp; the types of \sessp exclude $L$. 
From each $\CAL \in \{\HOp, \HO, \pi \}$, $\CAL^{-\mathsf{sh}}$ 
excludes shared name types ($\chtype{S}$ and $\chtype{L}$), 
from name type~$C$.

\begin{remark}[Restriction on Types for Values]
	The syntax for value types is restricted
	to disallow types of the form:
	\begin{enumerate}[$\bullet$]
		\item	$\chtype{\chtype{U}}$: shared names
			cannot carry shared names; and

		\item  $\shot{U}$: abstractions do not
			bind higher-order variables.
	\end{enumerate}
\end{remark}

The difference between the syntax of process
in \HOp with the syntax of processes in~\cite{tlca07,MostrousY15}
is also reflected on the two corresponding type syntax;
the type structure  in~\cite{tlca07,MostrousY15}, 
supports the arrow types of the form $U \sharedop T$ and 
$U \lollipop T$, where $T$ denotes an arbitrary type of a term 
(i.e.~a value or a process).

%%%%%%%%%%%%%%%%%%%%%%%%%%%%%%%%%%%%%%%%%%%%%%%%%%
%  DUALITY
%%%%%%%%%%%%%%%%%%%%%%%%%%%%%%%%%%%%%%%%%%%%%%%%%%

\subsection{Duality}

Duality is defined following the co-inductive
approach, as in~\cite{GH05,TGC14}.
We first require the notion of type equivalence.
\begin{definition}[Type Equivalence]\myrm
\label{def:type_equiv}
	Define function $F(\Re): \mathsf{T} \longrightarrow \mathsf{T}$:
	\[
		\begin{array}{rcl}
			F(\Re) 	&=&	\set{(\tinact, \tinact)} \\
				&\cup&	\set{(\chtype{S}, \chtype{T}) \bnfbar S\ \Re\ T} \cup \set{(\chtype{L_1}, \chtype{L_2}) \bnfbar L_1\ \Re\ L_2}\\
				&\cup&	\set{(\shot{C_1}, \shot{C_2}), (\lhot{C_1}, \lhot{C_2}) \bnfbar C_1\ \Re\ C_2}\\
				&\cup&	\set{(\btout{U_1} S, \btout{U_2} T)\,,\, (\btinp{U_1} S, \btinp{U_1} T) \bnfbar U_1\ \Re\ U_2, S\ \Re\ T}\\
				&\cup&	\set{(\btsel{l_i: S_i}_{i \in I} \,,\, \btsel{l_i: T_i}_{i \in I}) \bnfbar  S_i\ \Re\ T_i}\\
				&\cup&	\set{(\btbra{l_i: S_i}_{i \in I}\,,\, \btbra{l_i: T_i}_{i \in I}) \bnfbar S_i\ \Re\ T_i}\\
				&\cup&	\set{(S\,,\, T) \bnfbar S\subst{\trec{t}{S}}{\vart{t}}\ \Re\ T)} \\
				&\cup&	\set{(S\,,\, T) \bnfbar S\ \Re\ T\subst{\trec{t}{T}}{\vart{t}})}
		\end{array}
	\]	
	\noi Standard arguments ensure that $F$ is monotone, thus the greatest fixed point
	of $F$ exists. Let type equivalence be defined as $\iso = \nu X. F(X)$.
\end{definition}
\noi In essence, type equivalence is a co-inductive definition that
equates types up-to recursive unfolding.
We may now define the duality relation in terms of type equivalence.
\begin{definition}[Duality]\myrm
\label{def:type_dual}
	Define function $F(\Re): \mathsf{ST} \longrightarrow \mathsf{ST}$:
	\[
		\begin{array}{rcl}
			F(\Re) 	&=&	\set{(\tinact, \tinact)}\\
				&\cup&	\set{(\btout{U_1} S, \btinp{U_2} T)\,,\, (\btinp{U} S, \btout{U} T) \bnfbar U_1 \iso U_2, S\ \Re\ T}\\
				&\cup&	\set{(\btsel{l_i: S_i}_{i \in I} \,,\, \btbra{l_i: T_i}_{i \in I}) \bnfbar  S_i\ \Re\ T_i}\\
				&\cup&	\set{(\btbra{l_i: S_i}_{i \in I}\,,\, \btsel{l_i: T_i}_{i \in I}) \bnfbar S_i\ \Re\ T_i}\\
				&\cup&	\set{(S\,,\, T) \bnfbar S\subst{\trec{t}{S}}{\vart{t}}\ \Re\ T)}\\
				&\cup&	\set{(S\,,\, T) \bnfbar S\ \Re\ T\subst{\trec{t}{T}}{\vart{t}})}
		\end{array}
	\]	
	\noi Standard arguments ensure that $F$ is monotone, thus the greatest fixed point
	of $F$ exists. Let duality be defined as $\dualof = \nu X. F(X)$.
\end{definition}
Duality is applied co-inductively to session types
up-to recursive unfolding.
Dual session types are prefixed
on dual session type constructors
that carry equivalent types ($!$ is dual to $?$ and $\oplus$ is dual to~$\&$).
%The co-inductive definition of duality relates
%session types up-to recursive unfolding.

%%%%%%%%%%%%%%%%%%%%%%%%%%%%%%%%%%%%%%%%%%%%%%%%%%
%  TYPING SYSTEM
%%%%%%%%%%%%%%%%%%%%%%%%%%%%%%%%%%%%%%%%%%%%%%%%%%

\subsection{Type Environments and Judgements}
Following
\cite{tlca07,MostrousY15}, we define the typing environments.
\begin{definition}[Typing environment]\myrm\label{def:typeenv}
	We define the {\em shared type environment} $\Gamma$,
	the {\em linear type environment} $\Lambda$, and
	the {\em session type environment} $\Delta$ as:
	\[
	\begin{array}{llcl}
		\text{(Shared)}		& \Gamma  & \bnfis &	\emptyset \bnfbar \Gamma \cat x: \shot{C} \bnfbar \Gamma \cat u: \chtype{S} \bnfbar
								\Gamma \cat u: \chtype{L} \bnfbar \Gamma \cat \varp{X}: \Delta
		\\
		\text{(Linear)}		& \Lambda & \bnfis &	\emptyset \bnfbar \Lambda \cat x: \lhot{C}
		\\
		\text{(Session)}	& \Delta  & \bnfis &	\emptyset \bnfbar \Delta \cat u:S
	\end{array}
	\]
	We further require:
	\begin{enumerate}[i.]
		\item	Domains of $\Gamma, \Lambda, \Delta$ are pairwise distinct.
		\item	Weakening, contraction and exchange apply to shared environment $\Gamma$.
		\item	Exchange applies to linear environments $\Lambda$ and $\Delta$. 
	\end{enumerate}
\end{definition}
\noi We define typing judgements for values $V$
and processes $P$:
\[	\begin{array}{c}
		\Gamma; \Lambda; \Delta \proves V \hastype U \qquad \qquad \qquad \qquad \Gamma; \Lambda; \Delta \proves P \hastype \Proc
	\end{array}
\]
\noi The first judgement asserts that under environment $\Gamma; \Lambda; \Delta$
values $V$ have type $U$,
whereas the second judgement asserts that under environment $\Gamma; \Lambda; \Delta$
process $P$ has the typed process type $\Proc$.

%%%%%%%%%%%
%	Type system description
%%%%%%%%%%%

\subsection{Typing Rules}

% !TEX root = ../main.tex
\begin{figure}[!t]
\[
	\begin{array}{c}
		\trule{Sess}~~\Gamma; \emptyset; \set{u:S} \proves u \hastype S 
		\quad
		\trule{Sh}~~\Gamma \cat u : U; \emptyset; \emptyset \proves u \hastype U
		\quad
		\trule{LVar}~~\Gamma; \set{x: \lhot{C}}; \emptyset \proves x \hastype \lhot{C}
		\\[4mm]

		\trule{Prom}~~\tree{
			\Gamma; \emptyset; \emptyset \proves V \hastype \lhot{C}
		}{
			\Gamma; \emptyset; \emptyset \proves V \hastype \shot{C}
		} 
		\quad
		\trule{EProm}~~\tree{
			\Gamma; \Lambda \cat x : \lhot{C}; \Delta \proves P \hastype \Proc
		}{
			\Gamma \cat x:\shot{C}; \Lambda; \Delta \proves P \hastype \Proc
		}
		\\[6mm]
%		\trule{Pol}~~\tree{
%			I = \set{i \setbar k_i \in \tilde{k}, C_i \in \tilde{C}}
%			\qquad
%			\forall i \in I \quad \Gamma; \Lambda_i; \Delta_i \proves k_i \hastype C_i
%		}{
%			\Gamma; \bigcup_{i \in I} \Lambda_i; \bigcup_{i \in I} \Delta_i \proves \tilde{k} \hastype \tilde{C}
%		}
%		\\[6mm]
%
		\trule{Abs}~~\tree{
			\Gamma; \Lambda; \Delta_1 \proves P \hastype \Proc
			\quad
			\Gamma; \es; \Delta_2 \proves x \hastype C
		}{
			\Gamma; \Lambda; \Delta_1 \backslash \Delta_2 \proves \abs{x}{P} \hastype \lhot{C}
		}
		\\[6mm]

		\trule{App}~~\tree{
			\begin{array}{c}
				U = \lhot{C} \lor \shot{C}
				\quad
				\Gamma; \Lambda; \Delta_1 \proves V \hastype U
				\quad
				\Gamma; \es; \Delta_2 \proves u \hastype C
			\end{array}
		}{
			\Gamma; \Lambda; \Delta_1 \cat \Delta_2 \proves \appl{V}{u} \hastype \Proc
		} 
		\\[6mm]

%		\trule{Send}~~\tree{
%			\Gamma; \Lambda_1; \Delta_1 \proves P \hastype \Proc  \quad \Gamma; \Lambda_2; \Delta_2 \vdash V \hastype U  \quad (k:S \in \Delta_1 \cup \Delta_2)
%		}{
%			\Gamma; \Lambda_1 \cat \Lambda_2; (\Delta_1 \cat \Delta_2)\backslash\set{k:S} \cat k:\btout{U} S \proves \bout{k}{V} P \hastype \Proc
%		}
%		\\[4mm]

		\trule{Send}~~\tree{
			\Gamma; \Lambda_1; \Delta_1 \proves P \hastype \Proc
			\quad
			\Gamma; \Lambda_2; \Delta_2 \proves V \hastype U
			\quad
			u:S \in \Delta_1 \cat \Delta_2
		}{
			\Gamma; \Lambda_1 \cat \Lambda_2; ((\Delta_1 \cat \Delta_2) \backslash \set{u:S}) 
			\cat u:\btout{U} S \proves \bout{u}{V} P \hastype \Proc
		}
		\\[6mm]

		\trule{Rcv}~~\tree{
			\Gamma; \Lambda_1; \Delta_1 \cat u: S  \proves P \hastype \Proc
			\quad
			\Gamma; \Lambda_2; \Delta_2 \proves x \hastype C
		}{
			\Gamma \backslash x; \Lambda_1 \backslash \Lambda_2; \Delta_1 \backslash \Delta_2 
			\cat u: \btinp{C} S \vdash \binp{u}{x} P \hastype \Proc
		}
		\\[6mm]
%		\quad
%		\trule{RvH}~~\tree{
%			\Gamma; \Lambda_1; \Delta \cat u: S \proves P \hastype \Proc
%			\quad
%			\Gamma; \Lambda_2; \es \proves x \hastype L
%		}{
%			\Gamma \backslash x; \Lambda_1 \backslash \Lambda_2; \Delta \cat u: \btinp{L} S \proves \binp{u}{x} P \hastype \Proc
%		}
%		\\[6mm]
%
%		\trule{RcvS}~~\tree{
%			\Gamma; \Lambda; \Delta \cat k: S_1 \cat x: S_2 \proves P \hastype \Proc
%		}{
%			\Gamma; \Lambda; \Delta, k: \btinp{S_2} S_1  \vdash \binp{k}{x}P \hastype \Proc
%		}
%		\quad\quad 
%		\trule{RcvL}~~\tree{
%			\Gamma; \Lambda \cat X: \lhot{U}; \Delta \cat k: S  \proves P \hastype \Proc
%		}{
%			\Gamma; \Lambda; \Delta \cat k:\btinp{\lhot{U}}S  \proves \binp{k}{X}P \hastype \Proc
%		}
%		\\[4mm]
%		\trule{RcvShN}~~\tree{
%			\Gamma \cat x: \chtype{U}; \Lambda; \Delta \cat k: S_1  \proves P \hastype \Proc
%		}{
%			\Gamma; \Lambda; \Delta \cat k:\btinp{\chtype{U}}S_1  \proves \binp{k}{x}P \hastype \Proc
%		}		
%		\quad ~~
%		\trule{RcvSh}~~\tree{
%			\Gamma \cat X: \shot{U}; \Lambda; \Delta \cat k: S_1  \proves P \hastype \Proc
%		}{
%			\Gamma; \Lambda; \Delta \cat k:\btinp{\shot{U}}S_1  \proves \binp{k}{X}P \hastype \Proc
%		}
%		\\[4mm]

		\trule{Req}~~\tree{
			\begin{array}{c}
				\Gamma; \es; \es \proves u \hastype U_1
				\quad
				\Gamma; \Lambda; \Delta_1 \proves P \hastype \Proc
				\quad
				\Gamma; \es; \Delta_2 \proves V \hastype U_2
				\\
				(U_1 = \chtype{S} \Leftrightarrow U_2 = S)
				\lor
				(U_1 = \chtype{L} \Leftrightarrow U_2 = L)
			\end{array}
		}{
			\Gamma; \Lambda; \Delta_1 \cat \Delta_2 \proves \bout{u}{V} P \hastype \Proc
		}
		\\[6mm]

%		\trule{ReqH}~~\tree{
%			\Gamma; \es; \es \proves k \hastype \chtype{U}
%			\quad
%			\Gamma; \Lambda_1; \Delta_1 \proves P \hastype \Proc
%			\quad
%			\Gamma; \Lambda_2; \Delta_2 \proves (x) Q \hastype U
%		}{
%			\Gamma; \Lambda_1 \cat \Lambda_1; \Delta_1 \cat \Delta_2 \proves \bout{k}{(x) Q} P \hastype \Proc
%		}
%		\\[6mm]

		\trule{Acc}~~\tree{
			\begin{array}{c}
				\Gamma; \emptyset; \emptyset \proves u \hastype U_1
				\quad
				\Gamma; \Lambda_1; \Delta_1 \proves P \hastype \Proc
				\quad
				\Gamma; \Lambda_2; \Delta_2 \proves x \hastype U_2
				\\
				(U_1 = \chtype{S} \Leftrightarrow U_2 = S)
				\lor
				(U_1 = \chtype{L} \Leftrightarrow U_2 = L)
			\end{array}
		}{
			\Gamma; \Lambda_1 \backslash \Lambda_2; \Delta_1 \backslash \Delta_2 \proves \binp{u}{x} P \hastype \Proc
		}
		\\[6mm]

%		\trule{AcH}~~\tree{
%			\Gamma; \emptyset; \emptyset \proves u \hastype \chtype{L}
%			\quad
%			\Gamma; \Lambda_1; \Delta \proves P \hastype \Proc
%			\quad
%			\Gamma; \Lambda_2; \es \proves x \hastype L
%		}{
%			\Gamma \backslash X; \Lambda_1 \backslash \Lambda_2; \Delta \proves \binp{u}{x} P \hastype \Proc
%		}
%		\\[6mm]

		\trule{Bra}~~\tree{
			 \forall i \in I \quad \Gamma; \Lambda; \Delta \cat u:S_i \proves P_i \hastype \Proc
		}{
			\Gamma; \Lambda; \Delta \cat u: \btbra{l_i:S_i}_{i \in I} \proves \bbra{u}{l_i:P_i}_{i \in I}\hastype \Proc
		}
		\qquad\quad 
	 	\trule{Sel}~~\tree{
			\Gamma; \Lambda; \Delta \cat u: S_j  \proves P \hastype \Proc \quad j \in I
		}{
			\Gamma; \Lambda; \Delta \cat u:\btsel{l_i:S_i}_{i \in I} \proves \bsel{u}{l_j} P \hastype \Proc
		}
		\\[6mm]

%		\trule{Conn}~~\tree{
%			\Gamma; \Lambda; \Delta \cat x:S \proves P \hastype \Proc  \quad \Gamma; \emptyset; \emptyset \proves a \hastype \chtype{S}
%		}{
%			\Gamma; \Lambda; \Delta \proves \binp{a}{x} P \hastype \Proc
%		}
%		\quad
%		\trule{ConnL}~~\tree{
%			\Gamma \cat a : \chtype{\lhot{U}}; \Lambda \cat X: \lhot{U}; \Delta \proves P \hastype \Proc
%		}{
%			\Gamma \cat a : \chtype{\lhot{U}}; \Lambda; \Delta \proves \binp{a}{X} P \hastype \Proc
%		}
%		\\[4mm]
%
%		\trule{ConnSh}~~\tree{
%			\Gamma  \cat x:\chtype{U}; \Lambda; \Delta \proves P \hastype \Proc  \quad \Gamma; \emptyset; \emptyset \proves a \hastype \chtype{U}
%		}{
%			\Gamma; \Lambda; \Delta \proves \binp{a}{x} P \hastype \Proc
%		}
%		\quad
%		\trule{ConnS}~~\tree{
%			\Gamma \cat a : \chtype{\shot{U}} \cat X: \shot{U}; \Lambda; \Delta \proves P \hastype \Proc
%		}{
%			\Gamma \cat a : \chtype{\shot{U}} \cat X: \shot{U}; \Lambda; \Delta \proves \binp{a}{X} P \hastype \Proc
%		}
%		\\[4mm]

		\trule{Res}~~\tree{
			\Gamma\cat a:\chtype{S} ; \Lambda; \Delta \proves P \hastype \Proc
		}{
			\Gamma; \Lambda; \Delta \proves \news{a} P \hastype \Proc}
		\qquad\quad
		\trule{ResS}~~\tree{
			\Gamma; \Lambda; \Delta \cat s:S_1 \cat \dual{s}: S_2 \proves P \hastype \Proc \quad S_1 \dualof S_2
		}{
			\Gamma; \Lambda; \Delta \proves \news{s} P \hastype \Proc
		}
		\\[6mm]

		\trule{Par}~~\tree{
			\Gamma; \Lambda_{1}; \Delta_{1} \proves P_{1} \hastype \Proc \quad \Gamma; \Lambda_{2}; \Delta_{2} \proves P_{2} \hastype \Proc
		}{
			\Gamma; \Lambda_{1} \cat \Lambda_2; \Delta_{1} \cat \Delta_2 \proves P_1 \Par P_2 \hastype \Proc
		}
		\qquad\quad
		\trule{End}~~\tree{
			\Gamma; \Lambda; \Delta  \proves P \hastype T \quad u \not\in \dom{\Gamma, \Lambda,\Delta}
		}{
			\Gamma; \Lambda; \Delta \cat u: \tinact  \proves P \hastype \Proc
		}
		\\[6mm]

		\trule{Nil}~~\Gamma; \emptyset; \emptyset \proves \inact \hastype \Proc
		\qquad \quad
		\trule{RVar}~~\Gamma \cat \varp{X}: \Delta; \emptyset; \Delta  \proves \varp{X} \hastype \Proc
		\qquad\quad 
%	 	\trule{Rec}~~\tree{
%			\Gamma \cat \rvar{X}: \Delta; \emptyset; \emptyset  \proves P \hastype \Proc
%		}{
%			\Gamma ; \emptyset; \emptyset  \proves \recp{X}{P} \hastype \Proc
%		}
%		\\[4mm]

	 	\trule{Rec}~~\tree{
			\Gamma \cat \varp{X}: \Delta; \emptyset; \Delta  \proves P \hastype \Proc
		}{
			\Gamma ; \emptyset; \Delta  \proves \recp{X}{P} \hastype \Proc
		}

%		\\[4mm]
%		\trule{PSend}~~\tree{
%			\Gamma; \Lambda; \Delta \cat n: S \proves P \hastype \Proc \qquad \forall i \in I, \Gamma; \es; \Delta_i \proves m_i \hastype C_i
%		}{
%			\Gamma; \Lambda; ((\Delta\cat\tilde{\Delta_i})\backslash n:S) \cat n: \btout{\tilde{C_i}_{i \in I}} S\proves \bout{n}{\tilde{m_i}_{i \in I}} P \hastype \Proc
%		}
%		\\[4mm]
%
%		\trule{PRcv}~~\tree{
%			\Gamma; \Lambda; \Delta \cat n: S \proves P \hastype \Proc \qquad \forall i \in I, \Gamma_i; \es; \Delta_i \proves x: C_i 
%		}{
%			\Gamma\backslash\tilde{\Gamma_i}; \Lambda; \Delta\backslash\tilde{\Delta_i} \cat n: \btinp{\tilde{C_i}_{i \in I}} S \proves \binp{n}{\tilde{x_i}_{i \in I}} P \hastype \Proc
%		}
%		\\[4mm]
%
%		\trule{PAbs}~~\tree{
%			\Gamma; \Lambda; \Delta \proves P \hastype \Proc \quad \forall i \in I, \Gamma; \es; \Delta_i \proves x_i \hastype C_i
%		}{
%			\Gamma; \Lambda; \Delta\backslash\tilde{\Delta_i} \proves \abs{\tilde{x_i}_{i \in I}}{P} \hastype \lhot{\tilde{C_i}_{i \in I}}
%		}
%		\\[4mm]
%
%		\trule{App}~~\tree{(U = \lhot{\tilde{C_i}}) \lor (U = \shot{\tilde{C_i}}) \quad
%			\Gamma; \Lambda; \Delta \proves X \hastype U  \quad \forall i \in I, \Gamma; \es; \Delta_2 \proves k_i \hastype C_i
%		}{
%			\Gamma; \Lambda; \Delta \cat \tilde{\Delta_i} \proves \appl{X}{\tilde{k_i}} \hastype \Proc
%		} 
%		\\[4mm]
	\end{array}
\]
\caption{Typing Rules for $\HOp$.\label{fig:typerulesmy}}
\end{figure}

The type relation is defined in \figref{fig:typerulesmy}.
%Types for session names/variables $u$ and
%directly derived from the linear part of the typing
%environment, i.e.~type maps $\Delta$ and $\Lambda$.
Rule $\trule{Session}$ requires the minimal session environment $\Delta$ to type session 
$u$ with type $S$.
Rule $\trule{LVar}$ requires 
the minimal linear environment $\Lambda$ to type 
higher-order variable $x$ with type $\shot{C}$.
Rule $\trule{Shared}$
assigns the value type $U$
to shared names or shared variables $u$ 
if the map $u:U$ exists in environment~$\Gamma$. 
Rule $\trule{Shared}$ also requires that the linear environment is
empty.
The type $\shot{C}$ for shared higher-order values $V$
is derived using rule $\trule{Prom}$, where we require
a value with linear type to be typed without a linear
environment present in order to be used as a shared type.
Rule $\trule{EProm}$ allows to freely use a linear
type variable as shared type variable. 
%A value consisting of a tuple of names/variables is typed using the $\trule{Pol}$ rule,
%which requires theto type and combine each value in the tuple.
Abstraction values are typed with rule $\trule{Abs}$.
The key type for an abstraction is the type for
the bound variables of the abstraction, i.e.,~for
bound variable with type $C$ the abstraction
has type $\lhot{C}$.
The dual of abstraction typing is application typing
governed by rule $\trule{App}$, where we expect
the type $C$ of an application name $u$ 
to match the type $\lhot{C}$ or $\shot{C}$
of the application variable $x$.

A process prefixed with a session send operator $\bout{u}{V} P$
is typed using rule $\trule{Send}$.
The type $U$ of a send value $V$ should appear as a prefix
on the session type $\btout{U} S$ of~$s$.
Rule $\trule{Rcv}$
defines the typing for the 
reception of values $\binp{u}{V} P$.
The type $U$ of a receive value should 
appear as a prefix on the session type $\btinp{U} S$ of $u$.
We use a similar approach with session prefixes
to type interaction between shared channels as defined 
in rules $\trule{Req}$ and $\trule{Acc}$,
where the type of the sent/received object 
($S$ and $L$, respectively) should
match the type of the sent/received subject
($\chtype{S}$ and $\chtype{L}$, respectively).
%In the case of rule $\trule{Req}$ we require
%a duality condition for the communication of session names.
Select and branch prefixes are typed using the rules
$\trule{Sel}$ and $\trule{Bra}$ respectively. Both
rules prefix the session type with the selection
type $\btsel{l_i: S_i}_{i \in I}$ and
$\btbra{l_i:S_i}_{i \in I}$.

The creation of a
shared name $a$ requires to add
its type in environment $\Gamma$ as defined in 
rule \trule{Res}. 
Creation of a session name $s$
creates two endpoints with dual types and adds them to
the session environment 
$\Delta$ as defined in rule \trule{ResS}. 
Rule \trule{Par} concatenates the linear environment of
the parallel components of a parallel operator
to create a type for the composed process.
The disjointness of environments $\Lambda$ and $\Delta$
is implied. Rule \trule{End} allows a form of weakening 
for the session environment $\Delta$, provided that
the name added in $\Delta$ has the inactive
type $\tinact$. The inactive process $\inact$ has an empty
linear environment. The recursive variable is typed
directly from the shared environment $\Gamma$ as
in rule \trule{RVar}.
The recursive operator requires that the body of
a recursive process matches the type of the recursive
variable as in rule \trule{Rec}.

\begin{comment}
\subsection{Order of Types}

In~\cite{tlca07} the type syntax for values includes the definition
$U_1 \sharedop U_2$ and $U_1 \lollipop U_2$, that
allows us to define types of arbitrary order $k$.
An abstraction of $k$-order types requires to extend the syntax
to include higher-order applications:
\[
	\abs{z}{\binp{z}{x} \appl{x}{\abs{y} Q}}
\]
with with the type of $\abs{y}{Q}$ being of order
$k-1$. The type of of such an abstraction in the current setting would
be $\shot{U}$ (or $\lhot{U}$) with the order of the type being defined
as the number of nested higher-order types~\cite{San96int}.

In the type system we develop for the \HOp we only have
types of the form $\shot{C}$.
If we maintain the definition of counting the order
of the type as the nesting of higher-order types we
can still express $k$-order types, e.g:
\[
	\shot{(\btinp{U} \tinact)}
\]
with $U$ being of order $k-1$.
An $k$-order abstraction in \HOp would be:
\[
	\abs{z}{\binp{z}{x} \binp{x}{y} \appl{y}{n}}
\]
with $y$ being of order $k-1$.

\begin{definition}[Order of Value Type]\rm
	\label{def:order_type}
	Let type $U$ and value $V$ such that $\Gamma; \Lambda; \Delta \proves V \hastype U$.
	The order of $U$ is the number of using rule $\trule{Abs}$
	in the typing derivation $\Gamma; \Lambda; \Delta \proves V \hastype U$.
\end{definition}
\end{comment}

\subsection{Type Soundness}

%We state results for type safety:
Type safety result are instances of more general
statements already proved by
Mostrous and Yoshida~\cite{tlca07,MostrousY15} in the asynchronous case.
\begin{lemma}[Substitution Lemma - Lemma C.10 in~\cite{MostrousY15}]\myrm
	\label{lem:subst}
	\begin{enumerate}[1.]
		\item	$\Gamma; \Lambda; \Delta \cat x:S  \proves P \hastype \Proc$ and
			$u \not\in \dom{\Gamma, \Lambda, \Delta}$
			implies
			$\Gamma; \Lambda; \Delta \cat u:S  \proves P\subst{u}{x} \hastype \Proc$.

		\item	$\Gamma \cat x:\chtype{U}; \Lambda; \Delta \proves P \hastype \Proc$ and
			$a \notin \dom{\Gamma, \Lambda, \Delta}$
			implies
			$\Gamma \cat a:\chtype{U}; \Lambda; \Delta \proves P\subst{a}{x} \hastype \Proc$.

		\item	If $\Gamma; \Lambda_1 \cat x:\lhot{C}; \Delta_1  \proves P \hastype \Proc$ 
			and $\Gamma; \Lambda_2; \Delta_2  \proves V \hastype \lhot{C}$ with 
			$\Lambda_1 \cat \Lambda_2$ and $\Delta_1 \cat \Delta_2$ defined,
			then $\Gamma; \Lambda_1 \cat \Lambda_2; \Delta_1 \cat \Delta_2  \proves P\subst{V}{x} \hastype \Proc$.

		\item	$\Gamma \cat x:\shot{C}; \Lambda; \Delta  \proves P \hastype \Proc$ and
			$\Gamma; \emptyset ; \emptyset \proves V \hastype \shot{C}$
			implies
			$\Gamma; \Lambda; \Delta \proves P\subst{V}{x} \hastype \Proc$.
		\end{enumerate}
\end{lemma}
\begin{proof}
	By induction on the typing for $P$, with a case analysis on the last used rule. 
	\qed
\end{proof}

We are interested in session environments 
which are \emph{balanced}:
%that whenever they contain dual endpoints their types are dual:
%
\begin{definition}[Balanced Session Environment]\label{d:wtenv}\myrm
	We say that session environment $\Delta$ is {\em balanced} if
	$s: S_1, \dual{s}: S_2 \in \Delta$ implies $S_1 \dualof S_2$.
\end{definition}
The type soundness relies on the following auxiliary definition:
\begin{definition}[Session Environment Reduction]\myrm
	\label{def:ses_red}
	The reduction relation $\red$ on session environments is defined as:
\[
	\begin{array}{rcl}
		\Delta \cat s: \btout{U} S_1 \cat \dual{s}: \btinp{U} S_2 &\red& \Delta \cat s: S_1 \cat \dual{s}: S_2
		\\
		\Delta \cat s: \btsel{l_i: S_i}_{i \in I} \cat \dual{s}: \btbra{l_i: S_i'}_{i \in I} &\red& \Delta \cat s: S_k \cat \dual{s}: S_k', \quad k \in I
	\end{array}
\]
	We write $\red^\ast$ for the multistep environment reduction.
\end{definition}
We now state the main soundness result as an instance
of type soundness from the system in~\cite{tlca07}.
It is worth noticing that in~\cite{tlca07} has a slightly richer
definition of structural congruence.
Also, their statement for subject reduction relies on an
ordering on typing associated to queues and other 
runtime elements. %(such extended typing is denoted as $\Delta$ by M\&Y).
Since we are dealing with synchronous semantics we can omit such an ordering.
The type soundness result implies soundness for the sub-calculi
\HO, \sessp, and $\CAL^{\minussh}$

\begin{theorem}[Type Soundness - Theorem 7.3 in~\cite{MostrousY15}]\myrm
	\label{thm:sr}
	\begin{enumerate}[1.]
		\item	(Subject Congruence)
			$\Gamma; \es; \Delta \proves P \hastype \Proc$
			and
			$P \scong P'$
			implies
			$\Gamma; \es; \Delta \proves P' \hastype \Proc$.

		\item	(Subject Reduction)
			$\Gamma; \es; \Delta \proves P \hastype \Proc$
			with
			balanced $\Delta$
			and
			$P \red P'$
			implies $\Gamma; \es; \Delta'  \proves P' \hastype \Proc$
			and either (i)~$\Delta = \Delta'$ or (ii)~$\Delta \red \Delta'$
			with $\Delta'$ balanced.
	\end{enumerate}
\end{theorem}

\begin{proof}
	See \appref{app:ts} (Page \pageref{app:ts}).
	\qed
\end{proof}

%\input{mostypes} %% Type system by Mostrous and Yoshida
%\input{typing}

% !TEX root = main.tex
\section{Behavioural Semantics for \HOp}
\label{sec:beh_sem}
\label{sec:behavioural}

\noi We develop a theory for observational equivalence over
session typed \HOp processes.
The theory follows the principles laid by the previous
work of the authors \cite{KYHH2015,KY2015,dkphdthesis}.
We introduce three different bisimilarities and prove
that
all of them coincide with typed, reduction-closed,
barbed congruence. 

\subsection{Labelled Transition Semantics}
\label{subsec:lts}

\myparagraph{Labels.}
We define an (early) typed labelled transition system
$P_1 \by{\ell} P_2$ (LTS for short) over
untyped processes.
Later on, using the \emph{environmental} transition semantics, 
we can define a typed transition relation to formalise 
how a process interacts with a process in its environment.
The interaction
is defined on action $\ell$:
\[
	\begin{array}{rcl}
		\ell	& \bnfis  & \tau 
		\bnfbar \news{\tilde{m}} \bactout{n}{V} 
		\bnfbar\bactinp{n}{V} 
		\bnfbar \bactsel{n}{l} 
		\bnfbar \bactbra{n}{l} 
	\end{array}
\]
\noi The internal action is defined by label $\tau$.
Output action $\news{\tilde{m}} \bactout{n}{V}$ denotes the output of value
$V$ over name $n$ with a possibly empty set of names $\tilde{m}$
being restricted (we may write $\bactout{n}{V}$  when $\tilde{m}$ is empty).
Dually, the action for the value input is $\bactinp{n}{V}$.
We also define actions for selecting a label
$l$, $\bactsel{n}{l}$ and branching on a label
$n$, $\bactbra{s}{l}$.
$\fn{\ell}$ and $\bn{\ell}$ denote 
sets of free/bound names in $\ell$, resp.

The dual action relation is the symmetric relation $\asymp$ that satisfies the rules:
\[
	\bactsel{n}{l} \asymp \bactbra{\dual{n}}{l}
	\qquad
	\news{\tilde{m}'} \bactout{n}{V} \asymp \bactinp{\dual{n}}{V}
\]
Dual actions occur on subjects that are dual between
them and carry the same object. Thus,
output actions are dual to input actions and 
select actions is dual to branch actions.

\begin{figure}[t]
	\[
	\begin{array}{c}
		\appl{(\abs{x}{P})}{u} \by{\tau} P \subst{u}{x}\ \ltsrule{App}
		\qquad
		\bout{n}{V} P \by{\bactout{n}{V}} P\ \ltsrule{Out}
		\qquad
		\binp{n}{x} P \by{\bactinp{n}{V}} P\subst{V}{x}\ \ltsrule{In}
%		\qquad
%		\bout{n}{\abs{x}{Q}} P \by{\bactout{n}{\abs{x}{Q}}} P\ \ltsrule{OutA}
		\\[4mm]

%		\binp{n}{X} P \by{\bactinp{n}{\abs{x}{Q}}} P\subst{\abs{x}Q}{X}\ \ltsrule{InA}
%		\qquad
		\bsel{s}{l}{P} \by{\bactsel{s}{l}} P\ \ltsrule{Sel}
		\qquad
		\tree{
			j \in I
		}
		{
			\bbra{s}{l_i:P_i}_{i \in I} \by{\bactbra{s}{l_j}} P_j
		}\ \ltsrule{Bra}
		\\[6mm]

		\tree{
			P \by{\ell} P' \quad n \notin \fn{\ell}
		}{
			\news{n} P \by{\ell} \news{n} P' 
		}\ \ltsrule{Res}
		\qquad
		\tree{
			P \scong_\alpha P'' \quad P'' \by{\ell} P'
		}{
			P \by{\ell} P'
		}\ \ltsrule{Alpha}
		\qquad
		\tree{
			P \subst{\recp{X}{P}}{\varp{X}} \by{\ell} P'% \qquad P \scong_\alpha P'' \quad P'' \by{\ell} P'
		}{
			\recp{X}{P} \by{\ell} P'
		}\ \ltsrule{Rec}
		\\[6mm]

		\tree{
			P \by{\news{\tilde{m}} \bactout{n}{V}} P' \quad m \in \fn{V}
		}{
			\news{m} P \by{\news{m\cat\tilde{m}} \bactout{n}{V}} P'
		}\ \ltsrule{Scope}
		\qquad
%		\tree{
%			P \by{\news{\tilde{m}} \bactout{n}{\abs{x} Q}} P' \quad m' \in \fn{\abs{x} Q}
%		}{
%			\news{m'} P \by{\news{m'\cat\tilde{m}} \bactout{n}{\abs{x} Q}} P'
%		}\ \ltsrule{ScopeA}
%		\qquad
		\tree{
			P \by{\ell_1} P' \qquad Q \by{\ell_2} Q' \qquad \ell_1 \asymp \ell_2
		}{
			P \Par Q \by{\tau} \newsp{\bn{\ell_1} \cup \bn{\ell_2}}{P' \Par Q'}
		}\ \ltsrule{Tau}
		\\[6mm]

		\tree{

			P \by{\ell} P' \quad \bn{\ell} \cap \fn{Q} = \es
		}{
			P \Par Q \by{\ell} P' \Par Q
		}\ \ltsrule{LPar}
		\qquad
		\tree{
			Q \by{\ell} Q' \quad \bn{\ell} \cap \fn{P} = \es
		}{
			P \Par Q \by{\ell} P \Par Q'
		}\ \ltsrule{RPar}
%		\\[6mm]
	\end{array}
	\]
	\caption{The Untyped (Early) Labelled Transition System. \label{fig:untyped_LTS}}
\end{figure}

\myparagraph{LTS over Untyped Processes.}
The labelled transition system (LTS) over untyped processes
is defined in \figref{fig:untyped_LTS}.
We write $P_1 \by{\ell} P_2$ with the usual meaning.
The rules are standard~\cite{KYHH2015,KY2015}.
An application requires a  silent step $\tau$ to substitute
the application name over the application abstraction as defined
in rule $\ltsrule{App}$.
A process with a send prefix can interact with the environment with a send
action that carries a value $V$ as in rule $\ltsrule{Out}$.
Dually, in rule $\ltsrule{In}$
an input prefixed process can observe a receive action of a value $V$.
Select and branch prefixed processes observe the select
and branch actions in rules $\ltsrule{Sel}$ and $\ltsrule{Bra}$, respectively,
and proceed according to the labels observed.
Rule $\ltsrule{Res}$ closes the LTS under the name creation
operator provided that the restricted name does not occur free in the observable action.
If a restricted name occurs free in an output action 
then the name is added as in the bound name list of the action
and the continuation process performs scope opening as described in rule $\ltsrule{Scope}$.
Rules $\ltsrule{LPar}$ and $\ltsrule{RPar}$ close the LTS under the parallel operator 
provided that the observable action does not shared any bound names with the 
parallel processes.
Rule $\ltsrule{Tau}$ states that if two parallel processes can perform dual actions
then the two actions  can synchronise to observe an internal transition. 
Finally, rule $\ltsrule{Alpha}$ closes  the LTS under alpha-renaming 
and rule $\ltsrule{Rec}$ handles recursion unfolding.

\subsection{Environmental Labelled Transition System}

\begin{figure}[t!]
	\[
%	\footnotesize
	\begin{array}{c}
		\eltsrule{SRv}~~\tree{
			\dual{s} \notin \dom{\Delta} \quad \Gamma; \Lambda'; \Delta' \proves V \hastype U
		}{
			(\Gamma; \Lambda; \Delta \cat s: \btinp{U} S) \by{\bactinp{s}{V}} (\Gamma; \Lambda\cat\Lambda'; \Delta\cat\Delta' \cat s: S)
		}
		\\[8mm]

		\eltsrule{ShRv}~~\tree{
			\Gamma; \es; \es \proves a \hastype \chtype{U}
			\quad
			\Gamma; \Lambda'; \Delta' \proves V \hastype U
		}{
			(\Gamma; \Lambda; \Delta) \by{\bactinp{a}{V}} (\Gamma; \Lambda\cat\Lambda'; \Delta\cat\Delta')
		}
		\\[8mm]

		\eltsrule{SSnd}~~\tree{
			\begin{array}{c}
				\dual{s} \notin \dom{\Delta}
				\qquad 
				\Gamma \cat \Gamma'; \Lambda'; \Delta' \proves V \hastype U
				\qquad
				\tilde{m} = m_1 \dots m_n
				\\
				\Gamma'; \es; \Delta_i \proves m_i \hastype U_i
				\qquad
				\Gamma'; \es; \Delta_i' \proves \dual{m}_i \hastype U'_i
				\qquad
				\Lambda' \subseteq \Lambda
				\qquad
				(\Delta_1 \backslash \bigcup_{i}\Delta_i) \subseteq (\Delta \cat s: S)
			\end{array}
		}{
			(\Gamma; \Lambda; \Delta \cat s: \btout{U} S) \by{\news{\tilde{m}} \bactout{s}{V}} (\Gamma \cat \Gamma'; \Lambda\backslash\Lambda';
			(\Delta \cat s: S \cat \bigcup_{i} \Delta'_i) \backslash \Delta')
		}
		\\[8mm]

		\eltsrule{ShSnd}~~\tree{
			\begin{array}{c}
				\Gamma \cat \Gamma' \cat a: \chtype{U}; \Lambda'; \Delta' \proves V \hastype U
				\qquad
				\tilde{m} = m_1 \dots m_n
				\\
				\Gamma'; \es; \Delta_i \proves m_i \hastype U_i
				\qquad
				\Gamma'; \es; \Delta'_i \proves \dual{m}_i \hastype U_i
				\qquad
				\Lambda' \subseteq \Lambda
				\qquad
				(\Delta_1 \backslash \bigcup_{i}\Delta_i) \subseteq \Delta
			\end{array}
		}{
			(\Gamma \cat a: \chtype{U} ; \Lambda; \Delta) \by{\news{\tilde{m}} \bactout{a}{V}} (\Gamma \cat \Gamma' \cat a: \chtype{U}; \Lambda\backslash\Lambda';
			(\Delta \cat \bigcup_{i}\Delta'_i) \backslash \Delta')
		}
		\\[8mm]
		\eltsrule{Sel}~~\tree{
			\dual{s} \notin \dom{\Delta} \quad j \in I
		}{
			(\Gamma; \Lambda; \Delta \cat s: \btsel{l_i: S_i}_{i \in I}) \by{\bactsel{s}{l_j}} (\Gamma; \Lambda; \Delta \cat s:S_j)
		}
		\\[8mm]
		\eltsrule{Bra}~~\tree{
			\dual{s} \notin \dom{\Delta} \quad j \in I
		}{
			(\Gamma; \Lambda; \Delta \cat s: \btbra{l_i: T_i}_{i \in I}) \by{\bactbra{s}{l_j}} (\Gamma; \Lambda; \Delta \cat s:S_j)
		}
		\\[8mm]

		\eltsrule{Tau}~~\tree{
			\Delta_1 \red \Delta_2 \vee \Delta_1 = \Delta_2
		}{
			(\Gamma; \Lambda; \Delta_1) \by{\tau} (\Gamma; \Lambda; \Delta_2)
		}
	\end{array}
	\]
	\caption{Labelled Transition Semantics for Typed Environments. \label{fig:envLTS}}
\end{figure}

\label{ss:elts}
\noi 
\figref{fig:envLTS}
defines a labelled transition relation between 
a triple of environments, 
denoted
$(\Gamma_1, \Lambda_1, \Delta_1) \by{\ell} (\Gamma_2, \Lambda_2, \Delta_2)$.
It extends the transition systems
in \cite{KYHH2015,KY2015} 
to higher-order sessions. 

\myparagraph{Input Actions} 
are defined by 
$\eltsrule{SRv}$ and $\eltsrule{ShRv}$
%describe the input action
($n$ session or shared name respectively $\bactinp{n}{V}$). 
We require the value $V$ has
the same type as name $s$ and $a$, respectively.  Furthermore we
expect the resulting type tuple to contain the values that consist
with value $V$. The condition $\dual{s} \notin \dom{\Delta}$
in $\eltsrule{SRv}$ ensures that 
the dual name $\dual{s}$ should not be
present in the session environment, since if it were present
the only communication that could take place is the interaction
between the two endpoints (using $\eltsrule{Tau}$ below).

\myparagraph{Output Actions} are defined by $\eltsrule{SSnd}$
and $\eltsrule{ShSnd}$.  
Rule $\eltsrule{SSnd}$ states the conditions for observing action
$\news{\tilde{m}} \bactout{s}{V}$ on a type tuple 
$(\Gamma, \Lambda, \Delta \cdot s: S)$. 
The session environment $\Delta$ with $s: S$ 
should include the session environment of sent value $V$, 
{\em excluding} the session environments of the name $n_j$ 
in $\tilde{m}$ which restrict the scope of value $V$. 
Similarly the linear variable environment 
$\Lambda'$ of $V$ should be included in $\Lambda$. 
Scope extrusion of session names in $\tilde{m}$ requires
that the dual endpoints of $\tilde{m}$ appear in
the resulting session environment. Similarly for shared 
names in $\tilde{m}$ that are extruded.  
All free values used for typing $V$ are subtracted from the
resulting type tuple. The prefix of session $s$ is consumed
by the action.
Similarly, an output on a shared name is described
by rule $\eltsrule{ShSnd}$ where we require that the name
is typed with $\chtype{U}$. Conditions for
the output $V$ are identical to those
% the requirements 
for rule~$\eltsrule{SSnd}$.
We sometimes annotate the output action 
$\news{\tilde{m}} \bactout{n}{V}$
with the type of $V$ as $\news{\tilde{m}} \bactout{n}{V:U}$.

\myparagraph{Other Actions}
Rules $\eltsrule{Sel}$ and $\eltsrule{Bra}$ describe actions for
select and branch. The only requirements for both
rules is that the dual endpoint is not present in the session
environment and the action labels are present
in the type.
Hidden transitions defined by rule $\eltsrule{Tau}$ 
do not change the session environment or they follow the reduction on session
environments (\defref{def:ses_red}).

\begin{proposition}[Environment Transition Weakening]\myrm
	\label{prop:env_tran_weak}
	Consider the LTS for typing environments in \figref{fig:envLTS}.
	If $(\Gamma_1; \Lambda_1; \Delta_1) \hby{\ell} (\Gamma_2; \Lambda_2; \Delta_2)$
	then
	$(\Gamma_2; \Lambda_1; \Delta_1) \hby{\ell} (\Gamma_2; \Lambda_2; \Delta_2)$.
\end{proposition}

\begin{proof}
	The proof 
	is by 
	case analysis
	on the definition of $\hby{\ell}$, 
	exploiting the structural properties (in particular, weakening)
	of shared environment $\Gamma$ (cf. \defref{def:typeenv}).
	\qed
\end{proof}

As a direct consequence of \propref{prop:env_tran_weak}
we can always make an observation on a type environment
without observing a change in the shared environment.

\myparagraph{Typed Transition System}
We define a typed labelled transition system over typed processes,
as a combination of the untyped LTS and the LTS for typed environments (cf. \figref{fig:untyped_LTS} and
\ref{fig:envLTS}):

\begin{definition}[Typed Transition System]\label{d:tlts}\myrm
	\label{def:typed_transition}
	We write
	$\horel{\Gamma}{\Delta_1}{P_1}{\by{\ell}}{\Delta_2}{P_2}$
	whenever
	$P_1 \by{\ell} P_2$,
	$(\Gamma, \emptyset, \Delta_1) \by{\ell} (\Gamma, \emptyset, \Delta_2)$
	and $\Gamma; \emptyset; \Delta_2 \proves P_2 \hastype \Proc$.

	We extend to $\By{}$ 
	and $\By{\hat{\ell}}$ 
	where we write 
	$\By{}$ for the reflexive and
	transitive closure of $\by{}$, $\By{\ell}$ for the transitions
	$\By{}\by{\ell}\By{}$ and $\By{\hat{\ell}}$ for $\By{\ell}$ if
	$\ell\not = \tau$ otherwise $\By{}$.
%	\noi We extend to $\By{}$ and $\By{\hat{\ell}}$ in the standard way.
\end{definition}

%\subsection{Behavioural Semantics}

\subsection{Reduction-Closed, Barbed Congruence}
\label{subsec:rc}

Equivalent processes require a notion of session type confluence,
defined over session environments $\Delta$, following \defref{def:ses_red}:
\begin{definition}[Session Environment Confluence]\myrm
	We denote $\Delta_1 \bistyp \Delta_2$ whenever $\exists \Delta$ such that
	$\Delta_1 \red^* \Delta$ and $\Delta_2 \red^* \Delta$.
\end{definition}
%
%Session environment confluence requires from session environment
%to be able to follow each other.

We define the notion of typed relation over typed processes; it 
includes properties common to all the equivalence relations
that we are going to define:
\begin{definition}[Typed Relation]\myrm
	We say that
	$\Gamma; \emptyset; \Delta_1 \proves P_1 \hastype \Proc\ \Re \ \Gamma; \emptyset; \Delta_2 \proves P_2 \hastype \Proc$
	is a typed relation whenever:
	\begin{enumerate}[i)]
		\item	$P_1$ and $P_2$ are closed processes;
		\item	$\Delta_1$ and $\Delta_2$ are balanced; and
		\item	$\Delta_1 \bistyp \Delta_2$.
	\end{enumerate}
	\noi We write
	$\horel{\Gamma}{\Delta_1}{P_1}{\ \Re\ }{\Delta_2}{P_2}$
	for $\Gamma; \emptyset; \Delta_1 \proves P_1 \hastype \Proc\ \Re\ \Gamma; \emptyset; \Delta_2 \proves P_2 \hastype \Proc$.
\end{definition}

Type relations relate only closed processes
(i.e., processes with no free variables)
with balanced session environments and the two session
environments are confluent.

We define the notions of barb~\cite{MiSa92} and typed barb:
\begin{definition}[Barbs]\myrm
	Let $P$ be a closed process.
	\begin{enumerate}
		\item	We write $P \barb{n}$ if $P \scong \newsp{\tilde{m}}{\bout{n}{V} P_2 \Par P_3}, n \notin \tilde{m}$.
			We write $P \Barb{n}$ if $P \red^* \barb{n}$.

		\item	We write $\Gamma; \emptyset; \Delta \proves P \barb{n}$ if
			$\Gamma; \emptyset; \Delta \proves P \hastype \Proc$ with $P \barb{n}$ and $\dual{n} \notin \Delta$.
			We write $\Gamma; \emptyset; \Delta \proves P \Barb{n}$ if $P \red^* P'$ and
			$\Gamma; \emptyset; \Delta' \proves P' \barb{n}$.			
	\end{enumerate}
\end{definition}
A barb $\barb{n}$ is an observable on an output prefix with subject $n$.
Similarly a weak barb $\Barb{n}$ is a barb after a number of reduction steps.
Typed barbs $\barb{n}$ (resp.\ $\Barb{n}$)
occur on typed processes $\Gamma; \emptyset; \Delta \proves P \hastype \Proc$
where we require that whenever $n$ is a session name, then the corresponding 
dual endpoint $\dual{n}$ is not present in the session type $\Delta$.

To define a congruence relation we define the notion of the context $\C$:
\begin{definition}[Context]\myrm
	A context $\C$ is defined on the grammar:
\[
	\begin{array}{rcl}
		\C &\bnfis& \hole \bnfbar \bout{u}{V} \C \bnfbar \bout{u}{\abs{x}{\C}} P \bnfbar \binp{u}{x} \C \bnfbar \recp{X}{\C} \bnfbar \appl{(\abs{x}{\C})}{u}\\
		& \bnfbar & \news{n} \C \bnfbar \C \Par P \bnfbar P \Par \C \bnfbar \bsel{u}{l} \C \bnfbar \bbra{k}{l_1: P_1, \cdots, l_i:\C, \cdots, l_n: P_n}
	\end{array}
\]
	Notation $\context{\C}{P}$ replaces every hole $\hole$ in $\C$ with $P$.
\end{definition}
\noi A context is a function that takes a process and returns a new process
according to the above syntax.

The first behavioural relation we define is reduction-closed, barbed congruence:
\begin{definition}[Reduction-closed, Barbed Congruence]\myrm
	Typed relation
	$\horel{\Gamma}{\Delta_1}{P_1}{\ \Re\ }{\Delta_2}{P_2}$
	is a barbed congruence whenever:
	\begin{enumerate}
		\item
		\begin{enumerate}[-]
			\item	If $P_1 \red P_1'$ then there exist $P_2', \Delta_2'$ such that $P_2 \red^* P_2'$ and
				$\horel{\Gamma}{\Delta_1'}{P_1'}{\ \Re\ }{\Delta_2'}{P_2'}$
			\item	If $P_2 \red P_2'$ then there exist $P_1', \Delta_1'$ such that $P_1 \red^* P_1'$ and
				$\horel{\Gamma}{\Delta_1'}{P_1'}{\ \Re\ }{\Delta_2'}{P_2'}$
		\end{enumerate}

		\item
		\begin{enumerate}[-]
			\item	If $\Gamma;\emptyset;\Delta_1 \proves P_1 \barb{s}$ then $\Gamma;\emptyset;\Delta_2 \proves P_2 \Barb{s}$.
			\item	If $\Gamma;\emptyset;\Delta_2 \proves P_2 \barb{s}$ then $\Gamma;\emptyset;\Delta_1 \proves P_1 \Barb{s}$.
		\end{enumerate}

		\item	$\forall \C$, then there exist $\Delta_1'', \Delta_2''$ such that $\horel{\Gamma}{\Delta_1''}{\context{\C}{P_1}}{\ \Re\ }{\Delta_2''}{\context{\C}{P_2}}$
	\end{enumerate}
	The largest such congruence is denoted with $\cong$.
\end{definition}
\noi Reduction-closed, barbed congruence is closed under reduction semantics and 
preserves barbs under any context, i.e.,~no barb observer can distinguish
between two related processes.

\subsection{Context Bisimulation}

The second behavioural relation we define is the labelled
characterisation of reduction-closed, barbed congruence,
called \emph{context bisimulation}~\cite{San96H}:
\begin{definition}[Context Bisimulation]\myrm
	\label{def:context_bis}
	Typed relation 
	$\Re$ is a {\em context bisimulation} if for all
	$\horel{\Gamma}{\Delta_1}{P_1}{\ \Re\ }{\Delta_2}{P_2}$,
	\begin{enumerate}[1.]
		\item Whenever	%$\forall \news{\tilde{m_1}} \bactout{n}{V_1}$ such that
			$
				\horel{\Gamma}{\Delta_1}{P_1}{\by{\news{\tilde{m_1}} \bactout{n}{V_1}}}{\Delta_1'}{P_2}
			$ there exist $Q_2$, $V_2$, and $\Delta'_2$ such that
			\[
				\horel{\Gamma}{\Delta_2}{Q_1}{\By{\news{\tilde{m_2}} \bactout{n}{V_2}}}{\Delta_2'}{Q_2}
			\]
			and $\forall R$ with $\set{x} = \fv{R}$, 
			then
			\[
				\horel{\Gamma}{\Delta_1''}{\newsp{\tilde{m_1}}{P_2 \Par R\subst{V_1}{x}}}
				{\ \Re\ }
				{\Delta_2''}{\newsp{\tilde{m_2}}{Q_2 \Par R\subst{V_2}{x}}}.
			\]
		\item For all 
			$
				\horel{\Gamma}{\Delta_1}{P_1}{\by{\ell}}{\Delta_1'}{P_2}
			$
			such that $\ell \not= \news{\tilde{m}} \bactout{n}{V}$, there exist 
			$Q_2$ and $\Delta'_2$ such that
			\[
				\horel{\Gamma}{\Delta_2}{Q_1}{\By{\hat{\ell}}}{\Delta_2'}{Q_2}
			\]
			and
			$\horel{\Gamma}{\Delta_1'}{P_2}{\ \Re\ }{\Delta_2'}{Q_2}$.

		\item	The symmetric cases of 1 and 2.
	\end{enumerate}
	The Knaster-Tarski theorem ensures that the largest context bisimulation exists,
	it is called \emph{context bisimilarity} and is denoted by $\wbc$.
\end{definition}

\subsection{Higher-Order Bisimulation and Characteristic Bisimulation ($\wb/\wbf$)}
\label{subsec:char_bis}

In the general case, contextual bisimulation 
is a hard relation to compute due to:
\begin{enumerate}[i)]
	\item	the universal quantifier over contexts in the output case
		(Clause~1 in \defref{def:context_bis}); and

	\item	a higher order input prefix can observe infinitely many
		different input actions, since infinitely many different
		processes can match the session type of an input prefix.
\end{enumerate}
To reduce the burden of the contextual bisimulation
%induced by 
%universal quantification,
%we introduce \emph{higher-order} and \emph{characteristic}  
%bisimulations, two tractable equivalences denoted  $\hwb$ and $\fwb$, respectively.
we take the following two steps: 
\begin{enumerate}[(a)]
	\item	we replace Clause~1 in \defref{def:context_bis} with a clause
		involving a more tractable process closure; and
	\item	we refine the transition rule for input in the LTS
		so to define a bisimulation relation
		without observing infinitely many actions on the same input prefix. 
\end{enumerate}
%

%the set of higher order input actions
%on an input prefix is infite since infinitely
%many processes can match the session type of
%the input prefix.

%The output action case of the contextual bisimulation
%can be equivalently be defined as:

\myparagraph{Trigger Processes with Session Communication.}
Concerning~(a), we exploit session types. 
First observe that closure $R\subst{V}{x}$ 
in Clause~1 in \defref{def:context_bis}
is context bisimilar to the process:
\begin{eqnarray}
	\label{equ:1}
	P = \newsp{s}{\appl{(\abs{z}{\binp{z}{x}{R}})}{s} \Par \bout{\dual{s}}{V} \inact}
\end{eqnarray}
\noi
In fact, we do have $P \wbc R\subst{V}{x}$, since 
application and session transitions are deterministic.  
Now let us consider process $T_{V}$ below,
where $t$ is a fresh name:
\begin{eqnarray}
	\label{equ:0}
	T_{V} = \hotrigger{t}{x}{s}{V}
\end{eqnarray}
\noi
Process $T_{V}$ can input the class of
abstractions $\abs{z}{\binp{z}{x} R}$ and
can simulate the closure of~\eqref{equ:1}:%$P$:
\begin{eqnarray}
	\label{equ:2}
	T_{V} \by{\bactinp{t}{\abs{z}{\binp{z}{x}R}}} P \wbc R \subst{V}{x}
\end{eqnarray}
%s
Processes such as $T_{V}$
input a value at a fresh name;
we will use this class of 
{\bf\em trigger processes} to define a
refined bisimilarity without the demanding 
output Clause~1 in \defref{def:context_bis}.
Given a fresh name $t$, we write:
\[
	\htrigger{t}{V} = \hotrigger{t}{x}{s}{V}
\]
%
% to 
%stand for a trigger process $T_{V}$ for value $V$.
We note that in contrast to previous
approaches~\cite{SaWabook,JeffreyR05} 
our {trigger processes} do {\em not}
use recursion or replication.
This is crucial to preserve linearity of session names.

\myparagraph{Characteristic Processes and Values.}
Concerning point (b), we limit the possible input abstractions
$\abs{x} P$ by exploiting session types.
We introduce the key concept of {\bf \emph{characteristic process/value}},
which is the 
simplest process/value that can inhabit a type.
As an example, consider $S = \btinp{\shot{S_1}} \btout{S_2} \tinact$. Type $S$ is a 
session type that
first inputs an abstraction (from type $S_1$ to a process), then 
outputs a value of type $S_2$, and terminates.
Then, the following process:
\[
	Q = \binp{u}{x} (\bout{u}{s_2} \inact \Par \appl{x}{s_1})
\]
\noi is a characteristic process for $S$ along name $u$.
In fact, it is easy to see that $Q$ is well-typed by session type $S$.
The following definition formalizes this intuition.
\begin{definition}[Characteristic Process]\myrm
	\label{def:characteristic_process}
	\noi Let name $u$ and type $U$. Then we define the {\em characteristic process}:
	$\mapchar{U}{u}$ and the {\em characteristic value} $\omapchar{U}$ as:
	\[
	\begin{array}{cc}
		\begin{array}{rclcl}
			\mapchar{\btinp{U} S}{u} &\defeq& \binp{u}{x} (\mapchar{S}{u} \Par \mapchar{U}{x})
			\\
			\mapchar{\btout{U} S}{u} &\defeq& \bout{u}{\omapchar{U}} \mapchar{S}{u} %& & n \textrm{ fresh}
			\\
			\mapchar{\btsel{l : S}}{u} &\defeq& \bsel{u}{l} \mapchar{S}{u}
			\\
			\mapchar{\btbra{l_i: S_i}_{i \in I}}{u} &\defeq& \bbra{u}{l_i: \mapchar{S_i}{u}}_{i \in I}
			\\
			\mapchar{\tvar{t}}{u} &\defeq& \varp{X}_{\vart{t}}
			\\
			\mapchar{\trec{t}{S}}{u} &\defeq& \recp{X_{\vart{t}}}{\mapchar{S}{u}}
			\\
			\mapchar{\tinact}{u} &\defeq& \inact
			\\
		\end{array}
		&
		\begin{array}{rcrclcl}
			&&\mapchar{\chtype{S}}{u} &\defeq& \bout{u}{\omapchar{S}} \inact
			\\
			&&\mapchar{\chtype{L}}{u} &\defeq& \bout{u}{\omapchar{L}} \inact
			\\
			\mapchar{\shot{C}}{x} &\defeq& \mapchar{\lhot{C}}{x} &\defeq& \appl{x}{\omapchar{C}}
			\\
			\\
			&&\omapchar{S} &\defeq& s && s \textrm{ fresh}
			\\
			\omapchar{\chtype{S}} &\defeq& \omapchar{\chtype{L}} &\defeq& a && a \textrm{ fresh}
			\\
			\omapchar{\shot{C}} &\defeq& \omapchar{\lhot{C}} &\defeq& \abs{x}{\mapchar{C}{x}}
		\end{array}
	\end{array}
	\]
\end{definition}
%
%
%Given a value type $U$, we write $\mapchar{U}{u} $
%for its characteristic process along name $u$.
%Similarly, given value type $U$, we write 
%$\omapchar{U}$ to denote its characteristic value.
\begin{proposition}\myrm
Characteristic processes and values are inhabitants of their associated type:
$ $
	\begin{enumerate}[$\bullet$]
		\item	$\Gamma; \emptyset; \Delta \cat u:S \proves \mapchar{S}{u} \hastype \Proc$
		\item	$U = \chtype{S}$ or $U = \chtype{L}$ implies $\Gamma \cat u:U; \emptyset; \Delta \proves \mapchar{U}{u} \hastype \Proc$
		\item	$\Gamma; \es; \Delta \proves \omapchar{U} \hastype U$
	\end{enumerate}
\end{proposition}

\begin{proof}
	By induction on the definition of $\mapchar{S}{u}$ and $\mapchar{U}{u}$.
	\qed
\end{proof}

\begin{corollary}\myrm
	If $\Gamma; \emptyset; \Delta \proves \mapchar{C}{u} \hastype \Proc$
	then
	$\Gamma; \es; \Delta \proves u \hastype C$.
\end{corollary}

We use the characteristic value $\omapchar{U}$
to limit input transitions.
Following the same reasoning as (\ref{equ:1})--(\ref{equ:2}), 
we can define an alternative trigger process, called
{\bf\em characteristic trigger process} with type 
$U$ to replace Clause~1 in \defref{def:context_bis}.
\begin{eqnarray}
	\label{eq:4}
	\ftrigger{t}{V}{U} \defeq \fotrigger{t}{x}{s}{\btinp{U} \tinact}{V}
\end{eqnarray}
\noi 
Thus, in contrast to the trigger process in~\eqref{equ:0},
the characteristic trigger process in~\eqref{eq:4}
does not involve a higher-order communication on $t$.

To refine the input transition system, we need to observe 
an additional value:
\[
	\abs{{x}}{\binp{t}{y} (\appl{y}{{x}})}
\]
called the {\bf\em trigger value}.
This is necessary, because it turns out
that a characteristic value 
alone as the observable input 
is not enough to define a sound bisimulation.
Roughly speaking, the trigger value is used
to observe/simulate application processes.

The intuition for usage of the trigger is
demonstrated in the next example.
\begin{example}
	\label{ex:motivation}
	First we demonstrate that observing a characteristic value
	input alone is not sufficient
	to define a sound bisimulation closure.
	Consider typed processes $P_1, P_2$:
	\begin{eqnarray}
		P_1 = \binp{s}{x} (\appl{x}{s_1} \Par \appl{x}{s_2}) 
		\qquad \qquad
		P_2 = \binp{s}{x} (\appl{x}{s_1} \Par \binp{s_2}{y} \inact) 
		\label{equ:6}
	\end{eqnarray}
	with
	\[
		\Gamma; \es; \Delta \cat s: \btinp{\shot{(\btinp{C} \tinact)}} \tinact \proves P_i \hastype \Proc \qquad (i \in \{1,2\}).
	\]
	If the above processes input and substitute over $x$
	the characteristic value 
	\[
		\omapchar{\shot{(\btinp{C} \tinact)}} = \abs{x}{\binp{x}{y} \inact}
	\] 
	\noi then both processes evolve into:
	\begin{eqnarray*}
		\Gamma; \es; \Delta \proves \binp{s_1}{y} \inact \Par \binp{s_2}{y} \inact \hastype \Proc
	\end{eqnarray*}
	\noi therefore becoming  context bisimilar.
	However, the processes in~(\ref{equ:6})
	are clearly {\em not} context bisimilar:
	there exist many input actions
	which may be used to distinguish them.
	For example, if 
	$P_1$ and $P_2$ input
	\[
		\abs{x}{\newsp{s_3}{\bout{a}{s_3} \binp{x}{y} \inact}}
	\]
	\noi with
	$\Gamma; \es; \Delta \proves s \hastype \tinact$,
	then their derivatives are not bisimilar.

	Observing only the characteristic value 
	results in an over-discriminating bisimulation.
	However, if a trigger value, 
	$\abs{{x}}{\binp{t}{y} (\appl{y}{{x}})}$
	is received on $s$, 
	then we can distinguish 
	processes in \eqref{equ:6}:  
	\begin{eqnarray*}
		\horel{\Gamma}{\Delta}{P_1}{&\By{\bactinp{s}{\abs{{x}}{\binp{t}{y} (\appl{y}{{x}})}}}&}{\Delta'}{\binp{t}{x} (\appl{x}{s_1}) \Par \binp{t}{x} (\appl{x}{s_2})}
		\\
		\horel{\Gamma}{\Delta}{P_2}{&\By{\bactinp{s}{\abs{{x}}{\binp{t}{y} (\appl{y}{{x}})}}}&}{\Delta''}{\binp{t}{x} (\appl{x}{s_1}) \Par \binp{s_2}{y} \inact}
	\end{eqnarray*}

	One question that arises here is whether the trigger value is enough
	to distinguish two processes, hence no need of 
	characteristic values as the input. 
	This is not the case since the trigger value
	alone also results in an over-discriminating bisimulation relation.
	In fact the  trigger value can be observed on any input prefix
	of {\em any type}. For example, consider the following processes:
	\begin{eqnarray}
		\Gamma; \es; \Delta \proves \newsp{s}{\binp{n}{x} (\appl{x}{s}) \Par \bout{\dual{s}}{\abs{x} P} \inact} \hastype \Proc\label{equ:7}
		\\
		\Gamma; \es; \Delta \proves \newsp{s}{\binp{n}{x} (\appl{x}{s}) \Par \bout{\dual{s}}{\abs{x} Q} \inact} \hastype \Proc\label{equ:8}
	\end{eqnarray}
	\noi if processes in \eqref{equ:7}/\eqref{equ:8}
	input the trigger value, we obtain processes:
	\begin{eqnarray*}
		\Gamma; \es; \Delta' \proves  \newsp{s}{\binp{t}{x} (\appl{x}{s}) \Par \bout{\dual{s}}{\abs{x} P} \inact} \hastype \Proc
		\\
		\Gamma; \es; \Delta' \proves  \newsp{s}{\binp{t}{x} (\appl{x}{s}) \Par \bout{\dual{s}}{\abs{x} Q} \inact} \hastype \Proc
	\end{eqnarray*}

	\noi thus we can easily derive a bisimulation closure if we 
	assume a bisimulation definition that allows only trigger value input.

	But if processes in \eqref{equ:7}/\eqref{equ:8}
	input the characteristic value $\abs{z}{\binp{z}{x} (\appl{x}{m})}$,  
	then they would become:
	\begin{eqnarray*}
		\Gamma; \es; \Delta \proves \newsp{s}{\binp{s}{x} (\appl{x}{m}) \Par \bout{\dual{s}}{\abs{x} P} \inact} \wbc \Delta \proves P \subst{m}{x}
		\\
		\Gamma; \es; \Delta \proves \newsp{s}{\binp{s}{x} (\appl{x}{m}) \Par \bout{\dual{s}}{\abs{x} Q} \inact} \wbc \Delta \proves Q \subst{m}{x}
	\end{eqnarray*}
	\noi which are not bisimilar if $P \subst{m}{x} \not\wb Q \subst{m}{x}$.
\end{example}

\noi We now define the \emph{refined} typed LTS. 
The new LTS is defined by considering a transition
rule for input in which admitted values are
trigger or characteristic values:
We formalise the restricted input action with the
definition of a new environment transition relation:
\[
	(\Gamma, \Lambda_1, \Delta_1) \hby{\ell} (\Gamma, \Lambda_2, \Delta_2)
\]
\noi The new rule is defined on top of the rules in~\figref{fig:envLTS}:
%The definition uses the same rules for defining the $\by{\ell}$ relation,
%but it requires the substitution of rules $\eltsrule{SRv}$ and
%$\eltsrule{ShRv}$ with the rule:
\begin{definition}[Refined Input Environment LTS]\myrm
\[
	\eltsrule{RRv}~~\tree {
		(\Gamma_1; \Lambda_1; \Delta_1) \by{\bactinp{n}{V}} (\Gamma_2; \Lambda_2; \Delta_2)
		\quad
		\begin{array}{crcl}
			& (V & \scong & \abs{z}{\binp{t}{x} (\appl{x}{z})} \wedge t \textrm{ fresh}) \\
			\vee & (V & \scong &  \omapchar{U}) \vee V = m
		\end{array}
	}{
		(\Gamma_1; \Lambda_1; \Delta_1) \hby{\bactinp{n}{V}} (\Gamma_2; \Lambda_1; \Delta_2)
	}
\]
\end{definition}
\noi 
Rule $\eltsrule{RRv}$ refines the input action to carry only
a characteristic value (fresh name or abstraction)
or a trigger value on a fresh name $t$.
This rule is defined on top of rules
$\eltsrule{SRv}$ and $\eltsrule{ShRv}$
in~\figref{fig:envLTS}.
% as a higher order input.
%Input of names remains as defined in the rules 
%$\eltsrule{SRv}$ and $\eltsrule{ShRv}$.
%Rule $\eltsrule{RRcv}$ restricts the higher-order input
%in relation $\hby{\ell}$;
%only characteristic processes and trigger processes
%are allowed to be received on a higher-order input.
%Non fresh names can still be received as in the definition of
%the $\by{\ell}$ relation.
%The conditions for input follow the conditions
%for the $\by{\ell}$ definition.
The new environment transition system $\hby{\ell}$
uses %
rule~$\eltsrule{RRv}$ as input rule.
 All other defining cases
of environment LTS $\hby{\ell}$ remain the same
as in~\figref{fig:envLTS}.

The new typed relation derived from the $\hby{\ell}$ environment LTS is
defined as:
\begin{definition}[Restricted Typed Transition]\myrm 
	\label{def:restricted_typed_transition}
	We write
	$\horel{\Gamma}{\Delta_1}{P_1}{\hby{\ell}}{\Delta_2}{P_2}$
	whenever
	$P_1 \by{\ell} P_2$,
	$(\Gamma, \emptyset, \Delta_1) \hby{\ell} (\Gamma, \emptyset, \Delta_2)$
	and
	$\Gamma; \emptyset; \Delta_2 \proves P_2 \hastype \Proc$.

	\noi We extend to $\Hby{}$ and $\Hby{\hat{\ell}}$ in the standard way.
\end{definition}

\begin{lemma}[Invariant]\label{l:invariant}
	If $\horel{\Gamma}{\Delta_1}{P_1}{\hby{\ell}}{\Delta_2}{P_2}$
	then $\horel{\Gamma}{\Delta_1}{P_1}{\by{\ell}}{\Delta_2}{P_2}$.
\end{lemma}

\begin{proof}
	The proof is straightforward from the definition of rule $\eltsrule{RRv}$.
\end{proof}

%As we discussed previously it is convenient to define
%the output action case of the bisimulation using the
%trigger process:
%%
%\[
%	\hotrigger{t}{x}{s}{V}
%\]
%%
%\noi for substituting value $\Gamma; \es; \Delta \proves V \hastype U$
%over an observer context.
%An alternative trigger that uses a first-order input prefix,
%can be defined, which is parametrised on the type $U$ of $V$:
%%
%\[
%	\fotrigger{t}{x}{s}{\btinp{U} \inact}{V}
%\]
%%
%Both triggers reduce in a deterministic way under
%typed transition $\hby{\ell}$, due to the
%freshness of $t$ and the restricted session $s$. Furthermore,
%both trigger transitions result in the process
%$\mapchar{U}{x} \subst{V}{x}$.

The next definition formalises the notion of a trigger process.
\begin{definition}[Trigger Process]\myrm
Let $t$, $V$, and $U$ be a name, a value, and a type, respectively. We have:
	\begin{center}
		\begin{tabular}{rclcl}
			Trigger Process & & $\htrigger{t}{V}$ &$\defeq$& $\hotrigger{t}{x}{s}{V}$\\
			Characteristic Trigger Process & & $\ftrigger{t}{V}{U}$ &$\defeq$& $\fotrigger{t}{x}{s}{\btinp{U} \tinact}{V}$
		\end{tabular}
	\end{center}
\end{definition}

\myparagraph{The Two Bisimulations.} We now define 
higher-order bisimulation, 
a more tractable bisimulation for $\HO$ and $\HOp$.
The two bisimulations differ on the fact that
they use the different 
trigger processes: $\htrigger{t}{V}$ and $\ftrigger{t}{V}{U}$.

\begin{definition}[Higher-Order Bisimulation]\myrm
	\label{def:bisim}
	Typed relation
	$\Re$ is a {\em higher-Order bisimulation} if for all
	$\horel{\Gamma}{\Delta_1}{P_1}{\ \Re\ }{\Delta_2}{Q_1}$, % implies:
	\begin{enumerate}[1.]
		\item	%$\forall \news{\tilde{m_1}} \bactout{n}{V_1}$ such that
		   Whenever 
			$
				\horel{\Gamma}{\Delta_1}{P_1}{\hby{\news{\tilde{m_1}} \bactout{n}{V_1}}}{\Delta_1'}{P_2}
			$
			there exist $Q_2$, $V_2$, $\Delta_2'$ such that
			\[
				\horel{\Gamma}{\Delta_2}{Q_1}{\Hby{\news{\tilde{m_2}} \bactout{n}{V_2}}}{\Delta_2'}{Q_2}
			\]
			and, for a fresh $t$, 
			\[
				\horel{\Gamma}{\Delta_1''}{\newsp{\tilde{m_1}}{P_2 \Par \htrigger{t}{V_1}}}
				{\ \Re\ }
				{\Delta_2''}{}{\newsp{\tilde{m_2}}{Q_2 \Par \htrigger{t}{V_2}}}.
			\]
		\item	For all 
			$
				\horel{\Gamma}{\Delta_1}{P_1}{\hby{\ell}}{\Delta_1'}{P_2}
			$
			such that $\ell \not= \news{\tilde{m}} \bactout{n}{V}$, there exist
			 $\exists Q_2$ and $\Delta_2'$ such that 
			\[
				\horel{\Gamma}{\Delta_1}{Q_1}{\Hby{\hat{\ell}}}{\Delta_2'}{Q_2}
			\]
			and
			$\horel{\Gamma}{\Delta_1'}{P_2}{\ \Re\ }{\Delta_2'}{Q_2}$.

		\item	The symmetric cases of 1 and 2.
	\end{enumerate}
	The Knaster-Tarski theorem ensures that the largest higher-order bisimulation exists;
	it is called \emph{higher-order bisimilarity} and is denoted by $\wb$.
\end{definition}
The higher-order bisimulation definition uses higher order input guarded triggers,
thus it cannot be used as an equivalence relation for the \sessp sub-calculus.
An alternative definition of the bisimulation---based on characteristic
output triggers---solves this problem.

\begin{definition}[Characteristic Bisimulation]\myrm
	\label{def:cbisim}
	Typed relation 
	$\Re$ is a {\em characteristic bisimulation} if whenever
	$\horel{\Gamma}{\Delta_1}{P_1}{\ \Re\ }{\Delta_2}{Q_1}$ implies:
	\begin{enumerate}
		\item Whenever 
		 %$\forall \news{\tilde{m_1}} \bactout{n}{V_1: U}$ such that %with \dk{$\Gamma; \es; \Delta \proves V_1 \hastype U$} such that
			$
				\horel{\Gamma}{\Delta_1}{P_1}{\hby{\news{\tilde{m_1}} \bactout{n}{V_1: U}}}{\Delta_1'}{P_2}
			$
			there exist $Q_2$, $V_2$, and $\Delta_2'$ such that %with \dk{$\Gamma; \es; \Delta \proves V_2 \hastype U$} such that
			\[
				\horel{\Gamma}{\Delta_2}{Q_1}{\Hby{\news{\tilde{m_2}} \bactout{n}{V_2: U}}}{\Delta_2'}{Q_2}
			\]
			and, for a fresh $t$, 
			\[
				\horel{\Gamma}{\Delta_1''}{\newsp{\tilde{m_1}}{P_2 \Par \ftrigger{t}{V_1}{U}}}
				{\ \Re\ }
				{\Delta_2''}{}{\newsp{\tilde{m_2}}{Q_2 \Par \ftrigger{t}{V_2}{U}}}.
			\]
		\item For all 
			$
				\horel{\Gamma}{\Delta_1}{P_1}{\hby{\ell}}{\Delta_1'}{P_2}
			$
			such that $\ell \not= \news{\tilde{m}} \bactout{n}{V}$, there exist 
			  $\exists Q_2$ and $\Delta_2'$ such that 
			\[
				\horel{\Gamma}{\Delta_1}{Q_1}{\hat{\Hby{\ell}}}{\Delta_2'}{Q_2}
			\]
			and
			$\horel{\Gamma}{\Delta_1'}{P_2}{\ \Re\ }{\Delta_2'}{Q_2}$.

		\item	The symmetric cases of 1 and 2.
	\end{enumerate}
	The Knaster-Tarski theorem ensures that the largest bisimulation exists;
	it is called \emph{characteristic bisimilarity} and is denoted by $\wbf$.
\end{definition}

The next result clarifies our choice of restricting
higher-order input actions with input triggers and
characteristic processes: if 
two processes $P$ and $Q$ are bisimilar 
under the substitution
of the characteristic abstraction and the trigger
input, then 
$P$ and $Q$ 
are bisimilar under any abstraction substitution.

\begin{lemma}[Process Substitution]\myrm
	\label{lem:proc_subst}
	If 
	\begin{enumerate}
		\item	$\horel{\Gamma}{\Delta_1'}{P \subst{\abs{z}{\binp{t}{y} (\appl{y}{z})}}{x}}{\wb}{\Delta_2'}{Q \subst{\abs{z}{\binp{t}{y} (\appl{y}{z})}}{x}}$,
			for some fresh $t$.

		\item	$\horel{\Gamma}{\Delta_1''}{P \subst{\omapchar{U}}{x}}{\wb}{\Delta_2''}{Q \subst{\omapchar{U}}{x}}$, 
			for some $U$.
	\end{enumerate}
	then $\forall R$ such that $\fv{R} = z$
\[
	\horel{\Gamma}{\Delta_1}{P \subst{\abs{z}{R}}{x}}{\wb}{\Delta_2}{Q \subst{\abs{z}{R}}{x}}
\]
\end{lemma}

\begin{proof}
	The details of the proof can be found in \lemref{app:lem:proc_subst} (Page~\pageref{app:lem:proc_subst}).
	\qed
\end{proof}

We now state our main theorem: typed bisimilarities collapse.
The following theorem justifies our choices
for the bisimulation relations, since
they coincide between them and they also
coincide with reduction closed, barbed congruence.

\begin{theorem}[Coincidence]\myrm
	\label{the:coincidence}
	Relations $\wbc, \fwb, \hwb$ and $\cong$ coincide.
\end{theorem}

\begin{proof}
	The full details of the proof are in \appref{app:sub_coinc}. There, 
	the proof is split into a series of lemmas:
	\begin{enumerate}[$-$]
\item	\lemref{lem:wb_eq_wbf} establishes $\wb\ =\ \wbf$.
\item	\lemref{lem:wb_is_wbc} exploits the process substitution result
	(\lemref{lem:proc_subst}) to prove that $\wb \subseteq \wbc$.
\item	\lemref{lem:wbc_is_cong} shows that $\wbc$ is a congruence
	which implies $\wbc \subseteq \cong$.
\item	\lemref{lem:cong_is_wb} shows  that $\cong \subseteq \wb$, using
	the technique developed in~\cite{Hennessy07}.
\end{enumerate}
%	
%	\lemref{lem:wb_eq_wbf} establishes $\wb\ =\ \wbf$.
%	\lemref{lem:wb_is_wbc} exploits the process substitution result
%	(\lemref{lem:proc_subst}) to prove that $\wb \subseteq \wbc$.
%	\lemref{lem:wbc_is_cong} shows that $\wbc$ is a congruence
%	which implies $\wbc \subseteq \cong$.
%	The final result comes from \lemref{lem:cong_is_wb} where
%	we use label testing to show that $\cong \subseteq \wb$ using
%	the technique developed in~\cite{Hennessy07}. 
	The formulation of input
	triggers in the bisimulation relation allows us to prove
	the latter result without using a matching operator.
	\qed
\end{proof}

%\begin{definition}[Session Transition]\myrm
%	Let well-typed \HOp process $\Gamma; \es; \Delta \proves P \hastype \Proc$.
%	We write {\em session transition} 
%	$\horel{\Gamma}{\Delta}{P}{\shby{\ell}}{\Delta'}{P'}$ if
%	$\horel{\Gamma}{\Delta}{P}{\hby{\ell}}{\Delta'}{P'}$
%	is derived by not using the environment label transition rules $\eltsrule{ShOut}$
%	and $\eltsrule{ShRv}$.
%\end{definition}
%
%\begin{lemma}[$\tau$-inertness]\myrm
%	\label{lem:tau_inert}
%	Let $\Gamma; \es; \Delta \proves P \hastype \Proc$.
%	$\horel{\Gamma}{\Delta}{P}{\shby{\ell}}{\Delta'}{P'}$
%	implies
%	$\horel{\Gamma}{\Delta}{P}{\wb}{\Delta'}{P'}$
%\end{lemma}

%A transition $\horel{\Gamma}{\Delta}{P}{\hby{\tau}}{\Delta'}{P'}$ is said
%		{\em deterministic} if it is derived using~$\ltsrule{App}$ or~$\ltsrule{Tau}$ 
%		(where $\subj{\ell_1}$ and $\subj{\ell_2}$ in the premise 
%		are dual endpoints), 
%		possibly followed by uses of  $\ltsrule{Alpha}$, $\ltsrule{Res}$, $\ltsrule{Rec}$, or $\ltsrule{Par${}_L$}/
%		\ltsrule{Par${}_R$}$.

We now define internal deterministic transitions as those associated to session synchronizations or to 
$\beta$-reductions: 
		
\begin{definition}[Deterministic Transition]\myrm
\label{def:dettrans}
	Let  $\Gamma; \es; \Delta \proves P \hastype \Proc$ be a balanced \HOp process. 
	Transition $\horel{\Gamma}{\Delta}{P}{\hby{\tau}}{\Delta'}{P'}$ is called:
	\begin{enumerate}[$-$]
		\item {\em Session transition} whenever the untyped transition $P \by{\tau} P'$ 
		is derived using 
			rule~$\ltsrule{Tau}$ 
		(where $\subj{\ell_1}$ and $\subj{\ell_2}$ in the premise 
		are dual endpoints), 
		possibly followed by uses of  $\ltsrule{Alpha}$, $\ltsrule{Res}$, $\ltsrule{Rec}$, or $\ltsrule{Par${}_L$}/
		\ltsrule{Par${}_R$}$.
		
		\item	{\em \betatran}	whenever the untyped transition $P \by{\tau} P'$
			is derived using rule $\ltsrule{App}$,
			possibly followed by uses of  $\ltsrule{Alpha}$, $\ltsrule{Res}$, $\ltsrule{Rec}$, or $\ltsrule{Par${}_L$}/
		\ltsrule{Par${}_R$}$.
	\end{enumerate}
We write
$\horel{\Gamma}{\Delta}{P}{\hby{\stau}}{\Delta'}{P'}$
and 
$\horel{\Gamma}{\Delta}{P}{\hby{\btau}}{\Delta'}{P'}$
to denote session and $\beta$-transitions, resp. Also, 
	 $\horel{\Gamma}{\Delta}{P}{\hby{\dtau}}{\Delta'}{P'}$ denotes
	either a session transition or a \betatran.
\end{definition}

Deterministic transitions imply the $\tau$-inertness property, which
is a property that ensures behavioural invariance on deterministic
transitions.

\begin{proposition}[$\tau$-inertness]\myrm
	\label{lem:tau_inert}
	Let  $\Gamma; \es; \Delta \proves P \hastype \Proc$ be a balanced \HOp process.
	Then
	\begin{itemize}
		\item	$\horel{\Gamma}{\Delta}{P}{\hby{\dtau}}{\Delta'}{P'}$ implies
			$\horel{\Gamma}{\Delta}{P}{\wb}{\Delta'}{P'}$.
		\item	$\horel{\Gamma}{\Delta}{P}{\Hby{\dtau}}{\Delta'}{P'}$ implies
			$\horel{\Gamma}{\Delta}{P}{\wb}{\Delta'}{P'}$.
	\end{itemize}
\end{proposition}

\begin{proof}
	The proof for Part 1 relies on the fact that processes of the\\
	form $\Gamma; \es; \Delta \proves \bout{s}{V} P_1 \Par \binp{\dual{s}}{x} P_2$
	cannot have any typed transition observables 
	(for both $s$ and $\dual{s}$ are defined in $\Delta$)
	and the fact
	that bisimulation is a congruence.
	See details in \appref{app:sub_tau_inert} (Page \pageref{app:sub_tau_inert}).
	The proof for Part 2 is straightforward from Part 1.
	\qed
\end{proof}

Processes that do not use shared names are inherently 
deterministic, and so they enjoy 
\emph{$\tau$-inertness} (in the sense of~\cite{DBLP:journals/tcs/GrooteS96}).

\begin{corollary}[$\CAL^{\minussh}$ $\tau$-inertness]\myrm
	\label{cor:tau_inert}
	Let $\Gamma; \es; \Delta \proves P \hastype \Proc$ be an $\CAL^{\minussh}$ process.
	\begin{itemize}
		\item	$\horel{\Gamma}{\Delta}{P}{\hby{\tau}}{\Delta'}{P'}$ if and only if $\horel{\Gamma}{\Delta}{P}{\hby{\dtau}}{\Delta'}{P'}$.
		\item	$\horel{\Gamma}{\Delta}{P}{\hby{\dtau}}{\Delta'}{P'}$ implies $\horel{\Gamma}{\Delta}{P}{\wb}{\Delta'}{P'}$.
	\end{itemize}
%	\begin{enumerate}
%		\item
%	If $\horel{\Gamma}{\Delta}{P}{\by{\tau}}{\Delta'}{P'}$ then $\horel{\Gamma}{\Delta}{P}{\wb}{\Delta'}{P'}$.
%		\item	If $P \red^* P'$ then $\Gamma; \es; \Delta \wb \Delta' \proves P \wb P'$.
%	\end{enumerate}
\end{corollary}

\begin{lemma}[Up-to Deterministic Transition]\myrm
	\label{lem:up_to_deterministic_transition}
	Let $\horel{\Gamma}{\Delta_1}{P_1}{\ \Re\ }{\Delta_2}{Q_1}$ such
	that if whenever:
	\begin{enumerate}
		\item	$\forall \news{\tilde{m_1}} \bactout{n}{V_1}$ such that
			$
				\horel{\Gamma}{\Delta_1}{P_1}{\hby{\news{\tilde{m_1}} \bactout{n}{V_1}}}{\Delta_3}{P_3}
			$
			implies that $\exists Q_2, V_2$ such that
			\[
				\horel{\Gamma}{\Delta_2}{Q_1}{\Hby{\news{\tilde{m_2}} \bactout{n}{V_2}}}{\Delta_2'}{Q_2}
			\]
			and
			\[
				\horel{\Gamma}{\Delta_3}{P_3}{\Hby{\dtau}}{\Delta_1'}{P_2}
			\]
			and for fresh $t$:
			\[
				\horel{\Gamma}{\Delta_1''}{\newsp{\tilde{m_1}}{P_2 \Par \htrigger{t}{V_1}}}
				{\ \Re\ }
				{\Delta_2''}{}{\newsp{\tilde{m_2}}{Q_2 \Par \htrigger{t}{V_2}}}
%				\mhorel{\Gamma}{\Delta_1''}{\newsp{\tilde{m_1}}{P_2 \Par \hotrigger{t}{x}{s}{V_1}}}
%				{\ \Re\ }
%				{\Delta_2''}{}{\newsp{\tilde{m_2}}{Q_2 \Par \hotrigger{t}{x}{s}{V_2}}}
			\]
		\item	$\forall \ell \not= \news{\tilde{m}} \bactout{n}{V}$ such that
			$
				\horel{\Gamma}{\Delta_1}{P_1}{\hby{\ell}}{\Delta_3}{P_3}
			$
			implies that $\exists Q_2$ such that 
			\[
				\horel{\Gamma}{\Delta_1}{Q_1}{\hat{\Hby{\ell}}}{\Delta_2'}{Q_2}
			\]
			and
			\[
				\horel{\Gamma}{\Delta_3}{P_3}{\Hby{\dtau}}{\Delta_1'}{P_2}
			\]
			and
			$\horel{\Gamma}{\Delta_1'}{P_2}{\ \Re\ }{\Delta_2'}{Q_2}$

		\item	The symmetric cases of 1 and 2.
	\end{enumerate}
	Then $\Re\ \subseteq\ \wb$.
\end{lemma}

\begin{proof}
	The proof is easy by considering the closure
	\[
		\Re^{\Hby{\dtau}} = \set{ \horel{\Gamma}{\Delta_1'}{P_2}{,}{\Delta_2'}{Q_1} \setbar \horel{\Gamma}{\Delta_1}{P_1}{\ \Re\ }{\Delta_2'}{Q_1},
		\horel{\Gamma}{\Delta_1}{P_1}{\Hby{\dtau}}{\Delta_1'}{P_2} }
	\]
	We verify that $\Re^{\Hby{\dtau}}$ is a bisimulation with
	the use of \propref{lem:tau_inert}.
	\qed
\end{proof}

%\end{document}

%\input{proposale}

% !TEX root = main.tex
%%%%%%%%%%%%%%%%%%%%%%%%%%%%%%%%%%%%%%%%%%%%%%%%%%%%%%%%%%%%%%%%%%%%%%%%%
% ENCODING DEFINITION AND PROPERTIESs
%%%%%%%%%%%%%%%%%%%%%%%%%%%%%%%%%%%%%%%%%%%%%%%%%%%%%%%%%%%%%%%%%%%%%%%%%

\section{Typed Encodings}
\label{s:expr}
\label{sec:enc}

This section defines the formal notion of \emph{encoding}, 
extending to a typed setting existing criteria for untyped processes (as in,
e.g.~\cite{Nestmann00,Palamidessi03,DBLP:conf/lics/PalamidessiSVV06,DBLP:journals/iandc/Gorla10,DBLP:conf/icalp/LanesePSS10,DBLP:journals/corr/abs-1208-2750}). 
We first define a typed calculus parameterised by a syntax, operational semantics, and typing.

\begin{definition}[Typed Calculus]\myrm
	\label{d:tcalculus}
	A \emph{typed calculus} $\tyl{L}$ is a tuple:
	\begin{center}
		\begin{tabular}{c}
			$\calc{\CAL}{\cal{T}}{\hby{\ell}}{\WB}{\proves}$
		\end{tabular}
	\end{center}
	\noi where $\CAL$ and $\cal{T}$ are sets of processes and types, 
	respectively; and $\hby{\ell}$, $\WB$, and $\proves$ 
	denote a transition system, a typed equivalence, and a typing system for $\CAL$, respectively. 
\end{definition}

\noi Our notion of encoding considers a mapping on processes, 
types, and transition labels.  

\begin{definition}[Typed Encoding]\myrm
        Let  $\tyl{L}_i=\calc{\CAL_i}{{\cal{T}}_i}{{\hby{\ell}}_i}{\WB_i}{\proves_i}$
        ($i=1,2$) be typed calculi, and let $\mathcal{A}_i$ be the
	set of labels used in relation ${\hby{\ell}}_i$.
	Given mappings $\map{\cdot}: \CAL_1 \to \CAL_2$, 
	$\mapt{\cdot}: {\cal{T}}_1 \to {\cal{T}}_2$, and 
	$\mapa{\cdot}: \mathcal{A}_1 \to \mathcal{A}_2$, 
	we write 
	$\enco{\map{\cdot}, \mapt{\cdot}, \mapa{\cdot}} : \tyl{L}_1 \to \tyl{L}_2$ to denote the \emph{typed encoding} of $\tyl{L}_1$ into $\tyl{L}_2$.
\end{definition}

\noi We will often assume that  $\mapt{\cdot}$ extends to typing
environments as expected. This way, e.g., $\mapt{\Delta \cat u:S} = \mapt{\Delta} \cat u:\mapt{S}$.

We introduce two classes of typed encodings, which 
serve different purposes.
Both  consist of syntactic and semantic criteria 
proposed for untyped processes~\cite{Palamidessi03,DBLP:journals/iandc/Gorla10,DBLP:conf/icalp/LanesePSS10},
here extended to account for (higher-order) session types.
First, for stating stronger positive encodability results, 
we define the notion of {\em precise} encodings.
Then, 
with the aim of proving strong non-encodability results, 
precise encodings are relaxed into the weaker {\em minimal} encodings. 

We first state the syntactic criteria. 
Let $\sigma$ denote a substitution of names for names 
(a renaming, in the usual sense). Given environments $\Delta$ and $\Gamma$,
we write $\sigma(\Delta)$ and $\sigma(\Gamma)$ to denote 
the effect of applying $\sigma$ on the 
domains of $\Delta$ and $\Gamma$
(clearly, $\sigma(\Gamma)$ concerns only shared names in $\Gamma$: process and recursion variables in $\Gamma$ are not affected by $\sigma$). 

\begin{definition}[Syntax Preserving Encoding]\myrm
	\label{def:sep}
	We say that 
	the typed encoding 
	$\enco{\map{\cdot}, \mapt{\cdot}, \mapa{\cdot}}: \tyl{L}_1 \to \tyl{L}_2$ is \emph{syntax preserving}
	if it is:
	
	\begin{enumerate}[1.]
		\item	\emph{Homomorphic wrt parallel},   if 
		$\mapt{\Gamma}; \emptyset; \mapt{\Delta_1 \cat \Delta_2} \proves_1 \map{P_1 \Par P_2} \hastype \Proc$
		then \\
		$\mapt{\Gamma}; \emptyset; \mapt{\Delta_1} \cat \mapt{\Delta_2} \proves_2 \map{P_1} \Par \map{P_2} \hastype \Proc$.

		\item	\emph{Compositional wrt restriction},  if 
		$\mapt{\Gamma}; \emptyset; \mapt{\Delta} \proves_1 \map{\news{n}P} \hastype \Proc$
		then \\
		$\mapt{\Gamma}; \emptyset; \mapt{\Delta} \proves_2 \news{n}\map{P} \hastype \Proc$.
		
		\item \emph{Name invariant},   if
		$\mapt{\sigma(\Gamma)}; \emptyset; \mapt{\sigma(\Delta)} \proves_1 \map{\sigma(P)} \hastype \Proc$
		then \\
		$\sigma(\mapt{\Gamma}); \emptyset; \sigma(\mapt{\Delta}) \proves_2 \sigma(\map{P}) \hastype \Proc$, 
		for any injective renaming  of names $\sigma$.
	\end{enumerate}
\end{definition}
Homomorphism wrt parallel composition (used in, e.g.,~\cite{Palamidessi03,DBLP:conf/lics/PalamidessiSVV06})
expresses that encodings should preserve the distributed topology of source processes. This criteria 
is appropriate for both encodability and non encodability results; in our setting, 
it admits an elegant formulation, also induced by rules for typed composition.
Compositionality wrt restriction is 
also naturally supported by typing and turns out to be 
useful in our encodability results (see the following section).
Our name invariance criteria follows the one given in~\cite{DBLP:journals/iandc/Gorla10,DBLP:conf/icalp/LanesePSS10}. 
Next we define semantic criteria for typed encodings.

\begin{definition}[Semantic Preserving Encoding]\myrm
\label{def:ep}
	Let
	$\tyl{L}_i=\calc{\CAL_i}{{\cal{T}}_i}{\hby{\ell}}{\WB_i}{\proves_i}$
	($i=1,2$) be typed calculi. 
	We say that
	$\enco{\map{\cdot}, \mapt{\cdot}, \mapa{\cdot}}: \tyl{L}_1 \to \tyl{L}_2$
	is a \emph{semantic preserving encoding}
	if it satisfies the properties below.
	Given a label $\ell \neq \tau$, we write 
	$\subj{\ell}$
	to denote the \emph{subject} of the action.
	
	\begin{enumerate}[1.]
		\item	\emph{Type Preservation}:
			if
			$\Gamma; \emptyset; \Delta \proves_1 P \hastype \Proc$ then 
			$\mapt{\Gamma}; \emptyset; \mapt{\Delta} \proves_2 \map{P} \hastype \Proc$,  
			for any   $P$ in $\CAL_1$.
		\item	\emph{Subject preserving}: if $\subj{\ell} = u$ then $\mathsf{sub}(\mapa{\ell}) =u$.

		\item	\emph{Operational Correspondence}:
			If $\Gamma; \emptyset; \Delta \proves_1 P \hastype \Proc$ then
			\begin{enumerate}
				\item	Completeness: 
					If  
					$\stytraargi{\Gamma}{\ell_1}{\Delta}{P}{\Delta'}{P'}{1}{1}$
					then $\exists \ell_2, Q, \Delta''$ s.t. \\
					(i)~$\wtytraargi{\mapt{\Gamma}}{\ell_2}{\mapt{\Delta}}{\map{P}}{\mapt{\Delta''}}{Q}{2}{2}$,
					(ii)~$\ell_2 = \mapa{\ell_1}$, 
					and \\
					(iii)~${\mapt{\Gamma}};{\mapt{\Delta''}}\proves_2 {Q}{\WB_2}{\mapt{\Delta'}}\proves_2 {\map{P'}}$.
				
				\item	Soundness:   
					If  $\wtytraargi{\mapt{\Gamma}}{\ell_2}{\mapt{\Delta}}{\map{P}}{\mapt{\Delta''}}{Q}{2}{2}$
					then $\exists \ell_1, P', \Delta'$ s.t.  \\
					(i)~$\stytraargi{\Gamma}{\ell_1}{\Delta}{P}{\Delta'}{P'}{1}{1}$,
					(ii)~$\ell_2 = \mapa{\ell_1}$, and
					(iii)~
					${\mapt{\Gamma}};{\mapt{\Delta'}}\proves_2 {\map{P'}}{\WB_2}
					{\mapt{\Delta''}}\proves_2 {Q}$.
		\end{enumerate}
		
		\item	\emph{Full Abstraction:} \\
			\wbbarg{\Gamma}{}{\Delta_1}{P}{\Delta_2}{Q}{1}
			if and only if
			\wbbarg{\mapt{\Gamma}}{}{\mapt{\Delta_1}}{\map{P}}{\mapt{\Delta_2}}{\map{Q}}{2}.
	\end{enumerate}
\end{definition}

\noi Type preservation is a distinguishing criteria in our notion of encoding: 
it enables us to focus on encodings which retain the communication structures denoted by (session) types.
The other semantic
criteria build upon analogous definitions in the untyped setting, as we explain now. 
Operational correspondence, standardly divided into completeness and soundness criteria, is based
in the formulation given in~\cite{DBLP:journals/iandc/Gorla10,DBLP:conf/icalp/LanesePSS10}. 
Soundness ensures that the source process is mimicked 
by its associated encoding; completeness concerns the opposite direction.
Rather than reductions, completeness and soundness rely on 
the typed LTS of \defref{def:restricted_typed_transition}; labels are considered up to  mapping
$\mapa{\cdot}$, which offers flexibility when comparing different subcalculi of \HOp.
We require that $\mapa{\cdot}$ preserves communication subjects, in accordance with the
criteria in~\cite{DBLP:conf/icalp/LanesePSS10}.
It is worth stressing that 
the operational correspondence statements given in
the next section for our  encodings 
are tailored to the specifics of each encoding, and so they
are actually stronger than the criteria given above.
Finally, following~\cite{SangiorgiD:expmpa,DBLP:conf/lics/PalamidessiSVV06,Yoshida96},
we consider full abstraction as an encodability criteria: this results into 
stronger encodability results. 
%The completeness direction of full abstraction is dropped when we prove the negative result. 
From the criteria in \defref{def:sep} and \defref{def:ep}
we have the following derived criteria: 

\begin{proposition}[Derived Criteria]\myrm
	\label{p:barbpres}
	Let
	$\enco{\map{\cdot}, \mapt{\cdot}, \mapa{\cdot}}: \tyl{L}_1 \to \tyl{L}_2$
	be a typed encoding.
	Suppose the encoding is both
	operational complete (cf.~\defref{def:ep}-3(a)) 
	and subject preserving (cf.~\defref{def:ep}-2).
	Then, it is also \emph{barb preserving}, i.e., 
	$\Gamma; \Delta \proves_1 P \barb{n}$
	implies
	$\mapt{\Gamma}; \mapt{\Delta} \proves_2 \map{P} \Barb{n}$.
\end{proposition}

\begin{proof}
	The proof follows from the definition of barbs,
	operational completeness, and subject preservation.
	\qed
\end{proof}

We may now define \emph{precise} and \emph{minimal} typed criteria: 

\begin{definition}[Typed Encodings: Precise and Minimal]\myrm
	\label{def:goodenc}
	We say that  the typed encoding 
	$\enco{\map{\cdot}, \mapt{\cdot}, \mapa{\cdot}}: \tyl{L}_1 \to \tyl{L}_2$ is 
	\begin{enumerate}[(i)]
		\item	\emph{precise}, if it is both syntax and semantic preserving
			(cf.~\defref{def:sep} and \defref{def:ep}).

		\item	\emph{minimal}, if it is syntax preserving 
			(cf.~\defref{def:sep}),
			%barb preserving, 
			and operational complete (cf.~\defref{def:ep}-3(a)).
	\end{enumerate}
\end{definition}

Precise encodings offer more detailed criteria and used for positive 
encodability results (\secref{sec:positive}).
In contrast, minimal encodings contains only 
some of the criteria of precise encodings:    
this reduced notion will be used 
for the negative result in \secref{s:negative}. 

Further we have:
\begin{proposition}[Composability of Precise Encodings]\myrm
	\label{prop:enc_composability}
	Let %encodings 
	$\enco{\pmap{\cdot}{1}, \tmap{\cdot}{1}, \mapa{\cdot}^{1}}: \tyl{L}_1 \to \tyl{L}_2$
	and 
	$\enco{\pmap{\cdot}{2}, \tmap{\cdot}{2}, \mapa{\cdot}^{2}}: \tyl{L}_2 \to \tyl{L}_3$
	be two precise typed encodings.
	Then their composition, denoted 
	$\enco{\pmap{\cdot}{2} \circ \pmap{\cdot}{1}, \tmap{\cdot}{2} \circ \tmap{\cdot}{1}, \mapa{\cdot}^{2}\circ \mapa{\cdot}^{1}}: \tyl{L}_1 \to \tyl{L}_3$
	is also a precise encoding.
\end{proposition}

\begin{proof}
	Straightforward application of the definition of each property, with the left-to-right direction of
	full abstraction being crucial.\qed
\end{proof}

In \secref{sec:positive} %and \secref{sec:negative}, 
we consider the following concrete instances of typed calculi
(cf.~\defref{d:tcalculus}):

\begin{definition}[Concrete Typed Calculi]\myrm
	We define the following  typed calculi:
	\begin{eqnarray*}
	\tyl{L}_{\HOp}&=&\calc{\HOp}{{\cal{T}}_1}{\hby{\ell}}{\hwb}{\proves} \\
	\tyl{L}_{\HO}&=&\calc{\HO}{{\cal{T}}_2}{\hby{\ell}}{\hwb}{\proves} \\
	\tyl{L}_{\sessp}&=&\calc{\sessp}{{\cal{T}}_3}{\hby{\ell}}{\fwb}{\proves}
	\end{eqnarray*}
	where: 
	${\cal{T}}_1$, ${\cal{T}}_2$, 
	and ${\cal{T}}_3$
	are sets of types of $\HOp$, $\HO$, and $\sessp$, respectively;
	the typing $\proves$ is defined in 
	\figref{fig:typerulesmy};  
	LTSs are as in \defref{def:restricted_typed_transition};
	$\hwb$ is as in \defref{def:bisim}; 
	$\fwb$ is as in \defref{def:cbisim}.
\end{definition}

% !TEX root = main.tex
\section{Positive Expressiveness Results}

\label{sec:positive}
In this section we present a study of the expressiveness
of $\HOp$ and its subcalculi. 
We present two encodability results:
\begin{enumerate}[1.]
	\item	The higher-order name passing communications with recursions (\HOp) into
		the higher-order communication without name-passing nor 
		recursions (\HO) (\secref{subsec:HOp_to_HO}).

	\item	\HOp into the first-order name-passing communication
		with recursions (\sessp) (\secref{subsec:HOp_to_p}). 
\end{enumerate}
In each case we show that the encoding is precise.

We often omit $H$ and $C$ from $\wb$ and $\fwb$ for simplicity of the notations.

%In each case, encoding correctness is supported by
%type preservation and operational correspondence statements.
%Full abstraction results are conjectured, for the moment.

\begin{remark}[Polyadic \HOp]
	We can assume a semantic preserving encoding from the polyadic
	\HOp to the monadic \HOp. Polyadic \HOp assumes a polyadic
	extension of the \HOp semantics that defines values as
	$V \bnfis \tilde{u} \bnfbar \abs{\tilde{x}}{P}$
	and input prefix as $\binp{n}{\tilde{x}} P$.
	See \secref{subsec:pol_HOp} for the full definition of
	polyadic \HOp.
\end{remark}

\subsection{Encoding \HOp into \HO}
\label{subsec:HOp_to_HO}

We show that the subcalculus $\HO$ is expressive enough to
represent the the full \HOp calculus.

The main challenge is to encode (1) name passing 
and (2) recursions.
Name passing involves {\em packing} a name 
value as an abstraction send it and it and then
substitute on the receiving using a name appication.
The encoding on the recursion semantics are more complex;
A process is encoded as an abstraction with no free names
(i.e~a shared abstraction). We then use higher-order
passing to pass the process and duplicate the process.
One copy of the process is used to reconstitute the
original process and the other is used to enable another
duplicator procedure.
%we only use name abstraction passing.
%For (1), we pass  
%an % simple 
%abstraction which enables to use the name upon application. 
%For (2), we 
%copy a process upon reception; the case of linear abstraction passing
%presents a limitation 
%is \NY{delicate} 
%because 
%linear abstractions cannot be copied.
We handle the transformation of a process into a linear abstraction
with the definition of an
%a preliminary tool which is a mapping from
auxiliary mapping
from processes with free names to processes without free names
(but with free variables) (\defref{def:auxmap}). 
We first require an auxiliary definition:
\begin{definition}\myrm 
	Let $\vmap{\cdot}: 2^{\mathcal{N}} \longrightarrow \mathcal{V}^\omega$
	be a map of sequences lexicographically ordered 
	names to sequences of variables, defined
	inductively as:
	\[
		\vmap{\epsilon} = \epsilon \qquad \qquad \qquad \vmap{n \cat \tilde{m}} = x_n \cat \vmap{\tilde{m}}
	\]
\end{definition}

Given a process $P$, we write $\ofn{P}$ to denote the
\emph{sequence} of free names of $P$, lexicographically ordered.

The following auxiliary mapping transforms processes
with free names into abstractions and it is
used in \defref{def:enc:HOp_to_HO}.
\begin{definition}\myrm
	\label{def:auxmap}
	Let $\sigma$ be a set of session names.
	Define $\auxmap{\cdot}{\sigma}: \HOp \to \HOp$  as in \figref{fig:auxmap}.
\end{definition}

\begin{figure}[t]
\[
	\begin{array}{rcl}
		\auxmap{\news{n} P}{\sigma} &\bnfis& \news{n} \auxmap{P}{{\sigma \cat n}}
		\vspace{1mm} \\

		\auxmap{\bout{n}{\abs{x}{Q}} P}{\sigma} &\bnfis&
		\left\{
		\begin{array}{rl}
			\bout{x_n}{\abs{x}{\auxmap{Q}{\sigma}}} \auxmap{P}{\sigma} & n \notin \sigma\\
			\bout{n}{\abs{x}{\auxmap{Q}{\sigma}}} \auxmap{P}{\sigma} & n \in \sigma
		\end{array}
		\right.
		\vspace{1mm}	\\ 

		\auxmap{\binp{n}{X} P}{\sigma} &\bnfis&
		\left\{
		\begin{array}{rl}
			\binp{x_n}{X} \auxmap{P}{\sigma} & n \notin \sigma\\
			\binp{n}{X} \auxmap{P}{\sigma} & n \in \sigma
		\end{array}
		\right.

		\vspace{1mm}	\\ 
		\auxmap{\bsel{n}{l} P}{\sigma} &\bnfis&
		\left\{
		\begin{array}{rl}
			\bsel{x_n}{l} \auxmap{P}{\sigma} & n \notin \sigma\\
			\bsel{n}{l} \auxmap{P}{\sigma} & n \in \sigma
		\end{array}
		\right.
		\vspace{1mm} \\
		\auxmap{\bbra{n}{l_i:P_i}_{i \in I}}{\sigma} &\bnfis&

		\left\{
		\begin{array}{rl}
			\bbra{x_n}{l_i:\auxmap{P_i}{\sigma}}_{i \in I}  & n \notin \sigma\\
			\bbra{n}{l_i:\auxmap{P_i}{\sigma}}_{i \in I}  & n \in \sigma
		\end{array}
		\right.
		\vspace{1mm} \\
		\auxmap{\appl{x}{n}}{\sigma} &\bnfis&
		\left\{
		\begin{array}{rl}
			\appl{x}{x_n} & n \notin \sigma\\
			\appl{x}{n} & n \in \sigma\\
		\end{array}
		\right.
		\vspace{1mm} \\

		\auxmap{\appl{(\abs{x} P)}{n}}{\sigma} &\bnfis&
		\left\{
		\begin{array}{rl}
			\appl{(\abs{x} \auxmap{P}{\sigma})}{x_n} & n \notin \sigma\\
			\appl{(\abs{x} \auxmap{P}{\sigma})}{n} & n \in \sigma\\
		\end{array}
		\right.
		\vspace{1mm} \\

		\auxmap{\inact}{\sigma} &\bnfis& \inact
		\vspace{1mm} \\

		\auxmap{P \Par Q}{\sigma} &\bnfis& \auxmap{P}{\sigma} \Par \auxmap{Q}{\sigma}
	\end{array}
\]
\caption{\label{fig:auxmap} The auxiliary map (cf.~\defref{def:auxmap}) 
used in the encoding of \HOp into \HO~(\defref{def:enc:HOp_to_HO}).}
\end{figure}

Given a process $P$ with $\fn{P} = m_1, \cdots, m_n$,
we are interested in its associated (polyadic) abstraction,
which is defined as $\abs{x_1, \cdots, x_n}{\auxmap{P}{\es} }$,
where $\vmap{m_j} = x_j$, for all $j \in \{1, \ldots, n\}$.
This transformation from processes into abstractions can be reverted by
using abstraction and application with an appropriate sequence of session names:
\begin{proposition}\myrm
	Let $P$ be a \HOp process with $\tilde{n} = \ofn{P}$.
	Also, suppose $\tilde{x} = \vmap{\tilde{n}}$.
	Then $P \scong \appl{x}{\tilde{n}}\subst{\abs{\tilde{x}}\auxmap{P}{\emptyset}}{x}$.
\end{proposition}

\begin{proof}
	\noi The proof is an easy induction on the map $\auxmap{P}{\es}$.
	We show a case since other cases are similar.

	\noi - Case: $\auxmap{\bout{n}{m} P}{\es} = \bout{x_n}{x_m} \auxmap{P}{\es}$

	\noi We rewrite substitution as:
	$\appl{x}{\tilde{n}} \subst{\abs{\tilde{x}}{\bout{x_n}{y_m} \auxmap{P}{\es}}}{x} \scong (\bout{x_n}{y_m} P) \subst{\tilde{x}}{\tilde{n}}$

	\noi If consider that $x_n, y_m \in \vmap{\tilde{n}}$ then from the definition of $\vmap{\cdot}$ we
	get that $n, m \in \tilde{n}$. Furthermore by the fact that $\tilde{n}$ and $\vmap{\tilde{n}}$ are
	ordered, substitution becomes:
	$\bout{n}{m} \auxmap{P}{\es} \subst{\tilde{x}}{\tilde{n}}$.

	\noi The rest of the cases are similar.
	\qed
\end{proof}

We are now ready to define the encoding of \HOp into strict process-passing.
Note that we assume polyadicity in abstraction and application.
Given a session environment $\Delta = \{n_1:S_1, \ldots, n_m:S_m\}$, 
in the following definition we write
$\tilde{S}_{\Delta}$ to stand for $S_1, \ldots, S_m$.
\begin{definition}[Encoding \HOp into \HO]\myrm
	\label{def:enc:HOp_to_HO}
	Let $f$ be a function from recursion variables to sequences of name variables.
	Define the typed encoding $\enco{\pmapp{\cdot}{1}{f}, \tmap{\cdot}{1}, \mapa{\cdot}^{1}}: \tyl{L}_{\HOp} \to \tyl{L}_{\HO}$,
	where mappings $\map{\cdot}^{1}$, $\mapt{\cdot}^{1}$, $\mapa{\cdot}^{1}$
	are as in \figref{fig:enc:HOp_to_HO}.
	We assume that the mapping $\tmap{\cdot}{1}$ on types is extended to 
	session environments $\Delta$
	and
	shared environments $\Gamma$ 
	as follows:
	\[
	\begin{array}{rcll}
	    \mapt{\Delta \cat s: S}^{1} & =  & \mapt{\Delta}^{1} \cat s:\mapt{S}^{1} & \\
		\tmap{\Gamma \cat u: \chtype{S}}{1} & =  & \tmap{\Gamma}{1} \cat u:\chtype{\tmap{S}{1}} & \\
		\tmap{\Gamma \cat u: \chtype{L}}{1} & = &  \tmap{\Gamma}{1} \cat u:\chtype{\tmap{L}{1}} & \\
		\tmap{\Gamma \cat \varp{X}:\Delta}{1} & = & \tmap{\Gamma}{1} \cat x:\shot{(\tilde{S}_{\Delta}\,,\,S^*)} & 
		\quad\text{(where $ 
%		S^* = \trec{t}{\big((\tilde{S}_{\Delta}\,,\, \btinp{\vart{t}}\tinact)\big)}
		S^* = \trec{t}{\btinp{\shot{(\tilde{S}_{\Delta}\,,\,\vart{t})}} \tinact}$)}
	\end{array}
	\]
%\end{remark}
\end{definition}

\begin{figure}[h!]
\[
	\begin{array}{rclcrcl}
		\multicolumn{7}{l}{\textrm{\bf Terms}}
		\\
		\pmapp{\bout{u}{v} P}{1}{f}	&\defeq&	\bout{u}{ \abs{z}{\,\binp{z}{x} (\appl{x}{v}}) } \pmapp{P}{1}{f}
		& &
		\pmapp{\binp{u}{k} Q}{1}{f}	&\defeq&	\binp{u}{x} \newsp{s}{\appl{x}{s} \Par \bout{\dual{s}}{\abs{x}{\pmapp{Q}{1}{f}}} \inact}
		\\
		\pmapp{\bout{u}{\abs{x}{Q}} P}{1}{f} &\defeq& \bout{u}{\abs{x}{\pmapp{Q}{1}{f}}} \pmapp{P}{1}{f}
		& &
		\pmapp{\binp{u}{\underline{x}} P}{1}{f}	&\defeq&	\binp{u}{\underline{x}} \pmapp{P}{1}{f}
		\\
		\pmapp{\recp{X}{P}}{1}{f} &\defeq&
		\multicolumn{5}{l}{
			\newsp{s}{\binp{s}{x} \pmapp{P}{1}{{f,\{\varp{X}\to \tilde{n}\}}} \Par
			\bout{\dual{s}}{\abs{(\vmap{\tilde{n}}, y)} \,{\binp{y}{z_X} \auxmap{\pmapp{P}{1}{{f,\{\varp{X}\to \tilde{n}\}}}}{\es}}} \inact}
			\quad \tilde{n} = \ofn{P}
		}
		\\
		\pmapp{\varp{X}}{1}{f} &\defeq&
		\multicolumn{5}{l}{
			\newsp{s}{\appl{z_X}{(\tilde{n}, s)} \Par \bbout{\dual{s}}{ \abs{(\vmap{\tilde{n}},y)}\,\,{\appl{z_X}{(\vmap{\tilde{n}}, y)}}} \inact}
			\qquad \qquad \qquad \qquad \quad \tilde{n} = f(\varp{X})
		}
		\\
		\pmapp{\bsel{s}{l} P}{1}{f}	&\defeq&	\bsel{s}{l} \pmapp{P}{1}{f}
		& & 
		\pmapp{\bbra{s}{l_i: P_i}_{i \in I}}{1}{f} &\defeq& \bbra{s}{l_i: \pmapp{P_i}{1}{f}}_{i \in I}
		\\
		\pmapp{\appl{x}{u}}{1}{f}	&\defeq&	\appl{x}{u}
		& &
		\pmapp{(\appl{\abs{x}{P})}{u}}{1}{f}	&\defeq&	\appl{(\abs{x}{\pmapp{P}{1}{f}})}{u}
		\\
		\pmapp{P \Par Q}{1}{f}		&\defeq&	\pmapp{P}{1}{f} \Par \pmapp{Q}{1}{f}
		& &
		\pmapp{\news{n} P}{1}{f}	&\defeq&	\news{n} \pmapp{P}{1}{f}
		\\
		\pmapp{\inact}{1}{f}		&\defeq&	\inact
		\\
		\multicolumn{7}{l}{\textrm{\bf Types}}
		\\
		\tmap{C}{1}_\mathsf{v}		&\defeq&
		\multicolumn{5}{l}{
			\left\{
			\begin{array}{rcl}
				\lhot{(\btinp{\lhot{\tmap{C}{1}}} \tinact)} && \textrm{if } C = S\\
				\lhot{(\btinp{\shot{\tmap{C}{1}}} \tinact)} && \textrm{otherwise}\\
			\end{array}
			\right.
		}
		\\
		\tmap{\lhot{C}}{1}_\mathsf{v}	&\defeq& \lhot{\tmap{C}{1}}
		& & 
		\tmap{\shot{C}}{1}_\mathsf{v}	&\defeq& \shot{\tmap{C}{1}}
		\\
		\tmap{\chtype{S}}{1}		&\defeq& \chtype{\tmap{S}{1}}
		& &
		\tmap{\chtype{L}}{1}		&\defeq& \chtype{\tmap{L}{1}_\mathsf{v}}
		\\
		\tmap{\btout{U} S}{1}		&\defeq& \btout{\tmap{U}{\mathsf{v}}} \tmap{S}{1}
		& &
		\tmap{\btinp{U} S}{1}		&\defeq& \btinp{\tmap{U}{\mathsf{v}}} \tmap{S}{1}
		\\
		\tmap{\btsel{l_i: S_i}_{i \in I}}{1} &\defeq& \btsel{l_i: \tmap{S_i}{1}}_{i \in I}
		& &
		\tmap{\btbra{l_i: S_i}_{i \in I}}{1} &\defeq& \btbra{l_i: \tmap{S_i}{1}}_{i \in I}
		\\
		\tmap{\vart{t}}{1} &\defeq& \vart{t}
		& &
		\tmap{\trec{t}{S}}{1} &\defeq& \trec{t}{\tmap{S}{1}}
		\\
		\tmap{\tinact}{1} &\defeq& \tinact
		\\
		\multicolumn{7}{l}{\textrm{\bf Labels}}
		\\
		\mapa{\news{\tilde{m_1}}\bactout{n}{m}}^{1}	&\defeq&	\news{\tilde{m_1}}\bactout{n}{\abs{z}{\,\binp{z}{x} \appl{x}{m}} }
		& &
		\mapa{\bactinp{n}{m}}^{1}			&\defeq&	\bactinp{n}{\abs{z}{\,\binp{z}{x} \appl{x}{m}} }
		\\
		\mapa{\news{\tilde{m}}\bactout{n}{\abs{x}{P}}}^{1} &\defeq& \news{\tilde{m}}\bactout{n}{\abs{x}{\pmapp{P}{1}{\es}}}
		& &
		\mapa{\bactinp{n}{\abs{x}{P}}}^{1} &\defeq& \bactinp{n}{\abs{x}{\pmapp{P}{1}{\es}}}
		\\
		\mapa{\bactsel{n}{l} }^{1} &\defeq& \bactsel{n}{l} 
		& &
		\mapa{\bactbra{n}{l} }^{1} &\defeq& \bactbra{n}{l} 
		\\
		\mapa{\tau}^{1} &\defeq& \tau

	\end{array}
\]
	\caption{
		\label{fig:enc:HOp_to_HO}
		Typed encoding of \HOp into \HO (cf.~Defintion~\ref{def:enc:HOp_to_HO}).
%		Mappings 
%		$\map{\cdot}^2$,
%		$\mapt{\cdot}^2$, 
%		and 
%		$\mapa{\cdot}^2$
%		are homomorphisms for the other processes/types/labels. 
	}
\end{figure}

\noi Note that $\Delta$ in $\varp{X}:\Delta$ is mapped to a non-tail
recursive session type.
Non-tail
recursive session types have been studied in
\cite{DBLP:journals/corr/abs-1202-2086,TGC14};
to our knowledge,
this is the first application in the
context of higher-order session types.
%which carries type variable as the last argument.  
For a simplicity of the presentation, we use the polyadic name abstraction and passing.
Polyadic semantics will be formally encoded into \HO in \secref{subsec:pol_HOp}.

\noi We explain the mapping in \figref{def:enc:HOp_to_HO}, focusing 
on {\em name passing} ($\pmapp{\bout{u}{w} P}{1}{f}$ and $\pmapp{\binp{u}{x} P}{1}{f}$), and  
{\em recursion} ($\pmapp{\recp{X}{P}}{1}{f}$ and $\pmapp{\varp{X}}{1}{f}$). 

\myparagraph{Name passing}
A name $w$ is being passed as an input guarded abstraction;
the abstraction receives a higher-order
value and continues with the application of $w$ over
the received higher-order value.
%A name $m$ is being passed as an input
%guarded abstraction. 
%The input prefix receives an abstraction and
%continues with the application of $n$ over the received abstraction.
On the receiver side $\binp{u}{x} P$ 
the encoding realises a mechanism that i) receives
the input guarded abstraction, then ii) applies it on a fresh session endpoint $s$, 
and iii) uses
the dual endpoint $\dual{s}$ to send the continuation $P$ as the abstraction
$\abs{x}{P}$. 
\NY{Then} name substitution is achieved via name application.

\myparagraph{Recursion}
The encoding of a recursive process $\recp{X}{P}$  is delicate, for it 
must preserve the linearity of session endpoints. To this end, we:
\NY{i) record a mapping from recursive variable $X$ to process variables $z_X$;
ii)~encode the recursion body $P$ as a name abstraction
in which free names of $P$ are converted into name variables;
iii)~this higher-order value is embedded in an input-guarded 
``duplicator'' process; and 
iv)~make the encoding of process variable $x$ to 
simulate recursion unfolding by 
invoking the duplicator in a by-need fashion,
i.e.,~upon reception, abstraction $\auxmap{P}{\sigma}$ is duplicated
with one copy used to reconstitute the encoded recursion body $P$ through
the application of $\fn{P}$ and another copy used to re-invoke
the duplicator when needed. % to simulate recursion unfolding.
}
%The idea follows 
%a classical recursion encoding \cite{ThomsenB:plachoasgcfhop}.  
%A mapping of process $P$ is parallel composed, 
%and also being passed as an input
%guarded abstraction, parameterised also by a sequence of trigger names $\tilde{n}$. 
%We record a mapping from $z_X$ (which is a fresh variable of $X$) 
%to $\tilde{n}$, so that 
%when the abstraction is substituted to $z_\rvar{X}$ 
%(which occurs in the mapping of $P\subst{z_X}{X}$), 
%the correct $\tilde{n}$ is applied. In this way, we can 
%send and receive an abstraction which holds $P$, repeatedly. 

%In the higher-order setting, a name $v$ is being passed as an input
%guarded abstraction. The input prefix receives an abstraction and
%continues with the application of $v$ over the received abstraction.
%On the receiver side $\binp{u}{x} P$ 
%the encoding realizes a mechanism that (i) receives
%the input guarded abstraction, then (ii) applies it on a fresh session endpoint $s$, 
%and (iii) uses
%the dual endpoint $\dual{s}$ to send the continuation $P$ as the abstraction
%$\abs{x}{P}$. 
%As a result, name substitution is achieved via name application.

\begin{proposition}[Type Preservation, \HOp into \HO]\myrm
	\label{prop:typepres_HOp_to_HO}
	Let $P$ be a \HOp process.
	If $\Gamma; \emptyset; \Delta \proves P \hastype \Proc$ then 
	$\tmap{\Gamma}{1}; \emptyset; \tmap{\Delta}{1} \proves \pmapp{P}{1}{f} \hastype \Proc$. 
\end{proposition}

\begin{proof}
	By induction on the inference $\Gamma; \emptyset; \Delta \proves P \hastype \Proc$.
	Details in \propref{app:prop:typepres_HOp_to_HO} (Page~\pageref{app:prop:typepres_HOp_to_HO}).
	\qed
\end{proof}

The following proposition formalizes our strategy  for encoding
recursive definitions as passing of polyadic abstractions:
\begin{proposition}[Operational Correspondence for Recursive Processes]\myrm
	\label{prop:op_corr_HOprec_to_HO}
	Let $P$ and $P_1$ be \HOp processes s.t. 
	$P =\recp{X}{P'}$ and
	$P_1 = P'\subst{\recp{X}{P'}}{\varp{X}} \scong P$.

	\noi If %$P_1 \hby{\ell} P_2$ 
	$\stytra{\Gamma}{\ell}{\Delta}{P}{\Delta'}{P'}$
	then,  there exist
	processes $R_1$, $R_2$,  $R_3$, action $\ell'$,
	and mappings $f, f_1$, such that: 
	\begin{enumerate}[(i)]
		\item 
		%$\pmapp{P}{1}{f} \hby{\tau} \map{P'}^{1}_{f_{1}} \subst{R_3}{X} = R_1$;
		$\stytra{\mapt{\Gamma}^{1}}{\tau}{\mapt{\Delta}^{1}}{P}{\mapt{\Delta}^{1}}{\map{P'}^{1}\subst{R_3}{X}} = R_1$;
		\item 
		%$R_1 \Hby{\ell'} R_2$, with $\ell' = \mapa{\ell}^{1}$;
		$\wtytra{\mapt{\Gamma}^{1}}{\ell'}{\mapt{\Delta}^{1}}{R_1}{\mapt{\Delta}^{1}}{R_2} $,  with $\ell' = \mapa{\ell}^{1}$;
	
		\item $R_3 = \abs{\tilde{m}}\binp{z}{x}\auxmap{\map{P'}^{1}_{f_{1}}}{\sigma}$, with $\tilde{m} = \ofn{P'},z$)
		and
		$f_1 = f, \set{\varp{X} \to \ofn{P'}}$.
	\end{enumerate}
\end{proposition}

\begin{proof}[Sketch]
	Part~(1) follow directly from the definition of typed encoding for processes $\pmapp{\cdot}{1}{f}$ (\defref{def:enc:HOp_to_HO}),
	observing that the reduction occurs along a restricted name, and so the session environment remains unchanged.
	Part~(2) relies on  \propref{prop:op_corr_HOp_to_HO}.
	Part~(3) is immediate from \defref{def:enc:HOp_to_HO}.
	\qed
\end{proof}

The following proposition formalises completeness and
soundness results for the encoding of \HOp into \HO.
Recall that deterministic transitions $\stau$ and 
$\btau$ have been defined in \defref{def:dettrans}.
%We write $\by{\tau}_k$ to denote a sequence of $k$ $\tau$-transitions.

\begin{proposition}[Operational Correspondence, \HOp into \HO]\myrm
	\label{prop:op_corr_HOp_to_HO}
	Let $P$ be a \HOp process.
	If $\Gamma; \emptyset; \Delta \proves P \hastype \Proc$ then:
	\begin{enumerate}[1.]
		\item
			Suppose $\horel{\Gamma}{\Delta}{P}{\hby{\ell_1}}{\Delta'}{P'}$. Then we have:
			\begin{enumerate}[a)]
				\item
					If $\ell_1 \in \set{\news{\tilde{m}}\bactout{n}{m}, \,\news{\tilde{m}}\bactout{n}{\abs{x}Q}, \,\bactsel{s}{l}, \,\bactbra{s}{l}}$
					then $\exists \ell_2$ s.t. \\
					$\horel{\tmap{\Gamma}{1}}{\tmap{\Delta}{1}}{\pmapp{P}{1}{f}}{\hby{\ell_2}}{\tmap{\Delta'}{1}}{\pmapp{P'}{1}{f}}$
					and $\ell_2 = \mapa{\ell_1}^{1}$.
			
				\item
					If $\ell_1 = \bactinp{n}{\abs{y}Q}$ and
					$P' = P_0 \subst{\abs{y}Q}{x}$
					then $\exists \ell_2$ s.t. \\
					$\horel{\tmap{\Gamma}{1}}{\tmap{\Delta}{1}}{\pmapp{P}{1}{f}}{\hby{\ell_2}}{\tmap{\Delta'}{1}}{\pmapp{P_0}{1}{f}\subst{\abs{y}\pmapp{Q}{1}{\emptyset}}{x}}$
					and $\ell_2 = \mapa{\ell_1}^{1}$.
			
				\item
					If $\ell_1 = \bactinp{n}{m}$
					and 
					$P' = P_0 \subst{m}{x}$
					then $\exists \ell_2$, $R$ s.t. \\
					$\horel{\tmap{\Gamma}{1}}{\tmap{\Delta}{1}}{\pmapp{P}{1}{f}}{\hby{\ell_2}}{\tmap{\Delta'}{1}}{R}$,
					with $\ell_2 = \mapa{\ell_1}^{1}$, \\
					and
					$\horel{\tmap{\Gamma}{1}}{\tmap{\Delta'}{1}}{R}{\hby{\btau} \hby{\stau} \hby{\btau}}
					{\tmap{\Delta'}{1}}{\pmapp{P_0}{1}{f}\subst{m}{x}}$.
						
				\item
					If $\ell_1 = \tau$
					and $P' \scong \newsp{\tilde{m}}{P_1 \Par P_2\subst{m}{x}}$
					then $\exists R$ s.t. \\
					$\horel{\tmap{\Gamma}{1}}{\tmap{\Delta}{1}}{\pmapp{P}{1}{f}}{\hby{\tau}}{\mapt{\Delta}^{1}}{\newsp{\tilde{m}}{\pmapp{P_1}{1}{f} \Par R}}$,
					and\\ 
					$\horel{\tmap{\Gamma}{1}}{\tmap{\Delta}{1}}{\newsp{\tilde{m}}{\pmapp{P_1}{1}{f} \Par R}}{\hby{\btau} \hby{\stau} \hby{\btau}}
					{\mapt{\Delta}^{1}}{\newsp{\tilde{m}}{\pmapp{P_1}{1}{f} \Par \pmapp{P_2}{1}{f}\subst{m}{x}}}$.
			
				\item
					If $\ell_1 = \tau$
					and $P' \scong \newsp{\tilde{m}}{P_1 \Par P_2 \subst{\abs{y}Q}{x}}$
					then \\
					$\horel{\tmap{\Gamma}{1}}{\tmap{\Delta}{1}}{\pmapp{P}{1}{f}}{\hby{\tau}}
					{\tmap{\Delta_1}{1}}{\newsp{\tilde{m}}{\pmapp{P_1}{1}{f}\Par \pmapp{P_2}{1}{f}\subst{\abs{y}\pmapp{Q}{1}{\emptyset}}{x}}}$.
			
				\item
					If $\ell_1 = \tau$
					and $P' \not\scong \newsp{\tilde{m}}{P_1 \Par P_2 \subst{m}{x}} \land P' \not\scong \newsp{\tilde{m}}{P_1 \Par P_2\subst{\abs{y}Q}{x}}$
					then \\
					$\horel{\tmap{\Gamma}{1}}{\tmap{\Delta}{1}}{\pmapp{P}{1}{f}}{\hby{\tau}}{\tmap{\Delta'_1}{1}}{ \pmapp{P'}{1}{f}}$.
			\end{enumerate}
			
		\item	Suppose $\horel{\tmap{\Gamma}{1}}{\tmap{\Delta}{1}}{\pmapp{P}{1}{f}}{\hby{\ell_2}}{\tmap{\Delta'}{1}}{Q}$.
			Then we have:
			\begin{enumerate}[a)]
				\item 
					If $\ell_2 \in
					\set{\news{\tilde{m}}\bactout{n}{\abs{z}{\,\binp{z}{x} (\appl{x}{m})}}, \,\news{\tilde{m}} \bactout{n}{\abs{x}{R}}, \,\bactsel{s}{l}, \,\bactbra{s}{l}}$
					then $\exists \ell_1, P'$ s.t. \\
					$\horel{\Gamma}{\Delta}{P}{\hby{\ell_1}}{\Delta'}{P'}$, 
					$\ell_1 = \mapa{\ell_2}^{1}$, 
					and
					$Q = \pmapp{P'}{1}{f}$.
			
				\item 
					If $\ell_2 = \bactinp{n}{\abs{y} R}$ %(with $R \neq \binp{y}{x} \appl{x}{m}$)
					then either:
					\begin{enumerate}[(i)]
						\item	$\exists \ell_1, x, P', P''$ s.t. \\
							$\horel{\Gamma}{\Delta}{P}{\hby{\ell_1}}{\Delta'}{P' \subst{\abs{y}P''}{x}}$, 
							$\ell_1 = \mapa{\ell_2}^{1}$, $\pmapp{P''}{1}{\es} = R$, and $Q = \pmapp{P'}{1}{f}$.

						\item	$R \scong \binp{y}{x} (\appl{x}{m})$ and 
							$\exists \ell_1, z, P'$ s.t. \\
							$\horel{\Gamma}{\Delta}{P}{\hby{\ell_1}}{\Delta'}{P' \subst{m}{z}}$, 
							$\ell_1 = \mapa{\ell_2}^{1}$,
							and\\
							$\horel{\tmap{\Gamma}{1}}{\tmap{\Delta'}{1}}{Q}{\hby{\btau} \hby{\stau} \hby{\btau}}{\tmap{\Delta''}{1}}{\pmapp{P'\subst{m}{z}}{1}{f}}$
					\end{enumerate}
			
				\item 
					If $\ell_2 = \tau$ 
					then $\Delta' = \Delta$ and 
					either
					\begin{enumerate}[(i)]
						\item	$\exists P'$ s.t. 
							$\horel{\Gamma}{\Delta}{P}{\hby{\tau}}{\Delta}{P'}$,
							and $Q = \map{P'}^{1}_f$.	

						\item
							$\exists P_1, P_2, x, m, Q'$ s.t. 
							$\horel{\Gamma}{\Delta}{P}{\hby{\tau}}{\Delta}{\newsp{\tilde{m}}{P_1 \Par P_2\subst{m}{x}} }$, and\\
							$\horel{\tmap{\Gamma}{1}}{\tmap{\Delta}{1}}{Q}{\hby{\btau} \hby{\stau} \hby{\btau}}{\tmap{\Delta}{1}}{\pmapp{P_1}{1}{f} \Par \pmapp{P_2\subst{m}{x}}{1}{f}}$ 
%							$Q = \map{P_1}^{1}_f \Par Q'$, where $Q'  \Hby{} $.

%						\item $\exists P_1, P_2, x, R$ s.t. 
%						$\stytra{ \Gamma }{\tau}{ \Delta }{ P}{ \Delta}{ \news{\tilde{m}}(P_1 \Par P_2\subst{\abs{y}R}{x}) }$, and 
%						$Q = \map{\news{\tilde{m}}(P_1 \Par P_2\subst{\abs{y}R}{x})}^{1}_f$.
			\end{enumerate}
		    \end{enumerate}
		    
%		\item   
%			If  $\wtytra{\mapt{\Gamma}^{1}}{\ell_2}{\mapt{\Delta}^{1}}{\pmapp{P}{1}{f}}{\mapt{\Delta'}^{1}}{Q}$
%			then $\exists \ell_1, P'$ s.t.  \\
%			(i)~$\stytra{\Gamma}{\ell_1}{\Delta}{P}{\Delta'}{P'}$,
%			(ii)~$\ell_2 = \mapa{\ell_1}^{1}$, 
%			(iii)~$\wbb{\mapt{\Gamma}^{1}}{\ell}{\mapt{\Delta'}^{1}}{\pmapp{P'}{1}{f}}{\mapt{\Delta'}^{1}}{Q}$.
	\end{enumerate}
\end{proposition}

\begin{proof}
	The proof is a mechanical induction on the labelled Transition System.
	Parts (1) and (2) are proved separetely.
	The most demanding cases for the proof can be found in~\propref{app:prop:op_corr_HOp_to_HO}
	(page~\pageref{app:prop:op_corr_HOp_to_HO}).
	\qed
\end{proof}

\begin{proposition}[Full Abstraction, \HOp into \HO]\myrm
	\label{prop:fulla_HOp_to_HO}
	Let $P_1, Q_1$ be \HOp processes.
	$\horel{\Gamma}{\Delta_1}{P_1}{\hwb}{\Delta_2}{Q_1}$
	if and only if
	$\horel{\tmap{\Gamma}{1}}{\tmap{\Delta_1}{1}}{\pmapp{P_1}{1}{f}}{\hwb}{\tmap{\Delta_2}{1}}{\pmapp{Q_1}{1}{f}}$.
\end{proposition}

\begin{proof}
	The proof for the soundness direction considers
	closure that can be shown to be a bisimulation
	following the soundness direction of Operational Correspondence
	(\propref{prop:op_corr_HOp_to_HO}). Whenever needed
	the proof makes use of the $\tau$-inertness result
	(\propref{lem:tau_inert}).

	The proof for the completness direction also considers
	a closure shown to be a bisimulation
	up-to deterministic transition (\propref{lem:up_to_deterministic_transition})
	following the completeness direction of Operational Correspondence
	(\propref{prop:op_corr_HOp_to_HO}).

	Details of the proof can be found in~\propref{app:prop:fulla_HOp_to_HO}
	(page~\pageref{app:prop:fulla_HOp_to_HO}).
	\qed
\end{proof}

\begin{proposition}[Precise encoding of \HOp into \HO]\myrm
	\label{prop:prec:HOp_to_HO}
	The encoding from $\tyl{L}_{\HOp}$ to $\tyl{L}_{\HO}$
	is precise.
\end{proposition}

\begin{proof}
	Syntactic requirements are easily derivable from the
	definition of the mappings in \figref{fig:enc:HOp_to_HO}.
	Semantic requirements are a consequence of
	\propref{prop:typepres_HOp_to_HO}, \propref{prop:op_corr_HOp_to_HO}, and \propref{prop:fulla_HOp_to_HO}.
	\qed
\end{proof}

\begin{example}[Encode $\recp{X}{\bout{a}{m} \varp{X}}$ into \HO]

\noi {\bf Mapping:} Term mapping of \HOp process $\recp{X}{\bout{a}{m} \varp{X}}$
into a \HO process. We note initially $f = \emptyset$. The first application of the mapping
will give:
\[
\begin{array}{rcl}
	\pmapp{\recp{X}{\bout{a}{m} \varp{X}}}{1}{} &=&
	\newsp{s_1}{ \binp{s_1}{x} \pmapp{\bout{a}{m} \varp{x}}{1}{\varp{x} \rightarrow x_ax_m} \Par\\
	&&\bout{\dual{s_1}}{ \abs{(x_a, x_m, z)} \binp{z}{x} \auxmap{\pmapp{\bout{a}{m} \varp{x}}{1}{\varp{x} \rightarrow x_ax_m}}{\es} } \inact}
	\\
	\multicolumn{3}{l}{\textrm{with}}
	\\
	\pmapp{\bout{a}{m} \varp{x}}{1}{ \varp{x} \rightarrow x_ax_m} &=&
	\bout{a}{\abs{z}{\binp{z}{x} (\appl{x}{m})}} \pmapp{\varp{x}}{1}{\varp{x} \rightarrow x_ax_m}
	\\
	&=& \bout{a}{\abs{z}{\binp{z}{x} (\appl{x}{m})}} \newsp{s_2}{\appl{x}{(a,m, s_2)}  \Par \bout{\dual{s_2}}{\abs{(x_a, x_m, z)}{\appl{x}{(x_a, x_m, z)}}} \inact}
\end{array}
\]
\noi Furthermore:
\[
\begin{array}{l}
	\auxmap{\pmapp{\bout{a}{m} \varp{x}}{1}{\varp{x} \rightarrow x_ax_m}}{\es}\\
	\qquad \qquad \quad = \auxmap{\bout{a}{\abs{z}{\binp{z}{x} (\appl{x}{m})}} \newsp{s_2}{\appl{x}{(a,m, s_2)}  \Par \bout{\dual{s_2}}{\abs{(x_a, x_m, z)}{\appl{x}{(x_a, x_m, z)}}} \inact}}{\es}
	\\
	\qquad \qquad \quad = \bout{x_a}{\abs{z}{\binp{z}{x} (\appl{x}{x_m})}} \auxmap{\newsp{s_2}{\appl{x}{(a,m, s_2)}  \Par \bout{\dual{s_2}}{\abs{(x_a, x_m, z)}{\appl{x}{(x_a, x_m, z)}}} \inact}}{\es}
	\\
	\qquad \qquad \quad = \bout{x_a}{\abs{z}{\binp{z}{x} (\appl{x}{x_m})}} \newsp{s_2}{\appl{x}{(x_a,x_m, s_2)}  \Par \bout{\dual{s_2}}{\abs{(x_a, x_m, z)}{\appl{x}{(x_a, x_m, z)}}} \inact}
\end{array}
\]
\noi The whole encoding would be:
\[
\begin{array}{l}
	V = \abs{(x_a, x_m, z)} \binp{z}{x} \bout{x_a}{\abs{z}{\binp{z}{x} (\appl{x}{x_m})}} \newsp{s_2}{\appl{x}{(x_a,x_m, s_2)}  \Par \bout{\dual{s_2}}{\abs{(x_a, x_m, z)}{\appl{x}{(x_a, x_m, z)}}} \inact}\\
	\pmapp{\recp{X}{\bout{a}{m} \varp{X}}}{1}{} \scong \\
	\newsp{s_1}{\bout{\dual{s_1}}{V} \inact \Par \binp{s_1}{x} \bout{a}{\abs{z}{\binp{z}{x} (\appl{x}{m})}} \newsp{s_2}{\bout{\dual{s_2}}{\abs{(x_a, x_m, z)}{\appl{x}{(x_a, x_m, z)}}} \inact} \Par \appl{x}{(a,m, s_2)}}
\end{array}
\]

\noi {\bf Transition Semantics:} We can observe $\pmapp{\recp{X}{\bout{a}{m} \varp{X}}}{1}{}$ as:
\[
	\begin{array}{l}
		\pmapp{\recp{X}{\bout{a}{m} \varp{X}}}{1}{}\\
		\scong
		\\
		\newsp{s_1}{\bout{\dual{s_1}}{V} \inact \Par \binp{s_1}{x} \bout{a}{\abs{z}{\binp{z}{x} (\appl{x}{m})}} \newsp{s_2}{\bout{\dual{s_2}}{\abs{(x_a, x_m, z)}{\appl{x}{(x_a, x_m, z)}}} \inact} \Par \appl{x}{(a,m, s_2)}}
		\\
		\by{\tau}
		\\
		\bout{a}{\abs{z}{\binp{z}{x} (\appl{x}{m})}}\\
		\newsp{s_2}{\bout{\dual{s_2}}{V} \inact \Par \binp{s_2}{x} \bout{a}{\abs{z}{\binp{z}{x} (\appl{x}{m})}} \newsp{s_3}{\bout{\dual{s_3}}{\abs{(x_a, x_m, z)}{\appl{x}{(x_a, x_m, z)}}} \inact} \Par \appl{x}{(a,m, s_3)}}
		\\
		\scong_{\alpha}
		\\
		\bout{a}{\abs{z}{\binp{z}{x} (\appl{x}{m})}}\\
		\newsp{s_1}{\bout{\dual{s_1}}{V} \inact \Par \binp{s_1}{x} \bout{a}{\abs{z}{\binp{z}{x} (\appl{x}{m})}} \newsp{s_2}{\bout{\dual{s_2}}{\abs{(x_a, x_m, z)}{\appl{x}{(x_a, x_m, z)}}} \inact} \Par \appl{x}{(a,m, s_2)}}
		\\
		\scong
		\\
		\bout{a}{\abs{z}{\binp{z}{x} (\appl{x}{m})}} \pmapp{\recp{X}{\bout{a}{m} \varp{X}}}{1}{}
		\\
		\by{\bactout{a}{\abs{z}{\binp{z}{x} (\appl{x}{m})}}}
		\\
		\pmapp{\recp{X}{\bout{a}{m} \varp{X}}}{1}{}
	\end{array}
\]
\noi {\bf Typing Semantics:} We further show that $\pmapp{\recp{X}{\bout{a}{m} \varp{X}}}{1}{}$ is typable:
\begin{eqnarray}
	\tree{
		\begin{array}{l}
			\Gamma; \es; \es \proves a \hastype U_1 = \chtype{\lhot{\btinp{\lhot{U_2}} \tinact}} \\
			\Gamma; \es; \es \proves m \hastype U_2\\
			\Gamma; \es; s_2: \proves s_2 : \btinp{L} \tinact \proves s_2 \hastype \btinp{L} \tinact\\ 
			\Gamma; \es; \es \proves x \hastype \shot{(U_1, U_2, \btinp{L} \tinact)}
		\end{array}
	}{
		\Gamma; \es; s_2 : \btinp{L} \tinact \proves \appl{x}{(a,m, s_2)} \hastype \Proc
	}
	\label{ex:type1}
\end{eqnarray}
\begin{eqnarray}
	\tree{
		\tree{
			\begin{array}{l}
				\Gamma \cat x_a: U_1 \cat x_m: U_2; \es; \es \proves x_a \hastype U_1 = \chtype{\lhot{\btinp{\lhot{U_2}} \tinact}} \\
				\Gamma \cat x_a: U_1 \cat x_m: U_2 ; \es; \es \proves x_m \hastype U_2\\
				\Gamma; \es; z: \btinp{L} \tinact \proves z \hastype \btinp{L} \tinact\\
				\Gamma; \es; \es \proves x \hastype \shot{(U_1, U_2, \btinp{L} \tinact)}
			\end{array}
		}{
			\Gamma \cat x_a: U_1 \cat x_m: U_2 ; \es; z: \btinp{L} \tinact \proves \appl{x}{(x_a, x_m, z)} \hastype \Proc
		}
	}{
		\Gamma; \es; \es \proves \abs{(x_a, x_m, z)}{\appl{x}{(x_a, x_m, z)}} \hastype \shot{(U_1, U_2, \btinp{L} \tinact)}
	}
	\label{ex:type2}
\end{eqnarray}
\begin{eqnarray}
	\tree{
		\begin{array}{ll}
			\textrm{Result}~\eqref{ex:type2}
			\\
			\Gamma; \es; \dual{s_2} : \btout{\shot{(U_1, U_2, \btinp{L} \tinact)}} \tinact \proves \dual{s_2} \hastype \btout{\shot{(U_1, U_2, \btinp{L} \tinact)}} \tinact
		\end{array}
	}{
		\Gamma; \es; \dual{s_2} : \btout{\shot{(U_1, U_2, \btinp{L} \tinact)}} \tinact \proves \bout{\dual{s_2}}{\abs{(x_a, x_m, z)}{\appl{x}{(x_a, x_m, z)}}} \inact \hastype \Proc
	}
	\label{ex:type25}
\end{eqnarray}
\begin{eqnarray}
	\tree{
		\textrm{Result}~\eqref{ex:type1}
		\quad
		\textrm{Result}~\eqref{ex:type25}
		\quad
		\Delta = s_2: \btinp{L} \tinact \cat \dual{s_2}: \btout{\shot{(U_1, U_2, \btinp{L} \tinact)}} \tinact
	}{
		\Gamma; \es; \Delta \proves \bout{\dual{s_2}}{\abs{(x_a, x_m, z)}{\appl{x}{(x_a, x_m, z)}}} \inact \Par \appl{x}{(a,m, s_2)} \hastype \Proc
	}
	\label{ex:type3}
\end{eqnarray}
\begin{eqnarray}
	\tree{
		\begin{array}{l}
			\textrm{Result}~\eqref{ex:type3}
			\quad
			\btinp{L} \tinact \dualof \btout{\shot{(U_1, U_2, \btinp{L} \tinact)}} \tinact\\
			L = \shot{(U_1, U_2, \btinp{L} \tinact)} \textrm{ implies }\\
			\btinp{L} \tinact = \trec{t}{\btinp{\shot{(U_1, U_2, \vart{t})}} \tinact}
		\end{array}
	}{
		\Gamma; \es; \es \proves \newsp{s_2}{\bout{\dual{s_2}}{\abs{(x_a, x_m, z)}{\appl{x}{(x_a, x_m, z)}}} \inact \Par \appl{x}{(a,m, s_2)} \hastype \Proc}
	}
	\label{ex:type4}
\end{eqnarray}
\begin{eqnarray}
	\tree{
		\begin{array}{l}
			\textrm{Result}~\eqref{ex:type4}
			\\
			\Gamma; \es; \es \proves a \hastype \chtype{\lhot{\btinp{\lhot{U_2}} \tinact}}
			\\
			\Gamma; \es; \es \proves \abs{z}{\binp{z}{x} (\appl{x}{m})} \hastype \lhot{\btinp{\lhot{U_2}} \tinact}
		\end{array}
	}{
		\Gamma; \es; \es \proves \bout{a}{\abs{z}{\binp{z}{x} (\appl{x}{m})}} \newsp{s_2}{\bout{\dual{s_2}}{\abs{(x_a, x_m, z)}{\appl{x}{(x_a, x_m, z)}}} \inact \Par \appl{x}{(a,m, s_2)}} \hastype \Proc
	}
	\label{ex:type5}
\end{eqnarray}
\begin{eqnarray}
	\tree{
		\begin{array}{l}
			\textrm{Result}~\eqref{ex:type5}
			\quad
			\Gamma' = \Gamma \backslash x
			\\
			\Gamma; \es; \es \proves x \hastype \shot{(U_1, U_2, \trec{t}{\btinp{\shot{(U_1, U_2, t)}} \tinact})}
			\\
			\Gamma'; \es; \Delta \proves s_1 \hastype \btinp{\shot{(U_1, U_2, \trec{t}{\btinp{\shot{(U_1, U_2, \vart{t})}} \tinact})}} \tinact
		\end{array}
	}{
		\begin{array}{l}
			\Gamma'; \es; \Delta_1 \proves
			\\
			\quad\ \binp{s_1}{x} \bout{a}{\abs{z}{\binp{z}{x} (\appl{x}{m})}} \newsp{s_2}{\bout{\dual{s_2}}{\abs{(x_a, x_m, z)}{\appl{x}{(x_a, x_m, z)}}} \inact \Par \appl{x}{(a,m, s_2)}} \hastype \Proc
		\end{array}
	}
	\label{ex:type6}
\end{eqnarray}
\begin{eqnarray}
	\tree{
		\begin{array}{rcl}
			V &=& \abs{(x_a, x_m, z)} \binp{z}{x} \bout{x_a}{\abs{z}{\binp{z}{x} (\appl{x}{x_m})}}\\
			&& \newsp{s_2}{\appl{x}{(x_a,x_m, s_2)}  \Par \bout{\dual{s_2}}{\abs{(x_a, x_m, z)}{\appl{x}{(x_a, x_m, z)}}} \inact}
			\\
			\multicolumn{3}{l}{
				\Gamma'; \es; \es \proves V \hastype \shot{(U_1, U_2, \trec{t}{\btinp{\shot{(U_1, U_2, \vart{t})}} \tinact})}
			}
			\\
			\multicolumn{3}{l}{
				\Gamma'; \es; \Delta_2 \proves \dual{s_1} \hastype \btout{\shot{(U_1, U_2, \trec{t}{\btinp{\shot{(U_1, U_2, t)}} \tinact})}} \tinact
			}
		\end{array}
	}{
		\Gamma'; \es; \Delta_2 \proves \bout{\dual{s_1}}{V} \inact \hastype \Proc
	}
	\label{ex:type7}
\end{eqnarray}
\begin{eqnarray*}
	\tree{
		\tree{
			\textrm{Result}~\eqref{ex:type6} \quad \textrm{Result}~\eqref{ex:type7}
		}{
			\begin{array}{rcl}
				\Gamma; \es; \Delta_1 \cat \Delta_2 &\proves& \bout{\dual{s_1}}{V} \inact \Par \binp{s_1}{x} \bout{a}{\abs{z}{\binp{z}{x} (\appl{x}{m})}}\\
				&&\newsp{s_2}{\bout{\dual{s_2}}{\abs{(x_a, x_m, z)}{\appl{x}{(x_a, x_m, z)}}} \inact} \Par \appl{x}{(a,m, s_2)} \hastype \Proc
			\end{array}
		}
	}{
		\begin{array}{rcl}
			\Gamma; \es; \es &\proves& \newsp{s_1}{\bout{\dual{s_1}}{V} \inact \Par \binp{s_1}{x} \bout{a}{\abs{z}{\binp{z}{x} (\appl{x}{m})}}\\
			&&\newsp{s_2}{\bout{\dual{s_2}}{\abs{(x_a, x_m, z)}{\appl{x}{(x_a, x_m, z)}}} \inact} \Par \appl{x}{(a,m, s_2)}} \hastype \Proc
		\end{array}
	}
\end{eqnarray*}
\qed
\end{example}

\subsection{From \HOp to \sessp}
\label{subsec:HOp_to_p}

We now discuss the encodability of  $\HO$ into $\sessp$ where
we essentially follow the representability result put forward by 
Sangiorgi~\cite{San92,SaWabook}, but casted in the 
setting of session-typed communications. 
Intuitively, the strategy represents the exchange of a process 
with the exchange of a freshly generated \emph{trigger name}. 
Trigger names are used to activate copies of the process, 
which now becomes a persistent 
resource represented by an input-guarded replication.
In our calculi, a session name 
is a linear resource and cannot be replicated.
Consider the following (naive) adaptation of 
Sangiorgi's strategy in which session names are used are triggers and 
exchanged processes would be have to used exactly once:
\[
	\begin{array}{lcl}
		\pmap{\bout{u}{\abs{x}{Q}} P}{n} & \defeq &  \newsp{s}{\bout{u}{s} (\pmap{P}{n} \Par \binp{\dual{s}}{x} \pmap{Q}{n})} \\
		\pmap{\binp{u}{x} P}{n} & \defeq& \binp{u}{x} \pmap{P}{n}\\
		\pmap{\appl{x}{u}}{n} & \defeq & \bout{x}{u} \inact
	\end{array}
\]
with the remaining \HOp constructs being mapped homomorphically.
Although $\pmap{\cdot}{n}$ captures the correct semantics when
dealing with systems that allow only linear abstractions,
it suffers from non-typability in the presence
of shared abstractions. For instance,
mapping for $P = \bout{n}{\abs{x}{\bout{x}{m}\inact}} \inact \Par \binp{\dual{n}}{x} (\appl{x}{s_1} \Par \appl{x}{s_2})$
would be:
\[
	\pmap{P}{n} \defeq
	\newsp{s}{\bout{n}{s} \binp{\dual{s}}{x} \bout{x}{m} \inact \Par \binp{\dual{n}}{x} (\bout{x}{s_1} \inact \Par \bout{x}{s_2} \inact)}
\]
The above process is non typable since processes $(\bout{x}{s_1} \inact$ and $\bout{x}{s_2} \inact)$
cannot be put in parallel because they do not have disjoint session environments.

The correct approach would be to use replicated shared names
as triggers instead of session names, when dealing with shared abstractions. 
Below we write $\repl{} P$ as a shorthand notation for $\recp{X}{(P \Par \varp{X})}$.

\begin{definition}[Encoding \HOp to \sessp]\myrm
	\label{def:enc:HOp_to_p}
	Define encoding
	$\enco{\map{\cdot}^2, \mapt{\cdot}^2, \mapa{\cdot}^2}: \tyl{L}_{\HOp} \to \tyl{L}_{\sessp}$
	with mappings 
	$\map{\cdot}^{2}$, $\mapt{\cdot}^{2}$, $\mapa{\cdot}^{2}$ as
	in \figref{fig:enc:HOp_to_p}.
\end{definition}

\begin{figure}[t]
	\[
	\begin{array}{rcl}
		\pmap{\bout{u}{\abs{x}{Q}} P}{2} & \defeq &  \left\{
		\begin{array}{ll}
			\newsp{a}{\bout{u}{a} (\pmap{P}{2} \Par \repl{} \binp{a}{y} \binp{y}{x} \pmap{Q}{2})\,} & s \notin \fn{Q} \\
			\newsp{s}{\bout{u}{\dual{s}} (\pmap{P}{2} \Par \binp{s}{y} \binp{y}{x} \pmap{Q}{2})\,} & \textrm{otherwise} %\dk{Q \textrm{ linear}} \\
		\end{array}
		\right.
		\\
		\pmap{\binp{u}{x} P}{2} &\defeq&  \binp{u}{x} \pmap{P}{2}
		\\
		\pmap{\appl{x}{u}}{2} & \defeq & \newsp{s}{\bout{x}{s} \bout{\dual{s}}{u} \inact}
		\\
		\pmap{\appl{(\abs{x}{P})}{u}}{2} & \defeq & \newsp{s}{\binp{s}{x} \pmap{P}{2} \Par \bout{\dual{s}}{u} \inact}
%		\pmap{P}{2} \subst{u}{x}%\newsp{s}{\bout{a}{s} \bout{\dual{s}}{u} \inact}
%		\left\{
%		\begin{array}{ll}
%			\newsp{s}{\bout{\dual{a}}{s} \bout{\dual{s}}{u} \inact \Par \repl{} \binp{a}{y} \binp{y}{x} \pmap{P}{2}} & s \notin \fn{P} \\
%			\newsp{s}{\bout{\dual{a}}{s} \bout{\dual{s}}{u} \inact \Par \binp{a}{y} \binp{y}{x} \pmap{P}{2}} & \textrm{otherwise}
%		\end{array}
%		\right.
%		\left\{
%		\begin{array}{rcl}
%			v = x && \textrm{ if } V = x\\
%			v = m && \textrm{ if } V = \abs{x} P \wedge m \textrm{ fresh} \\
%		\end{array}
%		\right.
		\\
		\\
		\tmap{\btout{\shot{S}}S_1}{2} & \defeq & \bbtout{\chtype{\btinp{\tmap{S}{2}}\tinact}}\tmap{S_1}{2} \\
		\tmap{\btinp{\shot{S}}S_1}{2} & \defeq & \bbtinp{\chtype{\btinp{\tmap{S}{2}}\tinact}}\tmap{S_1}{2} \\

		\tmap{\btout{\lhot{S}}S_1}{2} & \defeq & \bbtout{\btinp{\tmap{S}{2}}\tinact}\tmap{S_1}{2} \\
		\tmap{\btinp{\lhot{S}}S_1}{2} & \defeq & \bbtinp{\btinp{\tmap{S}{2}}\tinact}\tmap{S_1}{2} \\
		\mapa{\news{\tilde{m}'}\bactout{n}{\abs{ x}{P}} }^{2} &  \defeq & \news{m} \bactout{n}{m} \\
		\mapa{\bactinp{n}{\abs{ x}{P}} }^{2} &  \defeq & \bactinp{n}{m} \quad m \textrm{ fresh}
	\end{array}
	\]
	\caption{
		Typed encoding of 
		\HOp to \sessp (\defref{def:enc:HOp_to_p}).
		\label{fig:enc:HOp_to_p}
		Mappings 
		$\map{\cdot}^3$,
		$\mapt{\cdot}^3$, 
		and 
		$\mapa{\cdot}^3$
		are homomorphisms for the other processes/types/labels. 
	}
\end{figure}

\begin{proposition}[Type Preservation, \HOp into \sessp]\myrm
	\label{prop:typepres_HOp_to_p}
	Let $P$ be a \HOp process. 
	If $\Gamma; \emptyset; \Delta \proves P \hastype \Proc$ then 
	$\mapt{\Gamma}^{2}; \emptyset; \mapt{\Delta}^{2} \proves \map{P}^{2} \hastype \Proc$.
\end{proposition}

\begin{proof}
	By induction on the inference $\Gamma; \emptyset; \Delta \proves P \hastype \Proc$. 
	Details in \propref{app:prop:typepres_HOp_to_p}
	(Page~\pageref{app:prop:typepres_HOp_to_p}).
	\qed
\end{proof}

\begin{remark}
	As stated in  \cite[Lem.\,5.2.2]{SangiorgiD:expmpa}, 
	due to the replicated trigger,  
	operational correspondence in \defref{def:ep} is refined to prove  
	full abstraction: 
	e.g., completeness of the case $\ell_1 \neq \tau$, is changed as follows.
	Suppose:
	\[
		\horel{\Gamma}{\Delta}{P}{\hby{\ell_1}}{\Delta'}{P'}{}
	\]
	If $\ell_1 = \news{\tilde{m}} \bactout{n}{\abs{ x}{R}}$, 
	then %$\exists \ell_2, Q$ s.t. 
	\[
		\horel{\tmap{\Gamma}{2}}{\tmap{\Delta}{2}}{\pmap{P}{2}}{\hby{\ell_2}}{\tmap{\Delta'}{2}}{Q}{}
	\]
	where  $\ell_2 = (\nu a)\bactout{n}{a}$ and
	$Q = \pmap{P' \Par  \repl{} \binp{a}{y} \binp{y}{x} R}{2}$.

	\noi Similarly, if  
	%$\stytraarg{\Gamma}{\ell_1}{\Delta}{P}{\Delta'}{P'}{}$
	%with 
	$\ell_1 = \bactinp{n}{\abs{ x}{R}}$, 
	then %$\exists \ell_2, Q$ s.t.
	\[
		\horel{\tmap{\Gamma}{2}}{\tmap{\Delta}{2}}{\pmap{P}{2}}{\hby{\ell_2}}{\tmap{\Delta'}{2}}{Q}{}
	\]
	where $\ell_2 = \bactout{n}{a}$ and
	$\pmap{P'}{2} \wb \newsp{a}{Q \Par  \repl{} \binp{a}{y} \binp{y}{x} \pmap{R}{2}}$.
	Soundness is stated in a symmetric way.
	%Operational correspondence for the encoding in~\defref{d:enc:hopitopi}
	%is different from that in~\defref{def:ep}, due to triggers. 
	%In particular,  completeness differs when $\ell_1 \neq \tau$.
	%This way, e.g., if  
	%$\stytraarg{\Gamma}{\ell_1}{\Delta}{P}{\Delta'}{P'}{}$
	%with $\ell_1 = (\nu \tilde{m})\bactout{n}{\abs{ x}{R}}$, 
	%then %$\exists \ell_2, Q$ s.t. 
	%$\stytraarg{\mapt{\Gamma}^2}{\ell_2}{\mapt{\Delta}^2}{\map{P}^2}{\mapt{\Delta'}^2}{Q}{}$,
	%where 
	%$\ell_2 = (\nu a)\bactout{n}{a}$ and
	%$Q = \pmap{P' \Par  \repl{} \binp{a}{y} \binp{y}{x} R}{2}$.
	%This 
	%statement, essential in proofs of full abstraction,
	%is the same given by Sangiorgi~\cite{SangiorgiD:expmpa}.
	%Completeness is as in~\defref{def:ep} when  $\ell_1 = \tau$.
	%See~\cite{KouzapasPY15} for details.
\end{remark}

This last remark is stated formally in the next proposition:
\begin{proposition}[Operational Correspondence, \HOp into \sessp]\myrm
	\label{prop:op_corr_HOp_to_p}
	Let $P$ be an  $\HOp$ process such that  $\Gamma; \emptyset; \Delta \proves P \hastype \Proc$.
	
	\begin{enumerate}[1.]
		\item Suppose $\horel{\Gamma}{\Delta}{P}{\hby{\ell_1}}{\Delta'}{P'}$.
		Then we have:
		\begin{enumerate}[a)]
			\item
				If  $\ell_1 = \news{\tilde{m}}\bactout{n}{\abs{x}Q}$,
				then $\exists \Gamma', \Delta'', R$ where either:
				\begin{enumerate}[-]
					\item 
						$\tmap{\Gamma}{2};\, \tmap{\Delta}{2} \proves  \pmap{P}{2} 
						\hby{\mapa{\ell_1}^{2}}
						\Gamma' \cdot \tmap{\Gamma}{2};\, \tmap{\Delta'}{2} \proves \pmap{P'}{2} \Par \repl{} \binp{a}{y} \binp{y}{x} \pmap{Q}{2}$
					\item 
						$\tmap{\Gamma}{2};\, \tmap{\Delta}{2} \proves \pmap{P}{2} 
						\hby{\mapa{\ell_1}^{2}}
						\tmap{\Gamma}{2};\, \Delta'' \proves \pmap{P'}{2} \Par \binp{s}{y} \binp{y}{x} \pmap{Q}{2}$
				\end{enumerate}

			\item
				If   
				$\ell_1 = \bactinp{n}{\abs{y}Q}$
				then $\exists R$ where
				either
				\begin{enumerate}[-]
					\item 
						$\tmap{\Gamma}{2};\, \tmap{\Delta}{2} \proves \pmap{P}{2} 
						\hby{\mapa{\ell_1}^{2}}
						\Gamma';\, \tmap{\Delta''}{2} \proves  R$, for some $ \Gamma'$
						and \\ 
						$\horel{\tmap{\Gamma}{2}}{\tmap{\Delta'}{2}}{\pmap{P'}{2}}{\wb}{\tmap{\Delta''}{2}}{\newsp{a}{R \Par \repl{} \binp{a}{y} \binp{y}{x} \pmap{Q}{2}}}$
					\item 
						$\tmap{\Gamma}{2};\, \tmap{\Delta}{2} \proves \pmap{P}{2}
						\hby{\mapa{\ell_1}^{2}}
						\tmap{\Gamma}{2};\, \tmap{\Delta''}{2} \proves R$, 
						and \\ 
						$\horel{\tmap{\Gamma}{2}}{\tmap{\Delta'}{2}}{\pmap{P'}{2}}{\wb}{\tmap{\Delta''}{2}}{\newsp{s}{R \Par \binp{s}{y} \binp{y}{x} \pmap{Q}{2}}}$  		
				\end{enumerate}

			\item	If
				$\ell_1 = \tau$ then either:

				\begin{enumerate}[-]
					\item	$\exists R$ such that
						\[
						\mhorel{\tmap{\Gamma}{2}}{\tmap{\Delta}{2}}{\pmap{P}{2}}
						{\hby{\tau}}
						{\tmap{\Delta'}{2}}{}{\newsp{\tilde{m}}{\pmap{P_1}{2} \Par \newsp{a}
						{\pmap{P_2}{2}\subst{a}{x} \Par \repl{} \binp{a}{y} \binp{y}{x} \pmap{Q}{2}}}}
						\]

					\item	$\exists R$ such that
						\[
						\mhorel{\tmap{\Gamma}{2}}{\tmap{\Delta}{2}}{\pmap{P}{2}}
						{\hby{\tau}}
						{\tmap{\Delta'}{2}}{}{\newsp{\tilde{m}}{\pmap{P_1}{2} \Par \newsp{s}
						{\pmap{P_2}{2}\subst{\dual{s}}{x} \Par \binp{s}{y} \binp{y}{x} \pmap{Q}{2}}}}
						\]

					\item	%$\ell_1 = \btau$ and
						$\tmap{\Gamma}{2};\, \tmap{\Delta}{2} \proves \pmap{P}{2}
						\hby{\tau}
						\tmap{\Gamma}{2};\, \tmap{\Delta'}{2} \proves \pmap{P'}{2}$

					\item	$\ell_1 = \btau$ and
						$\tmap{\Gamma}{2};\, \tmap{\Delta}{2} \proves \pmap{P}{2}
						\hby{\stau}
						\tmap{\Gamma}{2};\, \tmap{\Delta'}{2} \proves \pmap{P'}{2}$
				\end{enumerate}

%			\item	 
%				If  
%				%$\stytra{\Gamma}{\ell_1}{\Delta}{P}{\Delta'}{P_1 \Par P_2\subst{\abs{x}Q}{X}}$
%				$\ell_1 = \tau$ and $P' 	\not \scong \news{\tilde{m}}(P_1 \Par P_2\subst{\abs{x}Q}{X})$
%				then \\
%				$\mapt{\Gamma}^{2};\, \mapt{\Delta}^{2} \proves  \map{P}^{2}
%				\hby{\tau}
%				\mapt{\Gamma}^{2};\, \mapt{\Delta'}^{2} \proves  \map{P'}^{2}$.
				   			   
%			   then  $\exists \ell_2$ s.t. 
%			    $\wtytra{\mapt{\Gamma}^{3}}{\ell_2}{\mapt{\Delta}^{3}}{\map{P}^{3}}{\mapt{\Delta'}^{3}}{\map{P'}^{3}}$
%			    and $\ell_2 = \mapa{\ell_1}^{3}$.

			\item	 
				If  
				$\ell_1 \in \set{\bactsel{n}{l}, \bactbra{n}{l}}$
				%\not\in \set{\tau,\, \news{\tilde{m}}\bactout{n}{\abs{x}Q}, \, \bactinp{n}{\abs{x}Q}}$ 
				 then \\
				$\exists \ell_2 = \mapa{\ell_1}^{2}$ such that 
				$\mapt{\Gamma}^{2};\, \mapt{\Delta}^{2} \proves  \map{P}^{2}
				\hby{\ell_2}
				\mapt{\Gamma}^{2};\, \mapt{\Delta'}^{2} \proves  \map{P'}^{2}$.			
		\end{enumerate}
		
		%%%%%%% SOUNDNESSS
		\item Suppose 
		$\stytra{\mapt{\Gamma}^{2}}{\ell_2}{\mapt{\Delta}^{2}}{\map{P}^{2}}{\mapt{\Delta'}^{2}}{R}$.
			\begin{enumerate}[a)]
				\item %% soutput
					%\footnote{$\mapt{\Gamma}^{2}$ in the following three items need adjustments.}
					If  
					$\ell_2 = \news{m}\bactout{n}{m}$
					%$\stytra{\mapt{\Gamma}^{2}}{\news{m}\bactout{n}{m}}{\mapt{\Delta}^{2}}{\map{P}^{2}}{\mapt{\Delta'}^{2}}{R}$
					then 
					either 
					\begin{enumerate}[-]
					\item	$\exists P'$ such that $P \hby{\news{m} \bactout{n}{m}} P'$
						and $R = \pmap{P'}{2}$.

					\item	$\exists Q, P'$ such that $P \hby{\bactout{n}{\abs{x}Q}} P'$
						and $R = \map{P'}^{2} \Par \repl{} \binp{a}{y} \binp{y}{x} \pmap{Q}{2}$

					\item	$\exists Q, P'$ such that $P \hby{\bactout{n}{\abs{x}Q}} P'$
						and $R = \map{P'}^{2} \Par \binp{s}{y} \binp{y}{x} \pmap{Q}{2}$
					\end{enumerate}

				\item   %% sinput
					If  $\ell_2 = \bactinp{n}{m}$ 
					%$\stytra{\mapt{\Gamma}^{2}}{\bactinp{n}{m}}{\mapt{\Delta}^{2}}{\map{P}^{2}}{\mapt{\Delta'}^{2}}{R}$
					then either
					\begin{enumerate}[-]
					\item	$\exists P'$ such that $P \hby{\bactinp{n}{m}} P'$
						and $R = \pmap{P'}{2}$.

					\item	$\exists Q, P'$ such that
						$P \hby{\bactinp{n}{\abs{x}Q}} P'$\\
						and $\horel{\mapt{\Gamma}^{2}}{\mapt{\Delta'}^{2}}{\map{P'}^{2}}{\wb}{\mapt{\Delta'}^{2}}{\news{a}(R \Par \repl{} \binp{a}{y} \binp{y}{x} \pmap{Q}{2})}$
					\item	$\exists Q, P'$ such that
						$P \hby{\bactinp{n}{\abs{x}Q}} P'$\\
						and $\horel{\mapt{\Gamma}^{2}}{\mapt{\Delta'}^{2}}{\map{P'}^{2}}{\wb}{\mapt{\Delta'}^{2}}{\news{s}(R \Par \binp{s}{y} \binp{y}{x} \pmap{Q}{2})}$  
					\end{enumerate}
		
				\item   
					If  %$\stytra{\mapt{\Gamma}^{2}}{\tau}{\mapt{\Delta}^{2}}{\map{P}^{2}}{\mapt{\Delta'}^{2}}{R}$
					$\ell_2 = \tau$ 
					then $\exists P'$ such that
					$P \hby{\tau} P'$
					and $\horel{\mapt{\Gamma}^{2}}{\mapt{\Delta'}^{2}}{\map{P'}^{2}}{\wb}{\mapt{\Delta'}^{2}}{R}$.
				\item	 
					If  
					$\ell_2 \not\in \set{\bactout{n}{m}, \bactsel{n}{l}, \bactbra{n}{l}}$ 
					 then 
					$\exists \ell_1$ such that 
					$\ell_1 = \mapa{\ell_2}^{2}$ and \\
					$ \Gamma ;\, \Delta  \proves   P
					\hby{\ell_1}
					\Gamma ;\, \Delta  \proves   P'$.
		\end{enumerate}
	\end{enumerate}
\end{proposition}

\begin{proof}
	The proof is done by induction on the labelled transition system
	considering \defref{def:enc:HOp_to_p}.
	The most demaning cases are Part 1b and Part 2b where
	we require a further induction to proof bisimulation
	closure.

	Details of the proof of the most demanding cases can be
	found in \propref{app:prop:op_corr_HOp_to_p}
	(page \pageref{def:enc:HOp_to_p}).
	\qed
\end{proof}

\begin{proposition}[Full Abstraction, From \HOp to \sessp]\myrm
	\label{prop:fulla_HOp_to_p}
	Let $P_1, Q_1$ be \HOp processes.
	$\horel{\Gamma}{\Delta_1}{P_1}{\hwb}{\Delta_2}{Q_1}$
	if and only if
	$\horel{\tmap{\Gamma}{2}}{\tmap{\Delta_1}{2}}{\pmap{P_1}{2}}{\fwb}{\tmap{\Delta_2}{2}}{\pmap{Q_1}{2}}$.
\end{proposition}

\begin{proof}
%	The proof for the soundness direction considers
%	closure that can be shown to be a bisimulation
%	following the soundness direction of Operational Correspondence
%	(\propref{prop:op_corr_HOp_to_p}). Whenever needed
%	the proof makes use of the $\tau$-inertness result
%	(\propref{lem:tau_inert}).

%	The proof for the completness direction also considers
%	a closure shown to be a bisimulation
%	up-to deterministic transition (\propref{lem:up_to_deterministic_transition})
%	following the completeness direction of Operational Correspondence
%	(\propref{prop:op_corr_HOp_to_p}).
%
	Proof follows directly from \propref{prop:op_corr_HOp_to_p}. The cases
	of \propref{prop:op_corr_HOp_to_p} are used to create a
	bisimulation closure to prove the the soundness direction and
	a bisimulation up to determinate transition (\lemref{lem:up_to_deterministic_transition})
	to prove the
	completeness direction.
	\qed
\end{proof}

\begin{proposition}[Precise encoding of \HOp into \sessp]\myrm
	\label{prop:prec:HOp_to_p}
	The encoding from $\tyl{L}_{\HOp}$ to $\tyl{L}_{\sessp}$
	is precise.
\end{proposition}

\begin{proof}
	Syntactic requirements are easily derivable from the
	definition of the mappings in \figref{fig:enc:HOp_to_p}.
	Semantic requirements are a consequence of
	\propref{prop:typepres_HOp_to_p}, \propref{prop:op_corr_HOp_to_p}, and \propref{prop:fulla_HOp_to_p}.
	\qed
\end{proof}

% !TEX root = main.tex
\section{Negative Encodability Results}
	\label{s:negative}
	\label{sec:negative}

%In the encoding from $\HOp$ to $\sessp$ we showed that
%an easy and straightforward encoding would be to create
%a new shared name for every abstraction we want to pass
%in order to use it as a trigger that activate copies of
%the abstraction.
%
%At this point a reasonable question could be whether we can
%encode shared name behaviour to session name behaviour and at
%the same time maintain the type, operational and behavioural semantics.
%If such result holds then its impact would be much bigger than
%the encoding from $\HOp$ to $\sessp$, since it would
%allow us to have session type systems without shared names
%and still have the modelling convenience of shared names.

As most session calculi, 
\HOp includes communication on both shared and linear channels.
The former enables non determinism and unrestricted behavior; the latter allows to represent
deterministic and linear communication structures.
The expressive power of shared names is also illustrated by our 
encoding from \HOp into \sessp (\defref{def:enc:HOp_to_p}).
Shared and linear channels are fundamentally different; still, to the best of our knowledge,
the status of shared communication, in terms of expressiveness, has not been formalized for session calculi.

The above begs the question: 
%are shared names truly indispensable for communication, or could they be encoded using linear communication?
can we represent shared name interaction using session name interaction?
In this section we prove that shared names actually add expressiveness to \HOp,
for their behavior cannot be represented using purely deterministic processes.
To this end, we show the non existence of a minimal encoding 
(cf.~\defref{def:goodenc}(ii))
of shared name communication into linear 
communication. Recall that minimal encodings preserve barbs (\propref{p:barbpres}).

\begin{theorem}\myrm
	\label{thm:negative}
	Let $\CAL_1, \CAL_2 \in \set{\HOp, \HO, \sessp}$.
	There is no typed, minimal encoding from $\tyl{L}_{\CAL_1}$ into $\tyl{L}_{\CAL_2^{\minussh}}$
%	$\enco{\map{\cdot}, \mapt{\cdot}, \mapa{\cdot}}: \sessp \longrightarrow \HOp^{\minussh}$.
%	that enjoys: (i) homomorphism wrt parallel; (ii) barb preservation; (iii) operational completeness.
\end{theorem}

\begin{proof}
	Assume, towards a contradiction, that such a typed encoding indeed exists. 
	Consider the $\sessp$ process
	\[
		P = \breq{a}{s} \inact \Par \bacc{a}{x} \bsel{n}{l_1} \inact \Par \bacc{a}{x} \bsel{m}{l_2} \inact \qquad \text{(with $n \neq m$)}
	\]
	\noi such that 
	$\Gamma; \es; \Delta \proves P \hastype \Proc$.
	From process $P$ we have: %We then have both
	\begin{eqnarray}
		& & \horel{\Gamma}{\Delta}{P}{\hby{\tau}}{\Delta'}{\bsel{n}{l_1} \inact \Par \bacc{a}{x} \bsel{m}{l_2} \inact = P_1} \label{eq:nn3} \\
		& & \horel{\Gamma}{\Delta}{P}{\hby{\tau}}{\Delta'}{\bsel{m}{l_2} \inact \Par \bacc{a}{x} \bsel{n}{l_1} \inact = P_2} \label{eq:nn4}
	\end{eqnarray}
	Thus, by definition of typed barb we  have:
	\begin{eqnarray}
		\Gamma; \Delta' \proves P_1 \barb{n} & \land & 
		\Gamma; \Delta' \proves P_1 \nbarb{m} \label{eq:nn1} \\
		\Gamma; \Delta' \proves P_2 \barb{m} & \land & 
		\Gamma; \Delta' \proves P_2 \nbarb{n} \label{eq:nn2}
	\end{eqnarray}
	Consider now the $\HOp^{\minussh}$ process $\map{P}$.
	% = 
	% \map{\breq{a}{s} \inact} \Par \map{\bacc{a}{x} \bsel{n}{l_1} \inact} \Par \map{\bacc{a}{x} \bsel{m}{l_2}}$.
	By our assumption of operational completeness 
	(\defref{def:ep}-2(a)), 
	from \eqref{eq:nn3} with \eqref{eq:nn4}
	we infer that
	there exist $\HOp^{\minussh}$ processes $S_1$ and $S_2$ such that:
	%we have both:
	\begin{eqnarray}
		& & \horel{\mapt{\Gamma}}{\mapt{\Delta}}{\map{P}}{\Hby{\stau}}{\mapt{\Delta'}}{S_1 \WB \map{P_1}} \label{eq:n1} \\
		& & \horel{\mapt{\Gamma}}{\mapt{\Delta}}{\map{P}}{\Hby{\stau}}{\mapt{\Delta'}}{S_2 \WB \map{P_2}} \label{eq:n2}
		%\map{P} & \Hby{} &  S_1 \WB \map{P_1} \\
		%s\map{P} & \Hby{} & S_2 \WB \map{P_2}
	\end{eqnarray}
	By our assumption of barb preservation, 
	from \eqref{eq:nn1} with \eqref{eq:nn2}
	we infer: 
	\begin{eqnarray}
		\mapt{\Gamma}; \mapt{\Delta'} \proves \map{P_1} \Barb{n} & \land & 
		\mapt{\Gamma}; \mapt{\Delta'} \proves \map{P_1} \nBarb{m} \label{eq:n3} \\
		\mapt{\Gamma}; \mapt{\Delta'} \proves \map{P_2} \Barb{m} & \land & 
		\mapt{\Gamma}; \mapt{\Delta'} \proves \map{P_2} \nBarb{n} \label{eq:n4}
	\end{eqnarray}
	By definition of $\WB$, 
	by combining~\eqref{eq:n1} with~\eqref{eq:n3}
	and~\eqref{eq:n2} with~\eqref{eq:n4}, we infer barbs for $S_1$ and $S_2$:
	\begin{eqnarray}
		\mapt{\Gamma}; \mapt{\Delta'} \proves S_1 \Barb{n} & \land & 
		\mapt{\Gamma}; \mapt{\Delta'} \proves S_1 \nBarb{m} \label{eq:n5} \\
		\mapt{\Gamma}; \mapt{\Delta'} \proves S_2 \Barb{m} & \land & 
		\mapt{\Gamma}; \mapt{\Delta'} \proves S_2 \nBarb{n} \label{eq:n6}
	\end{eqnarray}
	That is, $S_1$ and $\map{P_1}$ 
	(resp. $S_2$ and $\map{P_2}$)
	have the same barbs.
	Now, by $\tau$-inertness (\propref{lem:tau_inert}), we have both 
	\begin{eqnarray}
		& & \horel{\mapt{\Gamma}}{\mapt{\Delta}}{S_1}{\WB}{\mapt{\Delta'}}{\map{P}} \label{eq:n7} \\
		& & \horel{\mapt{\Gamma}}{\mapt{\Delta}}{S_2}{\WB}{\mapt{\Delta'}}{\map{P}} \label{eq:n8}
	\end{eqnarray}
	Combining~\eqref{eq:n7} with~\eqref{eq:n8}, by transitivity of $\WB$,
	we have 
	\begin{equation}
		\horel{\mapt{\Gamma}}{\mapt{\Delta'}}{S_1}{\WB}{\mapt{\Delta'}}{S_2} \label{eq:n9}
	\end{equation}
	In turn, from~\eqref{eq:n9}
	we infer that 
	it must be the case that:
	\begin{eqnarray*}
		\mapt{\Gamma}; \mapt{\Delta'} \proves \map{P_1} \Barb{n} & \land & 
		\mapt{\Gamma}; \mapt{\Delta'} \proves \map{P_1} \Barb{m} \label{eq:n10} \\
		\mapt{\Gamma}; \mapt{\Delta'} \proves \map{P_2} \Barb{m} & \land & 
		\mapt{\Gamma}; \mapt{\Delta'} \proves \map{P_2} \Barb{n} \label{eq:n11}
	\end{eqnarray*}
	which clearly contradict \eqref{eq:n3} and \eqref{eq:n4} above.
	\qed
\end{proof}

%\input{examples}

% !TEX root = main.tex
\section{Extensions of \HOp}
\label{sec:extension}

This section studies (i) the extension of \HOp with higher-order applications/abstractions (denoted \HOpp), and (ii) the extension of \HOp with polyadicity (denoted \pHOp). In both cases, we detail required modifications in the syntax and types, and describe further encodability results.

\subsection{Encoding \HOpp into \HOp}

The \HOp calculus is purposefully minimal and allows
only name applications/abstractions (also referred to as \emph{first-order} applications/abstractions).
We now introduce \HOpp, the 
extension of \HOp with higher-order applications.
We show that \HOpp has a precise encoding
into \HOp (\propref{prop:prec:HOpp_to_HOp}). 
Therefore, since 
typed encodings are composable (\propref{prop:enc_composability}), 
\HOpp has a precise encoding to \HO and \sessp.
In turn, this latter result implies that \HO is powerful
enough to express full higher-order semantics.

\subsubsection{Modifications in Syntax, Reduction Semantics, and Types.}
%\myparagraph{Syntax and Operational Semantics:}
The syntax of \HOpp processes is obtained from the
syntax for processes given in \figref{fig:syntax} by replacing 
$\appl{V}{u}$ with $\appl{W}{V}$.
Reduction is then defined by 
the rules in \figref{fig:reduction}, excepting 
rule $\orule{App}$, which is replaced by 
the following rule 
$$
\orule{App$^+$} \qquad
\appl{(\abs{x}{P})}{V} \red P \subst{V}{x}
$$
%The structural congruence in \secref{subsec:reduction_semantics}
%is changed to include $\appl{(\abs{x} P)}{V} \scong P \subst{x}{V}$
%instead of $\appl{(\abs{x} P)}{u} \scong P \subst{x}{u}$.
%\myparagraph{Types}
The syntax of types in \figref{def:types}
is generalized by including $$L \bnfis \shot{U} \bnfbar \lhot{U}$$
instead of $L \bnfis \shot{C} \bnfbar \lhot{C}$.
Definitions of type equivalence/duality 
and typing environments ($\Gamma$ and $\Lambda$) are straightforward extensions of 
\defref{def:type_equiv}, \defref{def:type_dual},
and  \defref{def:typeenv}, respectively. 
The typing rules of \figref{fig:typerulesmy} are then modified accordingly:
most significant changes are required in rules $\trule{Abs}$ and $\trule{App}$ 
(for typing abstractions and applications, respectively), which  for \HOpp processes are modified as follows:
\[
	\begin{array}{c}
		\trule{Abs$^+$}~~\tree{
			\Gamma; \Lambda; \Delta_1 \proves P \hastype \Proc
			\quad
			\Gamma; \es; \Delta_2 \proves x \hastype U
		}{
			\Gamma; \Lambda; \Delta_1 \backslash \Delta_2 \proves \abs{x}{P} \hastype \lhot{U}
		}
		\\[6mm]

		\trule{App$^+$}~~\tree{
			\begin{array}{c}
				U = \lhot{U'} \lor \shot{U'}
				\quad
				\Gamma; \Lambda; \Delta_1 \proves V \hastype U
				\quad
				\Gamma; \es; \Delta_2 \proves W \hastype U'
			\end{array}
		}{
			\Gamma; \Lambda; \Delta_1 \cat \Delta_2 \proves \appl{V}{W} \hastype \Proc
		} 
	\end{array}
\]

%\iftodo\dk{prove subject reduction for the extension}\else\fi

With these modifications we can now state the extension of \thmref{thm:sr}:

\begin{theorem}[Type Soundness for \HOpp]\myrm
	\label{thm:sr_hopp}
	\begin{enumerate}[1.]
		\item	(Subject Congruence)
			$\Gamma; \es; \Delta \proves P \hastype \Proc$
			and
			$P \scong P'$
			implies
			$\Gamma; \es; \Delta \proves P' \hastype \Proc$.

		\item	(Subject Reduction)
			$\Gamma; \es; \Delta \proves P \hastype \Proc$
			with
			balanced $\Delta$
			and
			$P \red P'$
			implies $\Gamma; \es; \Delta'  \proves P' \hastype \Proc$
			and either (i)~$\Delta = \Delta'$ or (ii)~$\Delta \red \Delta'$
			with $\Delta'$ balanced.
	\end{enumerate}
\end{theorem}

\begin{proof}
Part (1) is as for \HOp processes.
Part (2) 
is also as before, but 
requires the expected generalization of parts (3) and (4) of the substitution lemma (\lemref{lem:subst}).
We describe the analysis when the reduction is inferred by rule \orule{App$^+$}. We have
	   $$
	   P = (\abs{x}{Q}) \, V   \red  Q \subst{V}{x} = P'
	   $$
	   Suppose $\Gamma;\, \emptyset ;\, \Delta \proves (\abs{x}{Q}) \, V \hastype \Proc$. 
	   We examine one possible way in which 
	   this assumption can be derived; other cases are similar or simpler:
	   \[
	   \tree{
	   \tree{\Gamma, x:\lhot{L_1};\, \emptyset ;\, \Delta \proves Q  \hastype \Proc \quad 
	   \Gamma, x:\lhot{L_1};\, \emptyset ;\, \es \proves x  \hastype \lhot{L_1}}
	   {
	   \Gamma;\, \emptyset ;\, \Delta \proves \abs{x}{Q}  \hastype \shot{(\shot{L_1})} }
	   \qquad
	   \tree{}{
	   \Gamma ;\, \es ;\, \emptyset \proves   V \hastype \shot{L_1}}
	   }{
	   \Gamma;\, \emptyset ;\, \Delta  \proves (\abs{x}{Q}) \, V \hastype \Proc
	   }
	   \]
	  Then, by combining premise
	   $\Gamma, x:\lhot{L_1};\, \emptyset ;\, \Delta \proves Q  \hastype \Proc$
	   with
	   the extended formulation of \lemref{lem:subst}(4),
	   we obtain 
	    $\Gamma;\, \emptyset ;\, \Delta   \proves Q\subst{V}{x}  \hastype \Proc$, as desired.
	\qed
\end{proof}

\noi 
As for the behavioural semantics of \HOpp, modifications are as expected.
The set of action labels remains the same.
In the untyped LTS, rule $\ltsrule{App}$
is replaced with rule
$\appl{\abs{x}{P}}{V} \by{\tau} P \subst{V}{x}$.
\defref{def:characteristic_process} (characteristic processes)
now includes
\begin{eqnarray*}
\mapchar{\shot{U}}{x} & \defeq & \mapchar{\lhot{U}}{x} \defeq \appl{x}{\omapchar{U}} \\
\omapchar{\shot{U}} & \defeq & \omapchar{\lhot{U}} \defeq \abs{x}{\mapchar{U}{x}}
\end{eqnarray*}
instead of $\mapchar{\shot{C}}{x} \defeq \mapchar{\lhot{C}}{x} \defeq \appl{x}{\omapchar{C}}$
and
$\omapchar{\shot{C}} \defeq \omapchar{\lhot{C}} \defeq \abs{x}{\mapchar{C}{x}}$, respectively.
\noi The rest of the definitions for the behavioural semantics is kept unchanged. 

\subsubsection{Encoding \HOpp into \HOp.}

We now present an encoding from \HOpp to \HOp.
\begin{definition}[Encoding from \HOpp to \HOp]\myrm
	\label{def:enc:HOpp_to_HOp}
	Let $\tyl{L}_{\HOpp}=\calc{\HOpp}{{\cal{T}}_4}{\hby{\ell}}{\wb}{\proves}$
where 
${\cal{T}}_4$ is a set of types of $\HOpp$;  
the typing $\proves$ is defined in 
\figref{fig:typerulesmy} with extended rules \trule{Abs} and \trule{App}. 
Then, mapping $\enco{\pmap{\cdot}{3}, \tmap{\cdot}{3}, \mapa{\cdot}^{3}}: \tyl{L}_\HOpp \to \tyl{L}_\HOp$
	is defined in \figref{fig:enc:HOpp_to_HOp}.
\end{definition}
\begin{figure}[t]
	\[
	\begin{array}{rcl}
		\pmap{\appl{x}{(\abs{y} P})}{3} &\defeq& \newsp{s}{\appl{x}{s} \Par \bout{\dual{s}}{\abs{y} \pmap{P}{3}} \inact}
		\\
		\pmap{\appl{(\abs{x} P)}{(\abs{y} Q)}}{3} &\defeq& \newsp{s}{\binp{s}{x}\pmap{P}{3} \Par \bout{\dual{s}}{\abs{y} \pmap{Q}{3}} \inact}
		\\
		\pmap{\bout{u}{\abs{\underline{x}}{Q}} P}{3} &\defeq& \bout{u}{\abs{z}{\binp{z}{\underline{x}} \pmap{Q}{3}}} \pmap{P}{3}
		\\
		\pmap{\bout{u}{\abs{k}{Q}} P}{3} &\defeq& \bout{u}{\abs{k}{\pmap{Q}{3}}} \pmap{P}{3}
		\\
		\\
		\tmap{\shot{L}}{3} &\defeq& \shot{\big(\btinp{\tmap{L}{3}} \tinact\big)}
		\\
		\tmap{\lhot{L}}{3} &\defeq& \lhot{\big(\btinp{\tmap{L}{3}} \tinact\big)}
		\\
		\tmap{\btout{\shot{L}} S}{3} &\defeq& \btout{\tmap{\shot{L}}{3}} \tmap{S}{3}
		\\
		\tmap{\btout{\lhot{L}} S}{3} &\defeq& \btout{\tmap{\lhot{L}}{3}} \tmap{S}{3}
		\\
		\tmap{\btinp{\shot{L}} S}{3} &\defeq& \btinp{\tmap{\shot{L}}{3}} \tmap{S}{3}
		\\
		\tmap{\btinp{\lhot{L}} S}{3} &\defeq& \btinp{\tmap{\lhot{L}}{3}} \tmap{S}{3}
		\\
		\\
		\mapa{\news{\tilde{m}} \bactout{n}{\abs{k}{P}}}^{3} &\defeq& \news{\tilde{m}} \bactout{n}{\abs{x}{\pmap{P}{3}}}
		\\
		\mapa{\bactinp{n}{\abs{k}{P}}}^{3} &\defeq& \bactinp{n}{\abs{x}{\pmap{P}{3}}}
		\\
		\mapa{\news{\tilde{m}} \bactout{n}{\abs{\underline{x}}{P}}}^{3} &\defeq& \news{\tilde{m}} \bactout{n}{\abs{z}{\binp{z}{x} \pmap{P}{3}}}
		\\
		\mapa{\bactinp{n}{\abs{\underline{x}}{P}}}^{3} &\defeq& \bactinp{n}{\abs{z}{\binp{z}{x} \pmap{P}{3}}}

	\end{array}
	\]
	\caption{Encoding of \HOpp into \HOp (cf.~\defref{def:enc:HOpp_to_HOp}).
	We assume that the rest of the encoding is homomorphic on the syntax of
	processes, types and labels, respectively. \label{fig:enc:HOpp_to_HOp}}
\end{figure}

\begin{proposition}[Type Preservation. From \HOpp to \HOp]\myrm
	\label{prop:typepres_HOpp_to_HOp}
	Let $P$ be a \HOpp process.
	If $\Gamma; \emptyset; \Delta \proves P \hastype \Proc$ then 
	$\tmap{\Gamma}{3}; \emptyset; \tmap{\Delta}{3} \proves \pmap{P}{3} \hastype \Proc$. 
\end{proposition}

\begin{proof}
	The proof is a mechanical induction on the structure of $P$.
	Details of the proof in \propref{app:prop:typepres_HOpp_to_HOp}
	(page~\pageref{app:prop:typepres_HOpp_to_HOp}).
	\qed
\end{proof}

\begin{proposition}[Operational Correspondence. From \HOpp to \HOp]\myrm
	\label{prop:op_corr_HOpp_to_HOp}
	\begin{enumerate}
		\item	Let $\Gamma; \es; \Delta \proves P$.
			$\horel{\Gamma}{\Delta}{P}{\hby{\ell}}{\Delta'}{P'}$ implies
			\begin{enumerate}[a)]
				\item	If $\ell \in \set{\news{\tilde{m}} \bactout{n}{\abs{x}{Q}}, \bactinp{n}{\abs{x}{Q}}}$ then
%					$\exists l' $ such that
					$\horel{\tmap{\Gamma}{3}}{\tmap{\Delta}{3}}{\pmap{P}{3}}{\hby{\ell'}}
					{\tmap{\Delta'}{3}}{\pmap{P'}{3}}$ with $\mapa{\ell}^{3} = \ell'$.

%				\item	If $\ell = \bactinp{n}{\abs{x: C}{Q}}$ then
%					$\horel{\tmap{\Gamma}{3}}{\tmap{\Delta}{3}}{\pmap{P}{3}}{\hby{\bactinp{n}{\abs{x: C}{\pmap{Q}{3}}}}}
%					{\tmap{\Delta'}{3}}{\pmap{P'}{3}}$.
%
%				\item	If $\ell = \news{\tilde{m}} \bactout{n}{\abs{x: L}{Q}}$ then
%					$\horel{\tmap{\Gamma}{3}}{\tmap{\Delta}{3}}{\pmap{P}{3}}{\hby{\news{\tilde{m}} \bactout{n}{\abs{z}{\binp{z}{x} \pmap{Q}{3}}}}}
%					{\tmap{\Delta'}{3}}{\pmap{P'}{3}}$.
%
%				\item	If $\ell = \bactinp{n}{\abs{x: L}{Q}}$ then
%					$\horel{\tmap{\Gamma}{3}}{\tmap{\Delta}{3}}{\pmap{P}{3}}{\hby{\bactinp{n}{\abs{z}{\binp{z}{x} \pmap{Q}{3}}}}}
%					{\tmap{\Delta'}{3}}{\pmap{P'}{3}}$.

				\item	If $\ell \notin \set{\news{\tilde{m}} \bactout{n}{\abs{x}{Q}}, \bactinp{n}{\abs{x}{Q}}, \tau}$ then
					$\horel{\tmap{\Gamma}{3}}{\tmap{\Delta}{3}}{\pmap{P}{3}}{\hby{\ell}}
					{\tmap{\Delta'}{3}}{\pmap{P'}{3}}$.

				\item	If $\ell = \btau$ then
					$\horel{\tmap{\Gamma}{3}}{\tmap{\Delta}{3}}{\pmap{P}{3}}{\hby{\tau}}
					{\Delta''}{R}$ and
					${\tmap{\Gamma}{3}}{\tmap{\Delta'}{3}}{\pmap{P'}{3}}{\wb}{\Delta''}{R}$.

				\item	If $\ell = \tau$ and $\ell \not= \btau$ then %and $\hby{\ell}$ is not a \betatran then
					$\horel{\tmap{\Gamma}{3}}{\tmap{\Delta}{3}}{\pmap{P}{3}}{\hby{\tau}}
					{\tmap{\Delta'}{3}}{\pmap{P'}{3}}$.
			\end{enumerate}

		\item	Let $\Gamma; \es; \Delta \proves P$.
			$\horel{\tmap{\Gamma}{3}}{\tmap{\Delta}{3}}{\pmap{P}{3}}{\hby{\ell}}
			{\tmap{\Delta''}{3}}{Q}$ implies
			\begin{enumerate}[a)]
				\item	If $\ell \in \set{\news{\tilde{m}} \bactout{n}{\abs{x}{Q}}, \bactinp{n}{\abs{x}{Q}}, \tau}$
					then
					$\horel{\Gamma}{\Delta}{P}{\hby{\ell'}}{\Delta'}{P'}$
%					and $\horel{\tmap{\Gamma}{3}}{\tmap{\Delta''}{3}}{Q}{\hby{\hat{\ell}}}{\tmap{\Delta'}{3}}{\pmap{P'}{3}}$
					with $\mapa{\ell'}^{3} = \ell$ and $Q \scong \pmap{P'}{3}$.

				\item	If $\ell \notin \set{\news{\tilde{m}} \bactout{n}{\abs{x}{R}}, \bactinp{n}{\abs{x}{R}}, \tau}$
					then
					$\horel{\Gamma}{\Delta}{P}{\hby{\ell}}{\Delta'}{P'}$ and $Q \scong \pmap{P'}{3}$.
%					and $\horel{\tmap{\Gamma}{3}}{\tmap{\Delta''}{3}}{Q}{\hby{\hat{\ell}}}{\tmap{\Delta'}{3}}{\pmap{P'}{3}}$.

				\item	If $\ell = \tau$ then
					either
					$\horel{\Gamma}{\Delta}{\Delta}{\hby{\tau}}{\Delta'}{P'}$ with $Q \scong \pmap{P'}{3}$\\
					or
					$\horel{\Gamma}{\Delta}{\Delta}{\hby{\btau}}{\Delta'}{P'}$ and
					$\horel{\tmap{\Gamma}{3}}{\tmap{\Delta''}{3}}{Q}{\hby{\btau}}
					{\tmap{\Delta''}{3}}{\pmap{P'}{3}}$.
			\end{enumerate}
	\end{enumerate}
\end{proposition}

\begin{proof}
	The proof is an induction on the labelled transition system.
	The most interesting cases can be found in
	\propref{app:prop:op_corr_HOpp_to_HOp} (page~\pageref{app:prop:op_corr_HOpp_to_HOp}).
	\qed
\end{proof}

\begin{proposition}[Full Abstraction. From \HOpp to \HOp]\myrm
	\label{prop:fulla_HOpp_to_HOp}
	Let $P, Q$ \HOpp processes with $\Gamma; \es; \Delta_1 \proves P \hastype \Proc$ and 
	$\Gamma; \es; \Delta_2 \proves Q \hastype \Proc$. \\
	Then 
	$\horel{\Gamma}{\Delta_1}{P}{\wb}{\Delta_2}{Q}$ if and only if $\horel{\tmap{\Gamma}{3}}{\tmap{\Delta_1}{3}}{\pmap{P}{3}}{\wb}{\tmap{\Delta_2}{3}}{\pmap{Q}{3}}$
\end{proposition}

\begin{proof}
	\noi {\bf Soundness Direction.}

	\noi We create the closure
	\[
		\Re = \set{\horel{\Gamma}{\Delta_1}{P}{\ ,\ }{\Delta_2}{Q} \setbar \horel{\tmap{\Gamma}{3}}{\tmap{\Delta_1}{3}}{\pmap{P}{3}}{\wb}{\tmap{\Delta_2}{3}}{\pmap{Q}{3}}}
	\]
	\noi	It is straightforward to show that $\Re$ is a bisimulation if we follow Part 2 of
		\propref{prop:op_corr_HOpp_to_HOp} for subcases a and b.
		In subcase c we make use of \propref{lem:tau_inert}.

	\noi {\bf Completeness Direction.}

	\noi We create the closure
	\[
		\Re = \set{\horel{\tmap{\Gamma}{3}}{\tmap{\Delta_1}{3}}{\pmap{P}{3}}{\ ,\ }{\tmap{\Delta_2}{3}}{\pmap{Q}{3}} \setbar \horel{\Gamma}{\Delta_1}{P}{\wb}{\Delta_2}{Q}}
	\]
	\noi	We show that $\Re$ is a bisimulation up to deterministic transitions
		by following Part 1 of \propref{prop:op_corr_HOpp_to_HOp}.
		The proof is straightforward for subcases a), b) and d).
		In subcase c) we make use of \lemref{lem:up_to_deterministic_transition}.
	\qed
\end{proof}

\begin{proposition}[Precise encoding of \HOpp into \HOp]\myrm
	\label{prop:prec:HOpp_to_HOp}
	The encoding from $\tyl{L}_{\HOpp}$ to $\tyl{L}_{\HOp}$
	is precise.
\end{proposition}

\begin{proof}
	Syntactic requirements are easily derivable from the
	definition of the mappings in \figref{fig:enc:HOpp_to_HOp}.
	Semantic requirements are a consequence of
	\propref{prop:typepres_HOpp_to_HOp}, \propref{prop:op_corr_HOpp_to_HOp}, and \propref{prop:fulla_HOpp_to_HOp}.
	\qed
\end{proof}

\subsection{Polyadic \HOp}
\label{subsec:pol_HOp}

\noi Embedding polyadic name passing 
into the monadic name passing is well-studied in the literature.    
Using the linear typing, 
the preciseness (full abstraction) can be obtained~\cite{Yoshida96}.
Here we describe an encoding of $\pHOp$ into $\HOp$.

%The \HOp is presented in monadic terms.
%Nevertheless, we use polyadic \HOp terms
%as syntax sugar to define the encoding from
%$\HOp$ to $\HO$ in \defref{def:enc:HOp_to_HO},
%since polyadic encoding for session types
%has already been studied (cf.~\cite{}).

%For clarity reasons we present the polyadic variants
%\pHOp of \HOp and show a fully abstract
%encoding from \pHOp to \HOp.

\subsubsection{Modifications in Syntax, Reduction Semantics, and Types.}
%\myparagraph{Syntax and Operational Semantics}
The syntax of 
\pHOp processes is obtained from the syntax for processes given in 
\figref{fig:syntax} by considering values 
$$V \bnfis \tilde{u} \bnfbar \abs{\tilde{x}}{P}$$
and input prefixes $\binp{n}{\tilde{x}} P$.
Thus, polyadicity arises both in (session) communications and abstractions. 
Reduction is then defined by the rules in \figref{fig:reduction}, excepting 
 rules $\orule{App}$ and $\orule{Pass}$ which are replaced by rules
 \begin{eqnarray*}
	\orule{App$^p$} & 
	\appl{(\abs{\tilde{x}}{P})}{\tilde{u}} \red P \subst{\tilde{u}}{\tilde{x}} \quad |\tilde{x}| = |\tilde{u}| \\
	\orule{Pass$^p$} & 
		\quad \bout{n}{V} P_1 \Par \binp{\dual{n}}{\tilde{x}} P_2 \red P_1 \Par P_2 \subst{V}{\tilde{x}} \quad |V| = |\tilde{x}|
\end{eqnarray*}
The syntax of types in \figref{def:types}
is modified to include
 \begin{eqnarray*}
 L & \bnfis & \shot{\tilde{C}} \bnfbar \lhot{\tilde{C}} \\
 U & \bnfis & \tilde{C} \bnfbar L
 \end{eqnarray*}
 instead of $L \bnfis \shot{C} \bnfbar \lhot{C}$ and
$U \bnfis C \bnfbar L$, respectively.

%In the current polyadic extension we do not
%require polyadic shared names since the
%polyadic encoding for shared names is
%less straightforward than the polyadic
%encoding for session names.
Definitions of type equivalence/duality 
and typing environments ($\Gamma$ and $\Lambda$) are straightforward extensions of 
\defref{def:type_equiv}, \defref{def:type_dual},
and  \defref{def:typeenv}, respectively. 
Following~\cite{tlca07,MostrousY15} the type system for \pHOp
disallows polyadicity along shared names. Based on these modifications, 
the typing rules of \figref{fig:typerulesmy} are adapted in the expected way. 
In order to type polyadic values, we 
rely on the following rule:
\[
	\trule{Pol}~~\tree{
		V = a_i \dots a_n \qquad \Gamma; \Lambda_i; \Delta_i \proves u_i \hastype C_i \qquad U = C_1 \dots C_n
	}{
		\Gamma; \bigcup_{i \in I} \Lambda_i; \bigcup_{i \in I} \Delta_i \proves V \hastype U
	}
\]
Other rules are adjusted in the expected way, in order to accommodate polyadic values.
Notice, however, that rules $\trule{Req}$ and $\trule{Acc}$ are kept unchanged, as they
are used to type monadic exchanges along shared name prefixes.
We now state type soundness for \pHOp; the proof is straightforward and omitted, for it follows closely the proof detailed in 
\appref{app:ts}.

\begin{theorem}[Type Soundness for \pHOp]\myrm
	\label{thm:sr_phop}
	\begin{enumerate}[1.]
		\item	(Subject Congruence)
			$\Gamma; \es; \Delta \proves P \hastype \Proc$
			and
			$P \scong P'$
			implies
			$\Gamma; \es; \Delta \proves P' \hastype \Proc$.

		\item	(Subject Reduction)
			$\Gamma; \es; \Delta \proves P \hastype \Proc$
			with
			balanced $\Delta$
			and
			$P \red P'$
			implies $\Gamma; \es; \Delta'  \proves P' \hastype \Proc$
			and either (i)~$\Delta = \Delta'$ or (ii)~$\Delta \red \Delta'$
			with $\Delta'$ balanced.
	\end{enumerate}
\end{theorem}

%\iftodo
%\dk{sketch subject reduction (just revise the related case)}
%\else\fi

As for the behavioral semantics for \pHOp, the 
set of action labels is kept unchanged. 
In fact, 
as $V$ now stands for $\tilde{u}$ and $\abs{\tilde{x}}{P}$, 
labels  $\news{\tilde{m}} \bactout{n}{V}$ and $\bactinp{n}{V}$
require no modification. 
The LTS for \pHOp is as for \HOp, 
excepting rule $\ltsrule{App}$ which is replaced with the rule:
\[
	\appl{(\abs{\tilde{x}} P)}{\tilde{u}} \by{\tau} P\subst{\tilde{u}}{\tilde{x}}
\]
The characteristic process and characteristic value 
definition (\defref{def:characteristic_process})
is extended to include the cases:
\begin{eqnarray*}
	\mapchar{C_1 \dots C_n}{u_1 \cdots u_n} & \defeq &  \mapchar{C_1}{x_1} \Par \dots \Par \mapchar{C_n}{x_n} \\
	\omapchar{U_1 \dots U_n}  & \defeq &  \omapchar{U_1}, \dots, \omapchar{U_n}
\end{eqnarray*}
Thus, a polyadic type is inhabited by process whose
parallel components inhabit type the individual components of the polyadic
type. A polyadic value type is inhabited by a list
of values which inhabit the individual components of the polyadic value.
The rest of the behavioural semantics remains unchanged.

%Other definitions are straightforwardly extended. 

\subsubsection{Encoding \pHOp into \HOp.}

We slightly modify \defref{def:ep} to capture that a 
label $\ell$ may be mapped into a sequence of labels $\tilde{\ell}$.
Also, \defref{def:ep} stays as the same
assuming that if 
$P \hby{\ell} P'$ and $\mapa{\ell} = \{\ell_1, \ell_2,  \cdots, \ell_m\}$ then
$\map{P} \Hby{\mapa{\ell}} \map{P'}$
should be understood as
$\map{P} \Hby{\ell_1} P_1 \Hby{\ell_2} P_2 \cdots \Hby{\ell_m} P_m =  \map{P'}$,
for some
$P_1, P_2, \ldots, P_m$.

Let $\tyl{L}_{\pHOp}=\calc{\pHOp}{{\cal{T}}_5}{\hby{\ell}}{\wb}{\proves}$
where 
${\cal{T}}_5$ is a set of types of $\HOpp$;  
the typing $\proves$ is defined in 
\figref{fig:typerulesmy} with polyadic types.

\begin{definition}[Encoding from \pHOp to \HOp]\myrm
	\label{def:enc:pHOp_to_HOp}
	Encoding $\enco{\pmap{\cdot}{4}, \tmap{\cdot}{4}, \mapa{\cdot}^{4}}: \tyl{L}_{\pHOp} \to \tyl{L}_{\HOp}$
	to be defined as in \figref{fig:enc:pHOp_to_HOp}.
\end{definition}

\begin{figure}[t]
	\[
		\begin{array}{rcl}
			\multicolumn{3}{l}{\textrm{\bf Terms}}
			\\
			\pmap{\bout{n}{u_1, \dots, u_n} P}{4} &\defeq& \bout{n}{u_1} \dots ; \bout{n}{u_n} \pmap{P}{4}
			\\
			\pmap{\binp{n}{x_1, \dots, x_n} P}{4} &\defeq& \binp{n}{x_1} \dots ; \binp{n}{x_n} \pmap{P}{4}
			\\
			\pmap{\bout{n}{\abs{x_1, \dots, x_n}{Q}} P}{4} &\defeq& \bout{n}{\abs{z}{\binp{z}{x_1} \dots; \binp{z}{x_n} \pmap{Q}{4}}} \pmap{P}{4}
			\\
			\pmap{\appl{x}{(u_1, \dots, u_n)}}{4} &\defeq& \newsp{s}{\appl{x}{s} \Par \bout{\dual{s}}{u_1} \dots; \bout{\dual{s}}{u_1} \inact}
			\\
			\pmap{\appl{(\abs{x}{P})}{(u_1, \dots, u_n)}}{4} &\defeq& \newsp{s}{\appl{(\abs{x}{\pmap{P}{4}})}{s} \Par \bout{\dual{s}}{u_1} \dots; \bout{\dual{s}}{u_1} \inact}
			\\
			\\
			\multicolumn{3}{l}{\textrm{\bf Types}}
			\\
			\tmap{\lhot{(C_1, \dots, C_n)}}{4} &\defeq& \lhot{(\btinp{C_1} \dots; \btinp{C_n} \tinact)}
			\\
			\tmap{\shot{(C_1, \dots, C_n)}}{4} &\defeq& \shot{(\btinp{C_1} \dots; \btinp{C_n} \tinact)}
			\\
			\tmap{\btout{L} S}{4} &\defeq& \btout{\tmap{L}{4}} \tmap{S}{4}
			\\
			\tmap{\btinp{L} S}{4} &\defeq& \btinp{\tmap{L}{4}} \tmap{S}{4}
			\\
			\tmap{\btout{C_1, \dots, C_n} S}{4} &\defeq& \btout{C_1} \dots; \btout{C_n} \tmap{S}{4}
			\\
			\tmap{\btinp{C_1, \dots, C_n} S}{4} &\defeq& \btinp{C_1} \dots; \btout{C_n} \tmap{S}{4}
			\\
			\\
			\multicolumn{3}{l}{\textrm{\bf Labels}}
			\\
			\mapa{\news{\tilde{m}'} \bactout{n}{m_1, \dots, m_n}}^{4} &\defeq& \news{\tilde{m_1}'} \bactout{n}{m_1} \dots \news{\tilde{m_n}'}\bactout{n}{m_n}
			\quad \left.
			\begin{array}{rcl}
				\tilde{m_i}' &=& m_i \Leftrightarrow m_i \in \tilde{m}' \wedge\\
				\tilde{m_i}' &=& \es \Leftrightarrow m_i \notin \tilde{m}'
			\end{array}
			\right.
			\\
			\mapa{\bactinp{n}{m_1, \dots, m_n}}^{4} &\defeq& \bactinp{n}{m_1} \dots \bactinp{n}{m_n}
			\\
			\mapa{\news{\tilde{m}} \bactout{n}{\abs{x_1, \dots, x_n}{P}}}{4} &\defeq& \news{\tilde{m}} \bactout{n}{\abs{z}{\binp{z}{x_1} \dots; \binp{z}{x_n} \pmap{P}{4}}}
			\\
			\mapa{\bactinp{n}{\abs{x_1, \dots, x_n}{P}}}{4} &\defeq& \bactinp{n}{\abs{z}{\binp{z}{x_1} \dots; \binp{z}{x_n} \pmap{P}{4}}}
			\\
			\mapa{\btau}^{4} &\defeq& \btau, \stau, \dots, \stau
			\\
			\mapa{\tau}^{4} &\defeq& \tau, \dots, \tau
		\end{array}
	\]
	\caption{Encoding of \pHOp into \HOp (cf.~\defref{def:enc:pHOp_to_HOp}).
	We assume that the rest of the encoding is homomorphic on the syntax of
	processes, types and labels, respectively. \label{fig:enc:pHOp_to_HOp}}
\end{figure}

\begin{proposition}[Type Preservation. From \pHOp to \HOp]\myrm
	\label{prop:typepres_pHOp_to_HOp}
	Let $P$ be a \pHOp process.
	If $\Gamma; \emptyset; \Delta \proves P \hastype \Proc$ then 
	$\tmap{\Gamma}{4}; \emptyset; \tmap{\Delta}{4} \proves \pmap{P}{4} \hastype \Proc$. 
\end{proposition}

\begin{proof}
	By induction on the inference $\Gamma; \emptyset; \Delta \proves P \hastype \Proc$.
	See \propref{app:prop:typepres_pHOp_to_HOp} (Page~\pageref{app:prop:typepres_pHOp_to_HOp}) for details.
	\qed
\end{proof}

\begin{proposition}[Operational Correspondence. From \pHOp to \HOp]\myrm
	\label{prop:op_cor:pHOp_to_HOp}
	\begin{enumerate}
		\item	Let $\Gamma; \es; \Delta \proves P$. Then
			$\horel{\Gamma}{\Delta}{P}{\hby{\ell}}{\Delta'}{P'}$ implies
			\begin{enumerate}[a)]
				\item	If $\ell = \news{\tilde{m}'} \bactout{n}{\tilde{m}}$ then
					$\horel{\tmap{\Gamma}{4}}{\tmap{\Delta}{4}}{\pmap{P}{4}}{\hby{\ell_1} \dots \hby{\ell_n}}{\tmap{\Delta'}{4}}{\pmap{P}{4}}$
					with $\mapa{\ell}^{4} = \ell_1 \dots \ell_n$.

				\item	If $\ell = \bactinp{n}{\tilde{m}}$ then
					$\horel{\tmap{\Gamma}{4}}{\tmap{\Delta}{4}}{\pmap{P}{4}}{\hby{\ell_1} \dots \hby{\ell_n}}{\tmap{\Delta'}{4}}{\pmap{P}{4}}$
					with $\mapa{\ell}^{4} = \ell_1 \dots \ell_n$.

				\item	If $\ell \in \set{\news{\tilde{m}} \bactout{n}{\abs{\tilde{x}}{R}}, \bactinp{n}{\abs{\tilde{x}}{R}}}$ then
%					$\exists l' $ such that
					$\horel{\tmap{\Gamma}{4}}{\tmap{\Delta}{4}}{\pmap{P}{4}}{\hby{\ell'}}
					{\tmap{\Delta'}{4}}{\pmap{P'}{4}}$ with $\mapa{\ell}^{4} = \ell'$.

				\item	If $\ell \in \set{\bactsel{n}{l}, \bactbra{n}{l}}$ then
					$\horel{\tmap{\Gamma}{4}}{\tmap{\Delta}{4}}{\pmap{P}{4}}{\hby{\ell}}
					{\tmap{\Delta'}{4}}{\pmap{P'}{4}}$.

				\item	If $\ell = \btau$ then either
					$\horel{\tmap{\Gamma}{4}}{\tmap{\Delta}{4}}{\pmap{P}{4}}{\hby{\btau} \hby{\stau} \dots \hby{\stau}}
					{\tmap{\Delta'}{4}}{\pmap{P'}{4}}$ with $\mapa{\ell} = \btau, \stau \dots \stau$.

				\item	If $\ell = \tau$ then %and $\hby{\ell}$ is not a \betatran then
					$\horel{\tmap{\Gamma}{4}}{\tmap{\Delta}{4}}{\pmap{P}{4}}{\hby{\tau} \dots \hby{\tau}}
					{\tmap{\Delta'}{4}}{\pmap{P'}{4}}$ with $\mapa{\ell}^{4} = \tau \dots \tau$.
			\end{enumerate}

		\item	Let $\Gamma; \es; \Delta \proves P$.
			$\horel{\tmap{\Gamma}{4}}{\tmap{\Delta}{4}}{\pmap{P}{4}}{\hby{\ell_1}}
			{\tmap{\Delta_1}{4}}{P_1}$ implies
			\begin{enumerate}[a)]
				\item	If $\ell \in \set{\bactinp{n}{m}, \bactout{n}{m}, \news{m} \bactout{n}{m}}$ then
					$\horel{\Gamma}{\Delta}{P}{\hby{\ell}}{\Delta'}{P'}$ and\\
					$\horel{\tmap{\Gamma}{4}}{\tmap{\Delta_1}{4}}{P_1}{\hby{\ell_2} \dots \hby{\ell_n}}
					{\tmap{\Delta'}{4}}{\tmap{P'}{4}}$ with $\mapa{\ell}^{4} = \ell_1 \dots \ell_n$.

				\item	If $\ell \in \set{\news{\tilde{m}} \bactout{n}{\abs{x}{R}}, \bactinp{n}{\abs{x}{R}}}$
					then
					$\horel{\Gamma}{\Delta}{P}{\hby{\ell'}}{\Delta'}{P'}$
					with $\mapa{\ell'}^{4} = \ell$ and $P_1 \scong \pmap{P'}{4}$.

				\item	If $\ell \in \set{\bactsel{n}{l}, \bactbra{n}{l}}$
					then
					$\horel{\Gamma}{\Delta}{P}{\hby{\ell}}{\Delta'}{P'}$ and $P_1 \scong \pmap{P'}{4}$.
%					and $\horel{\tmap{\Gamma}{3}}{\tmap{\Delta''}{3}}{Q}{\hby{\hat{\ell}}}{\tmap{\Delta'}{3}}{\pmap{P'}{3}}$.

				\item	If $\ell = \btau$ then
					$\horel{\Gamma}{\Delta}{P}{\hby{\btau}}{\Delta'}{P'}$ and
					$\horel{\tmap{\Gamma}{4}}{\tmap{\Delta_1}{4}}{P_1}{\hby{\stau} \dots \hby{\stau}}
					{\tmap{\Delta'}{4}}{\tmap{P'}{4}}$ with $\mapa{\ell}^{4} = \btau, \stau \dots \stau$.

				\item	If $\ell = \tau$ then
					$\horel{\Gamma}{\Delta}{P}{\hby{\tau}}{\Delta'}{P'}$ and
					$\horel{\tmap{\Gamma}{4}}{\tmap{\Delta_1}{4}}{P_1}{\hby{\tau} \dots \hby{\tau}}
					{\tmap{\Delta'}{4}}{\tmap{P'}{4}}$ with $\mapa{\ell}^{4} = \tau \dots \tau$.
			\end{enumerate}
	\end{enumerate}
\end{proposition}

\begin{proof}
	We present the proof for the dyadic case
	in \propref{app:prop:op_cor:pHOp_to_HOp}
	(Page~\pageref{app:prop:op_cor:pHOp_to_HOp}).
	The polyadic case proof is an generalisation
	of the dyadic case proof.
	\qed
\end{proof}

\begin{proposition}[Full Abstraction. From \HOpp to \HOp]\myrm
	\label{prop:fulla:pHOp_to_HOp}
	Let $P, Q$ \pHOp process with $\Gamma; \es; \Delta_1 \proves P \hastype \Proc$ and 
	$\Gamma; \es; \Delta_2 \proves Q \hastype \Proc$.
	$\horel{\Gamma}{\Delta_1}{P}{\wb}{\Delta_2}{Q}$ if and only if $\horel{\tmap{\Gamma}{4}}{\tmap{\Delta_1}{4}}{\pmap{P}{4}}{\wb}{\tmap{\Delta_2}{4}}{\pmap{Q}{4}}$
\end{proposition}

\begin{proof}
	The proof for both direction is a consequence of Operational Correspondence,
	\propref{prop:op_cor:pHOp_to_HOp}.

	\noi {\bf Soundness Direction.}

	\noi We create the closure
	\[
		\Re = \set{\horel{\Gamma}{\Delta_1}{P}{\ ,\ }{\Delta_2}{Q} \setbar \horel{\tmap{\Gamma}{4}}{\tmap{\Delta_1}{4}}{\pmap{P}{4}}{\wb}{\tmap{\Delta_2}{4}}{\pmap{Q}{4}}}
	\]
	\noi	It is straightforward to show that $\Re$ is a bisimulation if we follow Part 2 of
		\propref{prop:op_cor:pHOp_to_HOp}.
%		for subcases a and b.
%		In subcase c we make use of \propref{lem:tau_inert}.

	\noi {\bf Completeness Direction.}

	\noi We create the closure
	\[
		\Re = \set{\horel{\tmap{\Gamma}{4}}{\tmap{\Delta_1}{4}}{\pmap{P}{4}}{\ ,\ }{\tmap{\Delta_2}{4}}{\pmap{Q}{4}} \setbar \horel{\Gamma}{\Delta_1}{P}{\wb}{\Delta_2}{Q}}
	\]
%
%	\dk{Is the proof easy? do the proof}
	\noi	We show that $\Re$ is a bisimulation up to deterministic transitions
		by following Part 1 of \propref{prop:op_cor:pHOp_to_HOp}.
%		The proof is straightforward for subcases a), b) and d).
%		In subcase c) we make use of \lemref{lem:up_to_deterministic_transition}.
	\qed
\end{proof}

\begin{proposition}[Precise encoding of \HOpp into \HOp]\myrm
	\label{prop:prec:pHOp_to_HOp}
	The encoding from $\tyl{L}_{\pHOp}$ to $\tyl{L}_{\HOp}$
	is precise.
\end{proposition}

\begin{proof}
	Syntactic requirements are easily derivable from the
	definition of the mappings in \figref{fig:enc:pHOp_to_HOp}.
	Semantic requirements are a consequence of
	\propref{prop:typepres_pHOp_to_HOp}, \propref{prop:op_cor:pHOp_to_HOp}, and \propref{prop:fulla:pHOp_to_HOp}.
	\qed
\end{proof}

%\input{pencoding}  %% Positive results
%\input{governed}

%\input{lencoding}

%\input{osexample}

% !TEX root = main.tex
\section{Related Work}
\label{sec:related}

\myparagraph{Expressiveness in Concurrency.}
There is a vast literature on expressiveness studies for process calculi;
we refer to~\cite{DBLP:journals/entcs/Parrow08} for a survey
(see also~\cite[\S\,2.3]{PerezPhD10}). 
In particular, the expressive power of the $\pi$-calculus has received much attention.
Studies cover, e.g., 
relationships between first-order and higher-order concurrency~(see, e.g.,~\cite{SangiorgiD:expmpa,San96int}),
comparisons between 
synchronous and asynchronous communication~(see, e.g.,~\cite{Boudol92,Palamidessi03,BeauxisPV08}),
and
(non)encodability issues for different choice operators~(see, e.g.,~\cite{Nestmann00,DBLP:conf/esop/PetersNG13}).
To substantiate claims related to (relative) expressive power,
early works appealed to different definitions of encoding.
%For instance, Palamidessi~\cite{Palamidessi03} defines \emph{uniform encodings} as those encodings which are homomorphic wrt parallel composition, respect renamings, and respect a ``reasonable semantics.''
Later on, 
proposals of abstract 
frameworks which formalise the notion of encoding 
and state associated syntactic and semantic criteria 
were put forward; 
recent proposals are~\cite{DBLP:journals/iandc/Gorla10,DBLP:journals/tcs/FuL10,DBLP:journals/corr/abs-1208-2750}. 
These frameworks are applicable to different calculi, and 
have shown useful to clarify known results and to derive new ones.
Our formulation of (precise) typed encoding (\defref{def:goodenc}) 
builds upon existing proposals (including~\cite{Palamidessi03,DBLP:journals/iandc/Gorla10,DBLP:conf/icalp/LanesePSS10})
in order to account for the session type systems
associated to the process languages under comparison.

\myparagraph{Expressiveness of Higher-Order Process Calculi.}
Early expressiveness studies for higher-order calculi are~\cite{Tho90,SangiorgiD:expmpa}; 
more recent works include~\cite{BundgaardHG06,DBLP:conf/icalp/LanesePSS10,DBLP:journals/iandc/LanesePSS11,XuActa2012,DBLP:conf/wsfm/XuYL13}.
Due to the close relationship between higher-order process calculi and functional calculi, 
works devoted to encoding (variants of) the $\lambda$-calculus into (variants of) the $\pi$-calculus (see, e.g.,~\cite{San92,DBLP:journals/tcs/Fu99,DBLP:journals/iandc/YoshidaBH04,BHY,DBLP:conf/concur/SangiorgiX14}) are also worth mentioning.
The work~\cite{SangiorgiD:expmpa} gives an encoding of the higher-order $\pi$-calculus
into the first-order $\pi$-calculus which is fully abstract with respect to reduction-closed, barbed congruence. 
A basic form of input/output types is used in~\cite{DBLP:journals/tcs/Sangiorgi01}, where the encoding in~\cite{SangiorgiD:expmpa} is casted in the asynchronous setting, with output and applications coalesced in a single construct. Building upon~\cite{DBLP:journals/tcs/Sangiorgi01}, 
a simply typed encoding for synchronous processes is given in~\cite{SaWabook}; the reverse encoding (i.e.,  first-order communication into higher-order processes) is also studied there for an asynchronous, localised $\pi$-calculus (only the output capability of names can be sent around).
The work~\cite{San96int} studies hierarchies for calculi with \emph{internal} first-order mobility and 
with higher-order mobility without name-passing (similarly as the subcalculus \HO). 
The hierarchies are based on expressivity: 
formally defined according to the order of types needed in typing, 
they describe different ``degrees of mobility''.
Via fully abstract encodings, it is shown that that name- and process-passing calculi with equal order of types have the same expressiveness.
With respect to these previous results, our approach based on session types 
has several important consequences and allows us to derive new results.  Our study reinforces the intuitive view of ``encodings as protocols'', namely session protocols which enforce precise linear and shared disciplines for names, a distinction not investigated in~\cite{SangiorgiD:expmpa,DBLP:journals/tcs/Sangiorgi01}. 
In turn, the linear/shared distinction is central in proper definitions of trigger processes, which are essential to encodings and behavioural equivalences.
More interestingly, we showed that $\HO$, a  minimal higher-order session calculus (no name passing, only first-order application) suffices to encode $\sessp$ (the session calculus with name passing) but also 
$\HOp$  and 
its extension  with higher-order applications (denoted $\HOpp$). 
Thus, using session types all these calculi are shown to be equally expressive with fully abstract encodings.
To our knowledge, these are the first expressiveness results of this kind.

Other related works are~\cite{BundgaardHG06,XuActa2012,DBLP:journals/iandc/LanesePSS11}.
The paper~\cite{BundgaardHG06} proposes a fully abstract, contin\-u\-a\-tion-passing style encoding of the 
$\pi$-calculus into Homer, a rich higher-order process calculus with explicit locations, local names, and nested locations.
The work~\cite{XuActa2012} studies the encodability of the higher-order $\pi$-calculus (extended with a relabelling operator) into the first-order $\pi$-calculus; encodings in the reverse direction are also proposed, following \cite{Tho90}.
A minimal calculus of higher-order concurrency is studied in~\cite{DBLP:journals/iandc/LanesePSS11}: it lacks restriction,  name passing, output prefix (so  communication is asynchronous), and constructs for infinite behaviour. 
Nevertheless, this calculus (a sublanguage of \HO) is shown to be Turing complete. Moreover, 
strong bisimilarity is decidable and coincides with reduction-closed, barbed congruence. 

Building upon~\cite{ThomsenB:plachoasgcfhop},
the work~\cite{XuActa2012} studies 
the (non)encodability of the $\pi$-calculus into 
a higher-order $\pi$-calculus with a powerful 
name relabelling operator, which is 
shown to be essential in encoding name-passing. %, following \cite{Tho90}.
A core higher-order calculus is studied in~\cite{DBLP:journals/iandc/LanesePSS11}: 
it lacks restriction,  name passing, output prefix %(communication is asynchronous), 
and constructs for infinite behaviour. 
This calculus  has 
a simple notion of bisimilarity which coincides with reduction-closed, barbed congruence.
%be Turing complete, while 
%have a decidable notion of (strong) bisimilarity that coincides with barbed congruence. 
{
The absence of restriction plays a key role in the characterisations in~\cite{DBLP:journals/iandc/LanesePSS11};
hence, our characterisation of contextual equivalence for \HO (which has restriction)
cannot be derived from that in~\cite{DBLP:journals/iandc/LanesePSS11}.

%Our work is closely related in spirit to the expressiveness studies in~\cite{DBLP:conf/icalp/LanesePSS10,DBLP:conf/wsfm/XuYL13}.
In~\cite{DBLP:conf/icalp/LanesePSS10}
the core calculus in~\cite{DBLP:journals/iandc/LanesePSS11} is extended with restriction,
synchronous communication, and polyadicity. It is shown that 
synchronous communication can encode asynchronous communication, % (as in the first-order setting),
and that process passing polyadicity induces a hierarchy in expressive power. % (unlike the first-order setting).
%A further extension with process abstractions of order one
%(functions from processes to processes)
% is shown to strictly add expressive power with respect to passing of processes only.
The paper~\cite{DBLP:conf/wsfm/XuYL13} 
complements~\cite{DBLP:conf/icalp/LanesePSS10} 
by studying the expressivity %of second-order abstractions.
%with replication ($!P$).  
%The work \cite{DBLP:conf/wsfm/XuYL13} focuses  
%%name and process abstractions are distinguished and contrasted, also 
%on expressiveness of the hirarchy of polyadic abstraction parameters. 
%(the same kind of polyadicity present in \pHOp)
%By adapting the encodings in~\cite{DBLP:conf/icalp/LanesePSS10} 
%Polyadicity 
of 
second-order process abstractions.
Polyadicity is shown to induce an expressiveness hierarchy; 
also,
by adapting the encoding in~\cite{SangiorgiD:expmpa},
process abstractions are encoded into name abstractions.
In contrast, we 
give a fully abstract encoding of
 \PHOpp into \HO that preserves session types; this improves~\cite{DBLP:conf/icalp/LanesePSS10,DBLP:conf/wsfm/XuYL13}   
by enforcing linearity disciplines on process behaviour.
The focus of~\cite{DBLP:conf/icalp/LanesePSS10,DBLP:conf/wsfm/XuYL13} is on 
the expressiveness of untyped, higher-order processes; they
%Moreover,~\cite{DBLP:conf/icalp/LanesePSS10,DBLP:conf/wsfm/XuYL13}
do not address 
tractable equivalences for processes  (such as 
higher-order and characteristic bisimulations) which only require observation of finite %number of 
higher-order values,  
whose formulations rely on session types.}
%therefore, our work complements their  results. 
% by clarifying the status of typeful %, resource-aware 
%structured communications. % in trigger-based representations of process passing, both in encodings and  equivalences.

\myparagraph{Session Typed Processes.}
The works~\cite{DemangeonH11,Dardha:2012:STR:2370776.2370794} 
study encodings of binary session calculi into a linearly typed $\pi$-calculus. 
While~\cite{DemangeonH11}~gives a precise encoding of \sessp into a linear calculus 
(an extension of \cite{BHY}),  
the work~\cite{Dardha:2012:STR:2370776.2370794} 
gives the operational correspondence 
(without full abstraction, cf.~\defref{def:sep}-4)
for the first- and higher-order 
$\pi$-calculi into \cite{LinearPi}. 
They investigate an embeddability of two different typing systems;
by the result of \cite{DemangeonH11}, 
\HOpp is encodable  into the linearly typed $\pi$-calculi.     

The syntax of $\HOp$ is a subset of that in~\cite{tlca07,MostrousY15}.
The work~\cite{tlca07} develops a full higher-order session calculus
with process abstractions and applications; it admits the type 
$U=U_1 \rightarrow U_2 \dots U_n \rightarrow \Proc$ and its linear type 
$U^1$
which corresponds to $\shot{\tilde{U}}$ and $\lhot{\tilde{U}}$ in 
a super-calculus of $\HOpp$ and $\PHOp$. 
%in~\cite{MostrousY15} in the asynchronous setting.
%The session type
%system considered is influenced by the type systems for $\lambda$-calculi and
%uses type syntax of the form $U_1 \rightarrow U_2 \dots U_n \rightarrow \Proc$
%for shared values and $(U_1 \rightarrow U_2 \dots U_n \rightarrow \Proc)^{1}$
%for linear values.
%Such a type is expressed in $\HOpp$
%terms using the type $\shot{U}$ (respectively, $\lhot{U}$)
%with $U$ being a nested higher-order type; and 
%the $\HOp$ uses only types of the form
%$\shot{C}$ and $\lhot{C}$ with $C$ being a first-order channel type.
Our results show that
the calculus in~\cite{tlca07} is not only expressed but 
also reasoned in 
$\HO$ (with limited form of arrow types, $\shot{C}$ and $\lhot{C}$), via precise encodings. 
{None of the above works proposes tractable 
bisimulations for higher-order processes.}

\myparagraph{Other Works on Typed Behavioural Equivalences.}
Since types can limit
contexts (environments) where processes can interact, typed equivalences
usually offer {\em coarse} semantics than untyped semantics. 
The work  \cite{PiSa96b} demonstrated the IO-subtyping can equate 
the optimal encoding of the $\lambda$-calculus by Milner which was not i
n the untyped polyadic $\pi$-calculus \cite{MilnerR:funp}. 
After \cite{PiSa96b}, many works on typed $\pi$-calculi 
have investigated correctness of encodings of known concurrent and
sequential calculi in order to examine semantic
effects of proposed typing systems. 

The type discipline closely related
to session types is a family of linear typing systems. The
work \cite{LinearPi} first proposed a linearly typed reduction-closed, barbed congruence and 
reasoned a tail-call optimisation of higher-order functions which are
encoded 
as processes. 
The work \cite{Yoshida96} had
used a bisimulation of graph-based types to prove the full abstraction
of encodings of the polyadic synchronous $\pi$-calculus into the
monadic synchronous $\pi$-calculus. 
Later typed equivalences of a
family of linear and affine calculi \cite{BHY,DBLP:journals/iandc/YoshidaBH04,BergerHY05} 
were used to encode 
PCF \cite{Plotkin1977223,Milner19771}, the simply typed $\lambda$-calculi with sums and products, and system F \cite{GirardJY:protyp}
fully abstractly (a fully abstract encoding of the $\lambda$-calculi 
was an open problem in \cite{MilnerR:funp}).  
The work \cite{YHB02} proposed a new bisimilarity
method associated with linear type structure and strong
normalisation. It presented applications to reason secrecy in
programming languages. A subsequent work \cite{HY02} adapted these results
to a practical direction. It proposes new typing
systems for secure higher-order and multi-threaded programming 
languages. 
In these works, typed properties, linearity and liveness, 
play a fundamental role in the analysis. In general, linear types 
are suitable to encode ``sequentiality'' in the sense of 
\cite{HylandJME:fulapi,AbramskyS:fulap}.

\myparagraph{Typed Behavioural Equivalences.}
This work follows the principles 
for
session type behavioural semantics in 
\cite{KYHH2015,KY2015,DBLP:journals/iandc/PerezCPT14}
where a bisimulation is defined on a LTS 
that assumes a session typed
observer.
%The bisimilarity is characterised by the corresponding
%reduction-closed, barbed congruence using techniques derived from~\cite{Hennessy07}.
Our theory for higher-order session types 
differentiates from 
the work in~\cite{KYHH2015,KY2015}, which 
considers the first-order
binary and multiparty session types, respectively.
The work \cite{DBLP:journals/iandc/PerezCPT14} gives a behavioural theory 
for a 
logically motivated
language of binary sessions 
without shared names.
%Determinacy properties (confluence, $\tau$-inertness) are proven.

%The theory for higher-order session type quivalences is more challenging than
%their corresponding first-order bisimulation theory.
Our approach for the higher-order builds upon techniques by Sangiorgi~\cite{SangiorgiD:expmpa,San96H}
and Jeffrey and Rathke~\cite{JeffreyR05}.
The work %Sangiorgi as part of his Ph.D.~research
%\cite{San96H,SangiorgiD:expmpa}
\cite{SangiorgiD:expmpa}
introduced the first fully-abstract encoding from the higher-order 
$\pi$-calculus into the $\pi$-calculus. 
Sangiorgi's encoding is based on the idea of a replicated input-guarded process 
(called a trigger process). We use a similar 
replicated triggered process 
to encode \HOp into \sessp (\defref{def:enc:HOp_to_p}).
 Operational correspondence for
the triggered encoding is shown using a context bisimulation
with first-order labels.
%Although contextual bisimilarity has a satisfactory discriminative power,
%its use is hindered by the universal quantification on output clauses.
To deal with the issue of context bisimilarity, 
Sangiorgi proposes \emph{normal bisimilarity}, 
a tractable  equivalence without universal quantification. 
To prove that context and normal bisimilarities coincide,~\cite{SangiorgiD:expmpa} uses 
triggered processes.
%The encoding also motivates the definition of a form of
Triggered bisimulation is also defined on first-order labels
where the contextual bisimulation is restricted to arbitrary
trigger substitution. %rather than arbitrary process substitutions.
This
characterisation of context bisimilarity  was refined in~\cite{JeffreyR05} for
calculi with recursive types, not addressed in~\cite{San96H,SangiorgiD:expmpa} and
relevant in our work.
The
bisimulation in~\cite{JeffreyR05}
is based on an LTS which is extended with trigger meta-notation.
%for a full higher-order $\pi$-calculus that allows
%higher-order applications.
As in~\cite{San96H,SangiorgiD:expmpa}, 
the LTS in~\cite{JeffreyR05}
observes first-order triggered values instead of
higher-order values, offering a more direct characterisation of contextual equivalence
and lifting the restriction to finite types.

We contrast 
the approach in~\cite{JeffreyR05} and our approach based on 
higher-order and characteristic bisimilarities. 
Below we use the notations adopted in~\cite{JeffreyR05}.
 
\begin{enumerate}[i)]
	\item	The work \cite{JeffreyR05} extends the first-order
		LTS for a trigger interaction whereas 
		our work uses the higher-order LTS. 

	\item	The output of a higher-order value $\abs{x}{Q}$ on name
		$n$ in \cite{JeffreyR05}
		requires the output of
		a fresh trigger name $t$ (notation $\tau_t$)
		on channel $n$ 
		and then the introduction of a replicated triggered process
		(notation $(t \Leftarrow (x) Q)$). 
		Hence we have:
		\[
			P \by{\news{t} \bactout{n}{\tau_{t}}} P' \Par (t \Leftarrow (x) Q) \by{\bactinp{t}{v}} P' \Par \appl{(x) Q}{v} \Par (t \Leftarrow (x) Q) 
		\]
	In our characteristic bisimulation, we only observe
	an output of a value that can be either first- or higher-order as follows:
		\[
			P \hby{\bactout{n}{V}} P' 
		\]
		with $V \equiv \abs{x}{Q}$ or $V = m$.

		A non-replicated triggered process ($\htrigger{t}{V}$)
		appears in 
		the parallel context of the acting process when
		we compare two processes for behavioural equality
		(cf.~\defref{def:cbisim}).
		Using the LTS in
		\defref{def:typed_transition} we can obtain:
		\begin{eqnarray*}
			P' \Par \htrigger{t}{\abs{x}{Q}}
			&\by{\abs{z}{\binp{z}{y} \repl{} \binp{t}{x} (\appl{y}{x})}}&
			P' \Par \newsp{s}{\binp{s}{y} \repl{} \binp{t}{x} 
(\appl{y}{x}) \Par \bout{s}{\abs{x}{Q}} \inact}\\
			&\by{\tau}&
			P' \Par \repl{}\binp{t}{y} (\appl{(\abs{x}{Q})}{y})
		\end{eqnarray*}
		that simulates the approach in \cite{JeffreyR05}.

	In addition, the output of the characteristic bisimulation 
differentiates from
		the approach in \cite{JeffreyR05} as listed below:
		\begin{itemize}
			\item	The typed LTS predicts the case of linear
				output values and will never allow replication
				of such a value;
				if $V$ is linear the input action would have no replication
				operator, as
				$\abs{z}{\binp{z}{y} \binp{t}{x} (\appl{y}{x})}$.

			\item	The characteristic bisimulation introduces a uniform approach
				not only for
				higher-order values but for first-order values
				as well, i.e.~triggered process can accept any
				process that can substitute a first-order value as well.
				This is derived from the fact that the $\HOp$-calculus makes no use of a matching operator, in contrast
				to the calculus defined in \cite{JeffreyR05})
				where name matching is crucial to prove completeness
				of the bisimilarity relation.
				Instead of a matching operator, 
we use types: a characteristic value inhabiting a type
enables the simplest form of interactions 
				with the environment.

			\item	Our \HOp-calculus requires only first-order
				applications. Higher-order applications,
				as in \cite{JeffreyR05},
				are presented as an extension in the \HOpp
				calculus.

			\item	Our trigger process is non-replicated. 
				It guards the output
				value with a higher-order input prefix. The
				functionality of the input is then used to
				simulate the contextual bisimilarity that subsumes
				the replicated trigger approach (cf.~\secref{subsec:char_bis}).
				The transformation of an output action as an input
				action allows for treating an output
				using the restricted LTS (\defref{def:restricted_typed_transition}):
				\[
					P' \Par \htrigger{t}{\abs{x}{Q}} \hby{\bactinp{t}{\abs{x}{\mapchar{U}{x}}}}
					P' \Par \news{s}{ (\appl{\mapchar{U}{x}}{s} \Par \bout{\dual{s}}{\abs{x}{Q}} \inact)}
				\]
		\end{itemize}
		%
		%In essence we are transforming a replicated trigger into a process
		%that is input-prefixed on a fresh name that receives a higher-order
		%value;

	\item	The input of a higher-order value in the \cite{JeffreyR05}
		requires 
		the input of a meta-syntactic fresh trigger, which then
		substituted on the application variable, thus the meta-syntax
		is extended to represent applications, e.g.:
%		for triggered application instead
%		of higher-order applications:
		%
		\[
			\binp{n}{x} P \by{\bactinp{n}{\tau_k}} (\appl{(\abs{x}{P})}{\tau_k}) \by{\tau} P \subst{\tau_k}{x} 
		\]
		Every instance of process variable $x$ in $P$ being substituted
		with trigger value $\tau_k$ to give an application of the form $(\appl{\tau_k}{x})$.
		In contrast the approach in the characteristic bisimulation observes the
		triggered value
		$\abs{z}\binp{t}{x} (\appl{x}{z})$ as an input instead of the
		meta-syntactic trigger:
		\[
			\binp{n}{x} P \hby{\bactinp{n}{\abs{z}\binp{t}{x} (\appl{x}{z})}} P \subst{\abs{z}\binp{t}{x} (\appl{x}{z})}{x}
		\]
		Every instance of process variable $x$ in $P$
		is substituted to give application of the form
		$\appl{(\abs{z}{\binp{t}{x} (\appl{x}{z})})}{v}$
		Note that in the characteristic bisimulation, 
		we can also observe a characteristic process as an input.
		
	\item 	Triggered applications in~\cite{JeffreyR05}
		are observed as an output of the application
		value over the fresh trigger name:
%		using an output
%		lead into an output observation of the
%		application value over
%		the fresh trigger name.
		%
		\[
			\appl{\tau_k}{v} \by{\bactout{k}{v}} \inact
		\]
		In contrast in the characteristic bisimulation
		we have two kind of applications:
		i) the trigger value application allows us
		to simulate an application on a fresh trigger name.
		ii) the characteristic value application
		allows us to inhabit an application value 
		and observe the interaction its interaction
		with the environment as below: 
%		of 
		%instead of observing an 
		%application and its value as an action we observe:
%		i) allows us to observe a trigger name
%		through the trigger value application; and
%		ii) we observe the application
%		value by inhabiting it in the characteristic value
%		and observing the interaction of the corresponding
%		characteristic process with its environment.
		%
		\begin{eqnarray*}
			\appl{(\abs{z}{\binp{t}{x} (\appl{x}{z})})}{v} &\by{\tau}& \binp{t}{x} (\appl{x}{v})
			\by{\bactinp{t}{\abs{x}{\mapchar{U}{x}}}}
			\appl{(\abs{x}{\mapchar{U}{x}})}{v}
			\by{\tau} \mapchar{U}{x} \subst{v}{x}
		\end{eqnarray*}
\end{enumerate}

\bibliographystyle{plain}
\bibliography{session}

%%%%%%%%%%%%%%%%%%%%%%%%%%%%%%%%%%%%%%%%%%%%%%%%%%%%%%%%%%%%%%%%%%%%%%%%%%%%%
% Appendix.
%%%%%%%%%%%%%%%%%%%%%%%%%%%%%%%%%%%%%%%%%%%%%%%%%%%%%%%%%%%%%%%%%%%%%%%%%%%%%
\newpage
\appendix
% !TEX root = ../main.tex
\section{Type Soundness}
\label{app:ts}

We state type soundness of our system.
As our typed process framework is a sub-calculus of that considered
by Mostrous and Yoshida, the proof of type soundness requires notions
and properties which are specific instances of those already shown in~\cite{MostrousY15}.
We begin by stating weakening and strengthening lemmas,
which have standard proofs.

%%% Weakening
\begin{lemma}[Weakening - Lemma C.2 in~\cite{MostrousY15}]\rm
	\label{l:weak}
	\begin{enumerate}[$-$]
		\item	If $\Gamma; \Lambda; \Delta \proves P \hastype \Proc$
			and
			$x \not\in \dom{\Gamma,\Lambda,\Delta}$
			then
			$\Gamma\cat x: \shot{S}; \Lambda; \Delta \proves P \hastype \Proc$ 
	\end{enumerate}
\end{lemma}

\begin{lemma}[Strengthening - Lemmas C.3 and C.4 in~\cite{MostrousY15}]\rm
	\label{l:stren}
	\begin{enumerate}[$-$]
		\item	If $\Gamma \cat x: \shot{S}; \Lambda; \Delta \proves P \hastype \Proc$
			and
			$x \not\in \fpv{P}$ then
			$\Gamma; \Lambda; \Delta \proves P \hastype \Proc$

		\item	If $\Gamma; \Lambda; \Delta \cat s: \tinact \proves P \hastype \Proc$
			and
			$s \not\in \fn{P}$
			then
			$\Gamma; \Lambda; \Delta \proves P \hastype \Proc$
	\end{enumerate}
\end{lemma}

\begin{lemma}[Substitution Lemma - Lemma C.10 in~\cite{MostrousY15}]\rm
	\label{l:subst}
	\begin{enumerate}[1.]
		\item	Suppose $\Gamma; \Lambda; \Delta \cat x:S  \proves P \hastype \Proc$ and
			$s \not\in \dom{\Gamma, \Lambda, \Delta}$. 
			Then $\Gamma; \Lambda; \Delta \cat s:S  \vdash P\subst{s}{x} \hastype \Proc$.

		\item	Suppose $\Gamma \cat x:\chtype{U}; \Lambda; \Delta \proves P \hastype \Proc$ and
			$a \notin \dom{\Gamma, \Lambda, \Delta}$. 
			Then $\Gamma \cat a:\chtype{U}; \Lambda; \Delta   \vdash P\subst{a}{x} \hastype \Proc$.

		\item	Suppose $\Gamma; \Lambda_1 \cat x:\lhot{C}; \Delta_1  \proves P \hastype \Proc$ 
			and $\Gamma; \Lambda_2; \Delta_2  \proves V \hastype \lhot{C}$ with 
			$\Lambda_1, \Lambda_2$ and $\Delta_1, \Delta_2$ defined.  
			Then $\Gamma; \Lambda_1 \cat \Lambda_2; \Delta_1 \cat \Delta_2  \proves P\subst{V}{x} \hastype \Proc$.

		\item	Suppose $\Gamma \cat x:\shot{C}; \Lambda; \Delta  \proves P \hastype \Proc$ and
			$\Gamma; \emptyset ; \emptyset  \proves V \hastype \shot{C}$.
			Then $\Gamma; \Lambda; \Delta  \proves P\subst{V}{x} \hastype \Proc$.
		\end{enumerate}
\end{lemma}

\begin{proof}
	In all four parts, we proceed by induction on the typing for $P$,
	with a case analysis on the last applied rule. 
%	Parts (1) and (2) are standard and therefore omitted. 
%
%	In Part (3), we content ourselves by detailing only the case in
%	which the last applied rule is \trule{App}. 
%	Then we have $P = \appl{V}{u}$. By inversion on the first assumption 
%	we infer:
%	\[
%	\tree{
%	\tree{}{
%	\Gamma;\, \Lambda_1 \cat x:\lhot{C} ;\, \Delta_{11}   \proves V \hastype \lhot{C}} \quad
%	\tree{}{\Gamma;\,  \emptyset   ;\, \Delta_{12}  \proves u \hastype C}
%	}{
%	\Gamma;\,  \Lambda_1 \cat x:\lhot{C};\, \Delta_1    \proves \appl{V}{u} \hastype \Proc}
%	\]
%	where $\Delta_1 = \Delta_{11} \cat \Delta_{12}$.
%	By inversion on the second assumption we infer that either
%	(i)\,$V = y$ (for some   variable $y$) or 
%	(ii)\,$V = \abs{z}Q$, for some $Q$ such that
%%
%	\begin{equation}
%		\Gamma'; \Lambda_1 \cat x:\lhot{C} ; \Delta_{11} \cat \Delta' \proves Q \hastype \Proc \label{eq:subseq2}\\
%	\end{equation}
%%
%	In possibility\,(i), we have a simple substitution on process variables and the thesis follows easily. 
%	In possibility\,(ii), we observe that $P\subst{V}{X} = \appl{X}{\mytilde{k}}\subst{\abs{\mytilde{z}}Q}{X} = Q\subst{\mytilde{k}}{\mytilde{z}}$.
%	The thesis then follows by using Lemma~\ref{lem:subst}\,(1) with 
%	the second premise of the typing of $\appl{X}{\mytilde{k}}$
%	and \eqref{eq:subseq2} above to infer 
%	\begin{equation*}
%		\Gamma; \Lambda_2 ; \Delta_2 \cat \mytilde{k}:\mytilde{C}  \proves Q \subst{\mytilde{k}}{\mytilde{z}} \hastype \Proc .
%	\end{equation*}
%%
%	The proof of Part (4) follows similar lines as that of Part (3).
	\qed
\end{proof}

%\begin{definition}[Well-typed Session Environment]%\rm
%	Let $\Delta$ be a session environment.
%	We say that $\Delta$ is {\em well-typed} if whenever
%	$s: S_1, \dual{s}: S_2 \in \Delta$ then $S_1 \dualof S_2$.
%\end{definition}
%
%\begin{definition}[Session Environment Reduction]%\rm
%	We define the relation $\red$ on session environments as:
%	\begin{enumerate}[$-$]
%		\item	$\Delta \cat s: \btout{U} S_1 \cat \dual{s}: \btinp{U} S_2 \red \Delta \cat s: S_1 \cat \dual{s}: S_2$
%		\item	$\Delta \cat s: \btsel{l_i: S_i}_{i \in I} \cat \dual{s}: \btbra{l_i: S_i'}_{i \in I} \red \Delta \cat s: S_k \cat \dual{s}: S_k', \quad k \in I$.
%	\end{enumerate}
%\end{definition}

We now state the instance of type soundness that we
can derive from~\cite{MostrousY15}.
It is worth noticing 
the 
definition of structural congruence in~\cite{MostrousY15} is richer. 
Also, their statement for subject reduction relies on an 
ordering on typings associated to queues and other 
runtime elements (such extended typings are denoted $\Delta$ in~\cite{MostrousY15}).
Since we are working with synchronous communication we can omit such an ordering.

We now repeat the statement of
\thmref{thm:sr} in Page~\pageref{thm:sr}:

\begin{theorem}[Type Soundness - \thmref{thm:sr}]%\rm
	\begin{enumerate}[1.]
		\item	(Subject Congruence) Suppose $\Gamma; \Lambda; \Delta \proves P \hastype \Proc$.
			Then $P \scong P'$ implies $\Gamma; \Lambda; \Delta \proves P' \hastype \Proc$.

		\item	(Subject Reduction) Suppose $\Gamma; \es; \Delta \proves P \hastype \Proc$
			with
			balanced $\Delta$. \\
			Then $P \red P'$ implies $\Gamma; \es; \Delta'  \proves P' \hastype \Proc$
			and $\Delta = \Delta'$ or $\Delta \red \Delta'$.

	\end{enumerate}
\end{theorem}

\begin{proof}
	Part (1) is standard, using weakening and strengthening lemmas. Part (2) proceeds by induction on the last reduction rule used. Below, we give some details:
	\begin{enumerate}[1.]
	   \item
	   Case \orule{App}: Then we have
	   $$
	   P = (\abs{x}{Q}) \, u   \red  Q \subst{u}{x} = P'
	   $$
	   Suppose $\Gamma;\, \emptyset ;\, \Delta \proves (\abs{x}{Q}) \, u \hastype \Proc$. 
	   We examine one possible way in which 
	   this assumption can be derived; other cases are similar or simpler:
	   \[
	   \tree{
	   \tree{\Gamma;\, \emptyset ;\, \Delta \cat \{x:S\} \proves Q  \hastype \Proc \quad 
	   \Gamma';\, \emptyset ;\, \{x:S\} \proves x  \hastype S}
	   {
	   \Gamma;\, \emptyset ;\, \Delta \proves \abs{x}{Q}  \hastype \lhot{S} }
	   \qquad
	   \tree{}{
	   \Gamma;\, \emptyset ;\, \{u:S\} \proves   u \hastype S}
	   }{
	   \Gamma;\, \emptyset ;\, \Delta \cat u:S \proves (\abs{x}{Q}) \, u \hastype \Proc
	   }
	   \]
	  Then, by combining premise
	   $\Gamma;\, \emptyset ;\, \Delta \cat \{x:S\} \proves Q  \hastype \Proc$
	   with
	   the substitution lemma (\lemref{lem:subst}(1)),
	   we obtain 
	    $\Gamma;\, \emptyset ;\, \Delta \cat u:S \proves Q\subst{u}{x}  \hastype \Proc$, as desired.
	    
	    \item Case \orule{Pass}: 
	    There are several sub-cases, depending on the type of the communication 
	    subject $n$ and the type of the object $V$. We analyze two representative sub-cases:
	    
	    \begin{enumerate}[(a)]
	    \item $n$ is a shared name and $V$ is a name $v$. 
	    Then we have the following reduction: 
	    $$
	    P = \bout{n}{v} Q_1 \Par \binp{n}{x} Q_2  \red  Q_1 \Par Q_2 \subst{v}{x} = P'
	    $$
	    By assumption, we have 
	    the following typing derivation:
	    \[	    \hspace{-12mm}
	    \tree{
%	    \tree{
%	     \Gamma' \cat n:\chtype{S};\, \emptyset ;\, \emptyset  \proves n  \hastype \chtype{S}
%	     \quad
%	      \Gamma;\, \emptyset ;\, \Delta_1    \proves   Q_1  \hastype \Proc
%	      \quad
%	       \Gamma;\, \emptyset ;\, \{v:S\}  \proves v  \hastype S	    
%	    }{
%	    \Gamma;\, \emptyset ;\, \Delta_1 \cat \{v:S\}  \proves \bout{n}{v} Q_1  \hastype \Proc
%	    } 
		\eqref{eq:sound1}
	    \quad 
	    		\eqref{eq:sound2}
%	    	    \tree{
%	    \Gamma' \cat n:\chtype{S};\, \emptyset ;\, \emptyset  \proves n  \hastype \chtype{S}
%	     \quad
%	      \Gamma;\, \emptyset ;\, \Delta_3 \cat x:S    \proves   Q_2  \hastype \Proc
%	    }{
%	    \Gamma;\, \emptyset ;\, \Delta_3 \proves  \binp{n}{x} Q_2 \hastype \Proc
%	   }
	    }{
	    \Gamma;\, \emptyset ;\, \Delta_1 \cat \{v:S\} \cat \Delta_3 \proves \bout{n}{v} Q_1 \Par \binp{n}{x} Q_2 \hastype \Proc
	    }
	    \]
	    
	    where \eqref{eq:sound1} and \eqref{eq:sound2} are as follows:
	    \begin{eqnarray}
	      & \tree{
	     \Gamma' \cat n:\chtype{S};\, \emptyset ;\, \emptyset  \proves n  \hastype \chtype{S}
	     \quad
	      \Gamma;\, \emptyset ;\, \Delta_1    \proves   Q_1  \hastype \Proc
	      \quad
	       \Gamma;\, \emptyset ;\, \{v:S\}  \proves v  \hastype S	    
	    }{
	    \Gamma;\, \emptyset ;\, \Delta_1 \cat \{v:S\}  \proves \bout{n}{v} Q_1  \hastype \Proc
	    } & \label{eq:sound1}
	    \\
	    	    	&     \tree{
	    \Gamma' \cat n:\chtype{S};\, \emptyset ;\, \emptyset  \proves n  \hastype \chtype{S}
	     \quad
	      \Gamma;\, \emptyset ;\, \Delta_3 \cat x:S    \proves   Q_2  \hastype \Proc
	    }{
	    \Gamma;\, \emptyset ;\, \Delta_3 \proves  \binp{n}{x} Q_2 \hastype \Proc
	   } & 
		\label{eq:sound2}
	    \end{eqnarray}
	    
	    Now, by applying \lemref{lem:subst}(1) on $\Gamma;\, \emptyset ;\, \Delta_3 \cat x:S    \proves   Q_2  \hastype \Proc$
			we obtain 
	   $$
	   \Gamma;\, \emptyset ;\, \Delta_3 \cat v:S    \proves   Q_2\subst{v}{x}  \hastype \Proc
	   $$
	   
	   			and the case is completed by using rule~\trule{Par} with this judgment:
							\[		~~ 
				\tree{
					\Gamma; \emptyset; \Delta_1    \proves  
 					 Q_1 \hastype \Proc
					 \quad 
					\Gamma;\, \emptyset ;\, \Delta_3 \cat v:S    \proves   Q_2\subst{v}{x}  \hastype \Proc
					}{
					\Gamma; \emptyset; \Delta_1 \cat \Delta_3  \cat v:S \proves  
 					Q_1  \Par  Q_2\subst{v}{x} \hastype \Proc
					} 
			\]
			Observe how in this case the session environment does not reduce.\\
			
			%%%%%%%%%%%%%%%%%%%%%%%%
			
		\item $n$ is a shared name and $V$ is a higher-order value. 
	    Then we have the following reduction: 
	    $$
	    P = \bout{n}{V} Q_1 \Par \binp{n}{x} Q_2  \red  Q_1 \Par Q_2 \subst{V}{x} = P'
	    $$
	    By assumption, we have 
	    the following typing derivation (below, we write 
	    $L$ to stand for $\shot{C}$ and 
	    $\Gamma$ to stand for $ \Gamma' \setminus \{x:L\}$).
	    \[	    \hspace{-12mm}
	    \tree{
	     \eqref{eq:sound3}
%	    \tree{
%	     \Gamma;\, \emptyset ;\, \emptyset  \proves n  \hastype \chtype{L}
%	     \quad
%	      \Gamma;\, \emptyset ;\, \Delta_1    \proves   Q_1  \hastype \Proc
%	      \quad
%	       \Gamma;\, \emptyset ;\, \emptyset  \proves V  \hastype L	    
%	    }{
%	    \Gamma;\, \emptyset ;\, \Delta_1    \proves \bout{n}{V} Q_1  \hastype \Proc
%	    } 
	    \quad 
	    \eqref{eq:sound4}
%	    	    \tree{
%	    \Gamma' ;\, \emptyset ;\, \emptyset  \proves n  \hastype \chtype{L}
%	     \quad
%	      \Gamma';\, \emptyset ;\, \Delta_3    \proves   Q_2  \hastype \Proc
%	      	     \quad
%	      	    \Gamma' ;\, \emptyset ;\, \emptyset  \proves x  \hastype L
%	    }{
%	    \Gamma;\, \emptyset ;\, \Delta_3 \proves  \binp{n}{x} Q_2 \hastype \Proc
%	   }
	    }{
	    \Gamma;\, \emptyset ;\, \Delta_1 \cat \Delta_3 \proves \bout{n}{v} Q_1 \Par \binp{n}{x} Q_2 \hastype \Proc
	    }
	    \]
	    where \eqref{eq:sound3} and \eqref{eq:sound4} are as follows:
	    \begin{eqnarray}
	    & 	    \tree{
	     \Gamma;\, \emptyset ;\, \emptyset  \proves n  \hastype \chtype{L}
	     \quad
	      \Gamma;\, \emptyset ;\, \Delta_1    \proves   Q_1  \hastype \Proc
	      \quad
	       \Gamma;\, \emptyset ;\, \emptyset  \proves V  \hastype L	    
	    }{
	    \Gamma;\, \emptyset ;\, \Delta_1    \proves \bout{n}{V} Q_1  \hastype \Proc
	    } 
 & 	    \label{eq:sound3} \\
	    & 	    	    \tree{
	    \Gamma' ;\, \emptyset ;\, \emptyset  \proves n  \hastype \chtype{L}
	     \quad
	      \Gamma';\, \emptyset ;\, \Delta_3    \proves   Q_2  \hastype \Proc
	      	     \quad
	      	    \Gamma' ;\, \emptyset ;\, \emptyset  \proves x  \hastype L
	    }{
	    \Gamma;\, \emptyset ;\, \Delta_3 \proves  \binp{n}{x} Q_2 \hastype \Proc
	   }
 & 	    \label{eq:sound4}
	    \end{eqnarray}
	    
	    Now, by applying \lemref{lem:subst}(4) on 
	    $\Gamma' \setminus \{x:L\};\, \emptyset ;\, \Delta_3    \proves   Q_2  \hastype \Proc$
	    and
	    $\Gamma;\, \emptyset ;\, \emptyset  \proves V  \hastype L$
	    we obtain 
	   $$
	   \Gamma;\, \emptyset ;\, \Delta_3  \proves   Q_2\subst{V}{x}  \hastype \Proc
	   $$
	   
	   and the case is completed by using rule~\trule{Par} with this judgment:
							\[		~~ 
				\tree{
					\Gamma; \emptyset; \Delta_1    \proves  
 					 Q_1 \hastype \Proc
					 \quad 
					\Gamma;\, \emptyset ;\, \Delta_3     \proves   Q_2\subst{V}{x}  \hastype \Proc
					}{
					\Gamma; \emptyset; \Delta_1 \cat \Delta_3   \proves  
 					Q_1  \Par  Q_2\subst{V}{x} \hastype \Proc
					} 
			\]
			Observe how in this case the session environment does not reduce.\\

	\end{enumerate}

		\item	Case \orule{Sel}:
			The proof is standard, the session environment reduces.

%		\item	Case \orule{Sess}:
%			The proof is standard, exploiting induction hypothesis.
%			The session environment may remain invariant (channel restriction)  or reduce (name restriction).

		\item	Cases \orule{Par} and \orule{Res}:
			The proof is standard, exploiting induction hypothesis. 

		\item	Case \orule{Cong}:
			follows from \thmref{thm:sr}\,(1).
	\end{enumerate}
	\qed
\end{proof}

\section{Behavioural Semantics}

We present the proofs for the theorems in
\secref{sec:beh_sem}.

\subsection{Proof of \thmref{the:coincidence}}
\label{app:sub_coinc}

We split the proof of \thmref{the:coincidence} (Page \pageref{the:coincidence}) into 
several lemmas:
\begin{enumerate}[$-$]
\item	\lemref{lem:wb_eq_wbf} establishes $\wb\ =\ \wbf$.s
\item	\lemref{lem:wb_is_wbc} exploits the process substitution result
	(\lemref{lem:proc_subst}) to prove that $\wb \subseteq \wbc$.
\item	\lemref{lem:wbc_is_cong} shows that $\wbc$ is a congruence
	which implies $\wbc \subseteq \cong$.
\item	\lemref{lem:cong_is_wb} shows  that $\cong \subseteq \wb$.
\end{enumerate}

%By the combination of the lemmas, we can obtain the theorem.

\noi
We now proceed to state and proof these lemmas, together with some auxiliary results.

%%%%%%%%%%%%%%%%%%%%%%%%%%%%%%%%%%%%%%%%%%%%%%%%%%%%%%%%%
%  WB = WBF
%%%%%%%%%%%%%%%%%%%%%%%%%%%%%%%%%%%%%%%%%%%%%%%%%%%%%%%%%

\begin{lemma}\rm
	\label{lem:wb_eq_wbf}
	$\wb = \wbf$.
\end{lemma}

\begin{proof}
	\noi We only prove the direction $\wb \subseteq \wbf$. The
	direction $\wbf \subseteq \wb$ is similar.

	\noi Consider
	\[
		\Re = \set{\horel{\Gamma}{\Delta_1}{P}{\ ,\ }{\Delta_2}{Q} \setbar \horel{\Gamma}{\Delta_1}{P}{\wb}{\Delta_2}{Q}}
	\]
	We show that $\Re$ is a characteristic bisimulation.
	The proof does a case analysis on the transition label $\ell$.

	\noi - Case $\ell = \news{\tilde{m_1}} \bactout{n}{V_1}$ is the non-trivial case.

	\noi If
	\begin{eqnarray}
		\horel{\Gamma}{\Delta_1}{P}{\hby{\news{\tilde{m_1}} \bactout{n}{V_1}}}{\Delta_1'}{P'}
		\label{lem:wb_eq_wbf1}
	\end{eqnarray}
	then $\exists Q, V_2$ such that
	\begin{eqnarray}
		\horel{\Gamma}{\Delta_2}{Q}{\Hby{\news{\tilde{m_2}} \bactout{n}{V_2}}}{\Delta_2'}{Q'}
		\label{lem:wb_eq_wbf2}
	\end{eqnarray}
	and for fresh $t$:
	\[
		\mhorel{\Gamma}{\Delta_1'}{\newsp{\tilde{m_1}}{P' \Par \hotrigger{t}{x}{s}{V_1}}}
		{\wb}
		{\Delta_2}{}{\newsp{\tilde{m_2}}{Q' \Par \hotrigger{t}{x}{s}{V_2}}}
	\]
	From the last result we can derive that for $\Gamma; \es; \Delta \proves V_1 \hastype U$:
	\[
		\mhorel{\Gamma}{\Delta_1'}{\newsp{\tilde{m_1}}{P' \Par \hotrigger{t}{x}{s}{V_1}}}
		{\hby{\bactinp{t}{\map{\btinp{U} \tinact}^{x}}}}
		{\Delta_1''}{}{\newsp{\tilde{m_1}}{P' \Par \newsp{s}{\map{\btinp{U} \tinact}^{s} \Par \bout{\dual{s}}{V_1} \inact}}}
	\]
	\noi implies
	\[
		\mhorel{\Gamma}{\Delta_2'}{\newsp{\tilde{m_2}}{Q' \Par \hotrigger{t}{x}{s}{V_2}}}
		{\hby{\bactinp{t}{\map{\btinp{U} \tinact}^{x}}}}
		{\Delta_2''}{}{\newsp{\tilde{m_2}}{Q' \Par \newsp{s}{\map{\btinp{U} \tinact}^{s} \Par \bout{\dual{s}}{V_2} \inact}}}
	\]
	\noi and $\Gamma; \es; \Delta' \proves V_2 \hastype U$.

	\noi Transition~(\ref{lem:wb_eq_wbf1}) implies transition~(\ref{lem:wb_eq_wbf2}). It remains to
	show that for fresh $t$:
	\[
		\mhorel{\Gamma}{\Delta_1'}{\newsp{\tilde{m_1}}{P' \Par \fotrigger{t}{x}{s}{\btinp{U} \tinact}{V_1}}}
		{\wb}
		{\Delta_2}{}{\newsp{\tilde{m_2}}{Q' \Par \fotrigger{t}{x}{s}{\btinp{U} \tinact}{V_2}}}
	\]
	The freshness of $t$ implies that
	\[
		\mhorel{\Gamma}{\Delta_1'}{\newsp{\tilde{m_1}}{P' \Par \fotrigger{t}{x}{s}{\btinp{U} \tinact}{V_1}}}
		{\hby{\bactinp{t}{m'}}}
		{\Delta_1''}{}{\newsp{\tilde{m_1}}{P' \Par \newsp{s}{\map{\btinp{U} \tinact}^{s} \Par \bout{\dual{s}}{V_1} \inact}}}
	\]
	\noi and
	\[
		\mhorel{\Gamma}{\Delta_2'}{\newsp{\tilde{m_2}}{Q' \Par \fotrigger{t}{x}{s}{\btinp{U} \tinact}{V_2}}}
		{\hby{\bactinp{t}{m'}}}
		{\Delta_2''}{}{\newsp{\tilde{m_2}}{Q' \Par \newsp{s}{\map{\btinp{U} \tinact}^{s} \Par \bout{\dual{s}}{V_2} \inact}}}
	\]
	\noi which coincides with the transitions for $\wb$.

	\noi - The rest of the cases are trivial.

	\noi The direction $\wbf \subseteq \wb$ is very similar to the
	direction $\wb \subseteq \wbf$: it requires a case analysis
	on the transition label $\ell$. Again the non-trivial case is
	$\ell = \news{\tilde{m_1}} \bactout{n}{V_1}$.
	\qed
\end{proof}

%%%%%%%%%%%%%%%%%%%%%%%%%%%%%%%%%%%%%%%%%%%%%%%%%%%%%%%%%
%  LINEAR SUBSTITUTION
%%%%%%%%%%%%%%%%%%%%%%%%%%%%%%%%%%%%%%%%%%%%%%%%%%%%%%%%%

The next lemma implies a process substitution lemma as a corollary.
Given two processes that are bisimilar under trigerred substitution
and characteristic process substitution, we can prove that they are
bisimilar under every process substitution. This result is
the key result for proving the soundness of the bisimulation. 
%We prove for the case of the general polyadic abstractions. 

%We also use one of the equalities in the substitution lemmas. 
%However all results hold for other equivalences. 

\newcommand{\auxtr}[1]{\abs{\tilde{x}}{\binp{#1}{y} (\appl{y}{\tilde{x}})}}

\begin{lemma}[Linear Process Substitution]\rm
	\label{lem:subst_equiv}
	If 
	\begin{enumerate}
		\item	$\fpv{P_2} = \fpv{Q_2} = \set{x}$.
		\item	$\Gamma; x: U; \Delta_1''' \proves P_2 \hastype \Proc$ and $\Gamma; x: U; \Delta_2''' \proves Q_2 \hastype \Proc$.
		\item	$\horel{\Gamma}{\Delta_1'}{\newsp{\tilde{m_1}}{P_1 \Par P_2 \subst{\auxtr{t}}{x}}}
			{\wb}
			{\Delta_2'}{\newsp{\tilde{m_2}}{Q_1 \Par Q_2 \subst{\auxtr{t}}{x}}}$, \\
			for some fresh $t$.

		\item	$\horel{\Gamma}{\Delta_1''}{\newsp{\tilde{m_1}}{P_1 \Par P_2 \subst{\omapchar{U}}{x}}}
			{\wb}{\Delta_2''}{\newsp{\tilde{m_2}}{Q_1 \Par Q_2 \subst{\omapchar{U}}{x}}}$,\\
			for some $U$.
	\end{enumerate}
	then $\forall R$ such that $\fv{R} = \tilde{x}$
\[
	\horel{\Gamma}{\Delta_1}{\newsp{\tilde{m_1}}{P_1 \Par P_2 \subst{\abs{\tilde{x}}{R}}{x}}}
	{\wb}
	{\Delta_2}{\newsp{\tilde{m_2}}{Q_1 \Par Q_2 \subst{\abs{\tilde{x}}{R}}{x}}}
\]
\end{lemma}

\begin{proof}
	We create a bisimulation closure:
	\begin{eqnarray*}
		\Re &=&
			\set{\horel{\Gamma}{\Delta_1}{\newsp{\tilde{m_1}}{P_1 \Par P_2 \subst{\abs{\tilde{x}}{R}}{x}}}{,}
			{\Delta_2}{\newsp{\tilde{m_2}}{Q_1 \Par Q_2 \subst{\abs{\tilde{x}}{R}}{x}}} \setbar \\
			&& \quad \forall R \textrm{ such that } \fv{R} = \tilde{x}, \fpv{P_2} = \fpv{Q_2} = \set{x}\\
			&& \quad \Gamma; x: U; \Delta_1''' \proves P_2 \hastype \Proc, \Gamma; x: U; \Delta_2''' \proves Q_2 \hastype \Proc\\
			&& \quad \textrm{for fresh } t,\\
			&& \quad \horel{\Gamma}{\Delta_1'}{\newsp{\tilde{m_1}}{P_1 \Par P_2 \subst{\auxtr{t}}{x}}}{\wb}{\Delta_2}{\newsp{\tilde{m_2}}{Q_1 \Par Q_2 \subst{\auxtr{t}}{x}}},\\
			&& \quad \horel{\Gamma}{\Delta_1''}{\newsp{\tilde{m_1}}{P_1 \Par P_2 \subst{\omapchar{U}}{x}}}{\wb}{\Delta_2''}{\newsp{\tilde{m_2}}{Q_1 \Par Q_2 \subst{\omapchar{U}}{x}} \textrm{ for some } U}\\
			&&}
	\end{eqnarray*}
	\noi  We show that $\Re$ is a bisimulation up-to \betatran (\lemref{lem:tau_inert}).

	\noi We do a case analysis on the transition:
	\[
		\horel{\Gamma}{\Delta_1}{\newsp{\tilde{m_1}}{P_1 \subst{\abs{\tilde{x}}{R}}{x} \Par P_2\subst{\abs{\tilde{x}}{R}}{x} }}{\by{\ell_1}}{\Delta_1'}{P_1'}
	\]
%
	%%%%%%%%%%%%%%%%%%%%%%%%%%%%%%
	%          Case
	%%%%%%%%%%%%%%%%%%%%%%%%%%%%%%

	\noi - Case: $P_2 \not= \appl{x}{\tilde{n}}$ for some $\tilde{n}$.
	\begin{eqnarray*}
		&&	\horel{\Gamma}{\Delta_1}
	{\newsp{\tilde{m_1}}{P_1 \Par P_2 \subst{\abs{\tilde{x}}{R}}{x} }}
			{\hby{\ell_1}}
{\Delta_1'}{\newsp{\tilde{m_1'}}{P_1 \Par P_2' \subst{\abs{\tilde{x}}{R}}{x}}}
	\end{eqnarray*}

	\noi From the latter transition we obtain that
\[
		\mhorel{\Gamma}{\Delta_1}{\newsp{\tilde{m_1}}{P_1 \Par P_2 \subst{\auxtr{t}}{x}}}
		{\hby{\ell_1}}{\Delta_1'}{P' \scong}{\newsp{\tilde{m_1}}{P_1' \Par P_2' \subst{\auxtr{t}}{x}}}
\]
	\noi which implies
	\begin{eqnarray}
		\Gamma; \es; &\Delta_2& \proves \newsp{\tilde{m_2}}{Q_1 \Par Q_2 \subst{\auxtr{t}}{x}} \nonumber \\
		\Hby{\ell_2}&
		\Delta_2'& \proves Q' \scong \newsp{\tilde{m_2}}{Q_1' \Par Q_2' \subst{\auxtr{t}}{x}}
		\label{lem:subst_equiv1}
		\\
		&&\horel{\Gamma}{\Delta_1'}{P' \Par C_1}{\wb}{\Delta_2'}{Q' \Par C_2} \label{lem:subst_equiv2}
	\end{eqnarray}
	\noi Furthermore, we have:
\[
	\horel{\Gamma}{\Delta_1}{\newsp{\tilde{m_1}}{P_1 \Par P_2 \subst{\omapchar{U}}{x}}}{\hby{\ell_1}}
	{\Delta_1'}{P'' \scong \newsp{\tilde{m_1'}}{P_1' \Par P_2' \subst{\omapchar{U}}{x}}}
\]
	\noi which implies
	\begin{eqnarray}
		\Gamma; \es; &\Delta_2& \proves \newsp{\tilde{m_2}}{Q_1 \Par Q_2 \subst{\omapchar{U}}{x}} \nonumber \\
		\Hby{\ell_2} &\Delta_2'& \proves  Q'' \scong \newsp{\tilde{m_2}'}{Q_1' \Par Q_2' \subst{\omapchar{U}}{x}}
		\label{lem:subst_equiv3}
		\\
		&&\horel{\Gamma}{\Delta_1'}{P'' \Par C_1}{\wb}{\Delta_2'}{Q'' \Par C_2} \label{lem:subst_equiv4}
	\end{eqnarray}
	\noi From~(\ref{lem:subst_equiv1}) and~(\ref{lem:subst_equiv3}) we obtain that $\forall R$ with $\fv{R} = \tilde{x}$:
	\[
		\horel{\Gamma}{\Delta_2}{\newsp{\tilde{m_2}}{Q_1 \Par Q_2 \subst{\abs{\tilde{x}}{R}}{x}}}
		{\Hby{\ell_2}}
		{\Delta_2'}{\newsp{\tilde{m_2}'}{Q_1' \Par Q_2' \subst{\abs{\tilde{x}}{R}}{x}}}
	\]
	\noi The case concludes if we combine~(\ref{lem:subst_equiv2}) and~(\ref{lem:subst_equiv4}), to obtain that $\forall R$ with $\fv{R} = \tilde{x}$
	\[
		\horel{\Gamma'}{\Delta_1''}{\newsp{\tilde{m_1}'}{P_1' \Par P_2' \subst{\abs{\tilde{x}}{R}}{x}} \Par C_1}
		{\ \Re\ }
		{\Delta_2''}{\newsp{\tilde{m_2}'}{Q_1 \Par Q_2' \subst{\abs{\tilde{x}}{R}}{x}} \Par C_2}
	\]

	%%%%%%%%%%%%%%%%%%%%%%%%%%%%%%
	%          Case
	%%%%%%%%%%%%%%%%%%%%%%%%%%%%%%

	\noi - Case: $P_2 = \appl{x}{\tilde{n}}$ for some $\tilde{n}$.

%	\noi We do a case analysis on action $\ell$.
%	\noi - Subcase: $\ell \not= \tau$.

	\noi $\forall R$ with $\fv{R} = \tilde{x}$
	\[
		\mhorel{\Gamma}{\Delta_1}{\newsp{\tilde{m_1}}{P_1 \Par (\appl{x}{\tilde{n}}) \subst{\abs{\tilde{x}}{R}}{x}}}
		{\hby{\btau}}
		{\Delta_1'}{}{\newsp{\tilde{m_1'}}{P_1 \Par  R \subst{\tilde{n}}{\tilde{x}}}}
	\]
	\noi From the latter transition we get that:
	\nhorel{\Gamma}{\Delta_1}{\newsp{\tilde{m_1}}{P_1 \Par \appl{x}{\tilde{n}} \subst{\auxtr{t}}{x}}}
	{\hby{\btau} \hby{\bactinp{t}{\auxtr{t'}}}}
	{\Delta_1'}{\newsp{\tilde{m_1'}}{P_1 \Par \appl{x}{\tilde{n}} \subst{\auxtr{t'}}{x}}}
	{lem:subst_equiv5}
%
%	\begin{eqnarray}
%		\Gamma; \es; \Delta_1 &\hby{\bactinp{t}{\auxtr{t'}}}& \Delta_1' \proves
%		\newsp{\tilde{m_1}}{P_1 \Par \appl{X}{\tilde{n}} \subst{\auxtr{t}}{X}} \nonumber \\
%		&\hby{\bactinp{t}{\auxtr{t'}}}& 
%		\newsp{\tilde{m_1'}}{P_1 \Par \appl{X}{\tilde{n}} \subst{\auxtr{t'}}{X}}
%		\label{lem:subst_equiv5}
%	\end{eqnarray}
%
	\noi and $t'$ a fresh name. From the freshness of $t$,
	the determinacy of the application transition
	and the fact that $x$ is linear in $Q_2$
	it has to be the case that:
	\[
		\begin{array}{crll}
			&\Gamma; \es; \Delta_2'& \proves &
			\newsp{\tilde{m_2'}}{Q_1 \Par Q_2 \subst{\auxtr{t}}{x}}\\
			\Hby{} & & &
			\newsp{\tilde{m_2'}}{Q_1'' \Par Q_3 \Par \appl{x}{\tilde{m}} \subst{\auxtr{t}}{x}} \\
			\hby{\btau} \hby{\bactinp{t}{\auxtr{t'}}}
			& \Delta_2''& \proves& \newsp{\tilde{m_2'}}{Q_1' \Par \appl{x}{\tilde{m}} \subst{\auxtr{t'}}{x}}
		\end{array}
	\]
%
%
%	\[
%		\begin{array}{rcll}
%			\Gamma; \es; \Delta_2' &\Hby{\bactinp{t}{\auxtr{t'}}}& \Delta_2'' \proves &
%			\newsp{\tilde{m_2'}}{Q_1 \Par Q_2 \subst{\auxtr{t}}{X}}\\
%			&\Hby{}& &
%			\newsp{\tilde{m_2'}}{Q_1'' \Par Q_3 \Par \appl{X}{\tilde{m}} \subst{\auxtr{t}}{X}} \\
%			&\hby{\bactinp{t}{\auxtr{t'}}}& &
%			\newsp{\tilde{m_2'}}{Q_1' \Par \appl{X}{\tilde{m}} \subst{\auxtr{t'}}{X}} \\
%		\end{array}
%	\]
%
	\noi and
	\nhorel{\Gamma}{\Delta_1'}{\newsp{\tilde{m_1'}}{P_1 \Par \appl{x}{\tilde{n}} \subst{\auxtr{t'}}{x}}}
	{\wb}
	{\Delta_2'}{\newsp{\tilde{m_2}'}{Q_1' \Par \appl{x}{\tilde{m}} \subst{\auxtr{t'}}{x}}}
	{lem:subst_equiv6}
%
%
%	\begin{eqnarray}
%		\Gamma; \es; \Delta_1' &\wb& \Delta_2' \proves 
%		\newsp{\tilde{m_1'}}{P_1 \Par \appl{X}{\tilde{n}} \subst{\auxtr{t'}}{X}} \nonumber \\
%		& \wb &
%		\newsp{\tilde{m_2}'}{Q_1' \Par \appl{X}{\tilde{m}} \subst{\auxtr{t'}}{X}}
%		\label{lem:subst_equiv6}
%	\end{eqnarray} 
%
	\noi From the latter transition we can conclude that $\forall R$ with $\fv{R} = \set{x}$:
	\[
		\begin{array}{crll}
			&\Gamma; \es; \Delta_2'& \proves & 
			\newsp{\tilde{m_2'}}{Q_1 \Par Q_2 \subst{\abs{\tilde{x}}{R}}{x}}\\
			\Hby{} & & &
			\newsp{\tilde{m_2'}}{Q_1' \Par \appl{x}{\tilde{m}} \subst{\abs{\tilde{x}}{R}}{x}} \\
			\hby{\btau}
			&\Delta_2''& \proves & \newsp{\tilde{m_2'}}{Q_1' \Par R \subst{\tilde{m}}{\tilde{x}}}%\appl{x}{\tilde{m}} \subst{\abs{\tilde{x}}{R'}}{x}}
		\end{array}
	\]
%
%	\[
%		\begin{array}{rcll}
%			\Gamma; \es; \Delta_2' &\Hby{\ell}& \Delta_2'' \proves &
%			\newsp{\tilde{m_2'}}{Q_1 \Par Q_2 \subst{\abs{\tilde{x}}{R}}{X}}\\
%			&\Hby{}& &
%			\newsp{\tilde{m_2'}}{Q_1' \Par \appl{X}{\tilde{m}} \subst{\abs{\tilde{x}}{R}}{X}} \\
%			&\hby{\ell}& &
%			\newsp{\tilde{m_2'}}{Q_1' \Par \appl{X}{\tilde{m}} \subst{\abs{\tilde{x}}{R'}}{X}} \\
%		\end{array}
%	\]
%
	\noi From the definition of $S$ and~(\ref{lem:subst_equiv6}),
	we also conclude that
	\begin{eqnarray*}
		&& \horel{\Gamma}
		{\Delta_1'}{\newsp{\tilde{m_1'}}{P_1 \Par R 
\subst{\tilde{n}}{\tilde{x}}}}
		{\hby{\btau}\ \Re\ \stackrel{\btau}{\longleftarrow}}
		{\Delta_2'}{\newsp{\tilde{m_2}'}{Q_1' \Par R \subst{\tilde{m}}{\tilde{x}}}}
	\end{eqnarray*}
%
%	\noi the latter substitution can be rewritten as
%
%	\begin{eqnarray*}
%		&&\horel{\Gamma}
%		{\Delta_1'}{\newsp{\tilde{m_1'}}{P_1 \Par \appl{X}{s} \subst{\abs{\tilde{x}}{R'}}{X} \Par C}}
%		{\ \mathcal{S}\ }
%		{\Delta_2'}{\newsp{\tilde{m_2}'}{Q_1' \Par \appl{X}{s'} \subst{\abs{\tilde{x}}{R'}}{X} \Par C}}
%	\end{eqnarray*}
%
%	\noi with process $C$ being the process derived from action $\ell$
%	to complete the bisimulation closure.
	\qed
\end{proof}

We can generalise the result of the linear process substitution lemma
to prove process substitution (\lemref{lem:proc_subst}).
Intuitively, we can subsequently apply linear process substitution
to achieve process substitution.

%%%%%%%%%%%%%%%%%%%%%%%%%%%%%%%%%%%%%%
%	Process Substitution
%%%%%%%%%%%%%%%%%%%%%%%%%%%%%%%%%%%%%%

\begin{lemma}[Process Substitution]\rm
	\label{app:lem:proc_subst}
	If 
	\begin{enumerate}
		\item	$\horel{\Gamma}{\Delta_1'}{P \subst{\auxtr{t}}{x}}{\wb}{\Delta_2}{Q \subst{\auxtr{t}}{x}}$
			for some fresh $t$.

		\item	$\horel{\Gamma}{\Delta_1''}{P \subst{\omapchar{U}}{x}}{\wb}{\Delta_2''}{Q \subst{\omapchar{U}}{x}}$
			for some $U$.
	\end{enumerate}
	then $\forall R$ such that $\fv{R} = \tilde{x}$
\[
	\horel{\Gamma}{\Delta_1}{P \subst{\abs{\tilde{x}}{R}}{x}}{\wb}{\Delta_2}{Q \subst{\abs{\tilde{x}}{R}}{x}}
\]
\end{lemma}

\begin{proof}
	\noi We define a closure $\Re$ using the normal form of $P$ and $Q$

	\[
		\begin{array}{rcll}
			\Re &=& \set{\horel{\Gamma}{\Delta_1}{\newsp{\tilde{m_1}}{P_1 \subst{\abs{\tilde{x}}{R}}{x} \Par P_2 \subst{\abs{\tilde{x}}{R}}{x}}}{,}{\Delta_2}{\newsp{\tilde{m_2}}{Q_1 \subst{\abs{\tilde{x}}{R}}{x} \Par Q_2 \subst{\abs{\tilde{x}}{R}}{x}}} \setbar \\
%			\mathcal{S} &=& \set{\horel{\Gamma}{\Delta_1}{\newsp{\tilde{m_1}}{P_1 \subst{\abs{\tilde{x}}{R}}{x} \Par P_2 \subst{\abs{\tilde{x}}{R}}{x}}}{,}{\Delta_2}{\newsp{\tilde{m_2}}{Q_1 \subst{\abs{\tilde{x}}{R}}{x} \Par Q_2 \subst{\abs{\tilde{x}}{R}}{x}}} \setbar \\
			&& \qquad \forall R \textrm{ such that } \fv{R} = \tilde{x},\\
			&& \qquad \textrm{ for fresh } t,
			\mhorel{\Gamma}{\Delta_1'}
			{\newsp{\tilde{m_1}}{P_1 \subst{\auxtr{t}}{x} \Par P_2 \subst{\auxtr{t}}{x}}}
			{\wb}
			{\Delta_2'}{}{\newsp{\tilde{m_2}}{Q_1 \subst{\auxtr{t}}{x} \Par Q_2 \subst{\auxtr{t}}{x}}}\\
			&& \qquad \textrm{ for some } U, 
			\mhorel{\Gamma}{\Delta_1''}
			{\newsp{\tilde{m_1}}{P_1 \subst{\omapchar{U}}{x} \Par P_2 \subst{\omapchar{U}}{x}}}
			{\wb}
			{\Delta_2''}{}{\newsp{\tilde{m_2}}{Q_1 \subst{\omapchar{U}}{x} \Par Q_2 \subst{\omapchar{U}}{x}}} \\
			&&}
		\end{array}
	\]
%
%	\begin{eqnarray*}
%		\mathcal{S} &=& \set{\horel{\Gamma}{\Delta_1}{\newsp{\tilde{m_1}}{P_1 \subst{\abs{\tilde{x}}{R}}{X} \Par P_2 \subst{\abs{\tilde{x}}{R}}{X}}}{,}{\Delta_2}{\newsp{\tilde{m_2}}{Q_1 \subst{\abs{\tilde{x}}{R}}{X} \Par Q_2 \subst{\abs{\tilde{x}}{R}}{X}}} \setbar \\
%		&& \quad \forall R \textrm{ such that } \fv{R} = \tilde{x},\\
%		&& \quad \textrm{ for fresh } t, \Gamma; \es; \Delta_1' \wb \Delta_2' \proves \\
%		&& \quad \newsp{\tilde{m_1}}{P_1 \subst{\auxtr{t}}{X} \Par P_2 \subst{\auxtr{t}}{X}} \wb \newsp{\tilde{m_2}}{Q_1 \subst{\auxtr{t}}{X} \Par Q_2 \subst{\auxtr{t}}{X}}, \\
%		&& \quad \textrm{ for some } U, \Gamma; \es; \Delta_1'' \wb \Delta_2'' \proves \\
%		&& \quad \newsp{\tilde{m_1}}{P_1 \subst{\abs{\tilde{x}}{\map{U}^x}}{X} \Par P_2 \subst{\abs{\tilde{x}}{\map{U}^x}}{X}} \wb \newsp{\tilde{m_2}}{Q_1 \subst{\abs{\tilde{x}}{\map{U}^x}}{X} \Par Q_2 \subst{\abs{\tilde{x}}{\map{U}^x}}{X}} \\
%		&&}
%	\end{eqnarray*}
%
	\noi We show that $\Re$ is a bisimulation up to \betatran (\lemref{lem:tau_inert}).

	\noi - Case: $P_2 \not= \appl{x}{\tilde{n}}$ for some $\tilde{n}$.
	\nhorel	{\Gamma}{\Delta_1}{\newsp{\tilde{m_1}}{P_1 \subst{\abs{\tilde{x}}{R}}{x} \Par P_2 \subst{\abs{\tilde{x}}{R}}{x}}}
		{\hby{\ell_1}}
		{\Delta_1'}{\newsp{\tilde{m_1'}}{P_1 \subst{\abs{\tilde{x}}{R}}{x} \Par P_2' \subst{\abs{\tilde{x}}{R}}{x}}}
		{lem:subst_equiv11}
	\noi The case is similar to the first case of \lemref{lem:subst_equiv}.

	\noi - Case: $P_2 = \appl{x}{\tilde{n}}$ for some $\tilde{n}$.
\[
	\mhorel	{\Gamma}{\Delta_1}{\newsp{\tilde{m_1}}{P_1 \subst{\abs{\tilde{x}}{R}}{x} \Par \appl{x}{\tilde{n}} \subst{\abs{\tilde{x}}{R}}{x}}}
		{\hby{\btau}}{\Delta_1'}{}{\newsp{\tilde{m_1'}}{P_1 \subst{\abs{\tilde{x}}{R}}{x} \Par R \subst{\tilde{n}}{\tilde{x}}}} %\appl{x}{\tilde{n}} \subst{\abs{\tilde{x}}{R'}}{x}}}
%		{\hby{\ell_1}}{\Delta_1'}{}{\newsp{\tilde{m_1'}}{P_1 \subst{\abs{\tilde{x}}{R}}{x} \Par \appl{x}{\tilde{n}} \subst{\abs{\tilde{x}}{R'}}{x}}}
\]
	\noi From the latter transition we get that:
	\nhorel{\Gamma}{\Delta_1}
	{\newsp{\tilde{m_1}}{P_1 \subst{\auxtr{t}}{x} \Par \appl{x}{\tilde{n}} \subst{\auxtr{t}}{x}}}
	{\hby{\btau} \hby{\bactinp{t}{\auxtr{t'}}}}
	{\Delta_1'}{\newsp{\tilde{m_1}'}{P_1 \subst{\auxtr{t}}{x} \Par \appl{y}{\tilde{n}} \subst{\auxtr{t'}}{y}}}
%	{\hby{\bactinp{t}{\auxtr{t'}}}}
%	{\Delta_1'}{\newsp{\tilde{m_1}'}{P_1 \subst{\auxtr{t}}{x} \Par \appl{y}{\tilde{n}} \subst{\abs{\tilde{x}}{\binp{t'}{y} \appl{y}{\tilde{x}}}}{y}}}
	{cor:subst_equiv5}
%
%	\begin{eqnarray}
%		\Gamma; \es; \Delta_1 &\hby{\bactinp{t}{\auxtr{t'}}}& \Delta_1' \proves
%		\newsp{\tilde{m_1}}{P_1 \subst{\auxtr{t}}{X} \Par \appl{X}{\tilde{n}} \subst{\auxtr{t}}{X}}
%		\nonumber \\
%		&\hby{\bactinp{t}{\auxtr{t'}}}& 
%		\newsp{\tilde{m_1}'}{P_1 \subst{\auxtr{t}}{X} \Par \appl{Y}{\tilde{n}} \subst{\auxtr{t'}}{Y}}
%		\label{cor:subst_equiv5}
%	\end{eqnarray}
%
	\noi and $t'$ a fresh name. From the freshness of $t$
	and the determinacy of the application transition
	it has to be the case that:
	\[
		\begin{array}{crll}
			& \Gamma; \es; \Delta_2'& \proves &
			\newsp{\tilde{m_2}'}{Q_1 \subst{\auxtr{t}}{x} \Par Q_2 \subst{\auxtr{t}}{x}}\\
			\Hby{} && &
			\newsp{\tilde{m_2}'}{Q_1' \subst{\auxtr{t}}{x} \Par Q_2' \subst{\auxtr{t}}{x} \Par \appl{x}{\tilde{m}} \subst{\auxtr{t}}{x}} \\
			\hby{\btau} \hby{\bactinp{t}{\auxtr{t'}}}
			& \Delta_2''& \proves& \newsp{\tilde{m_2}'}{(Q_1' \Par Q_2') \subst{\auxtr{t}}{x} \Par \appl{y}{\tilde{m}} \subst{\auxtr{t'}}{y}}
%			\hby{\bactinp{t}{\abs{\tilde{x}}{\binp{t'}{y} \appl{y}{\tilde{x}}}}}
%			& \Delta_2''& \proves& \newsp{\tilde{m_2}'}{(Q_1' \Par Q_2') \subst{\auxtr{t}}{x} \Par \appl{y}{\tilde{m}} \subst{\abs{\tilde{x}}{\binp{t'}{y} \appl{y}{\tilde{x}}}}{y}}
		\end{array}
	\]
%
%	\[
%		\begin{array}{rcll}
%			\Gamma; \es; \Delta_2' &\Hby{\bactinp{t}{\auxtr{t'}}}& \Delta_2'' \proves &
%			\newsp{\tilde{m_2}'}{Q_1 \subst{\auxtr{t}}{X} \Par Q_2 \subst{\auxtr{t}}{X}}\\
%			&\Hby{}& &
%			\newsp{\tilde{m_2}'}{Q_1' \subst{\auxtr{t}}{X} \Par Q_2' \subst{\auxtr{t}}{X} \Par \appl{X}{\tilde{m}} \subst{\auxtr{t}}{X}} \\
%			&\hby{\bactinp{t}{\auxtr{t'}}}& &
%			\newsp{\tilde{m_2}'}{(Q_1' \Par Q_2') \subst{\auxtr{t}}{X} \Par \appl{Y}{\tilde{m}} \subst{\auxtr{t'}}{Y}}
%		\end{array}
%	\]
%
	Let $Q_3$ such that
	\[
		\mhorel{\Gamma}{\Delta}{\newsp{\tilde{m_2}'}{Q_1 \Par Q_3} \subst{\auxtr{t}}{x} \subst{\auxtr{t'}}{y}}
		{\Hby{}}
		{\Delta'}{}{\newsp{\tilde{m_2}'}{(Q_1' \Par Q_2') \subst{\auxtr{t}}{x} \Par \appl{y}{\tilde{m}} \subst{\auxtr{t'}}{y}}}
%		\mhorel{\Gamma}{\Delta}{\newsp{\tilde{m_2}'}{Q_1 \Par Q_3} \subst{\auxtr{t}}{x} \subst{\abs{\tilde{x}}{\binp{t'}{y} \appl{y}{\tilde{x}}}}{y}}
%		{\Hby{}}
%		{\Delta'}{}{\newsp{\tilde{m_2}'}{(Q_1' \Par Q_2') \subst{\auxtr{t}}{x} \Par \appl{y}{\tilde{m}} \subst{\abs{\tilde{x}}{\binp{t'}{y} \appl{y}{\tilde{x}}}}{y}}}
	\]
	\noi From \lemref{lem:subst_equiv} we get that $\forall R$ with $\fv{R} = \tilde{x}$
	\begin{eqnarray*}
		\mhorel{\Gamma}{\Delta_1'''}{\newsp{\tilde{m_1}'}{P_1 \subst{\auxtr{t}}{x} \Par \appl{y}{\tilde{n}} \subst{\abs{\tilde{x}}{R}}{y}}}
		{\wb}
		{\Delta'}{}{\newsp{\tilde{m_2}'}{(Q_1 \Par Q_3) \subst{\auxtr{t}}{x} \subst{\abs{\tilde{x}}{R}}{y}}}
%		\mhorel{\Gamma}{\Delta_1'''}{\newsp{\tilde{m_1}'}{P_1 \subst{\auxtr{t}}{X} \Par \appl{y}{\tilde{n}} \subst{\abs{\tilde{x}}{R}}{y}}}
%		{\wb}
%		{\Delta'}{}{\newsp{\tilde{m_2}'}{(Q_1 \Par Q_3) \subst{\auxtr{t}}{X} \subst{\abs{\tilde{x}}{R}}{y}}}
	\end{eqnarray*}
	\noi From~(\ref{lem:subst_equiv11}) we get that
\[
	\mhorel{\Gamma}{\Delta'}{\newsp{\tilde{m_1}'}{(Q_1 \Par Q_3) \subst{\auxtr{t}}{x} \subst{\abs{\tilde{x}}{R}}{y}}}
	{\Hby{} \hby{\btau}}
	{\Delta''}{}{\newsp{\tilde{m_2}'}{(Q_1' \Par Q_2') \subst{\auxtr{t}}{x} \Par R \subst{\tilde{m}}{\tilde{x}}}}%\appl{Y}{\tilde{m}} \subst{\auxtr{t'}}{Y}}}
%	{\Hby{\ell_2}}
%	{\Delta''}{}{\newsp{\tilde{m_2}'}{(Q_1' \Par Q_2') \subst{\auxtr{t}}{x} \Par R'}}%\appl{Y}{\tilde{m}} \subst{\abs{\tilde{x}}{\binp{t'}{Y} \appl{Y}{\tilde{x}}}}{Y}}}
\]
	\noi and from the definition of $\Re$
%	\noi From Lemma~\ref{lem:subst_equiv} we get that $\forall R$ with $\fv{R} = x$
%
	\[
		\mhorel{\Gamma}{\Delta_1''}{\newsp{\tilde{m_1}'}{P_1 \subst{\abs{\tilde{x}}{R}}{x} \Par \appl{y}{\tilde{n}} \subst{\abs{\tilde{x}}{R}}{y}}}
		{\hby{\btau}\ \Re\ \stackrel{\btau}{\longleftarrow}}
		{\Delta_2''}{}{\newsp{\tilde{m_2}'}{(Q_1' \Par Q_2') \subst{\abs{\tilde{x}}{R}}{x} \Par \appl{y}{\tilde{m}} \subst{\abs{\tilde{x}}{R}}{y}}}
%		\mhorel{\Gamma}{\Delta_1''}{\newsp{\tilde{m_1}'}{P_1 \subst{\abs{\tilde{x}}{R}}{x} \Par \appl{Y}{\tilde{n}} \subst{\abs{\tilde{x}}{R'}}{y}}}
%		{\ \mathcal{S}\ }
%		{\Delta_2''}{}{\newsp{\tilde{m_2}'}{(Q_1' \Par Q_2') \subst{\abs{\tilde{x}}{R}}{x} \Par \appl{y}{\tilde{m}} \subst{\abs{\tilde{x}}{R'}}{y}}}
	\]
	\noi as required.
%	\noi From here we apply Lemma~\ref{lem:subst_equiv} for each substituting instance of
%	abstraction $\abs{\tilde{x}}{R}$ to complete the proof.
	\qed
\end{proof}

%%%%%%%%%%%%%%%%%%%%%%%%%%%%%%%%%%%%%%%%%%%%%%%%%%%%%%%%%
%  WB IS WBC
%%%%%%%%%%%%%%%%%%%%%%%%%%%%%%%%%%%%%%%%%%%%%%%%%%%%%%%%%

\begin{lemma}\rm
	\label{lem:wb_is_wbc}
	$\wb\ \subseteq\ \wbc$
\end{lemma}

\begin{proof}
	Let
	\[
		\horel{\Gamma}{\Delta_1}{P_1}{\wb}{\Delta_2}{Q_1}
	\]
	The proof is divided on cases on the label $\ell$ for the transition:
	\begin{eqnarray}
		\horel{\Gamma}{\Delta_1}{P_1}{\hby{\ell}}{\Delta_1'}{P_2}
		\label{lem:wb_is_wbc1}
	\end{eqnarray}
	\noi - Case: $\ell \notin \set{ \news{\tilde{m_1}} \bactout{n}{\abs{\tilde{x}}{P}},  \news{\tilde{m_1}'} \bactout{n}{\tilde{m_1}}, \bactinp{n}{\abs{\tilde{x}}{P}} }$

	\noi For the latter $\ell$ and transition~in (\ref{lem:wb_is_wbc1}) we conclude that:	
	\[
		\horel{\Gamma}{\Delta_2}{Q_1}{\Hby{\ell}}{\Delta_2'}{Q_2}
	\]
	\noi and
	\[
		\horel{\Gamma}{\Delta_1'}{P_2}{\wb}{\Delta_2'}{Q_2}
	\]
	The above premise and conclusion coincides with defining cases for $\ell$ in $\wbc$.

	\noi - Case: $\ell = \bactinp{n}{\abs{\tilde{x}}{P}}$

	\noi Transition in~(\ref{lem:wb_is_wbc1}) concludes:
\[
	\begin{array}{l}
		\horel{\Gamma}{\Delta_1}{P_1}{\hby{\bactinp{n}{\abs{\tilde{x}}{\mapchar{U}{\tilde{x}}}}}}{\Delta_1'}{P_2 \subst{\abs{\tilde{x}}{\mapchar{U}{\tilde{x}}}}{x}}\\
		\horel{\Gamma}{\Delta_1}{P_1}{\hby{\bactinp{n}{\auxtr{t}}}}{\Delta_1''}{P_2 \subst{\auxtr{t}}{x}}
	\end{array}
\]
	\noi The last two transitions imply:
\[
	\begin{array}{l}
		\horel{\Gamma}{\Delta_2}{Q_1}{\Hby{\bactinp{n}{\abs{\tilde{x}}{\mapchar{U}{\tilde{x}}}}}}{\Delta_2'}{Q_2 \subst{\abs{\tilde{x}}{\mapchar{U}{\tilde{x}}}}{x}}\\
		\horel{\Gamma}{\Delta_2}{Q_1}{\Hby{\bactinp{n}{\auxtr{t}}}}{\Delta_2''}{Q_2 \subst{\auxtr{t}}{x}}
	\end{array}
\]
	\noi and
\[
	\begin{array}{l}
		\horel{\Gamma}{\Delta_1'}{P_2 \subst{\abs{\tilde{x}}{\mapchar{U}{\tilde{x}}}}{x}}{\wb}{\Delta_2'}{Q_2 \subst{\abs{\tilde{x}}{\mapchar{U}{\tilde{x}}}}{x}}\\
		\horel{\Gamma}{\Delta_1''}{P_2 \subst{\auxtr{t}}{x}}{\wb}{\Delta_2''}{Q_2 \subst{\auxtr{t}}{x}}
	\end{array}
\]
	\noi To conlude from (\ref{lem:proc_subst}) that
	$\forall R$ with $\fv{R} = \tilde{x}$
\[
	\horel{\Gamma}{\Delta_1'}{P_2 \subst{\abs{\tilde{x}}{R}}{x}}{\wb}{\Delta_2'}{Q_2 \subst{\abs{\tilde{x}}{R}}{x}}
\]
	\noi as required.

	\noi - Case: $\ell = \news{\tilde{m_1}} \bactout{n}{\abs{\tilde{x}}{P}}$

	\noi From transition~(\ref{lem:wb_is_wbc1}) we conclude:
\[
	\horel{\Gamma}{\Delta_2}{Q_1}{\Hby{\news{\tilde{m_2}} \bactout{n}{\abs{\tilde{x}}{Q}}}}{\Delta_2'}{Q_2}
\]
	\noi and for fresh $t$
\[
	\mhorel	{\Gamma}{\Delta_1'}{\newsp{\tilde{m_1}}{P_2 \Par \binp{t}{x} \newsp{s}{\appl{x}{s} \Par \bout{\dual{s}}{\abs{\tilde{x}}{P}} \inact}}}
		{\wb}
		{\Delta_2'}{}{\newsp{\tilde{m_2}}{Q_2 \Par \binp{t}{x} \newsp{s}{\appl{x}{s} \Par \bout{\dual{s}}{\abs{\tilde{x}}{Q}} \inact}}}
\]
	\noi From the  previous case we can conclude that $\forall R$ with $\fpv{R} = \set{x}$:
\[
	\begin{array}{rl}
		\Gamma; \es; &\Delta_1' \proves \newsp{\tilde{m_1}}{P_2 \Par \binp{t}{x} \newsp{s}{\appl{x}{s} \Par \bout{\dual{s}}{\abs{\tilde{x}}{P}} \inact}} \\
		\by{\bactinp{t}{\abs{z}{\binp{z}{x} R}}}& \newsp{\tilde{m_1}}{P_2 \Par \newsp{s}{\binp{s}{x} R \Par \bout{\dual{s}}{\abs{\tilde{x}}{P}} \inact}}\\
		\by{\tau} \quad &\Delta_1'' \proves \newsp{\tilde{m_1}}{P_2 \Par  R \subst{\abs{\tilde{x}}{P}}{x}}
	\end{array}
\]
	\noi and
\[
	\begin{array}{rl}
		\Gamma; \es; &\Delta_2' \proves \newsp{\tilde{m_2}}{Q_2 \Par \binp{t}{x} \newsp{s}{\appl{x}{s} \Par \bout{\dual{s}}{\abs{\tilde{x}}{Q}} \inact}} \\
		\by{\bactinp{t}{\abs{z}{\binp{z}{x} R}}} &\newsp{\tilde{m_2}}{Q_2 \Par \newsp{s}{\binp{s}{x} R \Par \bout{\dual{s}}{{\tilde{x}}{Q}} \inact}}\\
		\by{\tau} &\Delta_2'' \proves \newsp{\tilde{m_2}}{Q_2 \Par  R \subst{\abs{\tilde{x}}{Q}}{x}}
	\end{array}
\]
	\noi and furthermore it is easy to see that $\forall R$ with $\fpv{R} = X$:
\[
	\horel{\Gamma}{\Delta_1''}{\newsp{\tilde{m_1}}{P_2 \Par  R \subst{\abs{\tilde{x}}{P}}{x}}}{\wb}{\Delta_2}{\newsp{\tilde{m_2}}{Q_2 \Par R \subst{\abs{\tilde{x}}{Q}}{x}}}
\]
	\noi as required by the definition of $\wbc$.

	\noi - Case: $\ell = \news{\tilde{m_1}'} \bactout{n}{\tilde{m_1}}$

	The last case shares a similar argumentation with the previous case.
	\qed
\end{proof}

%%%%%%%%%%%%%%%%%%%%%%%%%%%%%%%%%%%%%%%%%%%%%%%%%%%%%%%%%
%  WB IS CONG
%%%%%%%%%%%%%%%%%%%%%%%%%%%%%%%%%%%%%%%%%%%%%%%%%%%%%%%%%

\begin{lemma}
	\label{lem:wbc_is_cong}
	$\wbc \subseteq \cong$.
\end{lemma}

\begin{proof}
	\noi We prove that $\wbc$ satisfies the defining properties of $\cong$. Let
	\[
		\horel{\Gamma}{\Delta_1}{P_1}{\wbc}{\Delta_2}{P_2}
	\]
	{\bf Reduction Closed:}
	\[
		\horel{\Gamma}{\Delta_1}{P_1}{\by{}}{\Delta_1'}{P_1'}
	\]
	\noi implies that 
	$\exists P_2'$ such that 
	\begin{eqnarray*}
		\horel{\Gamma}{\Delta_2}{P_2}{\By{}}{\Delta_2'}{P_2'}\\
		\horel{\Gamma}{\Delta_1}{P_1'}{\wbc}{\Delta_2'}{P_2'}
	\end{eqnarray*}
	\noi Same argument hold for the symmetric case, thus $\wbc$ is reduction closed.

	\noi {\bf Barb Preservation:}
	\begin{eqnarray*}
		\Gamma; \emptyset; \Delta_1 \proves P_1 \hastype \Proc \barb{n}
	\end{eqnarray*}
	implies that
	\begin{eqnarray*}
		P &\cong& \newsp{\tilde{m}}{\bout{n}{V_1} P_3 \Par P_4}\\
		\dual{n} &\notin& \Delta_1
	\end{eqnarray*}
	\noi From the definition of $\wbc$ we get that
\[
	\horel	{\Gamma}{\Delta_1}{\newsp{\tilde{m}}{\bout{n}{V_1} P_3 \Par P_4}}
		{\by{\news{s_1} \bactout{m}{V_1}}}
		{\Delta_1'}
		{\newsp{\tilde{m'}}{P_3 \Par P_4}}
\]
	\noi implies
	\begin{eqnarray*}
		\horel{\Gamma}{\Delta_2}{P_2}{\By{\news{m_2} \bactout{n}{V_2}}}{\Delta_2'}{P_2'}\\
	\end{eqnarray*}
	\noi From the last result we get that
	\begin{eqnarray*}
		\Gamma; \emptyset; \Delta_2 \proves P_2 \hastype \Proc \Barb{n}
	\end{eqnarray*}
	\noi as required.

	\noi {\bf Congruence:}

	\noi The congruence property requires that we check that $\wbc$
	is preserved under any context.
	The most interesting context case is parallel composition.

	\noi We construct a congruence relation. Let
	\[
	\begin{array}{rcl}
		\mathcal{S} &=&	\set{
				(\Gamma; \emptyset; \Delta_1 \cat \Delta_3 \proves \newsp{\tilde{n_1}}{P_1 \Par R} \hastype \Proc,
				\Gamma; \emptyset; \Delta_2 \cat \Delta_3 \proves \newsp{\tilde{n_2}}{P_2 \Par R})
				\setbar \\
		& &		\horel{\Gamma}{\Delta_1}{P_1}{\wbc}{\Delta_2}{P_2}, \forall \Gamma; \emptyset; \Delta_3 \proves R \hastype \Proc\\
		& &}
	\end{array}
	\]
	\noi We need to show that the above congruence is a bisimulation.
	To show that $\mathcal{S}$ is a bisimulation we do a case analysis on the structure
	of the $\by{\ell}$ transition.

	%%%%%%%%%%%%%%%
	% Case 1
	%%%%%%%%%%%%%%%

	\noi - Case: 
	\[
		\horel{\Gamma}{\Delta_1 \cat \Delta_3}{\newsp{\tilde{n_1}}{P_1 \Par R}}{\by{\ell}}{\Delta_1' \cat \Delta_3}{\newsp{\tilde{n_1'}}{P_1' \Par R}}
	\]

	\noi The case is divided into three subcases:

	\noi Subcase i: $\ell \notin \set{\news{\tilde{m}} \bactout{n}{\abs{\tilde{x}}{Q}}, \news{\tilde{mm_1}} \bactout{n}{\tilde{m_1}}}$

	\noi From the definition of typed transition we get:
	\[
		\horel{\Gamma}{\Delta_1}{P_1}{\by{\ell}}{\Delta_1'}{P_1'}
	\]
	\noi which implies that
	\begin{eqnarray}
		\horel{\Gamma}{\Delta_1}{P_2}{\By{\ell}}{\Delta_2'}{P_2'}
		\label{lem:wbc_is_cong1}\\
		\horel{\Gamma}{\Delta_1'}{P_1'}{\wbc}{\Delta_2''}{P_2'}
		\label{lem:wbc_is_cong2}
	\end{eqnarray}
	\noi From transition in~(\ref{lem:wbc_is_cong1}) we conclude that 
	\[
		\horel{\Gamma}{\Delta_2 \cat \Delta_3}{\newsp{\tilde{n_2}}{P_2 \Par R}}{\By{\ell}}{\Delta_2' \cat \Delta_3}{\newsp{\tilde{n_2}'}{P_2' \Par R}}
	\]
	\noi Furthermore from~(\ref{lem:wbc_is_cong2}) and the definition of $\mathcal{S}$ we conlude that
	\[
		\horel{\Gamma}{\Delta_1' \cat \Delta_3}{\newsp{\tilde{n_1}'}{P_1' \Par R}}{\ \mathcal{S}\ }{\Delta_2' \cat \Delta_3}{\newsp{\tilde{n_2}'}{P_2' \Par R}}
	\]

	\noi Subcase ii: $\ell = \news{\tilde{m_1}} \bactout{n}{\abs{\tilde{x}}{Q_1}}$

	\noi From the definition of typed transition we get
	\[
		\horel{\Gamma}{\Delta_1}{P_1}{\by{\news{\tilde{m_1}} \bactout{n}{\abs{\tilde{x}}{Q_1}}}}{\Delta_1'}{P_1'}
	\]
	\noi which implies that
	\begin{eqnarray}
		&& \horel{\Gamma}{\Delta_1}{P_2}{\By{\news{\tilde{m_2}} \bactout{n}{\abs{\tilde{x}}{Q_2}}}}{\Delta_2'}{P_2'}
		\label{lem:wbc_is_cong3} \\
		&&\forall Q, \set{x} \in \fpv{Q} \nonumber \\
%		\forall s'
		&& \horel{\Gamma}{\Delta_1''}{\newsp{\tilde{n_1}''}{P_1' \Par Q \subst{\abs{\tilde{x}}{Q_1}}{x}}}
		{\ \wbc\ }
		{\Delta_2''}{\newsp{\tilde{n_2}''}{P_2' \Par Q \subst{\abs{\tilde{x}}{Q_2}}{x}}}
		\label{lem:wbc_is_cong4}
	\end{eqnarray}
	\noi From transition~(\ref{lem:wbc_is_cong3}) conclude that 
	\[
		\horel{\Gamma}{\Delta_2 \cat \Delta_3}{\newsp{\tilde{n_2}}{P_2 \Par R}}{\By{\news{\tilde{m_2}} \bactout{n}{\abs{\tilde{x}}{Q_2}}}}{\Delta_2' \cat \Delta_3}{\newsp{\tilde{n_2}'}{P_2' \Par R}}
	\]
	\noi Furthermore from~(\ref{lem:wbc_is_cong4}) we conlude that $\forall Q$ with $\set{x} = \fpv{Q}$
	\[
		\horel{\Gamma}{\Delta_1'' \cat \Delta_3}{\newsp{\tilde{n_1}''}{P_1' \Par Q \subst{(\tilde{x}) Q_1}{x} \Par R}}{\ \mathcal{S}\ }{\Delta_2'' \cat \Delta_3}{\newsp{\tilde{n_2}''}{P_2' \Par Q \subst{\abs{\tilde{x}}{Q_2}}{x} \Par R}}
	\]
	- Subcase iii: $\ell = \news{\tilde{mm_1}} \bactout{n}{\tilde{m_1}}$

	\noi From the definition of typed transition we get that
	\[
		\horel{\Gamma}{\Delta_1}{P_1}{\by{\news{\tilde{mm_1}} \bactout{n}{\tilde{m_1}}}}{\Delta_1'}{P_1'}
	\]
	\noi which implies that $\exists P_2', s_2$ such that
	\begin{eqnarray}
		&& \horel{\Gamma}{\Delta_1}{P_2}{\By{\news{\tilde{mm_2}} \bactout{n}{\tilde{m_2}}}}{\Delta_2'}{P_2'}
		\label{lem:wbc_is_cong5}\\
		&&\forall Q, x = \fn{Q}, \nonumber \\%  &&
		&& \horel{\Gamma}{\Delta_1''}{\newsp{\tilde{n_1}}{P_1' \Par Q \subst{\tilde{m_1}}{\tilde{x}}}}{\ \wbc\ }{\Delta_2''}{\newsp{\tilde{n_2}}{P_2' \Par Q \subst{\tilde{m_2}}{\tilde{x}}}}
		\label{lem:wbc_is_cong6}
	\end{eqnarray}
	\noi From transition~(\ref{lem:wbc_is_cong5}) conclude that 
	\[
		\horel{\Gamma}{\Delta_2 \cat \Delta_3}{\newsp{\tilde{n_2}'}{P_2 \Par R}}{\By{\news{\tilde{mm_2}} \bactout{n}{\tilde{m_2}}}}{\Delta_2' \cat \Delta_3}{\newsp{\tilde{n_2}'''}{P_2' \Par R}}
	\]
	\noi Furthermore from~(\ref{lem:wbc_is_cong6}) we conlude that $\forall Q, x = \fn{Q}$
	\[
		\horel{\Gamma}{\Delta_1'' \cat \Delta_3}{\newsp{\tilde{n_1}''}{P_1' \Par Q \subst{\tilde{m_1}}{\tilde{x}} \Par R}}{\ \mathcal{S}\ }{\Delta_2'' \cat \Delta_3}{\newsp{\tilde{n_2}''}{P_2' \Par Q \subst{\tilde{m_2}}{\tilde{x}} \Par R}}
	\]
%
	%%%%%%%%%%%%%%%
	% Case 2
	%%%%%%%%%%%%%%%

	\noi - Case:
	\[
		\horel{\Gamma}{\Delta_1 \cat \Delta_3}{\newsp{\tilde{m_1}}{P_1 \Par R}}{\by{\ell}}{\Delta_1 \cat \Delta_3'}{\newsp{\tilde{m_1}'}{P_1 \Par R'}}
	\]
	\noi This case is divided into three subcases:

	\noi Subcase i: $\ell \notin \set{\news{\tilde{m}} \bactout{n}{\abs{\tilde{x}}{Q}}, \news{\tilde{mm_1}} \bactout{n}{\tilde{m_1}}}$

	\noi From the LTS we get that:
	\[
		\horel{\Gamma}{\Delta_3}{R}{\by{\ell}}{\Delta_3'}{R'}
	\]
	\noi Which in turn implies
	\begin{eqnarray*}
		\horel{\Gamma}{\Delta_2 \cat \Delta_3}{\newsp{\tilde{m_2}}{P_2 \Par R}}{\by{\ell}}{\Delta_2 \cat \Delta_3'}{\newsp{\tilde{m_2}'}{P_2 \Par R'}}
	\end{eqnarray*}
	\noi From the definition of $\mathcal{S}$ we conclude that
	\[
		\horel{\Gamma}{\Delta_1 \cat \Delta_3'}{\newsp{\tilde{m_1}'}{P_1 \Par R'}}{\ \mathcal{S}\ }{\Delta_2 \cat \Delta_3''}{\newsp{\tilde{m_2}'}{P_2 \Par R'}}
	\]
	\noi as required.

	\noi Subcase ii: $\ell = \news{\tilde{m_1}} \bactout{n}{\abs{\tilde{x}}{Q}}$

	\noi From the LTS we get that:
	\begin{eqnarray}
		& &	\horel{\Gamma}{\Delta_3}{R}{\by{\ell}}{\Delta_3'}{R'}
			\label{lem:wbc_is_cong7}\\
		& & 	\forall R_1, \set{x} = \fpv{R_1},
			\nonumber\\
%		\forall s'
		& &	\Gamma; \emptyset; \Delta_3'' \proves \newsp{\tilde{m}'}{R' \Par R_1 \subst{\abs{\tilde{x}}{Q}}{x}} \hastype \Proc
			\label{lem:wbc_is_cong8}
	\end{eqnarray}
	\noi From~(\ref{lem:wbc_is_cong7}) we get that
	\[
		\horel{\Gamma}{\Delta_2 \cat \Delta_3}{\newsp{\tilde{m_2}'}{P_2 \Par R}}{\by{\ell}}{\Delta_2 \cat \Delta_3'}{\newsp{\tilde{m_2}}{P_2 \Par R'}}
	\]
	\noi Furthermore from~(\ref{lem:wbc_is_cong8}) and the definition of $\mathcal{S}$ we conclude that
	$\forall R_1$ with $\set{x} \in \fpv{R_1}$
	\[
		\horel{\Gamma}{\Delta_1 \cat \Delta_3''}{\newsp{\tilde{m_1}}{P_1 \Par \newsp{\tilde{m}'}{R' \Par R_1 \subst{\abs{\tilde{x}}{Q}}{x}}}}
		{\ \mathcal{S}\ }
		{\Delta_2 \cup \Delta_3''}{\newsp{\tilde{m_2}}{P_2 \Par \newsp{\tilde{m}'}{R' \Par R_1 \subst{\abs{\tilde{x}}{Q}}{x}}}}
	\]
	\noi as required.

	\noi Subcase iii: $\ell = \news{\tilde{mm}} \bactout{n}{\tilde{m}}$

	\noi From the typed LTS we get that:
	\begin{eqnarray}
		& &	\horel{\Gamma}{\Delta_3}{R}{\by{\ell}}{\Delta_3'}{R'}
			\label{lem:wbc_is_cong9} \\
		& &	\forall Q, \tilde{x} = \fn{Q}, \nonumber\\
		& &	\Gamma; \emptyset; \Delta_3'' \proves \newsp{\tilde{m}'}{R' \Par Q \subst{\tilde{m}}{\tilde{x}}} \hastype \Proc
			\label{lem:wbc_is_cong10}
	\end{eqnarray}
	\noi From~(\ref{lem:wbc_is_cong9}), we obtain that
	\[
		\horel{\Gamma}{\Delta_2 \cat \Delta_3}{\newsp{\tilde{m_2}}{P_2 \Par R}}{\by{\ell}}{\Delta_2 \cat \Delta_3'}{\newsp{\tilde{m_2}}{P_2 \Par R'}}
	\]
	\noi Furthermore from~(\ref{lem:wbc_is_cong10}) and the definition of $\mathcal{S}$ we conclude that
	$\forall Q, \tilde{x} = \fn{Q}$
	\[
		\horel{\Gamma}{\Delta_1 \cat \Delta_3''}{\newsp{\tilde{m_1}}{P_1 \Par \newsp{\tilde{m}}{R' \Par Q \subst{\tilde{m}'}{\tilde{x}}}}}
		{\ \mathcal{S}\ }
		{\Delta_2 \cat \Delta_3''}{\newsp{\tilde{m_2}}{P_2 \Par \newsp{\tilde{m}'}{R' \Par Q \subst{\tilde{m}}{\tilde{x}}}}}
	\]
	\noi as required.

	%%%%%%%%%%%%%%%
	% Case 3
	%%%%%%%%%%%%%%%

	\noi - Case:
	\[
		\horel{\Gamma}{\Delta_1 \cat \Delta_3}{\newsp{\tilde{m_1}}{P_1 \Par R}}
		{\by{}}
		{\Delta_1' \cat \Delta_3'}{\newsp{\tilde{m_1}'}{P_1' \Par R'}}
	\]

	\noi This case is divided into three subcases:

	\noi Subcase i: $\horel{\Gamma}{\Delta_1}{P_1}{\by{\ell}}{\Delta_1'}{P_1'}$
	and $\ell \notin \set{\news{\tilde{m}} \bactout{n}{\abs{\tilde{x}}{Q}}, \news{\tilde{mm_1}} \bactout{n}{\tilde{m_1}}}$ implies
	\begin{eqnarray}
		\horel{\Gamma}{\Delta_3}{R}{\by{\dual{\ell}}}{\Delta_3}{R'}
		\label{lem:wbc_is_cong11} \\
		\horel{\Gamma}{\Delta_2}{P_2}{\By{\hat{\ell}}}{\Delta_2'}{P_2'}
		\label{lem:wbc_is_cong12}\\
		\horel{\Gamma}{\Delta_1'}{P_1'}{\wbc}{\Delta_2'}{P_2'}
		\label{lem:wbc_is_cong13}
	\end{eqnarray}
	\noi From~(\ref{lem:wbc_is_cong11}) and~(\ref{lem:wbc_is_cong12}) we get
	\[
		\horel{\Gamma}{\Delta_2 \cat \Delta_3}{\newsp{\tilde{m_2}}{P_2 \Par R}}{\By{}}{\Delta_2' \cat \Delta_3'}{\newsp{\tilde{m_2}'}{P_2' \Par R'}}
	\]
	\noi From~(\ref{lem:wbc_is_cong13}) and the definition of ($\mathcal{S}$) we get that
	\[
		\horel{\Gamma}{\Delta_1' \cat \Delta_3'}{\newsp{\tilde{m_1}'}{P_1' \Par R'}}{\ \mathcal{S}\ }{\Delta_2' \cat \Delta_3}{\newsp{\tilde{m_2}'}{P_2' \Par R'}}
	\]
	\noi as required.

	\noi Subcase ii:
	$\horel{\Gamma}{\Delta_1}{P_1}{\by{\news{\tilde{m_1}} \bactout{n}{\abs{\tilde{x}}{Q_1}}}}{\Delta_1'}{P_1'}$
	implies
	\begin{eqnarray}
		& & \horel{\Gamma}{\Delta_3}{R}{\by{\bactinp{n}{\abs{\tilde{x}} {Q_1}}}}{\Delta_3'}{R' \subst{\abs{\tilde{x}}{Q_1}}{x}}
		\label{lem:wbc_is_cong14}\\
		& & \horel{\Gamma}{\Delta_1 \cat \Delta_3}{\newsp{\tilde{m_1}}{P_1 \Par R}}{\by{}}{\Delta_1' \cat \Delta_3'}{\newsp{\tilde{m_1}''}{P_1' \Par R' \subst{\abs{\tilde{x}}{Q_1}}{x}}}
		\nonumber \\
		& & \horel{\Gamma}{\Delta_2}{P_2}{\By{\news{\tilde{m_2}} \bactout{n}{\abs{\tilde{x}}{Q_2}}}}{\Delta_2'}{P_2'}
		\label{lem:wbc_is_cong15}\\
		& & \forall Q, \set{x} = \fpv{Q}, \nonumber \\
		& & \horel{\Gamma}{\Delta_1''}{\newsp{\tilde{m_1}'}{P_1' \Par Q \subst{\abs{\tilde{x}}{Q_1}}{x}}}{\ \wbc\ }{\Delta_2''}{\newsp{\tilde{m_2}'}{P_2' \Par Q \subst{\abs{\tilde{x}}{Q_2}}{x}}}
		\label{lem:wbc_is_cong16}
	\end{eqnarray}
	From~(\ref{lem:wbc_is_cong14}) and the Substitution Lemma~(\lemref{lem:subst}) we obtain that
	\[
		\horel{\Gamma}{\Delta_3}{R}{\by{\bactinp{n}{\abs{\tilde{x}} {Q_2}}}}{\Delta_3''}{R' \subst{\abs{\tilde{x}}{Q_2}}{x}}
	\]
	%\dk{(prove that $\forall V, R \by{\bactinp{s}{V}} R'\subst{V}{x}$)}
	\noi to combine with~(\ref{lem:wbc_is_cong15}) and get
	\[
		\horel{\Gamma}{\Delta_2 \cat \Delta_3}{\newsp{\tilde{m_2}}{P_2 \Par R}}{\By{}}{\Delta_2' \cat \Delta_3''}{\newsp{\tilde{m_2}''}{P_2' \Par R' \subst{\abs{\tilde{x}}{Q_2}}{X}}}
	\]
	\noi In result in~(\ref{lem:wbc_is_cong16}), set $Q$ as $R'$ to obtain:
%
%	\noi From~\ref{lem:wbc_is_cong16} and the definition of $\mathcal{S}$ we get that
	\[
		\horel{\Gamma}{\Delta_1''}{\newsp{\tilde{m_1}'}{P_1' \Par R' \subst{\abs{\tilde{x}}{Q_1}}{x}}}
		{\ \mathcal{S}\ \Delta_2''}
		{\newsp{\tilde{m_2}'}{P_2' \Par R' \subst{\abs{\tilde{x}}{Q_2}}{x}}}
	\]

	\noi Subcase iii:
	$\horel{\Gamma}{\Delta_1}{P_1}{\by{\news{\tilde{mm_1}} \bactout{n}{\tilde{m_1}}}}{\Delta_1'}{P_1'}$
	\begin{eqnarray}
		& & \horel{\Gamma}{\Delta_3}{R}{\by{\bactinp{n}{\tilde{m_1}}}}{\Delta_3'}{R' \subst{\tilde{m_1}}{\tilde{x}}}
		\label{lem:wbc_is_cong24}\\
		& & \horel{\Gamma}{\Delta_1 \cup \Delta_3}{\newsp{\tilde{m_1}}{P_1 \Par R}}{\by{}}{\Delta_1' \cup \Delta_3'}{\newsp{\tilde{m_1}''}{P_1' \Par R' \subst{s_1}{x}}}
		\nonumber \\
		& & \horel{\Gamma}{\Delta_2}{P_2}{\By{\news{\tilde{mm_2}} \bactout{n}{\tilde{m_2}}}}{\Delta_2'}{P_2'}
		\label{lem:wbc_is_cong25}\\
		& & \forall Q, \set{x} = \fpv{Q}, \nonumber \\
		& & \horel{\Gamma}{\Delta_1''}{\newsp{\tilde{m_1}'}{P_1' \Par Q \subst{\tilde{m_1}}{\tilde{x}}}}
		{\ \wbc\ }
		{\Delta_2''}{\newsp{\tilde{m_2}'}{P_2' \Par Q \subst{\tilde{m_2}}{\tilde{x}}}}
		\label{lem:wbc_is_cong26}
	\end{eqnarray}
	From~(\ref{lem:wbc_is_cong24}) and the Substitution Lemma~(\lemref{lem:subst}) we get that
	\[
		\horel{\Gamma}{\Delta_3}{R}{\by{\bactinp{n}{\tilde{m_2}}}}{\Delta_3''}{R' \subst{\tilde{m_2}}{\tilde{x}}}
	\]
	%\dk{(prove that $\forall V, R \by{\bactinp{s}{V}} R'\subst{V}{x}$)}
	\noi to combine with~(\ref{lem:wbc_is_cong25}) and get
	\[
		\horel{\Gamma}{\Delta_2 \cat \Delta_3}{\newsp{\tilde{m_2}}{P_2 \Par R}}
		{\By{}}
		{\Delta_2' \cat \Delta_3''}{\newsp{\tilde{m_2}''}{P_2' \Par R' \subst{\tilde{m_2}}{\tilde{x}}}}
	\]
	\noi Set $Q$ as $R'$ in result in (\ref{lem:wbc_is_cong26}) to obtain
%
%	\noi From~\ref{lem:wbc_is_cong16} and the definition of $\mathcal{S}$ we get that
	\[
		\horel{\Gamma}{\Delta_1''}{\newsp{\tilde{m_1}'}{P_1' \Par R' \subst{\tilde{m_1}}{\tilde{x}}}}
		{\ \mathcal{S}\ }
		{\Delta_2''}{\newsp{\tilde{m_2}'}{P_2' \Par R' \subst{\tilde{m_2}}{\tilde{x}}}}
	\]
	\qed
\end{proof}

%%%%%%%%%%%%%%%%%%%%%%%%%%%%%%%%%%%%%%%%%%%%%%%%%%%%%%%%%
%  CONG IS WB
%%%%%%%%%%%%%%%%%%%%%%%%%%%%%%%%%%%%%%%%%%%%%%%%%%%%%%%%%

We prove the result $\cong \subseteq \wb$ following
the technique developed in~\cite{Hennessy07} and
refined for session types in~\cite{KYHH2015,KY2015}.

\begin{definition}[Definibility]\myrm
	Let $\Gamma; \emptyset; \Delta_1 \proves P \hastype \Proc$.
	A visible action $\ell$ is \emph{definable} whenever
	there exists (testing) process
	$\Gamma; \emptyset; \Delta_2 \proves T\lrangle{\ell, \suc} \hastype \Proc$
	with $\suc$ fresh name % and $N$ a set of names.
	such that:
	\begin{itemize}
		\item	If $\horel{\Gamma}{\Delta_1}{P}{\by{\ell}}{\Delta_1'}{P'}$ and
			$\ell \in \set{\bactsel{n}{\ell}, \bactbra{n}{\ell}, \bactinp{n}{\tilde{m}}, \bactinp{n}{\abs{\tilde{x}}{Q}}}$
			then:
		\[
			P \Par T\lrangle{\ell, \suc} \red P' \Par \bout{\suc}{\dual{m}} \inact \textrm{ and }
			\Gamma; \emptyset; \Delta_1' \cat \Delta_2' \proves P' \Par \bout{\suc}{\dual{m}} \inact
		\]
 		\item	If $\horel{\Gamma}{\Delta_1}{P}{\by{\news{\tilde{m}}\bactout{n}{V}}}{\Delta_1'}{P'}$,
			$t$ fresh
			and $\tilde{m}' \subseteq \tilde{m}$
			then:
			\begin{eqnarray*}
				& & P \Par T\lrangle{\news{\tilde{m}}\bactout{n}{V}, \suc} \red
				\newsp{\tilde{m}}{P' \Par \hotrigger{t}{x}{s}{V} \Par \bout{\suc}{\dual{n}, \tilde{m}'} \inact}\\
				& & \Gamma; \emptyset; \Delta_1' \cat \Delta_2' \proves
				\newsp{\tilde{m}}{P' \Par \hotrigger{t}{x}{s}{V} \Par  \bout{\suc}{\dual{n}, \tilde{m}'} \inact} \hastype \Proc\\
			\end{eqnarray*}

		\item	Let $\ell \in \set{\bactsel{n}{\ell}, \bactbra{n}{\ell}, \bactinp{n}{\tilde{m}}, \bactinp{n}{(\tilde{x}) Q}}$.
			If $P \Par T\lrangle{\ell, \suc} \red Q$ with			
			$\Gamma; \emptyset; \Delta \proves Q \hastype \Proc \barb{\suc}$ then 
			$\horel{\Gamma}{\Delta_1}{P}{\By{\ell}}{\Delta_1'}{P'}$
			and $Q \scong P' \Par \bout{\suc}{\dual{n}} \inact$.

		\item	If $P \Par T\lrangle{\news{\tilde{m}}\bactout{n}{V}, \suc} \red Q$
			with $\Gamma; \emptyset; \Delta \proves Q \hastype \Proc \barb{\suc}$ then
			$\horel{\Gamma}{\Delta_1}{P}{\By{\news{\tilde{m}}\bactout{n}{V}}}{\Delta_1'}{P'}$
			and $Q \scong \newsp{\tilde{m}}{P' \Par \hotrigger{t}{x}{s}{V} \Par \bout{\suc}{\dual{n}, \tilde{m}'} \inact}$
			with $t$ fresh and $\tilde{m}' \subseteq \tilde{m}$.
	\end{itemize}	
\end{definition}

We first show that every visible action $\ell$ is {\em definable}.

\begin{lemma}[Definibility]
	\label{lem:definibility}
	Every action $\ell$ is definable.
\end{lemma}

\begin{proof}
	\noi We define $T\lrangle{\ell, \suc}$:
	\begin{itemize}
		\item	$T\lrangle{\bactinp{n}{V}, \suc} = \bout{\dual{n}}{V} \bout{\suc}{\dual{n}} \inact$.

		\item	$T\lrangle{\bactbra{n}{l}, \suc} = \bsel{\dual{n}}{l} \bout{\suc}{\dual{n}} \inact$.

%		\item	$T\lrangle{\bactinp{n}{\abs{\tilde{x}} Q}, \suc} = \bout{\dual{n}}{\abs{\tilde{x}}{Q}} \bout{\suc}{\dual{n}} \inact$.

		\item	$T\lrangle{\news{\tilde{m}'} \bactout{n}{\tilde{m}}, \suc} = \binp{\dual{n}}{\tilde{x}} (\hotrigger{t}{x}{s}{\tilde{x}} \Par \bout{\suc}{\dual{n}, \tilde{m}''} \inact)$
			with $\tilde{m}'' \subseteq \tilde{m}'$.

		\item	$T\lrangle{\news{\tilde{m}} \bactout{n}{\abs{\tilde{x}}{Q}}, \suc} = \binp{\dual{n}}{y} (\hotrigger{t}{x}{s}{\abs{\tilde{x}}{(\appl{y}{\tilde{x}}})} \Par \bout{\suc}{\dual{n}, \tilde{m}'} \inact)$ with $\tilde{m}' \subseteq \tilde{m}$.

		\item	$T\lrangle{\bactsel{n}{l}, \suc} = \bbra{\dual{n}}{l: \bout{\suc}{\dual{n}} \inact), l_i: \newsp{a}{\binp{a}{y} \bout{\suc}{\dual{n}} \inact}}_{i \in I}$.
	\end{itemize}

	\noi Assuming a process 
	\[
		\Gamma; \emptyset; \Delta \proves P \hastype \Proc
	\] 
	\noi it is straightforward to verify that $\forall \ell$, $\ell$ is definable.
	\qed
\end{proof}

\begin{lemma}[Extrusion]\rm
	\label{lem:extrusion}
	If 
	\[
		\horel{\Gamma}{\Delta_1'}{\newsp{\tilde{m_1}'}{P \Par \bout{\suc}{\dual{n}, \tilde{m_1}''} \inact}}{\cong}{\Delta_2}{\newsp{\tilde{m_2}'}{Q \Par \bout{\suc}{\dual{n}, \tilde{m_2}''} \inact}}
	\]
	then
	\[
		\horel{\Gamma}{\Delta_1}{P}{\cong}{\Delta_2}{Q}
	\]
\end{lemma}

\begin{proof}
	\noi Let
	\begin{eqnarray*}
		\mathcal{S}	&=&
					\set{\Gamma; \es; \Delta_1 \proves P \hastype \Proc, \Gamma; \es; \Delta_2 \proves Q \hastype \Proc \setbar \\
				& &	\horel{\Gamma}{\Delta_1'}{\newsp{\tilde{m_1}'}{P \Par \bout{\suc}{\dual{n}, \tilde{m_1}''} \inact}}
					{\cong}{\Delta_2}{\newsp{\tilde{m_2}'}{Q \Par \bout{\suc}{\dual{n}, \tilde{m_2}''} \inact}} \\
		&&}
	\end{eqnarray*}
	\noi We show that $\mathcal{S}$ is a congruence.

	\noi {\bf Reduction closed:}

	\noi $P \red P'$
	implies
	$\newsp{\tilde{m_1}'}{P \Par \bout{\suc}{\dual{n}, \tilde{m_1}''} \inact} \red \newsp{\tilde{m_1}'}{P' \Par \bout{\suc}{\dual{n}, \tilde{m_1}''} \inact}$
	implies from the freshness of $\suc$
	$\newsp{\tilde{m_1}'}{P \Par \bout{\suc}{\dual{n}, \tilde{m_1}''} \inact} \Red \newsp{\tilde{m_1}'}{Q' \Par \bout{\suc}{\dual{n}, \tilde{m_2}''} \inact}$.
	which implies
	$Q \Red Q'$ as required.

	\noi {\bf Barb Preserving:}

	\noi Let $\Gamma; \es; \Delta_1 \proves P \barb{s}$. We analyse two cases.

	\noi - Case: $s \not= n$.

	\noi $\Gamma; \es; \Delta_1 \proves P \barb{s}$
	implies
	\[
		\Gamma; \es; \Delta_1' \proves \newsp{\tilde{m_1}'}{P \Par \bout{\suc}{\dual{n}, \tilde{m_1}''} \inact} \barb{s}
	\]
	\noi implies
	$\Gamma; \es; \Delta_2' \proves \newsp{\tilde{m_2}'}{Q \Par \bout{\suc}{\dual{n}, \tilde{m_2}''} \inact} \Barb{s}$
	implies from the freshness of $\suc$ that
	$\Gamma; \es; \Delta_2 \proves Q \Barb{s}$ as required.

	\noi - Case: $s = n$ and $\Gamma; \es; \Delta_1 \proves P \barb{n}$

	\noi We compose with $\binp{\dual{\suc}}{x, \tilde{y}} T\lrangle{\ell, \suc'}$ with $\subj{\ell} = x$ to get
	\[
		\Gamma; \es; \Delta_1' \proves \newsp{\tilde{m_1}'}{P \Par \bout{\suc}{\dual{n}, \tilde{m_1}''} \inact} \Par \binp{\dual{\suc}}{x, \tilde{y}} T\lrangle{\ell, \suc'}
	\]
	\noi Which implies from the fact that $\Gamma; \es; \Delta_1 \proves P \barb{n}$ that
	\[
		\newsp{\tilde{m_1}'}{P \Par \bout{\suc}{\dual{n}, \tilde{m_1}''} \inact} \Par \binp{\dual{\suc}}{x, \tilde{y}} T\lrangle{\ell, \suc'} \Red 
		\newsp{\tilde{m_1}'}{P' \Par \bout{\suc'}{\dual{n}, \tilde{m_1}''} \inact}
	\]
	\noi and furthermore
	\[
		\newsp{\tilde{m_2}'}{Q \Par \bout{\suc}{\dual{n}, \tilde{m_2}''} \inact} \Par \binp{\dual{\suc}}{x, \tilde{y}} T\lrangle{\ell, \suc'} \Red 
		\newsp{\tilde{m_2}'}{Q' \Par \bout{\suc'}{\dual{n}, \tilde{m_2}''} \inact}
	\]
	\noi The last reduction implies that
	$\Gamma; \es; \Delta_2 \proves Q \Barb{n}$ as required.

	\noi {\bf Congruence:}
	The key case of congruence is parallel composition.
	We define relation $\mathcal{C}$ as
	\begin{eqnarray*}
		\mathcal{C} &=&	\set{ \Gamma; \es; \Delta_1 \cat \Delta_3 \proves P \Par R \hastype \Proc,  \Gamma; \es; \Delta_2 \cat \Delta_3 \proves Q \Par R \hastype \Proc \setbar \\
		& &	\forall R,\\
		& &	\horel{\Gamma}{\Delta_1'}{\newsp{\tilde{m_1}'}{P \Par \bout{\suc}{\dual{n}, \tilde{m_1}''} \inact}}{\cong}{\Delta_2'}{\newsp{\tilde{m_2}'}{Q \Par \bout{\suc}{\dual{n}, \tilde{m_2}''} \inact}}}
	\end{eqnarray*}
	\noi We show that $\mathcal{C}$ is a congruence.

	\noi We distinguish two cases:

	\noi - Case: $\dual{n}, \tilde{m_1}'', \tilde{m_2}'' \notin \fn{R}$	

	\noi From the definition of $\mathcal{C}$ we can deduce that $\forall R$:
	\[
		\horel{\Gamma}{\Delta_1''}{\newsp{\tilde{m_1}'}{P \Par \bout{\suc}{\dual{n}, \tilde{m_1}''} \inact} \Par R}{\cong}{\Delta_2''}{\newsp{\tilde{m_2}'}{Q \Par \bout{\suc}{\dual{n}, \tilde{m_2}''} \inact} \Par R}
	\]
	\noi The conclusion is then trivial.

	\noi - Case: $\tilde{s} = \set{\dual{n}, \tilde{m_1}''} \cap \set{\dual{n}, \tilde{m_2}''} \in \fn{R}$

	\noi From the definition of $\mathcal{C}$ we can deduce that $\forall R^{y_1}$ such that $R = R^{y_1}\subst{\tilde{s}}{\tilde{y_1}}$
	and $\suc'$ fresh and $\set{\tilde{y}} = \set{\tilde{y_1}} \cup \set{\tilde{y_2}}$:
	\[
		\mhorel{\Gamma}{\Delta_1''}{\newsp{\tilde{m_1}'}{P \Par \bout{\suc}{\dual{n}, \tilde{m_1}''} \inact} \Par \binp{\dual{\suc}}{\tilde{y}} (R^{y_1} \Par \bout{\suc'}{\tilde{y_2}} \inact)}
		{\cong}
		{\Delta_2''}{}{\newsp{\tilde{m_2}'}{Q \Par \bout{\suc}{\dual{n}, \tilde{m_2}''} \inact} \Par \binp{\dual{\suc}}{\tilde{y}} (R^{y_1} \Par \bout{\suc'}{\tilde{y_2}} \inact)}
	\]
	\noi Applying reduction closeness to the above pair we get:
	\[
		\horel{\Gamma}{\Delta_1''}{\newsp{\tilde{m_1}'}{P \Par R \Par \bout{\suc'}{\tilde{s_2}} \inact}}{\cong}{\Delta_2''}{\newsp{\tilde{m_2}'}{Q \Par R \Par \bout{\suc'}{\tilde{s_2}} \inact}}
	\]
	\noi The conclusion then follows.
	\qed
\end{proof}

\begin{lemma}\rm
	\label{lem:cong_is_wb}
	$\cong \subseteq \wb$.
\end{lemma}

\begin{proof}
	\noi Let
	\[
		\horel{\Gamma}{\Delta_1}{P_1}{\cong}{\Delta_2}{P_2}
	\]
	\noi We distinguish two cases:

%% Case tau
	\noi - Case:
	\[
		\horel{\Gamma}{\Delta_1}{P_1}{\by{\tau}}{\Delta_1'}{P_1'}
	\]
	\noi The result follows the reduction closeness property of $\cong$ since
	\[
		\horel{\Gamma}{\Delta_2}{P_2}{\By{\tau}}{\Delta_2'}{P_2'}
	\]
	\noi and
	\[
		\horel{\Gamma}{\Delta_1'}{P_1'}{\cong}{\Delta_2'}{P_2'}
	\]

%% Case ell
	\noi - Case:
	\begin{eqnarray}
		\horel{\Gamma}{\Delta_1}{P_1}{\by{\ell}}{\Delta_1'}{P_1'}
		\label{lem:cong_is_wb1}
	\end{eqnarray}
	\noi We choose test $T\lrangle{\ell, \suc}$ to get
	\begin{eqnarray}
		\horel{\Gamma}{\Delta_1 \cat \Delta_3}{P_1 \Par T\lrangle{\ell, \suc}}{\cong}{\Delta_2 \cat \Delta_3}{P_2 \Par T\lrangle{\ell, \suc}}
		\label{lem:cong_is_wb2}
	\end{eqnarray}
	\noi From this point we distinguish three subcases:

%% Subcase i
	\noi Subcase i: $\ell \in \set{\bactinp{n}{\tilde{m}}, \bactinp{n}{\abs{\tilde{x}}{Q}}, \bactsel{n}{l}, \bactbra{n}{l}}$

	\noi By reducing~(\ref{lem:cong_is_wb1}), we obtain
	\begin{eqnarray*}
		&& P_1 \Par T\lrangle{\ell, \suc} \red P_1' \Par \bout{\suc}{\dual{n}} \inact \\
		&& \Gamma; \es; \Delta_1' \cat \Delta_3' \proves P_1' \Par \bout{\suc}{\dual{n}} \inact \barb{\suc}
	\end{eqnarray*}
	\noi implies from~(\ref{lem:cong_is_wb2})
	\begin{eqnarray*}
		&& \Gamma; \es; \Delta_2 \cat \Delta_3 \proves P_2 \Par T\lrangle{\ell, \suc} \Barb{\suc}
	\end{eqnarray*}
	\noi implies from Lemma~\ref{lem:definibility},
	\begin{eqnarray*}
		&& \horel{\Gamma}{\Delta_2}{P_2}{\By{\ell}}{\Delta_2'}{P_2'}\\
		&& P_2 \Par T \lrangle{\ell, \suc} \Red P_2' \Par \bout{\suc}{\dual{n}} \inact
	\end{eqnarray*}
	\noi and
	\[
		\horel{\Gamma}{\Delta_1' \cat \Delta_3'}{P_1' \Par \bout{\suc}{\dual{n}}}{\cong}{\Delta_2' \cat \Delta_3'}{P_2' \Par \bout{\suc}{\dual{n}} \inact}
	\]
	We then apply \lemref{lem:extrusion} to get
	\[
		\horel{\Gamma}{\Delta_1'}{P_1'}{\cong}{\Delta_2'}{P_2'}
	\]
	\noi as required.

%% Subcase ii
	\noi Subcase ii: $\ell = \news{\tilde{m_1}} \bactout{n}{\abs{\tilde{x}}{Q_1}}$

	\noi Note that $T\lrangle{\news{\tilde{m_1}} \bactout{n}{(\tilde{x}) Q_1}, \suc} = T\lrangle{\news{\tilde{m_2}} \bactout{n}{\abs{\tilde{x}}{Q_2}}, \suc}$

	\noi Transition~in (\ref{lem:cong_is_wb1}) becomes
	\begin{eqnarray}
		\horel{\Gamma}{\Delta_1}{P_1}{\by{\news{\tilde{m_1}} \bactout{n}{\abs{\tilde{x}}{Q_1}}}}{\Delta_1'}{P_1'}
		\label{lem:cong_is_wb3}
	\end{eqnarray}
	\noi If we use the test process $T\lrangle{\news{\tilde{m_1}} \bactout{n}{(\tilde{x}) Q_1}, \suc}$ we reduce to:%~\ref{lem:cong_is_wb1} we get
	\begin{eqnarray*}
		&& P_1 \Par T\lrangle{\news{\tilde{m_1}} \bactout{n}{\abs{\tilde{x}}{Q_1}}, \suc}
		\red
		\newsp{m_1}{P_1' \Par \hotrigger{t}{x}{s}{\abs{\tilde{x}}{Q_1}}} \Par \bout{\suc}{\dual{n}, \tilde{m_1}'} \inact \\
		&& \Gamma; \es; \Delta_1' \cat \Delta_3' \proves \newsp{m_1}{P_1' \Par \hotrigger{t}{x}{s}{\abs{\tilde{x}}{Q_1}}} \Par \bout{\suc}{\dual{n}, \tilde{m_1}'} \inact \barb{\suc}
	\end{eqnarray*}
	\noi implies from~(\ref{lem:cong_is_wb2})
	\[
		\Gamma; \es; \Delta_2 \cat \Delta_3 \proves P_2 \Par T\lrangle{\news{\tilde{m_2}} \bactout{n}{\abs{\tilde{x}}{Q_2}}, \suc} \Barb{\suc}
	\]
	\noi implies from \lemref{lem:definibility}
	\begin{eqnarray}
		&& \horel{\Gamma}{\Delta_2}{P_2}{\By{\news{\tilde{m_2}} \bactout{n}{\abs{\tilde{x}}{Q_2}}}}{\Delta_2'}{P_2'}
		\label{lem:cong_is_wb4}\\
		&& P_2 \Par T \lrangle{\ell, \suc} \Red \newsp{m_2}{P_2' \Par \hotrigger{t}{x}{s}{\abs{\tilde{x}}{Q_2}}} \Par \bout{\suc}{\dual{n}, \tilde{m_2}'} \inact \nonumber
	\end{eqnarray}
	\noi and
	\[
		\mhorel{\Gamma}{\Delta_1' \cat \Delta_3'}{\newsp{m_1}{P_1' \Par \hotrigger{t}{x}{s}{\abs{\tilde{x}}{Q_1}}} \Par \bout{\suc}{\dual{n}, \tilde{m_1}'} \inact}
		{\cong}
		{\Delta_2' \cat \Delta_3'}{}{\newsp{m_2}{P_2' \Par \hotrigger{t}{x}{s}{\abs{\tilde{x}}{Q_2}}} \Par \bout{\suc}{\dual{n}, \tilde{m_2}'} \inact}
	\]
	\noi We then apply \lemref{lem:extrusion} to get
	\[
		\mhorel{\Gamma}{\Delta_1'}{\newsp{m_1}{P_1' \Par \hotrigger{t}{x}{s}{\abs{\tilde{x}}{Q_1}}}}
		{\cong}
		{\Delta_2'}{}{\newsp{m_2}{P_2' \Par \hotrigger{t}{x}{s}{\abs{\tilde{x}}{Q_2}}}}
	\]
	\noi as required.

	\noi -Case: $\ell = \news{\tilde{s}} \bactout{n}{\tilde{m}}$

	\noi Follows similar arguments as the previous case.
	\qed
\end{proof}

%%%%%%%%%%%%%%%%%%%%%%%%%%%%%%%%%%%%%%%%%%%%%%%%%%%%%%%%%%%%%%
% Proof of the main theorem
%%%%%%%%%%%%%%%%%%%%%%%%%%%%%%%%%%%%%%%%%%%%%%%%%%%%%%%%%%%%%%

\begin{theorem}[Concidence]
	\begin{enumerate}
		\item	$\wbc\ =\ \wb$.
		\item	$\wbc\ =\ \cong$.
	\end{enumerate}
\end{theorem}

\begin{proof}
	\noi	\lemref{lem:wb_eq_wbf} proves $\wb\ =\ \wbf$.
		\lemref{lem:cong_is_wb} proves $\cong\ \subseteq\ \wb$.
		\lemref{lem:wb_is_wbc} proves $\wb\ \subseteq\ \wbc$.
		\lemref{lem:wbc_is_cong} proves $\wbc\ \subseteq\ \cong$.

	\noi From the above results, we conclude $\cong\ \subseteq\ \wb\ =\ \wbf\ \subseteq\ \wbc\ \subseteq\ \cong$. 
	\qed
\end{proof}

%%%%%%%%%%%%%%%%%%%%%%%%%%%%%%%%%%%%%%%%%%%%%%%%%%%%%%%%%%%%%%
% tau - Innertness
%%%%%%%%%%%%%%%%%%%%%%%%%%%%%%%%%%%%%%%%%%%%%%%%%%%%%%%%%%%%%%

\subsection{$\tau$-inertness}
\label{app:sub_tau_inert}

We prove Part 1 of \propref{lem:tau_inert}.

\begin{proposition}[$\tau$-inertness]\rm
	Let balanced \HOp process $\Gamma; \es; \Delta \proves P \hastype \Proc$.
	$\horel{\Gamma}{\Delta}{P}{\hby{\dtau}}{\Delta'}{P'}$ implies
	$\horel{\Gamma}{\Delta}{P}{\wb}{\Delta'}{P'}$.
\end{proposition}

\begin{proof}
	\noi The proof is done by induction on the structure of $\by{\tau}$
	which coincides the reduction $\red$.

	\noi Basic step:

	\noi - Case: $P = \appl{(\abs{x}{P})}{n}$:
	\[
		\horel{\Gamma}{\Delta}{\appl{(\abs{x}{P})}{n}}{\hby{\btau}}{\Delta'}{P \subst{n}{x}}
	\]
	\noi Bisimulation requirements hold since, there is no other transition to observe than ${\hby{\btau}}$.

	\noi - Case: $P = \bout{s}{V} P_1 \Par \binp{\dual{s}}{x} P_2$:
	\[
		\horel{\Gamma}{\Delta}{\bout{s}{V} P_1 \Par \binp{\dual{s}}{x} P_2}{\hby{\stau}}{\Delta'}{P_1 \Par P_2}
	\]
	\noi The proof follows from the fact that we can only observe a $\tau$
	action on typed process
	$\Gamma; \emptyset; \Delta \proves P \hastype \Proc$.
	Actions $\bactout{s}{V}$ and $\bactinp{\dual{s}}{V}$
	are forbiden by the LTS for typed environments.

	\noi It is easy to conclude then that $\horel{\Gamma}{\Delta}{P}{\wb}{\Delta'}{P'}$.

	\noi - Case: $P = \bsel{s}{l} P_1 \Par \bbra{\dual{s}}{l_i: P_i}_{i \in I}$

	\noi Similar arguments as the previous case.

	\noi Induction hypothesis:

	\noi If $P_1 \red P_2$ then $\horel{\Gamma_1}{\Delta_1}{P_1}{\wb}{\Delta_2}{P_2}$.

	\noi Induction Step:

	\noi - Case: $P = \news{s} P_1$
	\[
		\horel{\Gamma}{\Delta}{\news{s}{P_1}}{\hby{\stau}}{\Delta'}{\news{s} P_2}
	\]
	\noi From the induction hypothesis and the fact that bisimulation is a congruence
	we get that $\horel{\Gamma}{\Delta}{P}{\wb}{\Delta'}{P'}$.

	\noi - Case: $P = P_1 \Par P_3$
	\[
		\horel{\Gamma}{\Delta}{P_1 \Par P_3}{\hby{\stau}}{\Delta'}{P_2 \Par P_3}
	\]
	\noi From the induction hypothesis and the fact that bisimulation is a congruence
	we get that $\horel{\Gamma}{\Delta}{P}{\wb}{\Delta'}{P'}$.

	\noi - Case: $P \scong P_1$

	From the induction hypothesis and the fact that bisimulation is a congruence
	and structural congruence preserves $\wb$
	we get that $\horel{\Gamma}{\Delta}{P}{\wb}{\Delta'}{P'}$.

%	The proof for part two is an induction on the length of $\red^*$.
%	The basic step is trivial and the inductive step
%	deploys part 1 of this lemma and the fact that bisimulation is
%	transitive to conclude.
%	We can now conclude that
%	$P \wbc P'$ because $P \wbc P''$ and $P'' \wbc P'$.
	\qed
\end{proof}

% !TEX root = ../main.tex
\section{Expressiveness Results}

\subsection{Properties for $\enco{\pmapp{\cdot}{1}{f}, \tmap{\cdot}{1}, \mapa{\cdot}^{1}}: \HOp \to \HO$}
\label{app:enc_HOp_to_HO}

We repeat the statement of \propref{prop:typepres_HOp_to_HO}, 
as in Page \pageref{prop:typepres_HOp_to_HO}:

%% Type Preservation

\begin{proposition}[Type Preservation, \HOp into \HO]
	\label{app:prop:typepres_HOp_to_HO}
	Let $P$ be a \HOp process.
	If $\Gamma; \emptyset; \Delta \proves P \hastype \Proc$ then 
	$\mapt{\Gamma}^{1}; \emptyset; \mapt{\Delta}^{1} \proves \pmapp{P}{1}{f} \hastype \Proc$. 
\end{proposition}

\begin{proof}
	By induction on the   inference of $\Gamma; \emptyset; \Delta \proves P \hastype \Proc$. %\jp{TO BE ADJUSTED!}
	\begin{enumerate}[1.]
		%%%% Output of (linear) channel
		\item	Case $P = \bout{k}{n}P'$. There are two sub-cases.
			In the first sub-case $n = k'$ (output of a linear channel). Then  
			we have the following typing in the source language:
			{
			\[
				\tree{
					\Gamma; \emptyset; \Delta \cat k:S  \proves  P' \hastype \Proc \quad \Gamma ; \emptyset ; \{k' : S_1\} \proves  k' \hastype S_1}{
					\Gamma; \emptyset; \Delta \cat k':S_1 \cat k:\btout{S_1}S \proves  \bout{k}{k'} P' \hastype \Proc}
			\]
			}
			Thus, by IH we have
			$$
			\tmap{\Gamma}{1}; \emptyset ; \tmap{\Delta}{1} \cat k:\tmap{S}{1} \proves \pmap{P'}{1} \hastype \Proc
			$$
			Let us write $U_1$
			to stand for $\lhot{\btinp{\lhot{\tmap{S_1}{1}}}\tinact}$.
			The corresponding typing in the target language is as follows:
			\begin{eqnarray}
				\label{prop:sesspnr_to_HO_t1}
				\tree{
					\tree{
						\tree{
							\tree{
								\tmap{\Gamma}{1} ; \set{x : \lhot{\tmap{S_1}{1}}} ; \emptyset \proves x  \hastype \lhot{\tmap{S_1}{1}}
								\qquad 
								\tmap{\Gamma}{1} ; \emptyset ; \set{k' : \tmap{S_1}{1}} \proves  k' \hastype \tmap{S_1}{1}
							}{
								\tmap{\Gamma}{1} ; \set{x : \lhot{\tmap{S_1}{1}}} ; k' : \tmap{S_1}{1} \proves \appl{x}{k'} \hastype \Proc
							}
						}{
							\tmap{\Gamma}{1} ; \{x : \lhot{\tmap{S_1}{1}}\} ; k' : \tmap{S_1}{1} \cat z:\tinact \proves \appl{x}{k'} \hastype \Proc
						}
					}{
						\tmap{\Gamma}{1} ; \emptyset; k' : \tmap{S_1}{1} \cat z:\btinp{\lhot{\tmap{S_1}{1}}}\tinact \proves \binp{z}{x} (\appl{x}{k'}) \hastype \Proc
					}
				}{
					\tmap{\Gamma}{1} ; \emptyset; k' : \tmap{S_1}{1} \proves \abs{z}{\binp{z}{x} (\appl{x}{k'})} \hastype U_1
				}
			\end{eqnarray}
			\begin{eqnarray*}
				\tree{
					\tmap{\Gamma}{1}; \emptyset ; \tmap{\Delta}{1} \cat k:\tmap{S}{1} \proves \pmap{P'}{1} \hastype \Proc
					\qquad
					\tmap{\Gamma}{1} ; \emptyset; k' : \tmap{S_1}{1} \proves \abs{z}{\binp{z}{x} (\appl{x}{k'})} \hastype U_1 \ \eqref{prop:sesspnr_to_HO_t1}
				}{
					\tmap{\Gamma}{1}; \emptyset; \tmap{\Delta}{1} \cat k':\tmap{S_1}{1} \cat k:\btout{U_1}\tmap{S}{1} \proves  \bbout{k}{\abs{z}{\binp{z}{x} (\appl{x}{k'})}} \pmap{P'}{1} \hastype \Proc
				}
			\end{eqnarray*}
%
	
		%%%% Output of (shared) channel
			In the second sub-case, we have $n = a$ (output of a shared name). Then  
			we have the following typing in the source language:
			{
			\[
				\tree{
					\Gamma \cat a:\chtype{S_1}; \emptyset; \Delta \cat k:S  \proves
					P' \hastype \Proc \quad \Gamma \cat a:\chtype{S_1} ; \emptyset ; \emptyset \proves  a \hastype S_1
				}{
					\Gamma \cat a:\chtype{S_1} ; \emptyset; \Delta  \cat k:\bbtout{\chtype{S_1}}S \proves  \bout{k}{a} P' \hastype \Proc
				}
			\]
			}
			The typing in the target language is derived similarly as in the first sub-case. \\
	
		%%%% Input of (linear) channel 
		\item	Case $P = \binp{k}{x}Q$. We have two sub-cases, depending on the type of $x$.
			In the first case, $x$ stands for a linear channel.
			Then we have the following typing in the source language:
			{
			\[
				\tree{
					\Gamma; \emptyset; \Delta  \cat k:S \cat x:S_1 \proves   Q \hastype \Proc
				}{
					\Gamma; \emptyset; \Delta  \cat k:\btinp{S_1}S \proves  \binp{k}{x} Q \hastype \Proc
				}
			\]
			 }
			 Thus, by IH we have
			 $$
			 \tmap{\Gamma}{1}; \emptyset;  \tmap{\Delta}{1} \cat k:\tmap{S}{1}  \cat x:\tmap{S_1}{1} \proves  \pmap{Q}{1}   \hastype \Proc
			 $$
			 Let us write $U_1$ to stand for $\lhot{\btinp{\lhot{\tmap{S_1}{1}}}\tinact}$.
			 The corresponding typing in the target language is as follows:
			{\small
			\begin{eqnarray}
				\label{prop:sesspnr_to_HO_t2}
				\tree{
					\tmap{\Gamma}{1}; \{X: U_1\};   \emptyset \proves X \hastype U_1
					\qquad
					\tmap{\Gamma}{1}; \emptyset;   \cat s: \btinp{\lhot{\tmap{S_1}{1}}}\tinact \ \proves s \, \hastype  \btinp{\lhot{\tmap{S_1}{1}}} \tinact 
				}{
					\tmap{\Gamma}{1}; \{X: U_1\};   \cat s: \btinp{\lhot{\tmap{S_1}{1}}}\tinact \ \proves \appl{x}{s}  \hastype \Proc
				}
			\end{eqnarray}
			\begin{eqnarray}
				\label{prop:sesspnr_to_HO_t3}
				\tree{
					\tree{
						\tmap{\Gamma}{1}; \emptyset;  \emptyset \proves   \inact  \hastype \Proc
					}{
						\tmap{\Gamma}{1}; \emptyset;  \dual{s}: \tinact\proves   \inact  \hastype \Proc
					}
					\quad 
					\tree{
						\tmap{\Gamma}{1}; \emptyset;  \tmap{\Delta}{1} \cat k:\tmap{S}{1}  x:\tmap{S_1}{1} \proves \pmap{Q}{1}   \hastype \Proc
					}{
						\tmap{\Gamma}{1}; \emptyset;  \tmap{\Delta}{1} \cat k:\tmap{S}{1}   \proves \abs{x} \pmap{Q}{1}   \hastype \lhot{\tmap{S_1}{1}}
					}
				}{
					\tmap{\Gamma}{1}; \emptyset;  \tmap{\Delta}{1} \cat k:\tmap{S}{1}  \cat \dual{s}: \btout{\lhot{\tmap{S_1}{1}}}\tinact\proves  \bbout{\dual{s}}{\abs{x}{\pmap{Q}{1}}} \inact  \hastype \Proc
				}
			\end{eqnarray}
			\begin{eqnarray}
				\label{prop:sesspnr_to_HO_t4}
		 		\tree{
					\begin{array}{cl}
						\tmap{\Gamma}{1}; \{X: U_1\}; \cat s: \btinp{\lhot{\tmap{S_1}{1}}}\tinact \ \proves \appl{x}{s}  \hastype \Proc
						& \eqref{prop:sesspnr_to_HO_t2}
						\\
						\tmap{\Gamma}{1}; \emptyset; \tmap{\Delta}{1} \cat k:\tmap{S}{1} \cat \dual{s}: \btout{\lhot{\tmap{S_1}{1}}}\tinact \proves
						\bbout{\dual{s}}{\abs{x}{\pmap{Q}{1}}} \inact  \hastype \Proc
						& \eqref{prop:sesspnr_to_HO_t3}
					\end{array}
				}{
					\tmap{\Gamma}{1}; \{X: U_1\};  \tmap{\Delta}{1} \cat k:\tmap{S}{1} \cat s: \btinp{\lhot{\tmap{S_1}{1}}}\tinact \cat \dual{s}: \btout{\lhot{\tmap{S_1}{1}}}\tinact\proves \appl{x}{s} \Par \bbout{\dual{s}}{\abs{x}{\pmap{Q}{1}}} \inact  \hastype \Proc
			}
			\end{eqnarray}
			\begin{eqnarray*}
			\\
			 \tree{
				 \tree{
					\tmap{\Gamma}{1}; \{X: U_1\};  \tmap{\Delta}{1} \cat k:\tmap{S}{1} \cat s: \btinp{\lhot{\tmap{S_1}{1}}}\tinact \cat \dual{s}: \btout{\lhot{\tmap{S_1}{1}}}\tinact\proves \appl{x}{s} \Par \bbout{\dual{s}}{\abs{x}{\pmap{Q}{1}}} \inact  \hastype \Proc \quad \eqref{prop:sesspnr_to_HO_t4}
				}{
					\tmap{\Gamma}{1}; \{X: U_1\};  \tmap{\Delta}{1} \cat k:\tmap{S}{1} \proves \newsp{s}{\appl{x}{s} \Par \bbout{\dual{s}}{\abs{x}{\pmap{Q}{1}}} \inact}  \hastype \Proc
				}
			}{
				\tmap{\Gamma}{1}; \emptyset; \tmap{\Delta}{1}  \cat k:\btinp{U_1}\tmap{S}{1} \proves  \binp{k}{x} \newsp{s}{\appl{x}{s} \Par \bbout{\dual{s}}{\abs{x}{\pmap{Q}{1}}} \inact}  \hastype \Proc
			}
			\end{eqnarray*}
			 }
			 
			 In the second sub-case, $x$ stands for a shared name. Then we have the following typing in the source language:
			\[
			 \tree{
				\Gamma \cat x:\chtype{S_1} ; \emptyset; \Delta  \cat k:S \proves   Q \hastype \Proc
			 }{
				\Gamma ; \emptyset; \Delta  \cat k:\btinp{\chtype{S_1}}S \proves  \binp{k}{x} Q \hastype \Proc}
			 \]
			 The typing in the target language is derived similarly as in the first sub-case.	

		\item	Case $P_0 = \varp{X}$.
			Then we have the following typing in the source language:
			\[
				\Gamma \cat \varp{X}: \Delta ;\, \es ;\, \es \proves \varp{X} \hastype \Proc
			\]
			Then the typing of $\pmapp{\varp{X}}{1}{f}$ is as follows,
			assuming $f(\varp{X}) = \tilde{n}$ and $\tilde{x} = \vmap{\tilde{n}}$.
			Also, we write $\Delta_{\tilde{n}}$ 
			and $\Delta_{\tilde{x}}$ 
			to stand for 
			$n_1: S_1, \ldots, n_m: S_m$ and
			$x_1: S_1, \ldots, x_m: S_m$, respectively. 
			Below, we assume that $\Gamma = \Gamma' \cat X:\shot{\tilde{T}}$, 
			where  
			%$$\tilde{T} =  \trec{t}{\big(\tilde{S}, \btinp{\vart{t}}\tinact\big)}$$.
			\[
				\tilde{T} = \big(\tilde{S}, S^*\big) \qquad \quad
				S^* = \bbtinp{A}\tinact \qquad \quad
				A = \trec{t}{(\tilde{S}, \btinp{\vart{t}}\tinact)}
			\]
			\begin{eqnarray}
				\label{prop:sessp_to_HO_t1}
				\tree{
					\tree{
					}{
						\Gamma ;\, \es ;\, \es \proves z_X \hastype \shot{\tilde{T}}
					}
					\quad 
					\begin{array}{c}
						\Gamma ;\, \es ;\, \{n_i: S_i \} \proves n_i \hastype S_i \\
						\Gamma ;\, \es ;\, \{s: S^* \} \proves s\hastype S^*  \\
					\end{array}
				}{
					\Gamma  ;\, \es ;\, \Delta_{\tilde{n}}, s:\btinp{\shot{\tilde{T}}}\tinact
					\proves  
					\appl{z_X}{(\tilde{n}, s)} \hastype \Proc
				} 
			\end{eqnarray}
			\begin{eqnarray}
				\label{prop:sessp_to_HO_t2}
				\tree{
					\tree{
						\Gamma  ;\, \es ;\,   \es \proves \inact \hastype \Proc
					}{
						\Gamma  ;\, \es ;\,   \dual{s}: \tinact \proves \inact \hastype \Proc
					} 
					\quad
					\tree{
						\tree{
							\begin{array}{c}
								\Gamma ;\, \es ;\, \{x_i: S_i \} \proves x_i \hastype S_i \\
								\Gamma ;\, \es ;\, \{z: S^*  \} \proves z\hastype S^*  \\
								\Gamma ;\, \es ;\, \es \proves z_X \hastype \shot{\tilde{T}}  \\
							\end{array}
						}{
							\Gamma  ;\, \es ;\,   \Delta_{\tilde{x}}, \, z:S^*
							\proves 
							 {\appl{z_X}{( \tilde{x}, z)}} \hastype \Proc
						}
					}{
						\Gamma  ;\, \es ;\,   \es
						\proves 
						 \abs{(\tilde{x},z)}\,\,{\appl{z_X}{( \tilde{x}, z)}} \hastype \shot{\tilde{T}}
					} 	
				}{
					\Gamma  ;\, \es ;\,   \dual{s}: \btout{\shot{\tilde{T}}}\tinact
					\proves 
					\bbout{\dual{s}}{ \abs{(\tilde{x},z)}\,\,{\appl{z_X}{ (\tilde{x}, z)}}} \inact \hastype \Proc
				}
			\end{eqnarray}
			\[
			\tree{
				\tree{
					\begin{array}{cc}
						\Gamma  ;\, \es ;\, \Delta_{\tilde{n}}, s:\btinp{\shot{\tilde{T}}}\tinact
						\proves  
						\appl{z_X}{(\tilde{n}, s)} \hastype \Proc
						& \eqref{prop:sessp_to_HO_t1}
						\\ 
						\Gamma  ;\, \es ;\,   \dual{s}: \btout{\shot{\tilde{T}}}\tinact
						\proves 
						\bbout{\dual{s}}{ \abs{(\tilde{x},z)}\,\,{\appl{z_X}{ (\tilde{x}, z)}}} \inact \hastype \Proc
						& \eqref{prop:sessp_to_HO_t2}
					\end{array}
				}{
					\Gamma  ;\, \es ;\, \Delta_{\tilde{n}}, s:\btinp{\shot{\tilde{T}}}\tinact, \, \dual{s}: \btout{\shot{\tilde{T}}}\tinact
					\proves 
					\appl{z_X}{(\tilde{n}, s)} \Par 
					\bbout{\dual{s}}{ \abs{(\tilde{x},z)}\,\,{\appl{x}{ (\tilde{x}, z)}}} \inact \hastype \Proc
				}
			}{
				\Gamma  ;\, \es ;\, \Delta_{\tilde{n}}
				\proves 
				\newsp{s}{\appl{z_X}{(\tilde{n}, s)} \Par \bbout{\dual{s}}{ \abs{(\tilde{x},z)}\,\,{\appl{z_X}{ (\tilde{x}, z}})} \inact} \hastype \Proc
			}
			\]
		\item	Case $P_0 = \recp{X}{P}$. Then we have the following typing in the source language:
			\[
				\tree{
					\Gamma \cat \varp{X}:\Delta ;\, \es ;\,  \Delta \proves P \hastype \Proc
				}{
					\Gamma  ;\, \es ;\,  \Delta \proves \recp{X}{P} \hastype \Proc
				}
			\]
			Then we have the following typing in the target language ---we write $R$
			to stand for $\pmapp{P}{1}{{f,\{\varp{X}\to \tilde{n}\}} }$
			and $\tilde{x}$ to stand for $\vmap{\ofn{P}}$.
			\begin{eqnarray}
				\label{prop:sessp_to_HO_t4}
				\tree{
					\tree{
						\tmap{\Gamma}{1}\cat z_X:\shot{\tilde{T}};\, \es;\, \tmap{\Delta_{\tilde{n}}}{1}
						\proves
						 R  \hastype \Proc
					}{
						\tmap{\Gamma}{1}\cat z_X:\shot{\tilde{T}};\, \es;\, \tmap{\Delta_{\tilde{n}}}{1}, s:\tinact 
						\proves
						 R  \hastype \Proc
					}
				}{
					\tmap{\Gamma}{1};\, \es;\, \tmap{\Delta_{\tilde{n}}}{1}, s:\btinp{\shot{\tilde{T}}}\tinact 
					\proves
					\binp{s}{z_X} R  \hastype \Proc
				}
			\end{eqnarray}
			\begin{eqnarray}
				\label{prop:sessp_to_HO_t5}
				\tree{
					\tree{
						\tmap{\Gamma}{1};\, \es;\, \es
						\proves
						\inact \hastype \Proc
					}{
						\tmap{\Gamma}{1};\, \es;\, \dual{s}:\tinact
						\proves
						\inact \hastype \Proc
					} 
					\quad 
					\tree{
						\tree{
							\tree{
								\tmap{\Gamma}{1} \cat z_X: \shot{\tilde{T}};\, \es;\, \tmap{\Delta_{\tilde{x}}}{1}
								\proves
								{{\auxmap{R}{\es}}}  \hastype \Proc
							}{
								\tmap{\Gamma}{1} \cat z_X: \shot{\tilde{T}};\, \es;\, \tmap{\Delta_{\tilde{x}}}{1},
								y: \tinact
								\proves
								{{\auxmap{R}{\es}}}  \hastype \Proc
							}
						}{
							\tmap{\Gamma}{1};\, \es;\, \tmap{\Delta_{\tilde{x}}}{1}, \, y: \btinp{A}\tinact
							\proves
							{{\binp{y}{z_X} \auxmap{R}{\es}}}  \hastype \Proc
						}
					}{
						\tmap{\Gamma}{1};\, \es;\, \es
						\proves
						{\abs{(\tilde{x}, y) } \,{\binp{y}{z_X} \auxmap{R}{\es}}}  \hastype \shot{\tilde{T}}
					}
				}{
					\tmap{\Gamma}{1};\, \es;\, \dual{s}:\btout{\shot{\tilde{T}}}\tinact
					\proves
					\bbout{\dual{s}}{\abs{(\tilde{x}, y) } \,{\binp{y}{z_X} \auxmap{R}{\es}}} \inact \hastype \Proc
				}
			\end{eqnarray}
			\[
			\tree{
				\tree{
					\begin{array}{cc}
						\tmap{\Gamma}{1};\, \es;\, \tmap{\Delta_{\tilde{n}}}{1}, s:\btinp{\shot{\tilde{T}}}\tinact 
						\proves
						\binp{s}{z_X} R  \hastype \Proc
						& \eqref{prop:sessp_to_HO_t4}
						\\
						\tmap{\Gamma}{1};\, \es;\, \dual{s}:\btout{\shot{\tilde{T}}}\tinact
						\proves
						\bbout{\dual{s}}{\abs{(\tilde{x}, y) } \,{\binp{y}{z_X} \auxmap{R}{\es}}} \inact \hastype \Proc
						& \eqref{prop:sessp_to_HO_t5}
					\end{array}
				}{
					\tmap{\Gamma}{1};\, \es;\, \tmap{\Delta_{\tilde{n}}}{1}, s:\btinp{\shot{\tilde{T}}}\tinact , \dual{s}:\btout{\shot{\tilde{T}}}\tinact
					\proves
					\binp{s}{z_X} R \Par 
					\bbout{\dual{s}}{\abs{(\tilde{x}, y )} \,{\binp{y}{z_X} \auxmap{R}{\es}}} \inact \hastype \Proc
				}
			}{
				\tmap{\Gamma}{1};\, \es;\, \tmap{\Delta_{\tilde{n}}}{1} 
				\proves
				\newsp{s}{\binp{s}{z_X} R \Par 
				\bbout{\dual{s}}{\abs{(\tilde{x},y) } \,{\binp{y}{z_X} \auxmap{R}{\es}}} \inact} \hastype \Proc
			}
			\]
	\end{enumerate}
	\qed
\end{proof}

%%% Operational Correspondence

We repeat the statement of
\propref{prop:op_corr_HOp_to_HO}, 
as in Page \pageref{prop:op_corr_HOp_to_HO}:

\begin{proposition}[Operational Correspondence, \HOp into \HO]\rm
	\label{app:prop:op_corr_HOp_to_HO}
	Let $P$ be a \HOp process.
	If $\Gamma; \emptyset; \Delta \proves P \hastype \Proc$ then:
	\begin{enumerate}[1.]
		\item
			Suppose $\horel{\Gamma}{\Delta}{P}{\hby{\ell_1}}{\Delta'}{P'}$. Then we have:
			\begin{enumerate}[a)]
				\item
					If $\ell_1 \in \set{\news{\tilde{m}}\bactout{n}{m}, \,\news{\tilde{m}}\bactout{n}{\abs{x}Q}, \,\bactsel{s}{l}, \,\bactbra{s}{l}}$
					then $\exists \ell_2$ s.t. \\
					$\horel{\tmap{\Gamma}{1}}{\tmap{\Delta}{1}}{\pmapp{P}{1}{f}}{\hby{\ell_2}}{\tmap{\Delta'}{1}}{\pmapp{P'}{1}{f}}$
					and $\ell_2 = \mapa{\ell_1}^{1}$.
			
				\item
					If $\ell_1 = \bactinp{n}{\abs{y}Q}$ and
					$P' = P_0 \subst{\abs{y}Q}{x}$
					then $\exists \ell_2$ s.t. \\
					$\horel{\tmap{\Gamma}{1}}{\tmap{\Delta}{1}}{\pmapp{P}{1}{f}}{\hby{\ell_2}}{\tmap{\Delta'}{1}}{\pmapp{P_0}{1}{f}\subst{\abs{y}\pmapp{Q}{1}{\emptyset}}{x}}$
					and $\ell_2 = \mapa{\ell_1}^{1}$.
			
				\item
					If $\ell_1 = \bactinp{n}{m}$
					and 
					$P' = P_0 \subst{m}{x}$
					then $\exists \ell_2$, $R$ s.t. \\
					$\horel{\tmap{\Gamma}{1}}{\tmap{\Delta}{1}}{\pmapp{P}{1}{f}}{\hby{\ell_2}}{\tmap{\Delta'}{1}}{R}$,
					with $\ell_2 = \mapa{\ell_1}^{1}$, \\
					and
					$\horel{\tmap{\Gamma}{1}}{\tmap{\Delta'}{1}}{R}{\hby{\btau} \hby{\stau} \hby{\btau}}
					{\tmap{\Delta'}{1}}{\pmapp{P_0}{1}{f}\subst{m}{x}}$.
						
				\item
					If $\ell_1 = \tau$
					and $P' \scong \newsp{\tilde{m}}{P_1 \Par P_2\subst{m}{x}}$
					then $\exists R$ s.t. \\
					$\horel{\tmap{\Gamma}{1}}{\tmap{\Delta}{1}}{\pmapp{P}{1}{f}}{\hby{\tau}}{\mapt{\Delta}^{1}}{\newsp{\tilde{m}}{\pmapp{P_1}{1}{f} \Par R}}$,
					and\\ 
					$\horel{\tmap{\Gamma}{1}}{\tmap{\Delta}{1}}{\newsp{\tilde{m}}{\pmapp{P_1}{1}{f} \Par R}}{\hby{\btau} \hby{\stau} \hby{\btau}}
					{\mapt{\Delta}^{1}}{\newsp{\tilde{m}}{\pmapp{P_1}{1}{f} \Par \pmapp{P_2}{1}{f}\subst{m}{x}}}$.
			
				\item
					If $\ell_1 = \tau$
					and $P' \scong \newsp{\tilde{m}}{P_1 \Par P_2 \subst{\abs{y}Q}{x}}$
					then \\
					$\horel{\tmap{\Gamma}{1}}{\tmap{\Delta}{1}}{\pmapp{P}{1}{f}}{\hby{\tau}}
					{\tmap{\Delta_1}{1}}{\newsp{\tilde{m}}{\pmapp{P_1}{1}{f}\Par \pmapp{P_2}{1}{f}\subst{\abs{y}\pmapp{Q}{1}{\emptyset}}{x}}}$.
			
				\item
					If $\ell_1 = \tau$
					and $P' \not\scong \newsp{\tilde{m}}{P_1 \Par P_2 \subst{m}{x}} \land P' \not\scong \newsp{\tilde{m}}{P_1 \Par P_2\subst{\abs{y}Q}{x}}$
					then \\
					$\horel{\tmap{\Gamma}{1}}{\tmap{\Delta}{1}}{\pmapp{P}{1}{f}}{\hby{\tau}}{\tmap{\Delta'_1}{1}}{ \pmapp{P'}{1}{f}}$.
			\end{enumerate}
			
		\item	Suppose $\horel{\tmap{\Gamma}{1}}{\tmap{\Delta}{1}}{\pmapp{P}{1}{f}}{\hby{\ell_2}}{\tmap{\Delta'}{1}}{Q}$.
			Then we have:
			\begin{enumerate}[a)]
				\item 
					If $\ell_2 \in
					\set{\news{\tilde{m}}\bactout{n}{\abs{z}{\,\binp{z}{x} (\appl{x}{m})}}, \,\news{\tilde{m}} \bactout{n}{\abs{x}{R}}, \,\bactsel{s}{l}, \,\bactbra{s}{l}}$
					then $\exists \ell_1, P'$ s.t. \\
					$\horel{\Gamma}{\Delta}{P}{\hby{\ell_1}}{\Delta'}{P'}$, 
					$\ell_1 = \mapa{\ell_2}^{1}$, 
					and
					$Q = \pmapp{P'}{1}{f}$.
			
				\item 
					If $\ell_2 = \bactinp{n}{\abs{y} R}$ %(with $R \neq \binp{y}{x} \appl{x}{m}$)
					then either:
					\begin{enumerate}[(i)]
						\item	$\exists \ell_1, x, P', P''$ s.t. \\
							$\horel{\Gamma}{\Delta}{P}{\hby{\ell_1}}{\Delta'}{P' \subst{\abs{y}P''}{x}}$, 
							$\ell_1 = \mapa{\ell_2}^{1}$, $\pmapp{P''}{1}{\es} = R$, and $Q = \pmapp{P'}{1}{f}$.

						\item	$R \scong \binp{y}{x} (\appl{x}{m})$ and 
							$\exists \ell_1, z, P'$ s.t. \\
							$\horel{\Gamma}{\Delta}{P}{\hby{\ell_1}}{\Delta'}{P' \subst{m}{z}}$, 
							$\ell_1 = \mapa{\ell_2}^{1}$,
							and\\
							$\horel{\tmap{\Gamma}{1}}{\tmap{\Delta'}{1}}{Q}{\hby{\btau} \hby{\stau} \hby{\btau}}{\tmap{\Delta''}{1}}{\pmapp{P'\subst{m}{z}}{1}{f}}$
					\end{enumerate}
			
				\item 
					If $\ell_2 = \tau$ 
					then $\Delta' = \Delta$ and 
					either
					\begin{enumerate}[(i)]
						\item	$\exists P'$ s.t. 
							$\horel{\Gamma}{\Delta}{P}{\hby{\tau}}{\Delta}{P'}$,
							and $Q = \map{P'}^{1}_f$.	

						\item
							$\exists P_1, P_2, x, m, Q'$ s.t. 
							$\horel{\Gamma}{\Delta}{P}{\hby{\tau}}{\Delta}{\newsp{\tilde{m}}{P_1 \Par P_2\subst{m}{x}} }$, and\\
							$\horel{\tmap{\Gamma}{1}}{\tmap{\Delta}{1}}{Q}{\hby{\btau} \hby{\stau} \hby{\btau}}{\tmap{\Delta}{1}}{\pmapp{P_1}{1}{f} \Par \pmapp{P_2\subst{m}{x}}{1}{f}}$ 
%							$Q = \map{P_1}^{1}_f \Par Q'$, where $Q'  \Hby{} $.

%						\item $\exists P_1, P_2, x, R$ s.t. 
%						$\stytra{ \Gamma }{\tau}{ \Delta }{ P}{ \Delta}{ \news{\tilde{m}}(P_1 \Par P_2\subst{\abs{y}R}{x}) }$, and 
%						$Q = \map{\news{\tilde{m}}(P_1 \Par P_2\subst{\abs{y}R}{x})}^{1}_f$.
			\end{enumerate}
		    \end{enumerate}
		    
%		\item   
%			If  $\wtytra{\mapt{\Gamma}^{1}}{\ell_2}{\mapt{\Delta}^{1}}{\pmapp{P}{1}{f}}{\mapt{\Delta'}^{1}}{Q}$
%			then $\exists \ell_1, P'$ s.t.  \\
%			(i)~$\stytra{\Gamma}{\ell_1}{\Delta}{P}{\Delta'}{P'}$,
%			(ii)~$\ell_2 = \mapa{\ell_1}^{1}$, 
%			(iii)~$\wbb{\mapt{\Gamma}^{1}}{\ell}{\mapt{\Delta'}^{1}}{\pmapp{P'}{1}{f}}{\mapt{\Delta'}^{1}}{Q}$.
	\end{enumerate}
\end{proposition}

\begin{proof}

By transition induction. We consider parts (1) and (2) separately:

\noi \textbf{Part (1) - Completeness}. We consider two representative cases, the rest is similar or simpler:
\begin{enumerate}[1.]
	%%  Output 
	\item	Subcase  (a): $P =\bout{s}{n} P'$ and $\ell_1 = \bactout{s}{n}$ (the case $\ell_1 = \news{n}\bactout{s}{n}$ is similar). By assumption, $P$ is well-typed. 
		We may have:
		\[
			\tree{
				\Gamma; \emptyset; \Delta_0 \cat s:S_1  \proves  P' \hastype \Proc \quad 
				\Gamma ; \emptyset ; \{n{:} S\}  \proves   n \hastype S }{
				\Gamma; \emptyset; \Delta_0 \cat n{:}S \cat s:\btout{S}S_1 \proves \bout{s}{n} P' \hastype \Proc}
		\]
		\noi for some $S, S_1, \Delta_0$.
		%such that $\Delta = \Delta_0 \cat k_1{:}T  \cat k:\btout{T}S$.
		We may then have the following transition:
		\[
			\stytra{\Gamma}{\ell_1}{\Delta_0 \cat n{:}S \cat s:\btout{S}S_1 }{\bout{s}{n} P'}{\Delta_0 \cat s{:}S_1 }{P'}
		\]
		\noi The encoding of the source judgment for $P$ is as follows:
		\[
			\mapt{\Gamma}^{1}; \emptyset; \mapt{\Delta_0 \cat n{:}S \cat s:\btout{S}S_1}^{1} \proves \map{\bout{s}{n} P'}^{1} \hastype \Proc
		\]
		\noi which, using \defref{def:enc:HOp_to_HO} can be expressed as 
		\[
			\mapt{\Gamma}^{\mathsf{p}}; \emptyset; \mapt{\Delta_0} 
			\cat n{:}\mapt{S}^{1} 
			\cat s: \btout{\lhot{\btinp{\lhot{\tmap{S}{1}}}\tinact}} \tmap{S_1}{1}
			\proves 
			\bbout{s}{ \abs{z}{\,\binp{z}{x} (\appl{x}{n})} } \pmap{P'}{1}
			\hastype \Proc
		\]
		\noi Now, $\mapa{\ell_1}^{1} = \bactout{s}{\abs{z}{\,\binp{z}{x} \appl{x}{n}}\, } $. 
		We may infer the following  transition for $\map{P}^{1}$:
		\begin{eqnarray*}
			& & \mapt{\Gamma}^{1}; \emptyset; \mapt{\Delta}^{1} 
			\proves 
			\bbout{s}{ \abs{z}{\,\binp{z}{x} (\appl{x}{n})} } \pmap{P'}{1}
			\hastype \Proc \\
			& \hby{\mapa{\ell_1}^{1}} & \mapt{\Gamma}^{1}; \emptyset; \mapt{\Delta_0}^{1} 
			\cat s:  \tmap{S_1}{1}
			\proves  \pmap{P'}{1}
			\hastype \Proc \\
			& = & \mapt{\Gamma}^{1}; \emptyset; \mapt{\Delta_0 \cat s:  S_1}^{1}
			\proves  \pmap{P'}{1}
			\hastype \Proc 
		\end{eqnarray*}
		\noi from which the thesis follows easily.

	\item	Subcase (c): $P = \binp{n}{x} P'$	and $\ell_1 = \bactinp{n}{m}$.
		By assumption $P$ is well-typed.
		We may have:
		\[
			\tree{
				\Gamma; \emptyset; \Delta_0 \cat x:S \cat n:S_1  \proves  P' \hastype \Proc \quad 
				\Gamma ; \emptyset ; \{x: S\}  \proves   x\hastype S}{
				\Gamma; \emptyset; \Delta_0 \cat   n:\btinp{S}S_1 \proves \binp{n}{x} P' \hastype \Proc}
		\]
		for some  $S, S_1, \Delta_0$.
%		such that $\Delta = \Delta_0 \cat k:\btinp{T}S$.
		We may infer the following typed transition:
		\[
			\Gamma; \emptyset; \Delta_0 \cat   n:\btinp{S}S_1 \proves \binp{n}{x} P' \hastype \Proc
			\hby{\bactinp{n}{m}}
			\Gamma; \emptyset; \Delta_0 \cat  n:S_1 \cat m:S \proves   P'\subst{m}{x} \hastype \Proc
		\]
		The encoding of the source judgment for $P$ is as follows:
		\begin{eqnarray*}
			& & \mapt{\Gamma}^{1}; \emptyset; \mapt{ \Delta_0 \cat   n:\btinp{S}S_1 }^{1} \proves 
			\map{P}^{1}
			\hastype \Proc \\
			& = & \mapt{\Gamma}^{1}; \emptyset; \mapt{ \Delta_0 }^{1} \cat   n: \btinp{\lhot{\btinp{\lhot{\tmap{S}{1}}}\tinact}} \tmap{S_1}{1} \proves 
			\binp{n}{x} \newsp{s}{(\appl{x}{s}) \Par \bbout{\dual{s}}{\abs{x}{\pmap{P'}{1}}} \inact}
			\hastype \Proc
		\end{eqnarray*}
		Now, 
		$\mapa{\ell_1}^{1} = \bactinp{n}{\abs{z}{\,\binp{z}{x} (\appl{x}{m})}\, }$
		and it is immediate to infer the following 
		transition for $\map{P}^{1}$:
		\begin{eqnarray*}
			&  & \mapt{\Gamma}^{1}; \emptyset; \mapt{ \Delta_0 }^{1} \cat   
			n: \btinp{\lhot{\btinp{\lhot{\tmap{S}{1}}}\tinact}} \tmap{S_1}{1} \proves 
			\binp{n}{x} \newsp{s}{(\appl{x}{s}) \Par \bbout{\dual{s}}{\abs{x}{\pmap{P'}{1}}} \inact}
			\hastype \Proc \\
			& \hby{\mapa{\ell_1}^{1}}  & \mapt{\Gamma}^{1}; \emptyset; \mapt{ \Delta_0 }^{1} \cat   
			n:  \tmap{S_1}{1} \cat m:  \tmap{S}{1} \proves 
			 \newsp{s}{(\appl{x}{s}) \Par \bbout{\dual{s}}{\abs{x}{\pmap{P'}{1}}} \inact}\subst{\abs{z}{\,\binp{z}{x} (\appl{x}{m})}}{x}
			\hastype \Proc 
		\end{eqnarray*}
		Let us write $R$ to stand for process 
		$\newsp{s}{(\appl{x}{s}) 
		\Par 
		\bbout{\dual{s}}{\abs{x}{\pmap{P'}{1}}} \inact}\subst{\abs{z}{\,\binp{z}{x} (\appl{x}{m})}}{x}$. 
		%$\newsp{s}{\appl{X}{s} \Par \bbout{\dual{s}}{\abs{x}{\pmap{Q}{1}}} \inact}\subst{\abs{z}{\,\binp{z}{X} \appl{X}{k_1}}}{X}$.
		We then have:
		\begin{eqnarray*}
		R & \by{\tau} & \newsp{s}{\binp{s}{x} (\appl{x}{m})
							\Par 
							\bbout{\dual{s}}{\abs{x}{\pmap{P'}{1}}} \inact} \\
		& \by{\tau} &  \appl{(\abs{x}{\pmap{P'}{1}})}{m} \Par \inact \\
		& \by{\tau} & \pmap{P'}{1}\subst{m}{x}
		\end{eqnarray*}
		and so the thesis follows.

		%%%%%%%%%%%
		%%  Recursion
		%%%%%%%%%%%

%	\item	Case $P =\recp{X}{P'}$ and $P = \varp{X}$.
%
%		It follows similar arguments with the previous cases
%		and uses Prop.~\ref{prop:op_corr_HOprec_to_HO} whenever necessary.
		
\end{enumerate}
\noi \textbf{Part (2) - Soundness}. We consider two representative cases, the rest is similar or simpler:
\begin{enumerate}[1.]
		%%%%%%%%%%%
		%%  Output 
		%%%%%%%%%%%
	\item Subcase (a): $P = \bout{n}{m} P'$ and $\ell_2 = \bactout{n}{\abs{z}{\,\binp{z}{x} (\appl{x}{m})}}$
	(the case $\ell_2 = \news{m}\bactout{n}{\abs{z}{\,\binp{z}{x} (\appl{x}{m})}}$ is similar).
		%,  $\map{P}^{1} = \bbout{k}{ \abs{z}{\,\binp{z}{X} \appl{X}{k'}} } \pmap{P'}{1}$.
		Then 
		we have: % the following typed transition for $\map{P}^{1}$:
		\[
			\mapt{\Gamma}^{1};\, \emptyset;\, \mapt{\Delta_0}^{1} \cat 
			n: \btout{\lhot{\btinp{\lhot{\tmap{S}{1}}}\tinact}} \tmap{S_1}{1} 
			\proves 
			 \bbout{n}{ \abs{z}{\,\binp{z}{x} (\appl{x}{m})} } \pmap{P'}{1} 
			 \hastype \Proc
		\]
		for some $S, S_1$, and $\Delta_0$. 
		We may infer the following typed transition for $\pmap{P}{1}$:
		\begin{eqnarray*}
			& & \mapt{\Gamma}^{1};\, \mapt{\Delta_0}^{1} \cat n: \btout{\lhot{\btinp{\lhot{\tmap{S}{1}}}\tinact}} \tmap{S_1}{1} 
			\proves 
			 \bbout{n}{ \abs{z}{\,\binp{z}{x} (\appl{x}{m})} } \pmap{P'}{1} 
			 \\
			%& & \bbout{k}{ \abs{z}{\,\binp{z}{X} \appl{X}{k'}} } \pmap{P'}{1} \hby{\bactout{k}{\abs{z}{\,\binp{z}{X} \appl{X}{k'}}}} \pmap{P'}{1}  \\
			&\hby{\ell_2}& 
			\mapt{\Gamma}^{1};\, \mapt{\Delta_0}^{1} \cat n: \tmap{S_1}{1} 
			\proves  \pmap{P'}{1} 
		\end{eqnarray*}
%
		%with $\ell_2 = \bactout{k}{\abs{z}{\,\binp{z}{X} \appl{X}{k'}}}$.
		Now, in the source term $P$ we can infer the following transition 
		\[
		\Gamma;\,  \Delta_0 \cat n:\btout{S} S_1 \proves \bout{n}{m} P'
		 \hby{\bactout{n}{m}} 
		 \Gamma;\,  \Delta_0 \cat n: S_1 \proves P'
		\]
		and thus the thesis follows easily by noticing that 
		$\mapa{\bactout{n}{m}}^{1} = \bactout{n}{\abs{z}{\,\binp{z}{x} (\appl{x}{m})}}$.

		%%%%%%%%%%%
		%% Input
		%%%%%%%%%%%
	\item	Subcase (c): $P = \binp{n}{x} P'$ and $\ell_2 = \bactinp{n}{\abs{y}\binp{y}{x} (\appl{x}{m})}$.
		Then we have
		\[
			\mapt{\Gamma}^{1};\, \emptyset;\, \mapt{\Delta_0}^{1} \cat 
			n: \btinp{\lhot{\btinp{\lhot{\tmap{S}{1}}}\tinact}} \tmap{S_1}{1}
			\proves
			\binp{n}{x} \newsp{s}{(\appl{x}{s})
			\Par 
			\bbout{\dual{s}}{\abs{x}{\pmap{P'}{1}}} \inact}
			\hastype \Proc
		\]
		for some $S$, $S_1$, $\Delta_0$.
		We may infer the following typed transitions for $\pmap{P}{1}$:
		\begin{eqnarray*}
			& & 
			\mapt{\Gamma}^{1};\, %\emptyset;\, 
			\mapt{\Delta_0}^{1} \cat 
			n: \btinp{\lhot{\btinp{\lhot{\tmap{S}{1}}}\tinact}} \tmap{S_1}{1}
			\proves
			\binp{n}{x} \newsp{s}{(\appl{x}{s}) 
							\Par 
							\bbout{\dual{s}}{\abs{x}{\pmap{P'}{1}}} \inact} \\
			& \hby{\ell_2} & 
			\mapt{\Gamma}^{1};\, %\emptyset;\, 
			\mapt{\Delta_0}^{1} \cat 
			n:\tmap{S_1}{1}
			\cat m:\tmap{S_1}{1}
			\proves
			\newsp{s}{(\appl{x}{s}) 
				\Par 
				\bbout{\dual{s}}{\abs{x}{\pmap{P'}{1}}} \inact} \subst{\abs{z}\binp{z}{x}\appl{x}{m}}{x} \\
			& = & 
			\mapt{\Gamma}^{1};\, %\emptyset;\, 
			\mapt{\Delta_0}^{1} 
			\cat n:\tmap{S_1}{1}
			\cat m:\tmap{S}{1}
			\proves
			\newsp{s}{\binp{s}{x}(\appl{x}{m}) 
				\Par 
				\bbout{\dual{s}}{\abs{x}{\pmap{P'}{1}}} \inact}  \\
			& \hby{\tau} & 
			\mapt{\Gamma}^{1};\, %\emptyset;\, 
			\mapt{\Delta_0}^{1} 
			\cat n:\tmap{S_1}{1}
			\cat m:\tmap{S}{1}
			\proves
			\appl{(\abs{x}{\pmap{P'}{1}})}{m}   \\
			& \hby{\tau} & 
			\mapt{\Gamma}^{1};\, %\emptyset;\, 
			\mapt{\Delta_0}^{1} 
			\cat n:\tmap{S_1}{1}
			\cat m:\tmap{S}{1}
			\proves
			\pmap{P'}{1}\subst{m}{x}   
		\end{eqnarray*}
%
		%with $\ell_2 = \bactinp{k}{\abs{z}{\,\binp{z}{X} \appl{X}{k_1}}}$.
		Now, in the source term $P$ we can infer the following transition 
		\[
			\Gamma;\,  \Delta_0 \cat n:\btinp{S} S_1 \proves \binp{n}{x} P'
			\hby{\bactinp{n}{m}} 
			\Gamma;\,  \Delta_0 \cat n: S_1 \cat m: S \proves P'\subst{m}{x}
		\]
		and the thesis follows.
%		 easily by noticing that $\mapa{\bactinp{k}{k_1}}^{1} = \bactinp{k}{\abs{z}{\,\binp{z}{X} \appl{X}{k_1}}}$.

		%%%%%%%%%%%
		%%  Recursion
		%%%%%%%%%%%
%	\item	Case $P =\recp{X}{P'}$ and $P = \varp{X}$.
%
%		It follows similar arguments with the previous case
%		and uses Prop.~\ref{prop:op_corr_HOprec_to_HO} whenever nescessary.
\end{enumerate}
\qed
\end{proof}

%%%%%%%% Full Abstraction

We repeat the statement of
\propref{prop:fulla_HOp_to_HO}, 
as in Page~\pageref{prop:fulla_HOp_to_HO}:

\begin{proposition}[Full Abstraction, \HOp into \HO]\rm
	\label{app:prop:fulla_HOp_to_HO}
	$\horel{\Gamma}{\Delta_1}{P_1}{\wb}{\Delta_2}{Q_1}$
	if and only if
	$\horel{\tmap{\Gamma}{1}}{\tmap{\Delta_1}{1}}{\pmapp{P_1}{1}{f}}{\wb}{\tmap{\Delta_2}{1}}{\pmapp{Q_2}{1}{f}}$.
\end{proposition}

\begin{proof}
	\noi {\bf Proof of Soundness Direction.}

	\noi Let
	\[
		\Re = \set{\horel{\Gamma}{\Delta_1}{P_1}{\wb}{\Delta_2}{Q_1} \setbar \horel{\mapt{\Gamma}^{1}}{\mapt{\Delta_1}^{1}}{\pmapp{P_1}{1}{f}}{\wb}{\mapt{\Delta_2}^{1}}{\pmapp{Q_1}{1}{f}}}
	\]
	\noi	The proof considers a case analysis on the transition $\hby{\ell}$ and
		uses the soundness direction of operational correspondence (cf.~\propref{prop:op_corr_HOp_to_HO}).
		We give an interesting case. The others are similar of easier.

	\noi	- Case: $\ell = \news{\tilde{m_1}'} \bactout{n}{m_1}$.

	\noi \propref{prop:op_corr_HOp_to_HO} implies that
	\[
		\horel{\Gamma}{\Delta_1}{P_1}{\hby{\news{\tilde{m_1}'} \bactout{n}{m_1}}}{\Delta_1'}{P_2}
	\]
	\noi implies
	\[
		\horel{\mapt{\Gamma}^{1}}{\mapt{\Delta_1}^{1}}{\pmapp{P_1}{1}{f}}{\hby{\news{\tilde{m_1}'} \bactout{n}{\abs{z}{\binp{z}{x} (\appl{x}{m_1})}}}}{\mapt{\Delta_1'}^{1}}{\pmapp{P_2}{1}{f}}
	\]
	\noi that in combination with the definition of $\Re$ we get
	\begin{eqnarray}
		\horel{\mapt{\Gamma}^{1}}{\mapt{\Delta_2}^{1}}{\pmapp{Q_1}{1}{f}}{\Hby{\news{\tilde{m_2}'} \bactout{n}{\abs{z}{\binp{z}{x} (\appl{x}{m_2})}}}}{\mapt{\Delta_2'}^{1}}{\pmapp{Q_2}{1}{f}}
		\label{prop:HOp_to_HO:full_abs11}
	\end{eqnarray}
	\noi and
	\[
		\mhorel{\mapt{\Gamma}^{1}}{\mapt{\Delta_1'}^{1}}{\newsp{\tilde{m_1}'}{\pmapp{P_2}{1}{f} 
			\Par 
			\hotrigger{t}{x}{s}{\abs{z}{\binp{z}{x} (\appl{x}{m_1})}} }}
		{\wb}{\mapt{\Delta_2'}^{1}}{}{\newsp{\tilde{m_2}'}{\pmapp{Q_2}{1}{f} 
			\Par 
				\hotrigger{t}{x}{s}{\abs{z}{\binp{z}{x} (\appl{x}{m_2})}}}}
	\]
	\noi We rewrite the last result as
	\[
		\mhorel{\mapt{\Gamma}^{1}}{\mapt{\Delta_1'}^{1}}{\pmapp{\newsp{\tilde{m_1}'}{P_2 
			\Par \hotrigger{t}{x}{s}{m_1}}}{1}{f}}
		{\wb}{\mapt{\Delta_2'}^{1}}{}{\pmapp{\newsp{\tilde{m_2}'}{Q_2 \Par \hotrigger{t}{x}{s}{m_2}}}{1}{f}}
	\]
	\noi to conclude that
	\[
		\mhorel{\Gamma}{\Delta_1'}{\newsp{\tilde{m_1}'}{P_2 \Par \hotrigger{t}{x}{s}{m_1}}}
		{\ \Re\ }{\Delta_2'}{}{\newsp{\tilde{m_2}'}{Q_2 \Par \hotrigger{t}{x}{s}{m_2}}}
	\]
	\noi as required

	\noi {\bf Proof of Completeness Direction.}

	\noi Let
	\[
		\Re = \set{\horel{\mapt{\Gamma}^{1}}{\mapt{\Delta_1}^{1}}{\pmapp{P_1}{1}{f}}{,}{\mapt{\Delta_2}^{1}}{\pmapp{Q_1}{1}{f}} \setbar \horel{\Gamma}{\Delta_1}{P_1}{\wb}{\Delta_2}{Q_1}}
	\]
	We show that $\Re \subset \wb$ by a case analysis on the action $\ell$

	\noi - Case: $\ell \notin \set{\news{\tilde{m}} \bactout{n}{\abs{x}{P}},\, \bactinp{n}{\abs{x}{P}}}$.

	\noi The proof of \propref{prop:op_corr_HOp_to_HO} implies that
	\[
		\horel{\mapt{\Gamma}^{1}}{\mapt{\Delta_1}^{1}}{\pmapp{P_1}{1}{f}}{\hby{\ell}}{\mapt{\Delta_1'}^{1}}{\pmapp{P_2}{1}{f}}
	\]
	\noi implies
	\[
		\horel{\Gamma}{\Delta_1}{P_1}{\hby{\ell}}{\Delta_1'}{P_2}
	\]
	\noi From the latter transition and the definition of $\Re$ we imply
	\begin{eqnarray}
		&&\horel{\Gamma}{\Delta_2}{Q_1}{\Hby{\ell}}{\Delta_2'}{Q_2}
		\label{prop:HOp_to_HO:full_abs1}
		\\
		&&\horel{\Gamma}{\Delta_1'}{P_2}{\wb}{\Delta_2'}{Q_2}
		\label{prop:HOp_to_HO:full_abs2}
	\end{eqnarray}
	\noi From~\ref{prop:HOp_to_HO:full_abs1} and \propref{prop:op_corr_HOp_to_HO} we get
	\[
		\horel{\mapt{\Gamma}^{1}}{\mapt{\Delta_2}^{1}}{\pmapp{Q_1}{1}{f}}{\Hby{\ell}}{\mapt{\Delta_2'}^{1}}{\pmapp{Q_2}{1}{f}}
	\]
	\noi Furthermore, from~\ref{prop:HOp_to_HO:full_abs2} and the definition of $\Re$ we get
	\[
		\horel{\mapt{\Gamma}^{1}}{\mapt{\Delta_1'}^{1}}{\pmapp{P_2}{1}{f}}{\ \Re\ }{\mapt{\Delta_2'}^{1}}{\pmapp{Q_2}{1}{f}}
	\]
	\noi as required.

	\noi - Case: $\ell = \news{\tilde{m}} \bactout{n}{\abs{x}{P}}$

	\noi There are two subcases:

	\noi -Subcase:

	\noi The proof of \propref{prop:op_corr_HOp_to_HO} implies that
	\[
		\horel{\mapt{\Gamma}^{1}}{\mapt{\Delta_1}^{1}}{\pmapp{P_1}{1}{f}}{\hby{\ell}}{\mapt{\Delta_1'}^{1}}{\pmapp{P_2}{1}{f}}
	\]
	\noi implies
	\[
		\horel{\Gamma}{\Delta_1}{P_1}{\hby{\ell}}{\Delta_1'}{P_2}
	\]
	\noi where the proof is similar with the previous case.

	\noi - Subcase:

	\noi The proof of \propref{prop:op_corr_HOp_to_HO} implies that
	\[
		\horel{\mapt{\Gamma}^{1}}{\mapt{\Delta_1}^{1}}{\pmapp{P_1}{1}{f}}{\hby{\news{\tilde{m_1}'} \bactout{n}{\abs{z}{\binp{z}{x} (\appl{x}{m_1})}}}}{\mapt{\Delta_1'}^{1}}{\pmapp{P_2}{1}{f}}
	\]
	\noi implies
	\[
		\horel{\Gamma}{\Delta_1}{P_1}{\hby{\news{\tilde{m_1}'} \bactout{n}{m_1}}}{\Delta_1'}{P_2}
	\]
	\noi From the latter transition and the definition of $\Re$ we imply
	\begin{eqnarray}
		&&\horel{\Gamma}{\Delta_2}{Q_1}{\Hby{\news{\tilde{m_2}'} \bactout{n}{m_2}}}{\Delta_2'}{Q_2}
		\label{prop:HOp_to_HO:full_abs3}
	\end{eqnarray}
	\noi and
	\begin{eqnarray}
		& \Gamma; \es; \Delta_1' & \proves \newsp{\tilde{m_1}'}{P_2 \Par \hotrigger{t}{x}{s}{m_1}} \nonumber \\
		& \wb & \Delta_2' \proves \newsp{\tilde{m_2}'}{Q_2 \Par \hotrigger{t}{x}{s}{m_2}}
		\label{prop:HOp_to_HO:full_abs4}
	\end{eqnarray}
	\noi From~\eqref{prop:HOp_to_HO:full_abs3} and \propref{prop:op_corr_HOp_to_HO} we get
	\[
		\horel{\mapt{\Gamma}^{1}}{\mapt{\Delta_2}^{1}}{\pmapp{Q_1}{1}{f}}{\Hby{\news{\tilde{m_2}'} \bactout{n}{\abs{z}{\binp{z}{x} (\appl{x}{m_2})}}}}{\mapt{\Delta_2'}^{1}}{\pmapp{Q_2}{1}{f}}
	\]
	\noi Furthermore, from~\eqref{prop:HOp_to_HO:full_abs4} and the definition of $\Re$ we get
	\[
		\mhorel{\mapt{\Gamma}^{1}}{\mapt{\Delta_1'}^{1}}{\pmapp{\newsp{\tilde{m_1}'}{P_2 \Par \hotrigger{t}{x}{s}{m_1}}}{1}{f}}
		{\ \Re\ }{\mapt{\Delta_2'}^{1}}{}{\pmapp{\newsp{\tilde{m_2}'}{Q_2 \Par \hotrigger{t}{x}{s}{m_2}}}{1}{f}}
	\]
	\noi as required.

	\noi - Case: $\ell = \bactinp{n}{\abs{x}{P}}$

	\noi We have two subcases.

	\noi - Subcase: Similar with the first subcase of the previous case.

	\noi - Subcase:
	\noi The proof of \propref{prop:op_corr_HOp_to_HO} implies that
	\[
		\horel{\mapt{\Gamma}^{1}}{\mapt{\Delta_1}^{1}}{\pmapp{P_1}{1}{f}}{\hby{\bactinp{n}{\abs{z}{ \binp{z}{x} (\appl{x}{s})}}}}{\mapt{\Delta_1''}^{1}} R %{\pmapp{P_2}{1}{f}}
	\]
	\noi implies
	\begin{eqnarray}
		\horel{\Gamma}{\Delta_1}{P_1}{\hby{\bactinp{n}{m_1}}}{\Delta_1'}{P_2}
		\label{prop:HOp_to_HO:full_abs7}
	\end{eqnarray}
	\noi and
	\begin{eqnarray}
		\horel{\mapt{\Gamma}^{1}}{\mapt{\Delta_1''}^{1}}{R}{\hby{\stau}}{\mapt{\Delta_1'}^{1}}{\pmapp{P_2}{1}{f}}
		\label{prop:HOp_to_HO:full_abs8}
	\end{eqnarray}
%
%	\noi With the last transition happening on a restricted session channel.
%	From \dk{Lemma~\ref{lem:tau_inert}} we can conclude that
%	\begin{eqnarray}
%		\horel{\mapt{\Gamma}^{1}}{\mapt{\Delta_1''}^{1}}{R}{\wb}{\mapt{\Delta_1'}^{1}}{\pmapp{P_2}{1}{f}}
%		\label{prop:HOp_to_HO:full_abs9}
%	\end{eqnarray}
%
	\noi From the transition~\eqref{prop:HOp_to_HO:full_abs7} and the definition of $\Re$ we imply
	\begin{eqnarray}
		&&\horel{\Gamma}{\Delta_2}{Q_1}{\Hby{\bactinp{n}{m_2}}}{\Delta_2'}{Q_2}
		\label{prop:HOp_to_HO:full_abs5}
		\\
		&&\horel{\Gamma}{\Delta_1'}{P_2}{\wb}{\Delta_2'}{Q_2}
		\label{prop:HOp_to_HO:full_abs6}
	\end{eqnarray}
	\noi From~\eqref{prop:HOp_to_HO:full_abs5} and \propref{prop:op_corr_HOp_to_HO} we get
	\[
		\horel{\mapt{\Gamma}^{1}}{\mapt{\Delta_2}^{1}}{\pmapp{Q_1}{1}{f}}{\Hby{\bactinp{n}{\abs{z}{\binp{z}{x} (\appl{x}{s})}}}}{\mapt{\Delta_2'}^{1}}{\pmapp{Q_2}{1}{f}}
	\]
	\noi Furthermore, from~\ref{prop:HOp_to_HO:full_abs6} and the definition of $\Re$ we get
	\[
		\horel{\mapt{\Gamma}^{1}}{\mapt{\Delta_1'}^{1}}{\pmapp{P_2}{1}{f}}{\ \Re\ }{\mapt{\Delta_2'}^{1}}{\pmapp{Q_2}{1}{f}}
	\]
	\noi If we consider result~\eqref{prop:HOp_to_HO:full_abs8} we get:
	\[
		\horel{\mapt{\Gamma}^{1}}{\mapt{\Delta_1''}^{1}}{R}{\hby{\stau}\ \Re\ }{\mapt{\Delta_2'}^{1}}{\pmapp{Q_2}{1}{f}}
	\]
	where following \lemref{lem:up_to_deterministic_transition} we show that $R$ is a bisimulation an up to $\Hby{\stau}$.
	\qed
\end{proof}

%%%%%%%%%%%%%%%%%%%%%%%%%%%%%%%%%%%%%%%%%%%%%%%%%
% HOp TO SESSP
%%%%%%%%%%%%%%%%%%%%%%%%%%%%%%%%%%%%%%%%%%%%%%%%%

\subsection{Properties for $\enco{\pmap{\cdot}{2}, \tmap{\cdot}{2}, \mapa{\cdot}^{2}}: \HOp \to \sessp$}
\label{app:enc:HOp_to_sessp}

We repeat the statement of \propref{prop:typepres_HOp_to_p},
as in Page \pageref{prop:typepres_HOp_to_p}:

\begin{proposition}[Type Preservation, \HOp into \sessp]\rm
	\label{app:prop:typepres_HOp_to_p}
	Let $P$ be a \HOp process. \\
	If $\Gamma; \emptyset; \Delta \proves P \hastype \Proc$ then 
	$\mapt{\Gamma}^{2}; \emptyset; \mapt{\Delta}^{2} \proves \map{P}^{2} \hastype \Proc$.
\end{proposition}

%\begin{proposition}[Type Preservation, Higher-Order into First-Order]
%Let $P$ be an  $\HO$ process. 
%If			$\Gamma; \emptyset; \Delta \proves P \hastype \Proc$ then 
%			$\mapt{\Gamma}^{2}; \emptyset; \mapt{\Delta}^{2} \proves \map{P}^{2} \hastype \Proc$. 
%\end{proposition}

\begin{proof}
	By induction on the inference $\Gamma; \emptyset; \Delta \proves P \hastype \Proc$.
%	By induction on the structure of \HO process $P$.  \jp{TO BE ADJUSTED!}
	\begin{enumerate}[1.]

	%%%% Output of (linear) channel
		\item	Case $P = \bbout{k}{\abs{x}{Q}}P$. Then we have two possibilities, depending on the typing for $\abs{x}Q$.
			The first case concerns a linear typing, and  
			we have the following typing in the source language:
			\[
				\tree{
					\Gamma; \emptyset; \Delta_1 \cat k:S  \proves  P \hastype \Proc
					\quad
					\tree{
						\Gamma ; \emptyset ; \Delta_2\cat x:S_1 \proves  Q \hastype \Proc
					}{
						\Gamma ; \emptyset ; \Delta_2 \proves  \abs{x}Q \hastype \lhot{S_1}
					}
				}{
					\Gamma; \emptyset; \Delta_1 \cat \Delta_2 \cat k:\btout{\lhot{S_1}}S \proves  \bbout{k}{\abs{x}{Q}} P \hastype \Proc
				}
			\]
			This way, by IH we have
			$$
			\tmap{\Gamma}{2}; \es ; \tmap{\Delta_2}{2}, x:\tmap{S_1}{2}
									\proves 
									\pmap{Q}{2} \hastype \Proc
			$$
			Let us write 
			 $U_1$ to stand for 
			$\chtype{\btinp{\tmap{S_1}{2}}\tinact}$.
			The corresponding typing in the target language is as follows: 
			\begin{eqnarray*}
				\tmap{\Gamma_1}{2} & = & \tmap{\Gamma}{2} \cup a:\chtype{\btinp{\tmap{S_1}{2}}\tinact} \\
				\tmap{\Gamma_2}{2} & = & \tmap{\Gamma_1}{2} \cup \varp{X}:\tmap{\Delta_2}{2}
			\end{eqnarray*}
			Also $(*)$ stands for $\tmap{\Gamma_1}{2}; \es ; \es \proves a \hastype U_1$; 
			$(**)$ stands for $\tmap{\Gamma_2}{2}; \es ; \es \proves a \hastype U_1$; and
			$(***)$ stands for $\tmap{\Gamma_2}{2}; \es ; \es \proves \varp{X} \hastype \Proc$.
			\begin{eqnarray}
				\label{prop:HO_to_sessp_t1}
				\tree{
					\tree{
						\tree{
						}{
							(***)
						} 
						\quad 
						\tree{
							\tree{
								\tree{
									\tree{
									}{
										\tmap{\Gamma_2}{2}; \es ; \tmap{\Delta_2}{2},  x:\tmap{S_1}{2}
										\proves 
										\pmap{Q}{2} \hastype \Proc
									}
								}{
									\tmap{\Gamma_2}{2}; \es ; \tmap{\Delta_2}{2}, y:\tinact, x:\tmap{S_1}{2}
									\proves 
									\pmap{Q}{2} \hastype \Proc
								}
							}{
								\tmap{\Gamma_2}{2}; \es ; \tmap{\Delta_2}{2}, y: \btinp{\tmap{S_1}{2}}\tinact
								\proves 
								\binp{y}{x}\pmap{Q}{2} \hastype \Proc
							} 
							\quad 
							\tree{
							}{
								(**)
							}
						}{
							\tmap{\Gamma_2}{2}; \es ; \tmap{\Delta_2}{2} 
							\proves 
							\binp{a}{y}\binp{y}{x}\pmap{Q}{2} \hastype \Proc
						} 
					}{
						\tmap{\Gamma_2}{2}; \es ; \tmap{\Delta_2}{2} 
						\proves 
						\binp{a}{y}\binp{y}{x}\pmap{Q}{2} \Par \varp{X} \hastype \Proc
					}
				}{
					\tmap{\Gamma_1}{2}; \es ; \tmap{\Delta_2}{2} 
					\proves 
					\recp{X}{(\binp{a}{y}\binp{y}{x}\pmap{Q}{2} \Par \varp{X})} \hastype \Proc
				}
			\end{eqnarray}
			\begin{eqnarray}
				\label{prop:HO_to_sessp_t2}
				\tree{
					\begin{array}{c}
						\tmap{\Gamma_1}{2}; \es ; \tmap{\Delta_1}{2}, k:\tmap{S}{2} 
						\proves 
						\pmap{P}{2}  \hastype \Proc
						\\
						\tmap{\Gamma_1}{2}; \es ; \tmap{\Delta_2}{2} 
						\proves 
						\recp{X}{(\binp{a}{y}\binp{y}{x}\pmap{Q}{2} \Par \varp{X})} \hastype \Proc
						\quad \eqref{prop:HO_to_sessp_t1}
					\end{array}
				}{
					\tmap{\Gamma_1}{2}; \es ; \tmap{\Delta_1, \Delta_2}{2}, k:\tmap{S}{2} 
					\proves 
					\pmap{P}{2} \Par 
					\recp{X}{(\binp{a}{y}\binp{y}{x}\pmap{Q}{2} \Par \varp{X})} \hastype \Proc
				}
			\end{eqnarray}
			\[
				\tree{
					\tree{
						\begin{array}{c}
							\tmap{\Gamma_1}{2}; \es ; \es \proves a \hastype U_1
							\\
							\tmap{\Gamma_1}{2}; \es ; \tmap{\Delta_1, \Delta_2}{2}, k:\tmap{S}{2} 
							\proves 
							\pmap{P}{2} \Par 
							\recp{X}{(\binp{a}{y}\binp{y}{x}\pmap{Q}{2} \Par \varp{X})} \hastype \Proc
							\quad \eqref{prop:HO_to_sessp_t2}
						\end{array}
					}{
						\tmap{\Gamma_1}{2}; \es ; \tmap{\Delta_1, \Delta_2}{2}, k:\bbtout{U_1}\tmap{S}{2} 
						\proves 
						\bout{k}{a}(\pmap{P}{2} \Par 
						\recp{X}{(\binp{a}{y}\binp{y}{x}\pmap{Q}{2} \Par \varp{X}))} \hastype \Proc
					}
				}{
					\tmap{\Gamma}{2}; \es ; \tmap{\Delta_1, \Delta_2}{2}, k:\bbtout{U_1}\tmap{S}{2} 
					\proves 
					\newsp{a}{\bout{k}{a}( 
					\pmap{P}{2} \Par 
					\recp{X}{(\binp{a}{y}\binp{y}{x}\pmap{Q}{2} \Par \varp{X}))}} \hastype \Proc
				}
			\]
			In the second case, $\abs{x}Q$ has a shared type. We have the following typing in the source language:
			\[
				\tree{
					\Gamma; \emptyset; \Delta \cat k:S  \proves  P \hastype \Proc
					\quad 
					\tree{
						\tree{
							\Gamma ; \emptyset ; \cat x:S_1 \proves  Q \hastype \Proc
						}{
							\Gamma ; \emptyset ; \es \proves  \abs{x}Q \hastype \lhot{S_1}
						}
					}{
						\Gamma ; \emptyset ; \es \proves  \abs{x}Q \hastype \shot{S_1}
					}
				}{
					\Gamma; \emptyset; \Delta  \cat k:\btout{\shot{S_1}}S \proves  \bbout{k}{\abs{x}{Q}} P \hastype \Proc
				}
			\]
			The corresponding typing in the target language can be derived similarly as in the first case.
	
		\item	Case $P = \binp{k}{x} P$. Then there are two cases, depending on the type of $X$. 
			In the first case,
			we have the following typing in the source language:
			\[
				\tree{
					\Gamma \cat x : \shot{S_1};\, \emptyset ;\, \Delta \cat k:S \proves  P \hastype \Proc
				}{
					\Gamma;\, \emptyset;\, \Delta\cat k:\btinp{\shot{S_1}}S \proves  \binp{k}{x} P \hastype \Proc
				}
			\]
			The corresponding typing in the target language is as follows:
			% --- we write $\Gamma_0$ to stand for $\Gamma \setminus \{X: \lhot{S_1}\}$.
%
			\[
				\tree{
					\tree{}{\tmap{\Gamma}{2} \cat x : \chtype{\btinp{\tmap{S_1}{2}}\tinact};\, \emptyset ;\, \Delta \cat k:\tmap{S}{2} \proves  \tmap{P}{2} \hastype \Proc}
				}{
					\tmap{\Gamma}{2};\, \emptyset; \, \tmap{\Delta}{2}\cat k:\bbtinp{\chtype{\btinp{\tmap{S_1}{2}}\tinact}}\tmap{S}{2} \proves
					\binp{k}{x} \pmap{P}{2} \hastype \Proc
				}
			\]
			In the second case,  
			we have the following typing in the source language:
			\[
				\tree{
					\Gamma;\, \{x : \lhot{S_1}\};\, \emptyset ;\, \Delta \cat k:S \proves  P \hastype \Proc
				}{
					\Gamma;\, \emptyset;\, \Delta\cat k:\btinp{\lhot{S_1}}S \proves  \binp{k}{x} P \hastype \Proc
				}
			\]
			The corresponding typing in the target language is as follows:
			% --- we write $\Gamma_0$ to stand for $\Gamma \setminus \{X: \lhot{S_1}\}$.
%
			\[
				\tree{
					\tmap{\Gamma}{2} \cat x : \chtype{\btinp{\tmap{S_1}{2}}\tinact};\, \emptyset ;\, \Delta \cat k:\tmap{S}{2} \proves  \tmap{P}{2} \hastype \Proc
				}{
					\tmap{\Gamma}{2};\, \emptyset;\, \tmap{\Delta}{2}\cat k:\bbtinp{\chtype{\btinp{\tmap{S_1}{2}}\tinact}}\tmap{S}{2} \proves
					\binp{k}{x} \pmap{P}{2} \hastype \Proc
				}
			\]
		\item	Case $P = \appl{x}{k}$. Also here we have two cases, depending on whether $X$ has linear or shared type.
			In the first case, $x$ is linear and
			we have the following typing in the source language:
			\[
				\tree{
					\Gamma ;\, \{x : \lhot{S_1}\};\,  \es \proves  X \hastype \lhot{S_1} \quad \Gamma; \es ; \{k:S_1\} \proves k \hastype S_1
				}{
					\Gamma;\, \{x : \lhot{S_1}\};\, k:S_1 \proves  \appl{x}{k} \hastype \Proc}
			\]
			Let us write
			$\tmap{\Gamma_1}{2}$ to stand for $\tmap{\Gamma}{2} \cat x:\chtype{\btout{\tmap{S_1}{2}}\tinact}$.
			The corresponding typing in the target language is as follows:
			\begin{eqnarray}
				\label{prop:HO_to_sessp_t11}
				\tree{
					\tree{
						\tmap{\Gamma_1}{2};\, \es;\,  \es \proves  \inact \hastype \Proc
					}{
						\tmap{\Gamma_1}{2};\, \es;\,  \dual{s}:\tinact \proves  \inact \hastype \Proc
					}
					\quad 
						\tmap{\Gamma_1}{2};\, \es;\, \{k:\tmap{S_1}{2}\} \proves  k \hastype \tmap{S_1}{2} 
				}{
					\tmap{\Gamma_1}{2};\, \es;\,\, k:\tmap{S_1}{2},\,  \dual{s}:\btout{\tmap{S_1}{2}}\tinact \proves  \bout{\dual{s}}{k}\inact \hastype \Proc
				}
			\end{eqnarray}
			\[
				\tree{
					\tree{
						\begin{array}{c}
							\tmap{\Gamma_1}{2};\, \es;\,\, k:\tmap{S_1}{2},\,  \dual{s}:\btout{\tmap{S_1}{2}}\tinact \proves
							\bout{\dual{s}}{k}\inact \hastype \Proc
							\quad \eqref{prop:HO_to_sessp_t11}
							\\
							\tmap{\Gamma_1}{2} ;\, \es ;\, \es \proves x \hastype \chtype{\btout{\tmap{S_1}{2}}\tinact}
						\end{array}
					}{
						\tmap{\Gamma_1}{2};\, \es;\, k:\tmap{S_1}{2}, s:\btinp{\tmap{S_1}{2}}\tinact , \dual{s}:\btout{\tmap{S_1}{2}}\tinact
						\proves
						\bout{x}{s}\bout{\dual{s}}{k}\inact \hastype \Proc
					}
				}{
					\tmap{\Gamma_1}{2};\, \es;\, k:\tmap{S_1}{2} \proves  \news{s}{(\bout{x}{s}\bout{\dual{s}}{k}\inact)} \hastype \Proc
				}
	\]
			In the second case, $x$ is shared, and
			we have the following typing in the source language:
			\[
				\tree{
					\Gamma \cat  x : \lhot{S_1} ;\,  \es ;\,  \es \proves  x \hastype \shot{S_1} \quad \Gamma; \es ; k:S_1 \proves k \hastype S_1
				}{
					\Gamma \cat x : \shot{S_1};\, \es ;\, k:S_1 \proves  \appl{x}{k} \hastype \Proc
				}
			\]
			The associated typing in the target language is obtained similarly as in the first case. \qed
	\end{enumerate}
\end{proof}

We repeat the statement of
\propref{prop:op_corr_HOp_to_p}, 
as in Page \pageref{prop:op_corr_HOp_to_p}:

\begin{proposition}[Operational Correspondence, \HOp into \sessp]\myrm
	\label{app:prop:op_corr_HOp_to_p}
	Let $P$ be an  $\HOp$ process such that  $\Gamma; \emptyset; \Delta \proves P \hastype \Proc$.
	
	\begin{enumerate}[1.]
		\item Suppose $\horel{\Gamma}{\Delta}{P}{\hby{\ell_1}}{\Delta'}{P'}$.
		Then we have:
		\begin{enumerate}[a)]
			\item
				If  $\ell_1 = \news{\tilde{m}}\bactout{n}{\abs{x}Q}$,
				then $\exists \Gamma', \Delta'', R$ where either:
				\begin{enumerate}[-]
					\item 
						$\tmap{\Gamma}{2};\, \tmap{\Delta}{2} \proves  \pmap{P}{2} 
						\hby{\mapa{\ell_1}^{2}}
						\Gamma' \cdot \tmap{\Gamma}{2};\, \tmap{\Delta'}{2} \proves \pmap{P'}{2} \Par \repl{} \binp{a}{y} \binp{y}{x} \pmap{Q}{2}$
					\item 
						$\tmap{\Gamma}{2};\, \tmap{\Delta}{2} \proves \pmap{P}{2} 
						\hby{\mapa{\ell_1}^{2}}
						\tmap{\Gamma}{2};\, \Delta'' \proves \pmap{P'}{2} \Par \binp{s}{y} \binp{y}{x} \pmap{Q}{2}$
				\end{enumerate}

			\item
				If   
				$\ell_1 = \bactinp{n}{\abs{y}Q}$
				then $\exists R$ where
				either
				\begin{enumerate}[-]
					\item 
						$\tmap{\Gamma}{2};\, \tmap{\Delta}{2} \proves \pmap{P}{2} 
						\hby{\mapa{\ell_1}^{2}}
						\Gamma';\, \tmap{\Delta''}{2} \proves  R$, for some $ \Gamma'$
						and \\ 
						$\horel{\tmap{\Gamma}{2}}{\tmap{\Delta'}{2}}{\pmap{P'}{2}}{\wb}{\tmap{\Delta''}{2}}{\newsp{a}{R \Par \repl{} \binp{a}{y} \binp{y}{x} \pmap{Q}{2}}}$
					\item 
						$\tmap{\Gamma}{2};\, \tmap{\Delta}{2} \proves \pmap{P}{2}
						\hby{\mapa{\ell_1}^{2}}
						\tmap{\Gamma}{2};\, \tmap{\Delta''}{2} \proves R$, 
						and \\ 
						$\horel{\tmap{\Gamma}{2}}{\tmap{\Delta'}{2}}{\pmap{P'}{2}}{\wb}{\tmap{\Delta''}{2}}{\newsp{s}{R \Par \binp{s}{y} \binp{y}{x} \pmap{Q}{2}}}$  		
				\end{enumerate}

			\item	If
				$\ell_1 = \tau$ then either:

				\begin{enumerate}[-]
					\item	$\exists R$ such that
						\[
						\mhorel{\tmap{\Gamma}{2}}{\tmap{\Delta}{2}}{\pmap{P}{2}}
						{\hby{\tau}}
						{\tmap{\Delta'}{2}}{}{\newsp{\tilde{m}}{\pmap{P_1}{2} \Par \newsp{a}
						{\pmap{P_2}{2}\subst{a}{x} \Par \repl{} \binp{a}{y} \binp{y}{x} \pmap{Q}{2}}}}
						\]

					\item	$\exists R$ such that
						\[
						\mhorel{\tmap{\Gamma}{2}}{\tmap{\Delta}{2}}{\pmap{P}{2}}
						{\hby{\tau}}
						{\tmap{\Delta'}{2}}{}{\newsp{\tilde{m}}{\pmap{P_1}{2} \Par \newsp{s}
						{\pmap{P_2}{2}\subst{\dual{s}}{x} \Par \binp{s}{y} \binp{y}{x} \pmap{Q}{2}}}}
						\]

					\item	%$\ell_1 = \btau$ and
						$\tmap{\Gamma}{2};\, \tmap{\Delta}{2} \proves \pmap{P}{2}
						\hby{\tau}
						\tmap{\Gamma}{2};\, \tmap{\Delta'}{2} \proves \pmap{P'}{2}$

					\item	$\ell_1 = \btau$ and
						$\tmap{\Gamma}{2};\, \tmap{\Delta}{2} \proves \pmap{P}{2}
						\hby{\stau}
						\tmap{\Gamma}{2};\, \tmap{\Delta'}{2} \proves \pmap{P'}{2}$
				\end{enumerate}

%			\item	 
%				If  
%				%$\stytra{\Gamma}{\ell_1}{\Delta}{P}{\Delta'}{P_1 \Par P_2\subst{\abs{x}Q}{X}}$
%				$\ell_1 = \tau$ and $P' 	\not \scong \news{\tilde{m}}(P_1 \Par P_2\subst{\abs{x}Q}{X})$
%				then \\
%				$\mapt{\Gamma}^{2};\, \mapt{\Delta}^{2} \proves  \map{P}^{2}
%				\hby{\tau}
%				\mapt{\Gamma}^{2};\, \mapt{\Delta'}^{2} \proves  \map{P'}^{2}$.
				   			   
%			   then  $\exists \ell_2$ s.t. 
%			    $\wtytra{\mapt{\Gamma}^{3}}{\ell_2}{\mapt{\Delta}^{3}}{\map{P}^{3}}{\mapt{\Delta'}^{3}}{\map{P'}^{3}}$
%			    and $\ell_2 = \mapa{\ell_1}^{3}$.

			\item	 
				If  
				$\ell_1 \in \set{\bactsel{n}{l}, \bactbra{n}{l}}$
				%\not\in \set{\tau,\, \news{\tilde{m}}\bactout{n}{\abs{x}Q}, \, \bactinp{n}{\abs{x}Q}}$ 
				 then \\
				$\exists \ell_2 = \mapa{\ell_1}^{2}$ such that 
				$\mapt{\Gamma}^{2};\, \mapt{\Delta}^{2} \proves  \map{P}^{2}
				\hby{\ell_2}
				\mapt{\Gamma}^{2};\, \mapt{\Delta'}^{2} \proves  \map{P'}^{2}$.			
		\end{enumerate}
		
		%%%%%%% SOUNDNESSS
		\item Suppose 
		$\stytra{\mapt{\Gamma}^{2}}{\ell_2}{\mapt{\Delta}^{2}}{\map{P}^{2}}{\mapt{\Delta'}^{2}}{R}$.
			\begin{enumerate}[a)]
				\item %% soutput
					%\footnote{$\mapt{\Gamma}^{2}$ in the following three items need adjustments.}
					If  
					$\ell_2 = \news{m}\bactout{n}{m}$
					%$\stytra{\mapt{\Gamma}^{2}}{\news{m}\bactout{n}{m}}{\mapt{\Delta}^{2}}{\map{P}^{2}}{\mapt{\Delta'}^{2}}{R}$
					then 
					either 
					\begin{enumerate}[-]
					\item	$\exists P'$ such that $P \hby{\news{m} \bactout{n}{m}} P'$
						and $R = \pmap{P'}{2}$.

					\item	$\exists Q, P'$ such that $P \hby{\bactout{n}{\abs{x}Q}} P'$
						and $R = \map{P'}^{2} \Par \repl{} \binp{a}{y} \binp{y}{x} \pmap{Q}{2}$

					\item	$\exists Q, P'$ such that $P \hby{\bactout{n}{\abs{x}Q}} P'$
						and $R = \map{P'}^{2} \Par \binp{s}{y} \binp{y}{x} \pmap{Q}{2}$
					\end{enumerate}

				\item   %% sinput
					If  $\ell_2 = \bactinp{n}{m}$ 
					%$\stytra{\mapt{\Gamma}^{2}}{\bactinp{n}{m}}{\mapt{\Delta}^{2}}{\map{P}^{2}}{\mapt{\Delta'}^{2}}{R}$
					then either
					\begin{enumerate}[-]
					\item	$\exists P'$ such that $P \hby{\bactinp{n}{m}} P'$
						and $R = \pmap{P'}{2}$.

					\item	$\exists Q, P'$ such that
						$P \hby{\bactinp{n}{\abs{x}Q}} P'$\\
						and $\horel{\mapt{\Gamma}^{2}}{\mapt{\Delta'}^{2}}{\map{P'}^{2}}{\wb}{\mapt{\Delta'}^{2}}{\news{a}(R \Par \repl{} \binp{a}{y} \binp{y}{x} \pmap{Q}{2})}$
					\item	$\exists Q, P'$ such that
						$P \hby{\bactinp{n}{\abs{x}Q}} P'$\\
						and $\horel{\mapt{\Gamma}^{2}}{\mapt{\Delta'}^{2}}{\map{P'}^{2}}{\wb}{\mapt{\Delta'}^{2}}{\news{s}(R \Par \binp{s}{y} \binp{y}{x} \pmap{Q}{2})}$  
					\end{enumerate}
		
				\item   
					If  %$\stytra{\mapt{\Gamma}^{2}}{\tau}{\mapt{\Delta}^{2}}{\map{P}^{2}}{\mapt{\Delta'}^{2}}{R}$
					$\ell_2 = \tau$ 
					then $\exists P'$ such that
					$P \hby{\tau} P'$
					and $\horel{\mapt{\Gamma}^{2}}{\mapt{\Delta'}^{2}}{\map{P'}^{2}}{\wb}{\mapt{\Delta'}^{2}}{R}$.
				\item	 
					If  
					$\ell_2 \not\in \set{\bactout{n}{m}, \bactsel{n}{l}, \bactbra{n}{l}}$ 
					 then 
					$\exists \ell_1$ such that 
					$\ell_1 = \mapa{\ell_2}^{2}$ and \\
					$ \Gamma ;\, \Delta  \proves   P
					\hby{\ell_1}
					\Gamma ;\, \Delta  \proves   P'$.
		\end{enumerate}
	\end{enumerate}
\end{proposition}

\begin{proof}
	\noi The proof is done by transition induction.
	We conside the two parts separately.

	\noi - Part 1

	\noi - Basic Step:
 
	\noi - Subcase: $P= \bout{n}{\abs{x}{Q}} P'$ 
	and also from \defref{def:enc:HOp_to_p}
	we have that\\
	$\pmap{P}{2} = \newsp{a}{\bout{n}{a} \pmap{P'}{2} \Par \repl{} \binp{a}{y} \binp{y}{x} \pmap{Q}{2}}$

	\noi Then
	\begin{eqnarray*}
		\Gamma; \es; \Delta \proves P &\hby{\bactout{n}{\abs{x}{Q}}} & \Delta' \proves P'\\
		\tmap{\Gamma}{2}; \es; \tmap{\Delta}{2} \proves \pmap{P}{2} &\hby{\news{a} \bactout{n}{a}}& \tmap{\Delta}{2} \proves \pmap{P'}{2} \Par \repl{} \binp{a}{y} \binp{y}{x} \pmap{Q}{2}
	\end{eqnarray*}
	\noi and from \defref{def:enc:HOp_to_p}
	\begin{eqnarray*}
		\mapa{\bactout{n}{\abs{x}{Q}}} = \news{a} \bactout{n}{a}
	\end{eqnarray*}
	\noi as required.

	\noi - Subcase: $P= \bout{n}{\abs{x}{Q}} P'$ 
	and also from \defref{def:enc:HOp_to_p}
	we have that\\
	$\pmap{P}{2} = \newsp{s}{\bout{n}{\dual{s}} \pmap{P'}{2} \Par \binp{s}{y} \binp{y}{x} \pmap{Q}{2}}$
	is similar as above. 

	\noi - Subcase $P = \binp{n}{x} P'$.

	\noi - From \defref{def:enc:HOp_to_p}
	we have that
	$\pmap{P}{2} = \binp{n}{x} \pmap{P'}{2}$

	\noi Then
	\begin{eqnarray*}
		\Gamma; \es; \Delta \proves P &\hby{\bactinp{n}{\abs{x}{Q}}}& \Delta' \proves P' \subst{\abs{x}{Q}}{x}\\
		\tmap{\Gamma}{2}; \es; \tmap{\Delta}{2} \proves \pmap{P}{2} &\by{\bactinp{n}{a}}& \tmap{\Delta''}{2} \proves R \subst{a}{x}
	\end{eqnarray*}
	\noi with
	\begin{eqnarray*}
		\mapa{\bactinp{n}{\abs{x}{Q}}}^{2} &=& \bactinp{n}{a}
	\end{eqnarray*}
	It remains to show that
	\begin{eqnarray*}
		\tmap{\Gamma}{2}; \es; \tmap{\Delta'}{2} \proves \pmap{P' \subst{\abs{x}{Q}}{x}}{2} \wb
		\tmap{\Delta''}{2} \proves \newsp{a}{R \subst{a}{x} \Par \repl{} \binp{a}{y} \binp{y}{x} \pmap{Q}{2}}
	\end{eqnarray*}
	\noi The proof is an induction on the syntax structure of $P'$.
	Suppose $P' = \appl{x}{m}$, then:
	\begin{eqnarray*}
		\pmap{\appl{x}{m} \subst{\abs{x}{Q}}{x}}{2} &=& \pmap{Q \subst{m}{x}}{2}\\
		\newsp{a}{R \subst{a}{x} \Par \repl{} \binp{a}{y} \binp{y}{x} \pmap{Q}{2}} &=& \newsp{a}{\newsp{s}{ \bout{x}{s} \bout{\dual{s}}{m} \inact} \subst{a}{x} \Par \repl{} \binp{a}{y} \binp{y}{x} \pmap{Q}{2}}
	\end{eqnarray*}
	\noi The second term can be deterministically reduced as:
	\begin{eqnarray*}
		\mhorel{\tmap{\Gamma}{2}}{\tmap{\Delta''}{2}}{\newsp{a}{\newsp{s}{ \bout{x}{s} \bout{\dual{s}}{m} \inact} \subst{a}{x} \Par \repl{} \binp{a}{y} \binp{y}{x} \pmap{Q}{2}}}
		{\hby{\tau} \hby{\stau}}
		{\tmap{\Delta''}{2}}{}{\newsp{a}{\pmap{Q \subst{m}{x}}{2} \Par \repl{} \binp{a}{y} \binp{y}{x} \pmap{Q}{2}}}
	\end{eqnarray*}
	\noi which is bisimilar with:
	\begin{eqnarray*}
		\pmap{Q \subst{m}{x}}{2}
	\end{eqnarray*}
	\noi because $a$ is fresh and cannot interact anymore.

	\noi An interesting inductive step case is parallel composition. Suppose $P' = P_1 \Par P_2$. We need to show that:
	\begin{eqnarray*}
		&& \tmap{\Gamma}{2}; \es; \tmap{\Delta'}{2} \proves \pmap{(P_1 \Par P_2) \subst{\abs{x}{Q}}{x}}{2} \wb
		\tmap{\Delta''}{2} \proves \newsp{a}{\pmap{P_1 \Par P_2}{2} \subst{a}{x} \Par \repl{} \binp{a}{y} \binp{y}{x} \pmap{Q}{2}}
	\end{eqnarray*}
	\noi We know that
	\begin{eqnarray*}
		\horel{\tmap{\Gamma}{2}}{\tmap{\Delta_1}{2}}{\pmap{P_1\subst{\abs{x}{Q}}{x}}{2}}{&\wb&}
		{\tmap{\Delta_1''}{2}}{\newsp{a}{\pmap{P_1}{2} \subst{a}{x} \Par \repl{} \binp{a}{y} \binp{y}{x} \pmap{Q}{2}}}\\
		\horel{\tmap{\Gamma}{2}}{\tmap{\Delta_2}{2}}{\pmap{P_2\subst{\abs{x}{Q}}{x}}{2}}{&\wb&}
		{\tmap{\Delta_1''}{2}}{\newsp{a}{\pmap{P_2}{2} \subst{a}{x} \Par \repl{} \binp{a}{y} \binp{y}{x} \pmap{Q}{2}}}
	\end{eqnarray*}
	\noi We conclude from the congruence of $\wb$.

	\noi - The rest of the cases for Part 1 are easy to follow using \defref{def:enc:HOp_to_p}.

	\noi - Part 2.

	\noi The proof for Part 2 is straightforward following \defref{def:enc:HOp_to_p}.
	We give some distinctive cases:

	\noi - Case $P = \bout{n}{\abs{x}{Q}} P'$
	\begin{eqnarray*}
		\horel{\Gamma}{\Delta}{P}{&\hby{\bactout{n}{\abs{x}{Q}}}&}{\Delta'}{P'}\\
		\horel{\tmap{\Gamma}{2}}{\tmap{\Delta}{2}}{\pmap{P}{2}}{& \hby{\news{a} \bactout{n}{a}}&}{\tmap{\Delta'}{2}}{\pmap{P'}{2} \Par \repl{} \binp{a}{y} \binp{y}{s} \pmap{Q}{2}}
	\end{eqnarray*}
	\noi as required.

	\noi - Case $P = \binp{n}{x} P'$
	\begin{eqnarray*}
		\horel{\Gamma}{\Delta}{P}{&\hby{\bactinp{n}{\abs{x}{Q}}}&}{\Delta'}{P' \subst{\abs{x}}{Q}}{x}\\
		\horel{\tmap{\Gamma}{2}}{\tmap{\Delta}{2}}{\pmap{P}{2}}{& \hby{\bactinp{n}{a}}&}{\tmap{\Delta''}{2}}{\pmap{P'}{2} \subst{a}{x}}
	\end{eqnarray*}
	\noi We now use a similar argumentation as the input case in Part 1 to prove that:
	\begin{eqnarray*}
		\horel{\Gamma}{\Delta'}{P' \subst{\abs{x}{Q}}{x}}
		{\wb}
		{\tmap{\Delta''}{2}}{\newsp{a}{\pmap{P'}{2} \subst{a}{x} \Par \repl{} \binp{a}{y} \binp{y}{x} \pmap{Q}{2}}}
	\end{eqnarray*}
	\qed
\end{proof}

%\begin{proposition}\rm
%	\label{app:enc_HO_to_sessp_oc}
%	Encoding $\encod{\cdot}{\cdot}{2}: \HO \to \sessp$ 
%	enjoys operational correspondence (cf. Def.~\ref{def:ep}\,(2)).
%\end{proposition}
%
%\begin{proof}[Sketch]
%For completeness, we 
%consider the \HO process $P = {\bbout{k}{\abs{x}{Q}} P_1} \Par \binp{k}{X} P_2$. We have that
%\[
%P \red P_1 \Par P_2 \subst{\abs{x}Q}{X}
%\]
%In the target language, this reduction is mimicked as follows:
%\begin{eqnarray*}
%\pmap{P}{2} & = & \newsp{a}{\bout{k}{a} (\pmap{P_1}{2} \Par \repl{} \binp{a}{y} \binp{y}{x} \pmap{Q}{2})\,} 
%                  \Par \binp{k}{x} \pmap{P_2}{2} \\
%            & \red & \newsp{a}{\pmap{P_1}{2} \Par \repl{} \binp{a}{y} \binp{y}{x} \pmap{Q}{2} 
%                  \Par  \pmap{P_2}{2}\subst{a}{x}}
%\end{eqnarray*}
%\qed
%\end{proof}
\subsection{Properties for $\enco{\pmap{\cdot}{3}, \tmap{\cdot}{3}, \mapa{\cdot}^{3}}: \HOpp \to \HOp$}
\label{app:HOpp_to_HOp}

We study the properties of the typed encoding in
\defref{def:enc:HOpp_to_HOp} (Page~\pageref{def:enc:HOpp_to_HOp}).

We repeat the statement of \propref{prop:typepres_HOpp_to_HOp},
as in Page~\pageref{prop:typepres_HOpp_to_HOp}:

\begin{proposition}[Type Preservation. From \HOpp to \HOp]\myrm
	\label{app:prop:typepres_HOpp_to_HOp}
	Let $P$ be a \HOpp process.
	If $\Gamma; \emptyset; \Delta \proves P \hastype \Proc$ then 
	$\tmap{\Gamma}{3}; \emptyset; \tmap{\Delta}{3} \proves \pmap{P}{3} \hastype \Proc$. 
\end{proposition}

\begin{proof}
	By induction on the inference of 
	$\Gamma; \emptyset; \Delta \proves P \hastype \Proc$.
	We detail some representative cases:
	\begin{enumerate}[1.]
		\item	Case $P = \bout{u}{\abs{\underline{x}}{Q}} P'$.
			Then we may have the following typing in \HOpp:
			\[
				\tree{
					\tree{}{\Gamma; \Lambda_1; \Delta_1 \cat u:S  \proves  P' \hastype \Proc} 
					\quad
					\tree{
						\tree{}{\Gamma \cat \underline{x}:L; \Lambda_2 ; \Delta_2 \proves  Q \hastype \Proc}
						\quad
						\tree{}{\Gamma \cat \underline{x}:L; \es ; \es \proves  \underline{x} \hastype L}
					}{
						\Gamma ; \Lambda_2 ; \Delta_2 \proves  \abs{\underline{x}:L} Q \hastype \lhot{L}
					}
				}{
					\Gamma; \Lambda_1 \cat \Lambda_2; \Delta_1 \cat \Delta_2 \cat  u: \btout{\lhot{L}} S \proves \bout{u}{\abs{\underline{x}}{Q}} P' \hastype \Proc
				}
			\]
			Thus, by IH we have:
			\begin{eqnarray}
				\tmap{\Gamma}{3}; \tmap{\Lambda_1}{3}; \tmap{\Delta_1}{3} \cat u:\tmap{S}{3} & \proves &  \pmap{P'}{3} \hastype \Proc
				\label{eq:hopppre1}
				\\
				\tmap{\Gamma}{3} \cat \underline{x}:\tmap{L}{3}; \tmap{\Lambda_2}{3} ; \tmap{\Delta_2}{3} & \proves & \pmap{Q}{3} \hastype \Proc
				\label{eq:hopppre2}
				\\
				\tmap{\Gamma}{3} \cat \underline{x}:\tmap{L}{3}; \es ; \es & \proves & \underline{x} \hastype \tmap{L}{3}
				\label{eq:hopppre3}
			\end{eqnarray}
			The corresponding typing in \HOp is as follows:
			\begin{eqnarray}
				\tree{
					\tree{
						\tree{}{\eqref{eq:hopppre2}}
					}{
						\tmap{\Gamma}{3} \cat x:\tmap{L}{3}; \tmap{\Lambda_2}{3};  \tmap{\Delta_2}{3} \cat z: \tinact \proves \pmap{Q}{3} \hastype \Proc
					}
					\qquad
					\tree{}{\eqref{eq:hopppre3}}
				}{
					\tmap{\Gamma}{3}; \tmap{\Lambda_2}{3}; \tmap{\Delta_2}{3} \cat z:\btinp{\tmap{L}{3}} \tinact \proves \binp{z}{\underline{x}} \pmap{Q}{3} \hastype \Proc
				}
				\label{eq:hopppre11}
			\end{eqnarray}
			{\small
			\[
				\tree{
					\tree{}{
						\eqref{eq:hopppre1}}
						\quad
						\tree{
							\eqref{eq:hopppre11}
							\qquad 
							\tree{}{\tmap{\Gamma}{3}; \es; z:\btinp{\tmap{L}{3}} \tinact \proves z \hastype \btinp{\tmap{L}{3}} \tinact}
					}{
						\tmap{\Gamma}{3}; \tmap{\Lambda_2}{3}; \tmap{\Delta_2}{3} \proves \abs{z}{\binp{z}{\underline{x}} \pmap{Q}{3}} \hastype \lhot{(\btinp{\tmap{L}{3}} \tinact)}
					}
				}{
					\tmap{\Gamma}{3}; \tmap{\Lambda_1}{3} \cat \tmap{\Lambda_2}{3}; \tmap{\Delta_1}{3} \cat \tmap{\Delta_2}{3} \cat u:\btout{\lhot{\btinp{\tmap{L}{3}} \tinact}}\tmap{S}{3} 
					\proves  \bout{u}{\abs{z}{\binp{z}{\underline{x}} \pmap{Q}{3}}} \pmap{P'}{3}
					\hastype \Proc
				}
			\]
			}

			\item Case $P =  \appl{(\abs{x} P)}{(\abs{y} Q)}$.
			We may have different possibilities for the types of each abstraction. 
			We consider only one of them, as the rest are similar:
			\[
			\tree{
			\tree{
			\tree{}{
			\Gamma \cat x:\shot{C}; \Lambda;  \Delta_1 \proves   P \hastype \Proc}
			}{
			\Gamma; \Lambda;  \Delta_1 \proves \abs{x} P \hastype \lhot{(\lhot{C})}
			} 
			\quad
			\tree{
			\tree{}{
			\Gamma; \es;  \Delta_2, y:C \proves  Q \hastype \Proc}
			}{
			\Gamma; \es;  \Delta_2 \proves \abs{y} Q \hastype \lhot{C}
			}
			}{
			\Gamma; \Lambda; \Delta_1 \cdot \Delta_2 \proves\appl{(\abs{x} P)}{(\abs{y} Q)} \hastype \Proc
			}
			\]

			Thus, by IH we have:
			\begin{eqnarray}
			\tmap{\Gamma}{3} \cat x:\tmap{\shot{C}}{3}; \tmap{\Lambda}{3}; \tmap{\Delta_1}{3}   & \proves &  \pmap{P}{3} \hastype \Proc
			\label{eq:hopppre4} \\
			\tmap{\Gamma}{3}  ; \es; \tmap{\Delta_1}{3}, y:\tmap{C}{3}   & \proves &  \pmap{Q}{3} \hastype \Proc
			\label{eq:hopppre5} 
			\end{eqnarray}

			The corresponding typing in \HOp is as follows --- recall that $\tmap{\lhot{C}}{3} = \lhot{\tmap{C}{3}}$.
			\begin{eqnarray}
				\tree{
					\tree{
						\tree{}{\eqref{eq:hopppre4}}
					}{
						\tmap{\Gamma}{3}\cat x: \tmap{\shot{C}}{3}; \tmap{\Lambda}{3}; \tmap{\Delta_1}{3} \cat s: \tinact \proves \pmap{P}{3} \hastype \Proc
					}
				}{
					\tmap{\Gamma}{3}; \tmap{\Lambda}{3}; \tmap{\Delta_1}{3} \cat  s:\btinp{\tmap{\lhot{C}}{3}}\tinact \proves \binp{s}{x}\pmap{P}{3} \hastype \Proc
				}
				\label{eq:hopppre12}
			\end{eqnarray}
			{\small
			\[
				\tree{
					\tree{
						\eqref{eq:hopppre12}
						\quad
						\tree{
							\tree{
								\tree{}{\eqref{eq:hopppre5}}
							}{
								\tree{
									\tmap{\Gamma}{3}; \es; \tmap{\Delta_2}{3} \cat y:\tmap{ C}{3} \proves \pmap{Q}{3} \hastype \Proc
								}{
									\tree{
										\tmap{\Gamma}{3}; \es; \tmap{\Delta_2}{3} \proves \abs{y}{\pmap{Q}{3}} \hastype \tmap{\lhot{C}}{3}
									}{
										\tmap{\Gamma}{3}; \es; \tmap{\Delta_2}{3} \cat \dual{s}: \tinact \proves \abs{y}{\pmap{Q}{3}} \hastype \tmap{\lhot{C}}{3}
									}
								}
							}
						}{
							\tmap{\Gamma}{3}; \es;   \tmap{\Delta_2}{3} \cat \dual{s}:\btout{\tmap{\lhot{C}}{3}}\tinact \proves \bout{\dual{s}}{\abs{y}{\pmap{Q}{3}}}\inact \hastype \Proc
						}
					}{
						\tmap{\Gamma}{3}; \tmap{\Lambda}{3}; \tmap{\Delta_1}{3} \cdot \tmap{\Delta_2}{3} \cat s:\btinp{\tmap{\lhot{C}}{3}}\tinact \cat \dual{s}:\btout{\tmap{\lhot{C}}{3}}\tinact
						\proves
						\binp{s}{x}\pmap{P}{3} \Par \bout{\dual{s}}{\abs{y}{\pmap{Q}{3}}}\inact \hastype \Proc
					}
				}{
					\tmap{\Gamma}{3}; \tmap{\Lambda}{3}; \tmap{\Delta_1}{3} \cdot \tmap{\Delta_2}{3} \proves \news{s}(\binp{s}{x}\pmap{P}{3} \Par \bout{\dual{s}}{\abs{y}{\pmap{Q}{3}}}\inact) \hastype \Proc
				}
			\]
			}

	\end{enumerate}
\qed
\end{proof}

We repeat the statement of \propref{prop:op_corr_HOpp_to_HOp},
as in Page~\pageref{prop:op_corr_HOpp_to_HOp}:

\begin{proposition}[Operational Correspondence. From \HOpp to \HOp]\myrm
	\label{app:prop:op_corr_HOpp_to_HOp}
	\begin{enumerate}
		\item	Let $\Gamma; \es; \Delta \proves P$.
			$\horel{\Gamma}{\Delta}{P}{\hby{\ell}}{\Delta'}{P'}$ implies
			\begin{enumerate}[a)]
				\item	If $\ell \in \set{\news{\tilde{m}} \bactout{n}{\abs{x}{Q}}, \bactinp{n}{\abs{x}{Q}}}$ then
%					$\exists l' $ such that
					$\horel{\tmap{\Gamma}{3}}{\tmap{\Delta}{3}}{\pmap{P}{3}}{\hby{\ell'}}
					{\tmap{\Delta'}{3}}{\pmap{P'}{3}}$ with $\mapa{\ell}^{3} = \ell'$.

%				\item	If $\ell = \bactinp{n}{\abs{x: C}{Q}}$ then
%					$\horel{\tmap{\Gamma}{3}}{\tmap{\Delta}{3}}{\pmap{P}{3}}{\hby{\bactinp{n}{\abs{x: C}{\pmap{Q}{3}}}}}
%					{\tmap{\Delta'}{3}}{\pmap{P'}{3}}$.
%
%				\item	If $\ell = \news{\tilde{m}} \bactout{n}{\abs{x: L}{Q}}$ then
%					$\horel{\tmap{\Gamma}{3}}{\tmap{\Delta}{3}}{\pmap{P}{3}}{\hby{\news{\tilde{m}} \bactout{n}{\abs{z}{\binp{z}{x} \pmap{Q}{3}}}}}
%					{\tmap{\Delta'}{3}}{\pmap{P'}{3}}$.
%
%				\item	If $\ell = \bactinp{n}{\abs{x: L}{Q}}$ then
%					$\horel{\tmap{\Gamma}{3}}{\tmap{\Delta}{3}}{\pmap{P}{3}}{\hby{\bactinp{n}{\abs{z}{\binp{z}{x} \pmap{Q}{3}}}}}
%					{\tmap{\Delta'}{3}}{\pmap{P'}{3}}$.

				\item	If $\ell \notin \set{\news{\tilde{m}} \bactout{n}{\abs{x}{Q}}, \bactinp{n}{\abs{x}{Q}}, \tau}$ then
					$\horel{\tmap{\Gamma}{3}}{\tmap{\Delta}{3}}{\pmap{P}{3}}{\hby{\ell}}
					{\tmap{\Delta'}{3}}{\pmap{P'}{3}}$.

				\item	If $\ell = \btau$ then
					$\horel{\tmap{\Gamma}{3}}{\tmap{\Delta}{3}}{\pmap{P}{3}}{\hby{\tau}}
					{\Delta''}{R}$ and
					${\tmap{\Gamma}{3}}{\tmap{\Delta'}{3}}{\pmap{P'}{3}}{\wb}{\Delta''}{R}$.

				\item	If $\ell = \tau$ and $\ell \not= \btau$ then %and $\hby{\ell}$ is not a \betatran then
					$\horel{\tmap{\Gamma}{3}}{\tmap{\Delta}{3}}{\pmap{P}{3}}{\hby{\tau}}
					{\tmap{\Delta'}{3}}{\pmap{P'}{3}}$.
			\end{enumerate}

		\item	Let $\Gamma; \es; \Delta \proves P$.
			$\horel{\tmap{\Gamma}{3}}{\tmap{\Delta}{3}}{\pmap{P}{3}}{\hby{\ell}}
			{\tmap{\Delta''}{3}}{Q}$ implies
			\begin{enumerate}[a)]
				\item	If $\ell \in \set{\news{\tilde{m}} \bactout{n}{\abs{x}{Q}}, \bactinp{n}{\abs{x}{Q}}, \tau}$
					then
					$\horel{\Gamma}{\Delta}{P}{\hby{\ell'}}{\Delta'}{P'}$
%					and $\horel{\tmap{\Gamma}{3}}{\tmap{\Delta''}{3}}{Q}{\hby{\hat{\ell}}}{\tmap{\Delta'}{3}}{\pmap{P'}{3}}$
					with $\mapa{\ell'}^{3} = \ell$ and $Q \scong \pmap{P'}{3}$.

				\item	If $\ell \notin \set{\news{\tilde{m}} \bactout{n}{\abs{x}{R}}, \bactinp{n}{\abs{x}{R}}, \tau}$
					then
					$\horel{\Gamma}{\Delta}{P}{\hby{\ell}}{\Delta'}{P'}$ and $Q \scong \pmap{P'}{3}$.
%					and $\horel{\tmap{\Gamma}{3}}{\tmap{\Delta''}{3}}{Q}{\hby{\hat{\ell}}}{\tmap{\Delta'}{3}}{\pmap{P'}{3}}$.

				\item	If $\ell = \tau$ then
					either
					$\horel{\Gamma}{\Delta}{\Delta}{\hby{\tau}}{\Delta'}{P'}$ with $Q \scong \pmap{P'}{3}$\\
					or
					$\horel{\Gamma}{\Delta}{\Delta}{\hby{\btau}}{\Delta'}{P'}$ and
					$\horel{\tmap{\Gamma}{3}}{\tmap{\Delta''}{3}}{Q}{\hby{\btau}}
					{\tmap{\Delta''}{3}}{\pmap{P'}{3}}$.
			\end{enumerate}
	\end{enumerate}
\end{proposition}

\begin{proof}
\begin{enumerate}
	\item The proof of Part 1 does a transition induction and
	considers the mapping as defined in \defref{def:enc:HOpp_to_HOp}.
	We give the most interesting cases.

	\begin{itemize}
		\item	Case: $P = \appl{(\abs{x}{Q_1})}{\abs{x}{Q_2}}$.

			$\horel{\Gamma}{\Delta}{\appl{(\abs{x}{Q_1})}{\abs{x}{Q_2}}}{\hby{\btau}}{\Delta}{Q_1 \subst{\abs{x}{Q_2}}{x}}$ implies
\[
			\horel{\tmap{\Gamma}{3}}{\tmap{\Delta}{3}}{\newsp{s}{\binp{s}{x} \pmap{Q_1}{3} \Par \bout{\dual{s}}{\abs{x}{\pmap{Q_2}{3}}} \inact}}{\hby{\stau}}
			{\tmap{\Delta'}{3}}{\pmap{Q_1}{3} \subst{\abs{x}{\pmap{Q_2}{3}}}{x}}
\]

		\item	Case: $P = \bout{n}{\abs{\underline{x}} Q} P$

			$\horel{\Gamma}{\Delta}{\bout{n}{\abs{\underline{x}} Q} P}{\hby{ \bactout{n}{\abs{x}{Q}}}}{\Delta}{P}$ implies

			$\horel{\tmap{\Gamma}{3}}{\tmap{\Delta}{3}}{\bout{n}{\abs{z} \binp{z}{x} \pmap{Q}{3}} \pmap{P}{3}}{\hby{ \bactout{n}{\abs{z}{\binp{z}{x} \pmap{Q}{3}} } }}{\Delta}{\pmap{P}{3}}$
		\item Other cases are similar. 
	\end{itemize}

	\item The proof of Part 2 also does a transition induction and
	considers the mapping as defined in \defref{def:enc:HOpp_to_HOp}.
	We give the most interesting cases.

	\begin{itemize}
		\item	Case: $P = \appl{(\abs{x}{Q_1})}{\abs{x}{Q_2}}$.
		\[
			\mhorel{\tmap{\Gamma}{3}}{\tmap{\Delta}{3}}{\newsp{s}{ \appl{(\abs{z}\binp{z}{x} \pmap{Q}{3})}{s}  \Par \bout{\dual{s}}{\abs{x}{Q_2}} \inact}}{\hby{\btau}}
			{\tmap{\Delta'}{3}}{}{\newsp{s}{\binp{s}{x} \pmap{Q}{3}  \Par \bout{\dual{s}}{\abs{x}{Q_2}} \inact}}
		\]
			\noi implies
			$\horel{\Gamma}{\Delta}{\appl{(\abs{x}{Q_1})}{\abs{x}{Q_2}}}{\hby{\btau}}{\Delta}{Q_1 \subst{\abs{x}{Q_2}}{x}}$ and
		\[
			\mhorel{\tmap{\Gamma}{3}}{\tmap{\Delta}{3}}{\newsp{s}{\binp{s}{x} \pmap{Q}{3}  \Par \bout{\dual{s}}{\abs{x}{Q_2}} \inact}}{\hby{\stau}}
			{\tmap{\Delta'}{3}}{}{\pmap{Q_1}{3} \subst{\abs{x}{\pmap{Q_2}{3}}}{x}}
		\]

\iftodo
		\item	Case: $P = \bout{n}{\abs{\underline{x}} Q} P$

			$\horel{\tmap{\Gamma}{3}}{\tmap{\Delta}{3}}{\bout{n}{\abs{z} \binp{z}{x} \pmap{Q}{3}} \pmap{P}{3}}{\hby{ \bactout{n}{\abs{z}{\binp{z}{x} \pmap{Q}{3}} } }}{\Delta}{\pmap{P}{3}}$ and

			$\horel{\Gamma}{\Delta}{\bout{n}{\abs{\underline{x}} Q} P}{\hby{ \bactout{n}{\abs{\underline{x}}{Q}}}}{\Delta}{P}$
%			\dk{this case is incomplete}
		\item Other cases are similar. 
	\end{itemize}
\end{enumerate}
\qed
\end{proof}

\subsection{Properties for $\enco{\pmap{\cdot}{4}, \tmap{\cdot}{4}, \mapa{\cdot}^{4}}: \pHOp \to \HOp$}
\label{app:pHOp_to_HOp}

We study the properties of the typed encoding in
\defref{def:enc:pHOp_to_HOp} (Page~\pageref{def:enc:pHOp_to_HOp}).

We repeat the statement of \propref{prop:typepres_pHOp_to_HOp}, as in Page~\pageref{prop:typepres_pHOp_to_HOp}:

\begin{proposition}[Type Preservation. From \pHOp to \HOp]\rm
	\label{app:prop:typepres_pHOp_to_HOp}
	Let $P$ be a \pHOp process.
	If $\Gamma; \emptyset; \Delta \proves P \hastype \Proc$ then 
	$\tmap{\Gamma}{4}; \emptyset; \tmap{\Delta}{4} \proves \pmap{P}{4} \hastype \Proc$. 
\end{proposition}

\begin{proof}
	By induction on the inference $\Gamma; \emptyset; \Delta \proves P \hastype \Proc$.
	We examine two representative cases, using biadic communications.

	\begin{enumerate}[1.]
		\item	Case $P = \bout{n}{V} P'$ and 
			$\Gamma; \emptyset; \Delta_1 \cat \Delta_2 \cat n:\btout{\lhot{(C_1,C_2)}} S \proves \bout{n}{V} P' \hastype \Proc$.
			Then either $V = y$ or $V = \abs{(x_1,x_2)}Q$, for some $Q$.
			The case $V = y$ is immediate; we give details for the case $V = \abs{(x_1,x_2)}Q$, for which we have the following typing:
			\[
				\tree{
					\tree{}{
						\Gamma; \emptyset; \Delta_1 \cat n:S \proves P' \hastype \Proc
					}
					\quad
					\tree{
						\Gamma; \emptyset; \Delta_2 \cat x_1: C_1 \cat x_2:C_2 \proves Q \hastype \Proc
					}{
						\Gamma; \emptyset; \Delta_2 \proves \abs{(x_1,x_2)}Q \hastype \lhot{(C_1,C_2)}
					}
				}{
					\Gamma; \emptyset; \Delta_1 \cat \Delta_2 \cat n:\btout{\lhot{(C_1,C_2)}} S \proves \bout{k}{\abs{(x_1,x_2)}Q} P \hastype \Proc
				}
		\]
		We now show the typing for $\pmap{P}{4}$.
		By IH we have both:
		\[
			\tmap{\Gamma}{4}; \emptyset; \tmap{\Delta_1}{4} \cat n: \tmap{S}{4} \proves \pmap{P'}{4} \hastype \Proc
			\qquad
			\tmap{\Gamma}{4}; \emptyset; \tmap{\Delta_2}{4} \cat x_1: \tmap{C_1}{4} \cat x_2:\tmap{C_2}{4} \proves \pmap{Q}{4} \hastype \Proc
		\]
		Let $L = \lhot{(C_1,C_2)}$. 
		By \defref{def:enc:pHOp_to_HOp} we have  
		$\tmap{L}{4} = \lhot{\big(\btinp{\tmap{C_1}{4}} \btinp{\tmap{C_2}{4}}\tinact\big)}$
		and
		$\pmap{P}{4} = \bbout{n}{\abs{z}\binp{z}{x_1}\binp{z}{x_2} \pmap{Q}{4}} \pmap{P'}{4}$.
		We can now infer the following typing derivation:
		\begin{eqnarray}
			\label{prop:tpres:pHOp_to_HOp1}
			\tree{
				\tree{
					\tree{
						\tree{
							\tree{}{
								\tmap{\Gamma}{4}; \emptyset; \tmap{\Delta_2}{4} \cat x_1: \tmap{C_1}{4} \cat x_2: \tmap{C_2}{4} \proves \pmap{Q}{4} \hastype \Proc
							}
						}{
							\tmap{\Gamma}{4}; \emptyset; \tmap{\Delta_2}{4} \cat x_1: \tmap{C_1}{4} \cat x_2: \tmap{C_2}{4} \cat z:\tinact \proves \pmap{Q}{4} \hastype \Proc
						}
					}{
						\tmap{\Gamma}{4}; \emptyset; \tmap{\Delta_2}{4} \cat x_1: \tmap{C_1}{4}\cat z:\btinp{\tmap{C_2}{4}}\tinact \proves \binp{z}{x_2} \pmap{Q}{4} \hastype \Proc
					}
				}{
					\tmap{\Gamma}{4}; \emptyset; \tmap{\Delta_2}{4} \cat z:\btinp{\tmap{C_1}{4}}\btinp{\tmap{C_2}{4}}\tinact \proves \binp{z}{x_1}\binp{z}{x_2} \pmap{Q}{4} \hastype \Proc
				}
			}{
				\tmap{\Gamma}{4}; \emptyset; \tmap{\Delta_2}{4} \proves \abs{z}\binp{z}{x_1}\binp{z}{x_2} \pmap{Q}{4} \hastype \lhot{(\tmap{C_1}{4},\tmap{C_2}{4})}
			}
		\end{eqnarray}
		\[
		\tree{
			\tree{}{
				\mapt{\Gamma}^{\mathsf{p}}; \emptyset; \mapt{\Delta_1}^{\mathsf{p}} \cat k:\mapt{S}^{\mathsf{p}} \proves \map{P'}^{\mathsf{p}} \hastype \Proc
			}
			\quad
			\eqref{prop:tpres:pHOp_to_HOp1}
		}{
			\tmap{\Gamma}{4}; \emptyset; \tmap{\Delta_1}{4} \cat \tmap{\Delta_2}{4} \cat n:\btout{\tmap{L}{4}} \tmap{S}{4} \proves \pmap{P}{4} \hastype \Proc
		}
		\]

		\item	Case $P = \binp{n}{x_1,x_2} P'$ 
			and
			$\Gamma; \emptyset; \Delta_1 \cat n: \btinp{(C_1, C_2)} S \proves \binp{n}{x_1,x_2} P' \hastype \Proc$.
			We have the following typing derivation:
			\[
				\tree{
					\Gamma; \emptyset; \Delta_1 \cat n:S \cat x_1: C_1 \cat x_2: C_2 \proves  P' \hastype \Proc
					\quad
					\Gamma; \emptyset;  \proves x_1, x_2 \hastype C_1,C_2
				}{
					\Gamma; \emptyset; \Delta_1 \cat n: \btinp{(C_1, C_2)} S \proves \binp{n}{x_1,x_2} P' \hastype \Proc
				}
		\]
		By \defref{def:enc:pHOp_to_HOp} we have 
		$\pmap{P}{4} = \binp{n}{x_1}\binp{k}{x_2} \pmap{P'}{4}$.
		By IH we have 
		\[
			\tmap{\Gamma}{4}; \emptyset; \tmap{\Delta_1}{4} \cat n:\tmap{S}{4} \cat x_1: \tmap{C_1}{4} \cat x_2: \tmap{C_2}{4} \proves  \pmap{P'}{4} \hastype \Proc
		\]
		and the following type derivation:
		\[
			\tree{
				\tree{
					\tree{}{
						\tmap{\Gamma}{4}; \emptyset; \tmap{\Delta_1}{4} \cat x_1:\tmap{C_1}{4} \cat x_2:\tmap{C_2}{4} \cat n:\tmap{S}{4} \proves \pmap{P'}{4} \hastype \Proc
					}
					%\quad
					%\tree{}{
					%\mapt{\Gamma}^{\mathsf{p}}; \emptyset; x_2:\tmap{C_2}{\mathsf{p}}  \proves  x_2 \hastype \tmap{C_2}{\mathsf{p}}}
				}{
					\tmap{\Gamma}{4}; \emptyset; \tmap{\Delta_1}{4} \cat x_1:\tmap{C_1}{4} \cat n:\btinp{\tmap{C_2}{4}}\tmap{S}{4} \proves \binp{n}{x_2}\pmap{P'}{4} \hastype \Proc
				}
				%\quad
				%\tree{}{
				%\mapt{\Gamma}^{\mathsf{p}}; \emptyset; x_1:\tmap{C_1}{\mathsf{p}}  \proves  x_1 \hastype \tmap{C_1}{\mathsf{p}}}
			}{
				\tmap{\Gamma}{4}; \emptyset; \tmap{\Delta_1}{4} \cat n:\btinp{\tmap{C_1}{4}}\btinp{\tmap{C_2}{4}}\tmap{S}{4} \proves \pmap{P}{4} \hastype \Proc
			}
		\]
	\end{enumerate}
	\qed
\end{proof}

We repeat the statement of \propref{prop:op_cor:pHOp_to_HOp}, as in Page~\pageref{prop:op_cor:pHOp_to_HOp}:

\begin{proposition}[Operational Correspondence. From \pHOp to \HOp]\myrm
	\label{app:prop:op_cor:pHOp_to_HOp}
	\begin{enumerate}
		\item	Let $\Gamma; \es; \Delta \proves P$. Then
			$\horel{\Gamma}{\Delta}{P}{\hby{\ell}}{\Delta'}{P'}$ implies
			\begin{enumerate}[a)]
				\item	If $\ell = \news{\tilde{m}'} \bactout{n}{\tilde{m}}$ then
					$\horel{\tmap{\Gamma}{4}}{\tmap{\Delta}{4}}{\pmap{P}{4}}{\hby{\ell_1} \dots \hby{\ell_n}}{\tmap{\Delta'}{4}}{\pmap{P}{4}}$
					with $\mapa{\ell}^{4} = \ell_1 \dots \ell_n$.

				\item	If $\ell = \bactinp{n}{\tilde{m}}$ then
					$\horel{\tmap{\Gamma}{4}}{\tmap{\Delta}{4}}{\pmap{P}{4}}{\hby{\ell_1} \dots \hby{\ell_n}}{\tmap{\Delta'}{4}}{\pmap{P}{4}}$
					with $\mapa{\ell}^{4} = \ell_1 \dots \ell_n$.

				\item	If $\ell \in \set{\news{\tilde{m}} \bactout{n}{\abs{\tilde{x}}{R}}, \bactinp{n}{\abs{\tilde{x}}{R}}}$ then
%					$\exists l' $ such that
					$\horel{\tmap{\Gamma}{4}}{\tmap{\Delta}{4}}{\pmap{P}{4}}{\hby{\ell'}}
					{\tmap{\Delta'}{4}}{\pmap{P'}{4}}$ with $\mapa{\ell}^{4} = \ell'$.

				\item	If $\ell \in \set{\bactsel{n}{l}, \bactbra{n}{l}}$ then
					$\horel{\tmap{\Gamma}{4}}{\tmap{\Delta}{4}}{\pmap{P}{4}}{\hby{\ell}}
					{\tmap{\Delta'}{4}}{\pmap{P'}{4}}$.

				\item	If $\ell = \btau$ then either
					$\horel{\tmap{\Gamma}{4}}{\tmap{\Delta}{4}}{\pmap{P}{4}}{\hby{\btau} \hby{\stau} \dots \hby{\stau}}
					{\tmap{\Delta'}{4}}{\pmap{P'}{4}}$ with $\mapa{\ell} = \btau, \stau \dots \stau$.

				\item	If $\ell = \tau$ then %and $\hby{\ell}$ is not a \betatran then
					$\horel{\tmap{\Gamma}{4}}{\tmap{\Delta}{4}}{\pmap{P}{4}}{\hby{\tau} \dots \hby{\tau}}
					{\tmap{\Delta'}{4}}{\pmap{P'}{4}}$ with $\mapa{\ell}^{4} = \tau \dots \tau$.
			\end{enumerate}

		\item	Let $\Gamma; \es; \Delta \proves P$.
			$\horel{\tmap{\Gamma}{4}}{\tmap{\Delta}{4}}{\pmap{P}{4}}{\hby{\ell_1}}
			{\tmap{\Delta_1}{4}}{P_1}$ implies
			\begin{enumerate}[a)]
				\item	If $\ell \in \set{\bactinp{n}{m}, \bactout{n}{m}, \news{m} \bactout{n}{m}}$ then
					$\horel{\Gamma}{\Delta}{P}{\hby{\ell}}{\Delta'}{P'}$ and\\
					$\horel{\tmap{\Gamma}{4}}{\tmap{\Delta_1}{4}}{P_1}{\hby{\ell_2} \dots \hby{\ell_n}}
					{\tmap{\Delta'}{4}}{\tmap{P'}{4}}$ with $\mapa{\ell}^{4} = \ell_1 \dots \ell_n$.

				\item	If $\ell \in \set{\news{\tilde{m}} \bactout{n}{\abs{x}{R}}, \bactinp{n}{\abs{x}{R}}}$
					then
					$\horel{\Gamma}{\Delta}{P}{\hby{\ell'}}{\Delta'}{P'}$
					with $\mapa{\ell'}^{4} = \ell$ and $P_1 \scong \pmap{P'}{4}$.

				\item	If $\ell \in \set{\bactsel{n}{l}, \bactbra{n}{l}}$
					then
					$\horel{\Gamma}{\Delta}{P}{\hby{\ell}}{\Delta'}{P'}$ and $P_1 \scong \pmap{P'}{4}$.
%					and $\horel{\tmap{\Gamma}{3}}{\tmap{\Delta''}{3}}{Q}{\hby{\hat{\ell}}}{\tmap{\Delta'}{3}}{\pmap{P'}{3}}$.

				\item	If $\ell = \btau$ then
					$\horel{\Gamma}{\Delta}{P}{\hby{\btau}}{\Delta'}{P'}$ and
					$\horel{\tmap{\Gamma}{4}}{\tmap{\Delta_1}{4}}{P_1}{\hby{\stau} \dots \hby{\stau}}
					{\tmap{\Delta'}{4}}{\tmap{P'}{4}}$ with $\mapa{\ell}^{4} = \btau, \stau \dots \stau$.

				\item	If $\ell = \tau$ then
					$\horel{\Gamma}{\Delta}{P}{\hby{\tau}}{\Delta'}{P'}$ and
					$\horel{\tmap{\Gamma}{4}}{\tmap{\Delta_1}{4}}{P_1}{\hby{\tau} \dots \hby{\tau}}
					{\tmap{\Delta'}{4}}{\tmap{P'}{4}}$ with $\mapa{\ell}^{4} = \tau \dots \tau$.
			\end{enumerate}
	\end{enumerate}
\end{proposition}

\begin{proof}
	The proof of both parts is by transition induction, following 
	the mapping defined in \defref{def:enc:HOpp_to_HOp}.
	We consider some representative cases, using biadic communication:
	\begin{enumerate}[$\bullet$]
	%\item 
	%% Biadic Output 
\item Case (1(a)), with $P =\bout{n}{m_1, m_2} P'$ and $\ell_1 = \bactout{n}{m_1, m_2}$. 
By assumption, $P$ is well-typed. 
As one particular possibility, we may have:
			\[
				\tree{
					\Gamma; \emptyset; \Delta_0 \cat n:S  \proves  P' \hastype \Proc \quad 
					\Gamma ; \emptyset ; m_1{:} S_1 \cat m_2{:}S_2 \proves  m_1,m_2 \hastype S_1,S_2}{
					\Gamma; \emptyset; \Delta_0 \cat m_1{:}S_1 \cat m_2{:}S_2 \cat n:\btout{S_1,S_2}S \proves  
					\bout{n}{m_1,m_2} P' \hastype \Proc}
			\]
for some $\Gamma, S, S_1, S_2, \Delta_0$, 
such that $\Delta = \Delta_0 \cat m_1{:}S_1 \cat m_2{:}S_2 \cat n:\btout{S_1,S_2}S$.
We may then have the following typed transition
$$
\stytra{\Gamma}{\ell_1}{\Delta_0 \cat m_1{:}S_1 \cat m_2{:}S_2 \cat n:\btout{S_1,S_2}S}{\bout{n}{m_1, m_2} P'}{\Delta_0 \cat n{:}S}{P'}
$$
The encoding of the source judgment for $P$ is as follows:
$$
\mapt{\Gamma}^{4}; \emptyset; \mapt{\Delta_0 \cat m_1{:}S_1 \cat m_2{:}S_2 \cat n:\btout{S_1,S_2}S}^{4} \proves \map{\bout{n}{m_1, m_2} P'}^{4} \hastype \Proc
$$
which, using \defref{def:enc:HOpp_to_HOp}, can be expressed as 
$$
\mapt{\Gamma}^{4}; \emptyset; \mapt{\Delta_0} 
\cat m_1{:}\mapt{S_1}^{4} \cat m_2{:}\mapt{S_2}^{4} 
\cat n:\btout{\mapt{S_1}^{4}}\btout{\mapt{S_2}^{4}}\mapt{S}^{4}
\proves 
\bout{n}{m_1}\bout{n}{m_2} \map{P'}^{4} 
\hastype \Proc
$$
Now, $\mapa{\ell_1}^{4} = \bactout{n}{m_1 }, \bactout{n}{ m_2}$. 
It is immediate to infer the following typed transitions for $\map{P}^{4}  = \bout{n}{m_1}\bout{n}{m_2} \map{P'}^{4} $:
\begin{eqnarray*}
& & \mapt{\Gamma}^{4}; 
\mapt{\Delta_0} \cat  m_1{:}\mapt{S_1}^{4} \cat m_2{:}\mapt{S_2}^{4} \cat
n:\btout{\mapt{S_1}^{4}}\btout{\mapt{S_2}^{4}}\mapt{S}^{4}
\proves 
\bout{n}{m_1}\bout{n}{m_2} \map{P'}^{4}  \\
& \hby{\bactout{n}{m_1}} & 
\mapt{\Gamma}^{4}; \mapt{\Delta_0} \cat  m_2{:}\mapt{S_2}^{4} \cat
n:\btout{\mapt{S_2}^{4}}\mapt{S}^{4}
\proves 
\bout{n}{m_2} \map{P'}^{4} \\
& \hby{\bactout{n}{m_2}} & 
\mapt{\Gamma}^{4}; \mapt{\Delta_0}  \cat n{:}\mapt{S}^{4}
\proves 
 \map{P'}^{4} \\
 & = & 
 \mapt{\Gamma}^{4}; \mapt{\Delta_0 \cat
n:S }^{4}
\proves 
 \map{P'}^{4}
\end{eqnarray*}
which concludes the proof for this case.

%% Biadic Abstraction Output 
\item Case (1(c)) with $P = \bbout{n}{\abs{(x_1, x_2)} Q} P' $ and $\ell_1 = \bactout{n}{\abs{(x_1, x_2)}{Q}}$. 
By assumption, $P$ is well-typed. 
We may have:
			\[
				\tree{
					\Gamma; \emptyset; \Delta_0 \cat n:S  \proves  P' \hastype \Proc \quad 
					\Gamma ; \emptyset ; \Delta_1 \proves  \abs{(x_1,x_2)}Q \hastype \lhot{(C_1,C_2)}}{
					\Gamma; \emptyset; \Delta_0 \cat \Delta_1 \cat n:\btout{\lhot{(C_1,C_2)}}S \proves  
					\bout{n}{\abs{(x_1,x_2)}Q} P' \hastype \Proc}
			\]
for some $\Gamma$, $S$, $C_1$, $C_2$, $\Delta_0$, $\Delta_1$, 
such that $\Delta = \Delta_0 \cat \Delta_1 \cat  n:\btout{\lhot{(C_1,C_2)}}S$.
(For simplicity, we consider only the case of a linear function.)
We may have the following typed transition:
$$
\stytra{\Gamma}{\ell_1}{\Delta_0 \cat \Delta_1 \cat n:\bbtout{\lhot{(C_1, C_2)}}S}{\bbout{n}{\abs{(x_1, x_2)} Q} P' }{\Delta_0 \cat n{:}S}{P'}
$$
The encoding of the source judgment is
$$
\mapt{\Gamma}^{4}; \emptyset; \mapt{\Delta_0 \cat \Delta_1 \cat n:\bbtout{\lhot{(C_1, C_2)}}S}^{4} \proves \map{\bbout{n}{\abs{(x_1, x_2)} Q} P' }^{4} \hastype \Proc
$$
which, using \defref{def:enc:HOpp_to_HOp}, can be equivalently expressed as 
$$
\mapt{\Gamma}^{4}; \emptyset; \mapt{\Delta_0 \cat \Delta_1} \cat
%n:\btout{\mapt{S_1}^{4}}\btout{\mapt{S_2}^{4}}\mapt{S}^{4}
n:\bbtout{
		\lhot{\big(\btinp{\tmap{C_1}{4}}\btinp{\tmap{C_2}{4}}\tinact\big)}}\mapt{S}^{4}
\proves 
\bbout{n}{\abs{z}\binp{z}{x_1} \binp{z}{x_2} \map{Q}^{4}} \map{P'}^{4}
\hastype \Proc
$$

Now, $\mapa{\ell_1}^{4} = \bactout{n}{\abs{z}\binp{z}{x_1}\binp{z}{x_2} \map{Q}^{4}}$. 
It is immediate to infer the following typed transition for $\map{P}^{4}  = \bbout{n}{\abs{z}\binp{z}{x_1} \binp{z}{x_2} \map{Q}^{4}} \map{P'}^{4}$:
\begin{eqnarray*}
& & \mapt{\Gamma}^{4}; \mapt{\Delta_0 \cat \Delta_1} \cat
%n:\btout{\mapt{S_1}^{4}}\btout{\mapt{S_2}^{4}}\mapt{S}^{4}
n:\bbtout{
		\lhot{\big(\btinp{\tmap{C_1}{4}}\btinp{\tmap{C_2}{4}}\tinact\big)}}\mapt{S}^{4}
\proves 
\bbout{n}{\abs{z}\binp{z}{x_1} \binp{z}{x_2} \map{Q}^{4}} \map{P'}^{4} \\
& \hby{\mapa{\ell_1}^{4}} & 
\mapt{\Gamma}^{4}; \mapt{\Delta_0} \cat
n:\mapt{S}^{4}, \,
\proves 
\map{P'}^{4} \\
 & = & 
 \mapt{\Gamma}^{4}; 
 \mapt{\Delta_0 \cat n:S}^{4}
\proves 
 \map{P'}^{4}
\end{eqnarray*}
which concludes the proof for this case.

%%%%%%%%%%%%%%%%%%%%%%%%%%%%%%% PART 2 %%%%%%%%%%%%%%%%%%%%%%%%%%%%%%%%%%%%%%%%%%%%

%% Biadic Input 
\item Case (2(a)), with $P =  \binp{n}{x_1, x_2} P' $, 
$\map{P}^{4} = 
		\binp{n}{x_1}  \binp{n}{x_2}  \map{P'}^{4}$.
%		We show that this case falls under part~(b) of the thesis (cf. Prop.~\ref{p:ocpotomo}). 		
%		and $\ell_2 = \bactinp{n}{m_1}, \bactinp{n}{m_2}$. Then w
		We have  the following typed transitions for $\map{P}^{4}$, for some $S$, $S_1$, $S_2$, and $\Delta$:
\begin{eqnarray*}
& & \mapt{\Gamma}^{4}; 
\mapt{\Delta}^{4} \cat 
n:\btinp{\tmap{S_1}{4}}\btinp{\tmap{S_2}{4}}\tmap{S}{4} \cat
\proves 
\binp{n}{x_1} \binp{n}{x_2}\map{P'}^{4} \\
& \hby{\bactinp{n}{m_1}} & 
\mapt{\Gamma}^{4}; 
\mapt{\Delta}^{4} \cat 
n:\btinp{\tmap{S_2}{4}}\tmap{S}{4} \cat
m_1:\mapt{S_1}^{4}
\proves 
\binp{n}{x_2}\map{P'}^{4} \subst{m_1}{x_1} \\
& \hby{\bactinp{n}{m_2}} & 
\mapt{\Gamma}^{4}; 
\mapt{\Delta}^{4} \cat n:\tmap{S}{4} \cat
m_1:  \mapt{S_1}^{4} \cat
m_2: \mapt{S_2}^{4}
\proves 
\map{P'}^{4} \subst{m_1}{x_1}\subst{m_2}{x_2} = Q
\end{eqnarray*}
Observe that the substitution lemma (Lemma~\ref{lem:subst}(1)) has been used twice.
%Considering Remarn~\ref{r:multilabels} 
It is then immediate to infer the label for the source transition:
$\ell_1 = \bactinp{n}{m_1,m_2}$. Indeed, $\mapa{\ell_1}^{4} = \bactinp{n}{m_1}, \bactinp{n}{m_2}$.
Now, in the source term $P$ we can infer the following transition:
$$
\stytra{\Gamma}{\ell_1}{\Delta \cat n:\btinp{S_1, S_2}S}{\binp{n}{x_1, x_2} P' }{\Delta\cat n{:}S \cat m_1:S_1 \cat m_2:S_2}{P'\subst{m_1,m_2}{x_1, x_2}}
$$
which concludes the proof for this case.

%% Biadic Abstraction Output 
\item Case (2(b)), with $P =  \bbout{n}{\abs{(x_1,x_2)} Q} P' $, 
$\map{P}^{4} = 
		\bbout{n}{\abs{z}\binp{z}{x_1}\binp{z}{x_2} \map{Q}^{4}} \map{P'}^{4}$.
		%We show that this case falls under part~(a) of the thesis (cf. Prop.~\ref{p:ocpotomo}). 
		We have the following  typed transition, for some $S$, $C_1$, $C_2$, and $\Delta$:
\begin{eqnarray*}
& & \mapt{\Gamma}^{4}; 
\mapt{\Delta}^{4}\cat n:\tmap{\bbtout{\lhot{(C_1,  C_2)}} S}{4}
\proves 
\bbout{n}{\abs{z}\binp{z}{x_1}\binp{z}{x_2} \map{Q}^{4}} \map{P'}^{4} \\
& \hby{\ell'_1} & 
\mapt{\Gamma}^{4}; 
\mapt{\Delta}^{4}\cat n:\tmap{ S}{4} 
\proves 
\map{P'}^{4} = Q
\end{eqnarray*}
where
$\ell'_1 = \bactout{n}{\abs{z}\binp{z}{x_1} \binp{z}{x_2} \map{Q}^{4}}$.
For simplicity, we consider only the case of linear functions.
It is then immediate to infer the label for the source transition:
$\ell_1 = \bactout{n}{\abs{(x_1,  x_2)}{Q}} $. 
Now, in the source term $P$ we can infer the following transition:
$$
\stytra{\Gamma}{\ell_1}{\Delta\cat n:\bbtout{\lhot{(C_1,  C_2)}} S}{ \bbout{n}{\abs{x_1,x_2} Q} P'}{\Delta\cat n{:}S}{P'}
$$
which concludes the proof for this case.

	\end{enumerate}
%\iftodo{
%	\dk{do some cases}
%}\else\fi
	\qed
\end{proof}

%\input{appendix/app-negative}

%\input{app-typingrule}
%\iflong 
%\input{app-definitions}
%\input{app-properties}
%\fi

\end{document}